%% file: main.tex
\documentclass[runningheads,a4paper]{llncs}

\def\BibTeX{{\rm B\kern-.05em{\sc i\kern-.025em b}\kern-.08emT\kern-.1667em\lower.7ex\hbox{E}\kern-.125emX}}
    
\usepackage{nicefrac}
\usepackage{siunitx}
\usepackage{array,framed}
\usepackage{booktabs}
\usepackage{
  color,
  float,
  epsfig,
  wrapfig,
  graphics,
  graphicx,
  subcaption
}
\usepackage{textcomp}
\usepackage{setspace}
 \usepackage{amsfonts}
\usepackage{latexsym,fancyhdr,url}
\usepackage{enumerate} 
\usepackage{multicol}                                
\usepackage[utf8]{inputenc}
\usepackage{algorithmicx, algorithm, algpseudocode}
\usepackage{caption} 
\usepackage{graphics}
\usepackage{xparse} 
\usepackage{xspace}
\usepackage{multirow}
\usepackage{csvsimple}
\usepackage{balance}
\usepackage{graphicx}
\usepackage{amsmath}
\usepackage{cuted}
\usepackage[bookmarksdepth=3,hypertexnames=false]{hyperref}
\usepackage{cleveref}
\usepackage{bookmark}
\makeatletter
\renewcommand{\part}[1]{%
  \clearpage%
  \refstepcounter{part}%
  \addcontentsline{toc}{part}{Part~\thepart: #1}%
  {\noindent\Large\bfseries Part~\thepart: #1\par}%
  \nopagebreak\vspace{1ex}%
}
\makeatother
\usepackage{datetime2}
\usepackage{lipsum}
\usepackage{placeins}
\usepackage{float}
\usepackage{paralist}
\usepackage{stmaryrd}
\usepackage{tabularx}
\usepackage{xifthen}
\usepackage[normalem]{ulem}
\newfloat{module}{htbp}{loa}
\floatname{module}{Module}
\crefname{module}{Module}{Modules}
\Crefname{module}{Module}{Modules}

\usepackage{etoolbox}
\usepackage[outerbars,color]{changebar}
\newif\ifcbactive
\BeforeBeginEnvironment{table}{\ifcbactive\cbend\fi}
\AfterEndEnvironment{table}{\ifcbactive\cbstart\fi}
\BeforeBeginEnvironment{module}{\ifcbactive\cbend\fi}
\AfterEndEnvironment{module}{\ifcbactive\cbstart\fi}
\BeforeBeginEnvironment{algorithm}{\ifcbactive\cbend\fi}
\AfterEndEnvironment{algorithm}{\ifcbactive\cbstart\fi}
\BeforeBeginEnvironment{figure}{\ifcbactive\cbend\fi}
\AfterEndEnvironment{figure}{\ifcbactive\cbstart\fi}

\input{defs}

\usepackage{
  tikz,
  pgfplots,
  pgfplotstable
}

\usetikzlibrary{
  shapes.geometric,
  arrows,
  external,
  pgfplots.groupplots,
  matrix
}

\pgfplotsset{compat=1.9}

\usepackage[a4paper,margin=2.5cm]{geometry}
\usepackage{mathtools,}

\DeclareMathAlphabet{\mathcal}{OMS}{cmsy}{m}{n}

\DeclareGraphicsExtensions{%
    .png,.PNG,%
    .pdf,.PDF,%
    .jpg,.mps,.jpeg,.jbig2,.jb2,.JPG,.JPEG,.JBIG2,.JB2}

\usepackage[many]{tcolorbox}
\tcbset{
  myalgorithm/.style={
    colback=white!95!white,
    colframe=black,
    sharp corners,
    boxrule=0.8pt,
    top=2mm,
    bottom=2mm,
    left=2mm,
    right=2mm,
    fonttitle=\bfseries,
    title=#1
  }
}

\begin{document}
\def\thetitle{\cadence: Extreme Pipelining with Multiple Concurrent Proposers}
\title{\thetitle}

\author{Kushal Babel \and Fatima Elsheimy \and Lioba Heimbach \and Mohammad Mussadiq Jalalzai \and Tobias Klenze \and \\Jovan Komatovic \and Jason Milionis \and Mike Setrin \and Victor Shoup \thanks{\textbf{Author order:} alphabetical by last name.}}

\institute{Category Labs}

\maketitle

\pagestyle{fancy}
\fancyhf{}
\fancyhead[R]{}
\fancyfoot[C]{\thepage}
\renewcommand{\headrulewidth}{0pt}
\thispagestyle{fancy}

\thispagestyle{fancy}
\bigskip

\begin{abstract}
We present \cadence, a Byzantine fault-tolerant multi-proposer consensus protocol with arbitrarily low block intervals, optimal resilience, and optimal fast-path latency.
\cadence divides time into equally spaced \emph{slots}, one block per slot, each finalized in its own consensus instance. Blocks do not build directly on their predecessor, which lets these instances run independently, so none waits for an earlier block to finish or even to propagate over the network; we call this \emph{extreme pipelining}, and it decouples the block interval from network latency.
\cadence also removes the single-leader monopoly over transaction inclusion and ordering: under \emph{multiple concurrent proposers} (MCP), several validators propose for each block, and it guarantees that, under synchrony, a transaction a correct proposer includes cannot be censored or deferred to a later block (\emph{short-term censorship resistance}), and that no proposer can craft its proposal in reaction to the others' (\emph{hiding}).

To realize extreme pipelining, we introduce a framework that turns any one-shot consensus meeting our slot-consensus specification into a multi-shot protocol with arbitrarily low block intervals. It is general and of independent interest. We instantiate it for the MCP setting with two protocols of our own: \chorus, an MCP slot consensus whose \fastpath finalizes a block in an optimal three communication rounds, with speculative finality one round earlier, and \conductor, an orchestrator that opens slots at an even cadence in normal operation, and more slowly under asynchrony to keep the number of open slots bounded.
To our knowledge, \cadence is the first MCP protocol to provide short-term censorship resistance and hiding at the fast-path latency of single-leader consensus.
We prove safety, liveness, short-term censorship resistance, and hiding under partial synchrony with optimal resilience ($n = 3f+1$).
Beyond the theory, we address the practical considerations of deploying \cadence and evaluate its latency in simulation: over Monad mainnet's $200$ globally distributed validators with five concurrent proposers per slot, the finalization latency averages $219$\,ms ($167$\,ms to speculative finality), and at a $100$\,ms block interval, a transaction waits on average only $50$\,ms to enter a proposal.
\end{abstract}

\input{introduction}


\input{p1_informal}

\FloatBarrier
\section{Deploying Cadence as Part of a Blockchain}
\label{sec:deploying}
\input{p1_deploying}

\section{Practical Considerations}
\label{sec:practical}
\input{p1_practical_considerations}

\section{Latency Evaluation}
\label{subsection:simulation}

\input{p1_latency_evaluation}


\input{related_work}

\section*{Acknowledgments}
We thank Andrei Constantinescu, Andrea Canidio, and Babak Poorebrahim Gilkalaye for helpful discussions and feedback.

\bibliographystyle{unsrt}
\bibliography{References.bib}

\newpage
\appendix
\crefname{subsection}{Appendix}{Appendices}
\Crefname{subsection}{Appendix}{Appendices}

\section*{Organization of the Appendix}
This appendix gives the formal treatment underlying the main text. It restates the MCP problem precisely (\S\ref{section:formal_problem_definition}), defines our extreme-pipelining framework and proves that it solves MCP given conforming building blocks (\S\ref{section:framework}), and then instantiates the two building blocks: \chorus (\S\ref{section:slot_agreement}) and \conductor (\S\ref{section:conductor-formal}).

\input{p2_problem_definition}

\input{p2_framework}

\input{p2_chorus}


\input{p2_conductor_proofs.tex}

\clearpage

\end{document}

%% file: defs.tex
\newcommand{\rulecmt}[1]{\Statex \textcolor{blue}{$\triangleright$~}\parbox[t]{\dimexpr\linewidth-1em\relax}{\hangindent=1.5em\hangafter=1 \textcolor{blue}{#1}}}
\newcommand{\cadence}{\textsc{Cadence}\xspace}    
\newcommand{\chorus}{\textsc{Chorus}\xspace}      
\newcommand{\conductor}{\textsc{Conductor}\xspace} 

\newcommand{\dtot}{\Delta_{\mathsf{tot}}}

\newcommand{\fastpath}{fast path\xspace}

\algdef{SE}[UPON]{Upon}{EndUpon}[1]{\textbf{upon} #1}{}

\newcommand{\statelabel}[1]{\State \hskip2em\parbox[t]{\dimexpr\linewidth-4em\relax}{\hangindent=1.5em\hangafter=1 #1}}
\newcommand{\substatelabel}[1]{\State \hskip4em\parbox[t]{\dimexpr\linewidth-6em\relax}{\hangindent=1.5em\hangafter=1 #1}}

\algtext*{EndFor}
\algtext*{EndIf}
\algtext*{EndWhile}
\algtext*{EndProcedure}
\algtext*{EndFunction}
\algtext*{EndUpon}

\newcommand{\payload}{\mathit{Payload}}

\newcommand{\roothash}{\textup{\texttt{merkle\_root}}\xspace}

\newcommand{\metablock}{meta-block\xspace}

\newcommand{\execblock}{execution block\xspace}
\newcommand{\keyitem}[1]{\par\smallskip\noindent\textbf{#1}}
\newcommand{\Encode}{\mathsf{Encode}}

\newcommand{\pid}{j}                     
\newcommand{\Entry}{\mathit{Entry}}      
\newcommand{\entries}{\mathsf{entries}}  
\newcommand{\CommitVote}{\textsc{CommitVote}}  
\newcommand{\CommitMsg}{\textsc{Commit}} 
\newcommand{\Ev}{\mathit{Ev}}            
\newcommand{\modcall}[1]{\mbox{\hyperref[#1]{Alg~\ref*{#1}}}}
\newcommand{\recoverProposals}{\ensuremath{\mathsf{recoverProposals}}}
\newcommand{\isDecoded}{\ensuremath{\mathsf{isDecoded}}}
\newcommand{\tryIngestChunk}{\ensuremath{\mathsf{tryIngestChunk}}}
\newcommand{\tryIngestShare}{\ensuremath{\mathsf{tryIngestShare}}}
\newcommand{\status}{\mathsf{status}}    
\newcommand{\Data}{\mathsf{data}}        
\newcommand{\Shares}{\mathsf{shares}}    
\newcommand{\mroot}{\rho}                
\newcommand{\cdata}{d}                   

\newcommand{\true}{\mathsf{true}}
\newcommand{\false}{\mathsf{false}}

\newcommand{\Slot}{\mathsf{Slot}}

\newcommand{\Time}{\mathsf{Time}}

\newcommand{\FastQC}{\mathsf{FastQC}}
\newcommand{\votetag}{\textsc{vote}}
\newcommand{\FallbackQC}{\mathsf{FallbackQC}}
\newcommand{\EquivCert}{\mathsf{EquivCert}}
\newcommand{\FBCert}{\mathsf{FBCert}}
\newcommand{\rhobot}{\rho_{\!\bot}}

\newcommand{\commitQC}{\mathsf{CommitQC}}

\newcommand{\TIBE}{\mathsf{TIBE}}

\newcommand{\Hash}{H}

\newcommand{\Sign}{\mathsf{Sign}}
\newcommand{\Verify}{\mathsf{Verify}}

\newcommand{\Enc}{\mathsf{Enc}}
\newcommand{\Dec}{\mathsf{Dec}}

\newcommand{\mpk}{\mathit{MPK}}
\newcommand{\dsk}{\mathit{DecryptShare}}

\newcommand{\Decode}{\mathsf{Decode}}

\newcommand{\MerkleRoot}{\mathsf{MerkleRoot}}
\newcommand{\MerkleProof}{\mathsf{MerkleProof}}
\newcommand{\VerifyMerkle}{\mathsf{VerifyMerkle}}

\newcommand{\slot}{s}              
\newcommand{\proc}{p}              
\newcommand{\deadline}{D}          
\newcommand{\blockint}{\tau}       
\newcommand{\numprop}{k}           

\algnewcommand{\BlueComment}[1]{\textcolor{blue}{\hfill\(\triangleright\) #1}}
\newcommand{\length}{\mathsf{length}}
\newcommand{\locallog}{\mathsf{log}}

\newcommand{\fproposers}{\mathsf{proposers}}
\newcommand{\fproposer}{\mathsf{proposer}}
\newcommand{\fnumber}{\mathsf{number}}
\newcommand{\fdeadline}{\mathsf{deadline}}
\newcommand{\fpayload}{\mathsf{payload}}
\newcommand{\fslot}{\mathsf{slot}}

\newcommand{\sourcemark}[1]{}

\newif\ifexplain
\explaintrue   
\explainfalse  

\ifexplain
    \newcommand{\explain}[2]{\textcolor{blue!70!black}{{\textsf{[#1-note:}} #2\textsf{]}}}
\else
    \newcommand{\explain}[2]{}
\fi

\newif\ifcomments
\commentstrue
\commentsfalse  

\newcommand{\kw}[1]{\ensuremath{\mathsf{#1}}}

\ifcomments
    \newcommand{\todo}[1]{\textcolor{red}{[ToDo: #1]}}
    
    \newcommand{\new}[1]{\textcolor{purple}{#1}}
    \newcommand{\old}[1]{\textcolor{cyan}{\sout{#1}}}
\else
    \newcommand{\todo}[1]{}
    
    \newcommand{\new}[1]{}
    \newcommand{\old}[1]{}
\fi

\ifcomments
    \newcommand{\kushal}[1]{\textcolor{orange}{[Kushal: #1]}}
    
    \newcommand{\victor}[1]{\textcolor{blue}{[Victor: #1]}}
    
    \newcommand{\mussadiq}[1]{\textcolor{purple}{[MJ: #1]}}
    
    \newcommand{\mike}[1]{\textcolor{olive}{[MS: #1]}}
    
    \newcommand{\tobias}[1]{\textcolor{red}{[TK: #1]}}
    \newcommand{\tobiaschange}[1]{\textcolor{red!30!orange}{#1}}
    \newcommand{\tobiasremove}[1]{\textcolor{red!30!orange}{\sout{#1}}}
    \newcommand{\sd}[1]{\textcolor{magenta}{[SD: #1]}}
    
    \newcommand{\jovan}[1]{\textcolor{teal}{[JK: #1]}}
    
    \newcommand{\lioba}[1]{\textcolor{pink}{[LH: #1]}}
    \newcommand{\fatima}[1]{\textcolor{violet}{[FA: #1]}}
    
    \newcommand{\mchen}[1]{\textcolor{blue!60!purple}{[MC: #1]}}

    \newcommand{\jason}[1]{\textcolor{cyan!70!blue!60!black}{[JM: #1]}}
    
\else
    \newcommand{\jason}[1]{}
    
    \newcommand{\kushal}[1]{}
    
    \newcommand{\victor}[1]{}
    
    \newcommand{\mussadiq}[1]{}
    
    \newcommand{\yd}[1]{}
    
    \newcommand{\tobias}[1]{}
    \newcommand{\tobiasremove}[1]{}
    \newcommand{\tobiaschange}[1]{{{#1}}}
    \newcommand{\sd}[1]{}
    
    \newcommand{\jovan}[1]{}
    
    \newcommand{\fatima}[1]{}
    
    \newcommand{\mchen}[1]{}
    
    \newcommand{\lioba}[1]{}
    \newcommand{\mike}[1]{}

\fi

\newcommand{\tobiaslater}[1]{\textcolor{red}{[TK: #1]}}
\renewcommand{\tobiaslater}[1]{}

\newcommand{\victorlater}[1]{\textcolor{blue}{[Victor: #1]}}
\renewcommand{\victorlater}[1]{}

\newif\ifformalpart
\formalparttrue   

%% file: introduction.tex
\section{Introduction} \label{section:introduction}

Blockchains increasingly host real-time financial markets, and serving them well places three requirements on the consensus layer.
We motivate each, then present our design.

\keyitem{Requirement 1: no proposer monopoly.}
Currently, most deployed protocols use a rotating leader schedule, where there is a uniquely identified entity (the leader) with the short-lived outsized ability to control the inclusion, exclusion and ordering of transactions.\footnote{Because existing blockchain protocols recognize this unique role of the leader, they typically rotate that role across blocks, thereby spreading power among validators over time; nonetheless, even short-lived control can enable rent extraction.}
Recent proposals have emerged to induce competition among the parties that contribute to a block, in what are called \emph{multiple concurrent proposer} (MCP) designs.
There, a proposer schedule rotates the role over time, much as a leader schedule does, but names a more diverse, small \emph{set} of validators per block who are allowed to \emph{simultaneously} make their proposals.
A combination arising from these different simultaneous proposals ends up forming the final block.
To this end, a number of properties form axes along which such designs can be evaluated:
\emph{Short-term censorship resistance}~\cite{mcp-why-and-how} guarantees that a correct proposer's proposal is included in the block it contributes to, not dropped or deferred to a later block; so a transaction that reaches a correct proposer is included without delay, as long as the proposer has space for it.
\emph{Hiding}~\cite{mcp-why-and-how} keeps each proposal concealed from the other proposers until it is too late for them to react with a proposal of their own.
In this way, a proposer cannot condition their proposal's contents on the other ones, and proposals are truly simultaneous.\footnote{DAG-based consensus protocols also have multiple proposers, but these designs provide different properties since there can exist an advantage in conditioning one's proposal on the others'.}

\keyitem{Requirement 2: low latency.}
On-chain applications such as trading demand low latency. A transaction passes through several stages, each adding latency: the time to reach a proposer, the wait to enter a proposal (on average half the block interval), the propagation of that proposal across the network, the consensus voting rounds needed to reach finality (whether speculative or full), and finally the execution that produces the result the application observes.

\keyitem{Requirement 3: short economic ticks.}
An economic tick is the shortest time to inclusion of a transaction in \emph{any} proposal.
Since transactions might provide---among others---price updates, finer granularity in transaction inclusion enables more frequent state updating; for example, it might allow reducing arbitrage profits in decentralized exchanges.~\cite{cexdex}
\tobias{An economic tick is the interval between consecutive updates to the on-chain state, such as price refreshes or liquidations. Because such an update takes effect only once the transaction carrying it is included, more frequent inclusion shortens the tick, letting on-chain applications track external conditions more closely, for example reducing arbitrage profits in decentralized exchanges~\cite{cexdex}.}

\medskip
Most recent MCP designs~\cite{mcp-why-and-how,Constellation} run a separate proposal-and-aggregation phase before the consensus: each proposal is first attested by a quorum, and these attestations limit the discretion of a single leader who assembles the block for an off-the-shelf consensus. This added phase costs two extra communication rounds, sacrificing the desideratum of low latency. Our proposal aims to avoid that cost by incorporating the multiple proposers directly into the consensus.

\paragraph{Our protocol.}
We present \cadence, an MCP Byzantine fault-tolerant protocol that runs among $n = 3f + 1$ validators, up to $f$ of which may be Byzantine, under partial synchrony.
Using synchronized clocks, \cadence divides time into \emph{slots}, each contributing exactly one block to the ledger
 in order of their slots,
and combines two components, illustrated in \Cref{fig:cadence-architecture}: \chorus, a single-shot MCP consensus run for each slot, and \conductor, which schedules the slots, normally at a regular interval~$\tau$ (the block interval).
Within a slot, multiple proposers contribute concurrently, as in any MCP protocol; \cadence additionally produces blocks concurrently across slots, each decided by its own independent single-shot consensus with no chaining between them, so the slots overlap. 
We call the latter \emph{extreme pipelining}.
The protocol also provides short-term censorship resistance and hiding, the latter by threshold-encrypting each proposal and releasing the decryption shares at the deadline.
On the fast path, it finalizes a slot in an optimal three communication rounds, even when some proposers are offline. A slot reaches speculative finality one round earlier, revertible only under provable equivocation.
We target blockchains with asynchronous execution~\cite{yin-sep-exec}, where validators commit to digests of proposals, each a list of transactions; invalid proposals and transactions are discarded deterministically afterward, when the committed data is turned into an \execblock, the ordered list of transactions to be executed.

Together, these features let \cadence meet all three requirements. Its multiple concurrent proposers remove the single-proposer monopoly, and short-term censorship resistance and hiding keep any one proposer from suppressing a transaction or reacting to the others' proposals (Requirement~1). For low latency (Requirement~2), \cadence shortens several of the stages a transaction passes through: having many proposers means one is likely near each user, shortening the hop to a proposer; its arbitrarily low block intervals shorten the wait to enter a proposal; and its fast path reaches (speculative) finality with optimal latency. Execution results are certified separately from the main protocol, as soon as they are available (\Cref{sec:deploy-certifying}). Those same low block intervals keep economic ticks short, since transactions enter proposals frequently (Requirement~3).

\paragraph{Key ideas.}
Five ideas underlie our design.

\keyitem{Idea 1: extreme pipelining.}
We make every slot an independent single-shot consensus instance, scheduled by \conductor using synchronized clocks.
Slots do not chain\footnote{We offer chain certificates as a practical add-on (\Cref{sec:chain-certification}), but they are not required for consensus.}: no certificate or artifact from an earlier slot is needed for a slot to start, and every opened slot finalizes \emph{exactly one} block.
The usual way to shrink the block interval is pipelining, overlapping the consensus of consecutive blocks; but traditional pipelining chains each block to its predecessor, so the next cannot begin until the previous proposal has propagated, leaving the interval floored at the network delay~$\Delta$. Removing this dependency on the previous block lets the block interval drop below $\Delta$, and because block production normally does not wait on any slot to finish, a slow or stalled slot does not hold up the rest. Traditional pipelining complicates protocols, in particular with multiple proposers.\footnote{If each proposal builds on an uncommitted parent, blocks that merge several proposals may build on different parents, so the merged block has no unique parent for the next block to chain to.} Extreme pipelining, however, makes the per-slot consensus \emph{simpler}: each \chorus instance is single-shot rather than multi-shot, reasoning about only one block at a time.

\keyitem{Idea 2: deadlines rather than start times.}
Waiting for all proposers would let the slowest one set the pace. In most major proof-of-stake blockchains, stake concentrates in a few mutually well-connected regions, with a long tail of distant validators; with many concurrent proposers, at least one is likely to be drawn from that tail.
In \cadence, rather than coordinating on when proposers start, validators rely on a synchronized clock to fix, for each slot, a common \emph{deadline} by which its proposals must have arrived: a fixed point in time, the same for every validator and, under steady operation, known in advance.
Each proposer chooses its own broadcast time to meet this deadline, according to its own network conditions: later when the network is fast, earlier when it detects delays.
A shared start time would require a common timeout that accommodates the slowest proposer. A shared deadline instead fixes when proposals must arrive, so that the slowest proposer no longer sets the pace for proposal dissemination.

\keyitem{Idea 3: a three-round fast path that tolerates offline proposers.}
The simplest fast path would finalize a slot only when all of its proposers are online and correct. But in that case, a single missing proposer would make the fast path unavailable.
Instead, in our protocol, validators judge each proposer separately at the deadline, voting yes or no on whether its proposal arrived.
Once a quorum forms for each proposer, one way or the other, validators cast commit votes, and a quorum of commit votes finalizes the slot, two rounds after the deadline. This is the good case, where every proposer is either online and correct, or offline.
Neither voting round uses a timeout, so a slot finalizes as fast as its votes propagate.
The fast path is unavailable if some proposal reaches only some correct validators but not others (\emph{partial dissemination}) or if a proposer sends conflicting proposals (\emph{equivocation}). In that case, \chorus finalizes through a slower \emph{fallback path}, which relies on an off-the-shelf agreement protocol.

\keyitem{Idea 4: agreement on digests, proposal dissemination in parallel.}
Following DispersedSimplex~\cite{DispersedSimplex}, 
 validators agree on Merkle roots of encodings of the proposals, not on the proposals themselves.
Traditionally, consensus disseminates the data before voting on it, which adds latency, in particular for large payloads.
Because asynchronous execution lets validators vote on digests rather than on the payload, dissemination can instead run in parallel with consensus. A committed Merkle root comes with an explicit guarantee: its data can be recovered, or, if invalid, deterministically rejected by all honest validators. 

\keyitem{Idea 5: bounding open slots under asynchrony.}
\conductor opens slots at a regular cadence, one every $\tau$, without waiting to see how the open ones fare; since each slot runs independently, many can be open at once. Under asynchrony, slot production must be throttled when finalization stalls, or the open slots would accumulate without bound; \conductor therefore groups slots into fixed-size \emph{windows} and opens a new window only once enough slots of the earlier windows have finalized, keeping the number of open slots bounded (\emph{boundedness}). Once the network stabilizes, it returns to a steady cadence, with consecutive deadlines spaced $\tau$ apart and known at least $\Delta$ in advance (\emph{recovery}).
\victorlater{The two properties are {\em bounded number of open slots}
and {\em bounded recovery time}. Both properties are ``boundedness''
properties, so called one ``boundedness'' and on ``recovery''
doesn't makse sense, at least to me.}

\input{figures/cadence-architecture-tikz}

\paragraph{A generic framework.}
Underlying \chorus and \conductor is our generic \emph{extreme-pipelining framework} that composes any \emph{slot consensus}, responsible for everything within a slot, with any \emph{orchestrator}, responsible for scheduling the slots. The design is modular: either component can be swapped out without affecting the other. For example, instantiating the orchestrator with one that sets the block interval $\tau$ adaptively does not affect the slot consensus. Conversely, instantiating the slot consensus with a $5f{+}1$ protocol reduces the fast-path latency at the cost of resilience, while leaving the slot scheduling unchanged. The framework is not specific to MCP, though our correctness proof is specific to it: we show that any conforming pair solves the MCP problem\ifformalpart~(\Cref{section:framework})\fi.

\paragraph{In practice.}
We intend \cadence as a deployable protocol, not only a theoretical one: we show how to deploy it as part of a blockchain (\Cref{sec:deploying}) and work through several further practical considerations (\Cref{sec:practical}).
We also evaluate its network latency in simulation, using estimated delays between Monad mainnet's $200$ globally distributed validators, with five proposers per slot: at a block interval of $100$\,ms, the end-to-end latency, the sum of the inclusion and finalization latencies, averages about $269$\,ms, or $217$\,ms to speculative finality (\Cref{subsection:simulation}).

\paragraph{Contributions.}
Our contributions are:
(i) the extreme-pipelining framework, comprising the slot-consensus and orchestrator abstractions and a proof that any pair of conforming instantiations solves the MCP problem;
(ii) \cadence, our instantiation of the framework, combining \chorus, an MCP slot consensus with a three-round fast path at optimal resilience, with \conductor, an orchestrator that schedules the slots at a regular cadence in normal operation and slows block production under asynchrony; and
(iii) the practical considerations for deploying \cadence as a real blockchain rather than a purely theoretical protocol.

\tobiaslater{
\tobias{This paragraph needs some fine-tuning after merging the contributions.} Two of these results are, to the best of our knowledge, the first of their kind.
First, \cadence is the first protocol to combine multiple concurrent proposers with \emph{extreme pipelining}: block intervals driven below the network delay $\Delta$ by running each slot as an independent single-shot consensus.
Second, \cadence is the first MCP protocol that finalizes in the optimal three communication rounds, the good-case latency of plain single-leader consensus, while providing short-term censorship resistance and hiding.
}

\paragraph{Paper organization.}
The rest of this paper is organized as follows.
We present the problem and our design informally: the MCP problem (\Cref{subsection:mcp-overview}), our extreme-pipelining framework (\Cref{subsection:cadence-overview}) and its components in \cadence: \chorus (\Cref{section:chorus-overview}) and \conductor (\Cref{section:conductor-overview}).
We then discuss deploying \cadence as part of a blockchain (\Cref{sec:deploying}) and further practical considerations (\Cref{sec:practical}), evaluate its latency (\Cref{subsection:simulation}), and review related work (\Cref{sec:related}).
The appendix gives the formal treatment: the formal problem definition (\Cref{section:formal_problem_definition}), the \cadence framework and its correctness (\Cref{section:framework}), \chorus (\Cref{section:slot_agreement}), and \conductor (\Cref{section:conductor-formal}).

%% file: figures/cadence-architecture-tikz.tex
\begin{figure}[t]
\centering
\begingroup
\usetikzlibrary{arrows.meta,decorations.pathreplacing}
\def\cadarchTauX{1.35}
\def\cadarchPropT{1.08}
\def\cadarchFastT{1.76}
\def\cadarchFallT{4.32}
\def\cadarchLaneH{0.78}
\def\cadarchLaneM{0.39}

\tikzset{
  >=Stealth,
  cadarchProp/.style={draw=orange!65!black, fill=orange!15, line width=0.6pt, rounded corners=1pt},
  cadarchPartial/.style={draw=orange!65!black, fill=orange!8, line width=0.6pt,
    densely dotted, rounded corners=1pt},
  cadarchFast/.style={draw=green!55!black, fill=green!10, line width=0.8pt, rounded corners=2pt},
  cadarchFallback/.style={draw=orange!55!black, fill=orange!5, line width=0.7pt,
    dashed, rounded corners=2pt},
  cadarchDeadline/.style={red!55!black, line width=1.05pt, line cap=round},
  cadarchGrid/.style={gray!18, line width=0.6pt},
  cadarchMeasure/.style={<->, gray!60!black, line width=0.6pt},
  cadarchEmit/.style={->, blue!25!gray, line width=0.7pt},
  cadarchBlock/.style={draw=blue!25!gray, fill=blue!4, line width=0.75pt,
    rounded corners=1pt, minimum width=0.70cm, minimum height=0.34cm,
    inner sep=1pt, font=\scriptsize\bfseries, text=blue!40!black},
  cadarchNote/.style={font=\footnotesize, text=gray!50!black},
  cadarchSlot/.style={font=\footnotesize\bfseries, text=gray!45!black, anchor=east},
}

\newcommand{\cadarchDeadlineFlag}[2]{%
  \draw[cadarchDeadline] (#1,#2+0.08) -- (#1,#2-\cadarchLaneH-0.06);
  \fill[red!55!black] (#1,#2+0.08) -- (#1+0.17,#2+0.02) --
    (#1,#2-0.04) -- cycle;
}

\newcommand{\cadarchProposalBars}[2]{%
  \draw[cadarchProp] (#1+0.10,#2-0.18) rectangle (#1+\cadarchPropT-0.05,#2-0.06);
  \draw[cadarchProp] (#1+0.18,#2-0.39) rectangle (#1+\cadarchPropT-0.05,#2-0.27);
  \draw[cadarchProp] (#1+0.14,#2-0.60) rectangle (#1+\cadarchPropT-0.05,#2-0.48);
}

\newcommand{\cadarchFastSlot}[4]{%
  \node[cadarchSlot] at (#1-0.18,#2-\cadarchLaneM) {#3};
  \cadarchProposalBars{#1}{#2}
  \cadarchDeadlineFlag{#1+\cadarchPropT}{#2}
  \draw[cadarchFast] (#1+\cadarchPropT+0.12,#2-\cadarchLaneH+0.10)
    rectangle (#1+\cadarchPropT+\cadarchFastT-0.08,#2-0.10);
  \draw[cadarchEmit] (#1+\cadarchPropT+\cadarchFastT-0.03,#2-\cadarchLaneM)
    -- (#1+\cadarchPropT+\cadarchFastT+0.28,#2-\cadarchLaneM);
  \node[cadarchBlock] at (#1+\cadarchPropT+\cadarchFastT+0.62,#2-\cadarchLaneM) {#4};
}

\newcommand{\cadarchFallbackSlot}[4]{%
  \node[cadarchSlot] at (#1-0.18,#2-\cadarchLaneM) {#3};
  \draw[cadarchProp] (#1+0.10,#2-0.18) rectangle (#1+\cadarchPropT-0.05,#2-0.06);
  \draw[cadarchProp] (#1+0.18,#2-0.39) rectangle (#1+\cadarchPropT-0.05,#2-0.27);
  \draw[cadarchFallback] (#1+\cadarchPropT+0.12,#2-\cadarchLaneH+0.10)
    rectangle (#1+\cadarchPropT+\cadarchFallT-0.08,#2-0.10);
  \draw[cadarchPartial] (#1+\cadarchPropT-0.32,#2-0.60)
    rectangle (#1+\cadarchPropT+0.28,#2-0.48);
  \cadarchDeadlineFlag{#1+\cadarchPropT}{#2}
  \draw[cadarchEmit] (#1+\cadarchPropT+\cadarchFallT-0.03,#2-\cadarchLaneM)
    -- (#1+\cadarchPropT+\cadarchFallT+0.28,#2-\cadarchLaneM);
  \node[cadarchBlock] at (#1+\cadarchPropT+\cadarchFallT+0.62,#2-\cadarchLaneM) {#4};
}

\newcommand{\cadarchLegend}[2]{%
  \begin{scope}[shift={(#1,#2)}]
    \draw[cadarchProp] (0,0) rectangle (0.42,0.15);
    \node[cadarchNote,anchor=west] at (0.50,0.075) {concurrent proposals};
    \draw[cadarchDeadline] (3.35,-0.03) -- (3.35,0.19);
    \node[cadarchNote,anchor=west] at (3.49,0.075) {deadline $\deadline_{\slot}$};
    \draw[cadarchFast] (5.30,0) rectangle (5.72,0.15);
    \node[cadarchNote,anchor=west] at (5.83,0.075) {fast finalize};
    \draw[cadarchFallback] (7.55,0) rectangle (7.97,0.15);
    \node[cadarchNote,anchor=west] at (8.08,0.075) {fallback finalize};
    \node[cadarchBlock, minimum width=0.42cm, minimum height=0.17cm,
      font=\tiny\bfseries] at (10.62,0.075) {$B$};
    \node[cadarchNote,anchor=west] at (10.95,0.075) {block};
  \end{scope}
}

\begin{tikzpicture}[font=\normalsize, x=1.28cm, y=0.72cm]
  \path[use as bounding box] (-1.05,-4.55) rectangle (11.05,1.95);

  \foreach \i in {0,1,2,3,4,5} {
    \pgfmathsetmacro{\x}{\cadarchPropT+\i*\cadarchTauX}
    \draw[cadarchGrid] (\x,0.92) -- (\x,-3.70);
  }

  \node[anchor=east, font=\small\bfseries, text=teal!55!black]
    at (0.92,1.25) {\conductor};
  \draw[teal!60!black, line width=1.1pt]
    (\cadarchPropT,1.25) -- (\cadarchPropT+4*\cadarchTauX,1.25);
  \foreach \i in {0,1,2,3,4} {
    \pgfmathsetmacro{\x}{\cadarchPropT+\i*\cadarchTauX}
    \fill[teal!60!black] (\x,1.25) circle (1.5pt);
    \draw[teal!60!black, line width=0.6pt] (\x,1.18) -- (\x,0.98);
  }
  \draw[cadarchMeasure] (\cadarchPropT,0.85) -- (\cadarchPropT+\cadarchTauX,0.85);
  \node[cadarchNote,anchor=north] at (\cadarchPropT+0.675,0.83) {$\tau$};
  \node[cadarchNote,anchor=west] at (3.05,0.78) {deadlines every $\tau$};

  \cadarchFastSlot{0*\cadarchTauX}{0.18}{$\slot$}{$B_{\slot}$}
  \cadarchFallbackSlot{1*\cadarchTauX}{-0.72}{$\slot{+}1$}{$B_{\slot{+}1}$}
  \cadarchFastSlot{2*\cadarchTauX}{-1.62}{$\slot{+}2$}{$B_{\slot{+}2}$}
  \cadarchFastSlot{3*\cadarchTauX}{-2.52}{$\slot{+}3$}{$B_{\slot{+}3}$}

  \draw[decorate, decoration={brace, amplitude=5pt, mirror},
    orange!65!black, line width=0.85pt] (8.03,-3.20) -- (8.03,0.18)
    node[midway, right=7pt, font=\footnotesize\bfseries, text=orange!60!black,
      align=left] {independent\\\chorus instances};

  \cadarchLegend{-0.85}{-4.25}
\end{tikzpicture}
\endgroup
\caption{\cadence architecture.
In normal operation, \conductor schedules consecutive slot deadlines $\tau$ apart.
Each deadline belongs to an independent single-shot \chorus instance. Here every slot
has three concurrent proposers; on finalizing, the slot emits one slot-numbered block
that merges the proposals it includes. In slot $\slot{+}1$ the last proposer disseminates
only partially (dotted), forcing the slower fallback path, so a later fast slot produces
its block first; the ledger order is still determined by slot number.}
\label{fig:cadence-architecture}
\end{figure}
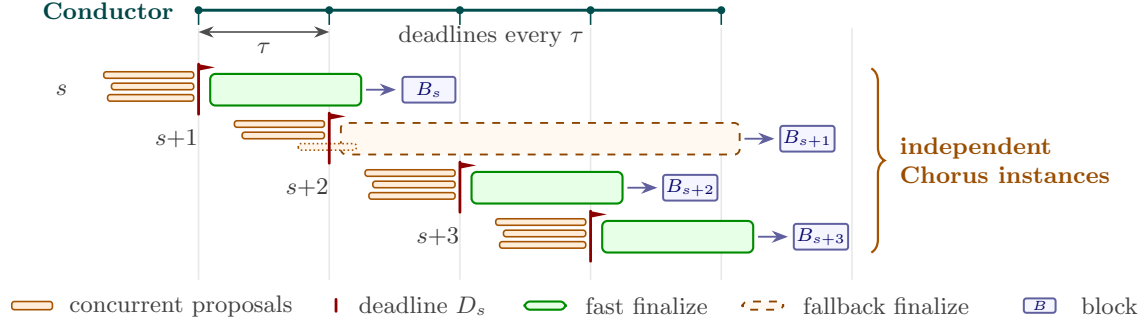

%% file: p1_informal.tex


\section{Multiple Concurrent Proposers (MCP): Problem Definition}
\label{subsection:mcp-overview}
\tobiaslater{I would advocate for this section to stay more high-level: our primary contribution is not the problem definition, and these properties are not novel. Two concrete suggestions. (1) Avoid GST in this informal part. Rather than tying guarantees to ``GST plus a grace period'', frame partial synchrony plainly: ``We assume partial synchrony: the network may delay messages arbitrarily for an unknown initial period, after which it delivers them within a bounded time. The protocol then settles a bounded time later, and the guarantees that depend on timing hold from then on.'' (2) Merge liveness and eventual stability into a single informal progress guarantee, for example: ``Once the network is in synchrony and the protocol has stabilized, blocks are produced at regular intervals and finalized within a bounded time.''}

We begin with a description of the problem we aim to solve.

\subsection{Setting}
We consider a set of $n$ validators, each of which maintains its own append-only ledger: an ordered sequence of blocks, where each block bundles transactions submitted by end users.
The goal is for these ledgers to remain mutually consistent and to keep growing, so that, to the end users, they appear as a single, ever-growing ledger.
We work in the standard Byzantine setting with $n = 3f + 1$ validators, of which up to $f > 0$ may be \emph{faulty} and behave arbitrarily (Byzantine), while the rest are \emph{correct} (or \emph{honest}).
We assume \emph{partial synchrony}: the network may behave asynchronously, delaying messages arbitrarily, up to some unknown moment called the \emph{global stabilization time} ($\mathrm{GST}$), after which it stabilizes and delivers messages within a known bound $\Delta$.
Throughout, $\delta \leq \Delta$ denotes the \emph{actual} network delay after $\mathrm{GST}$.
We further assume that validators have synchronized clocks, giving a global notion of time shared by all validators.
Finally, we assume that every correct validator begins executing at global time $0$.

Each block belongs to a \emph{slot}: slots are numbered by the positive integers, validators process them in increasing order, and each slot contributes exactly one block to the ledger.
A slot also has two associated attributes:
\begin{compactitem}
    \item a fixed set of $k$ \emph{proposers}, the validators entitled to propose that slot's contents;
    
    \item a \emph{deadline} (a point in global time), the cut-off by which proposers must ensure that validators have received sufficient information about their proposals for casting votes; the deadline is not hardcoded, but set by the protocol itself.\footnote{We occasionally refer to a slot's \emph{starting time}, which we define as its deadline minus $\Delta$, where $\Delta$ is the known bound on message delays after $\mathrm{GST}$.}
\end{compactitem}
The defining feature of the MCP problem, and what sets it apart from classical single-leader blockchains, is precisely that a slot may have \emph{several} proposers at once, rather than one: each proposer contributes its own proposal, and the slot's block should reflect the contributions of all of them.
The validators must keep agreeing on, and extending, a ledger of such blocks despite the Byzantine faults and network asynchrony.

\subsection{Guarantees}
Our MCP protocol must provide the following guarantees, closely following those introduced by Garimidi \emph{et al.}~\cite{mcp-why-and-how}.
Two of them --- censorship resistance and eventual stability --- take effect only once the network has stabilized: there is a \emph{grace period} $\mathcal{G}$ such that both hold from time $\mathrm{GST} + \mathcal{G}$ onwards.

\begin{itemize}
    \item \emph{Safety.} The ledgers of any two honest validators are always consistent, i.e., at any point in time, one is a prefix of the other, so the ledgers never fork.

    \item \emph{Liveness.} Every slot eventually contributes a block to every honest validator's ledger, so no slot stalls forever and the ledgers keep growing.

    \item \emph{Short-term censorship resistance.} An honest proposer's proposal cannot be suppressed after the grace period:
    for every slot whose starting time is at least $\mathrm{GST} + \mathcal{G}$, the proposals of all honest proposers become part of that slot's block. Note that this notion is stronger than the usual notion of censorship resistance, which only requires that proposed transactions are eventually included in the ledger.\footnote{Henceforth, we write ``censorship resistance'' to mean this stronger notion when clear from the context.}

    \item \emph{Hiding.} A faulty proposer cannot tailor its proposal to the honest proposers' proposals for the same slot: an honest proposer's proposal stays concealed until it is too late for a faulty proposer to build a new proposal and still have it included in that slot's block.

    \item \emph{Eventual stability.}
    From time $\mathrm{GST} + \mathcal{G}$ onwards, the slots come ``without gaps'': successive slots follow one another promptly, with no undue gaps between their deadlines.
\end{itemize}
\Cref{fig:mcp-overview} illustrates the problem together with these five guarantees.


\begin{figure}[h]
\centering
\scalebox{0.9}{%
\input{figures/overview-tikz}
}
\caption{The MCP problem and its five guarantees.
} 
\label{fig:mcp-overview}
\end{figure}
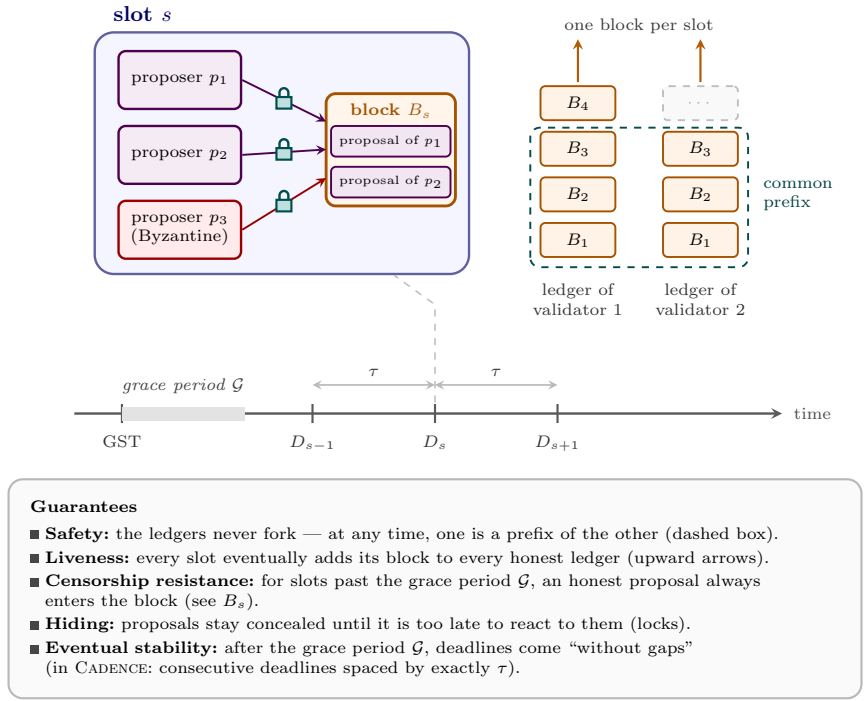

The first two guarantees, safety and liveness, are the classical guarantees expected of any blockchain protocol: the validators' ledgers must remain consistent, and they must keep growing.
The other two, short-term censorship resistance and hiding, are specific to the economic nature of blockchains.
Recall that the main goal of having several concurrent proposers is economic in the first place: the right to propose is spread across many parties precisely so that no single party controls what enters a block and profits from that control.
This entails that no honest proposer can be silenced even for a single block (otherwise its contribution could simply be dropped).
Furthermore, no faulty proposer can peek at others' proposals before committing to its own (otherwise it could profit by reacting to what it sees).
Censorship resistance guarantees the former, and hiding the latter. 
One timing aspect worth emphasizing: censorship resistance is required to hold not immediately at $\mathrm{GST}$, but only after a grace period $\mathcal{G}$ following it. In other words, slots whose starting time is within the grace period need not enjoy censorship resistance. 

This grace period is not incidental \ifformalpart: our formal definitions 
(\Cref{section:formal_problem_definition}) parametrize the guarantees by its 
length, and\fi\ifformalpart{} the shorter the $\mathcal{G}$ a protocol 
achieves\else the shorter the $\mathcal{G}$ a protocol achieves\fi, the 
sooner its guarantees kick in after network asynchrony --- a direct measure 
of the protocol's quality.


Hiding, in turn, rules out the adverse selection this peeking would otherwise enable: a Byzantine proposer that could see an honest proposal would react to the transactions it carries only when doing so is profitable, capturing the gains for itself and leaving their senders the losses. Since the adversary never sees an honest proposal in time, it cannot react selectively in this way.
We note that our hiding property slightly relaxes that of Garimidi et al.~\cite{mcp-why-and-how} while still capturing its purpose: there, a proposal must remain concealed until the ledger up to and including its slot is irrevocably committed; we require concealment only until it is too late for a faulty proposer to react to the honest proposal.

The last guarantee, eventual stability, concerns the slot schedule and is particularly applicable when performing extreme pipelining across slots.
Intuitively, after the grace period $\mathcal{G}$, the proposals for slots should be created at a steady pace: there should be no undue gaps between slots, which gives the protocol a steady high-frequency ``economic tick''. 
In \cadence, we make this intuition concrete by requiring that the slot deadlines eventually become \emph{$\tau$-spaced}, for a fixed duration $\tau$: from some point on, the deadlines of consecutive slots are separated by exactly $\tau$. 

\section{Our Extreme-Pipelining Framework}
\label{subsection:cadence-overview}

We now describe our \emph{extreme-pipelining framework},\footnote{This framework is not specific to the MCP problem: it supports ``slot-based'' consensus in general and may be of independent interest. We nevertheless present it through the lens of MCP, which is the focus of this work.} which allows for producing blocks at an arbitrarily high rate, no matter how long any individual block takes to propagate or finalize.

The framework rests on splitting the problem into two largely independent concerns, each addressed by its own abstract building block.
The first concern is what happens \emph{within} a single slot: how the proposers of a slot get their proposals into that slot's block, and how all validators come to agree on that block.
The second concern is what happens \emph{across} slots: how fast the slots follow one another (how we set their deadlines), and how many slots validators may have underway at once.
We capture the first concern in a primitive we call \emph{slot consensus}, and the second in a primitive we call the \emph{orchestrator}.
The framework itself is then little more than the glue that wires these two primitives together: it runs a single orchestrator and one instance of slot consensus per slot.
This decomposition keeps the design modular and general: each primitive can be designed, analyzed, and swapped out on its own, and any valid pair of primitives composes into a correct MCP protocol.
\cadence, our concrete MCP protocol, is this framework instantiated with concrete protocols for the two primitives: \chorus as the slot consensus and \conductor as the orchestrator. 

\subsection{Slot Consensus}
\label{subsection:slot-consensus-overview}

At its core, slot consensus is a one-shot consensus primitive, responsible for a single slot: it has all honest validators agree on a single finalized block for the slot. It is also where the MCP-specific work occurs: slot consensus gathers the slot proposers' proposals and enforces the two economic guarantees for that slot.

\subsubsection{Interface \& Guarantees.}
The interface of slot consensus is simple: the slot's proposers \kw{propose} their proposals as input, and each validator may eventually \kw{finalize} a block for the slot as output.
Roughly speaking, slot consensus is required to satisfy the following four guarantees:
\tobiaslater{I don't think this list adds much in the informal part: here we should focus on what the system does and how it works, not decompose its properties into those of the black-box primitives it is composed of.}

\begin{compactitem}
    \item \emph{Agreement.} Honest validators never finalize conflicting blocks for the slot.

    \item \emph{Termination.} Every honest validator eventually finalizes a block for the slot.

    \item \emph{Proposal inclusion.} If the slot's starting time (the slot's deadline minus $\Delta$) is past $\mathrm{GST}$, the finalized block contains the proposals of all honest proposers that proposed by the slot's starting time.

    \item \emph{Hiding.} The proposals stay concealed until it is too late for a faulty proposer to react to them: it is impossible for a faulty proposer to first learn the honest proposals and then, equipped with that knowledge, craft a new proposal of its own that still gets included in the finalized block.
\end{compactitem}
These four guarantees should already look familiar: they mirror, at the level of a single slot, every global MCP guarantee except eventual stability, and they are precisely what we will later rely on to argue that our framework solves the MCP problem. 
\Cref{fig:slot-consensus-overview} collects the interface and the guarantees of slot consensus in one place.

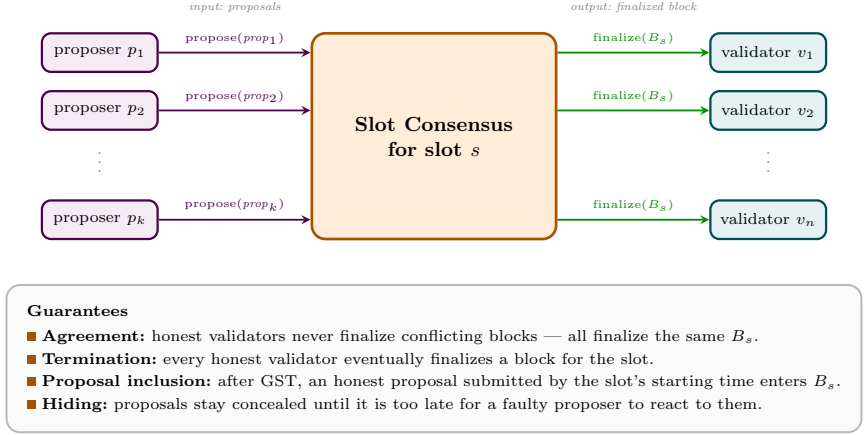
\begin{figure}[h]
\centering
\scalebox{0.85}{%
\begin{tikzpicture}[
  scmain/.style={
    draw=orange!65!black, line width=1.2pt, fill=orange!15,
    rounded corners=6pt, minimum width=3.8cm, minimum height=3.2cm,
    align=center, font=\small\bfseries
  },
  proposer/.style={
    draw=violet!60!black, line width=1pt, fill=violet!8,
    rounded corners=4pt, minimum width=1.8cm, minimum height=0.6cm,
    align=center, font=\scriptsize
  },
  validator/.style={
    draw=teal!60!black, line width=1pt, fill=teal!10,
    rounded corners=4pt, minimum width=1.8cm, minimum height=0.6cm,
    align=center, font=\scriptsize
  },
  proparr/.style={->, >=stealth, violet!60!black, thick},
  finarr/.style={->, >=stealth, green!55!black, thick},
  propscard/.style={
    draw=gray!55, line width=0.9pt, fill=gray!4,
    rounded corners=6pt, align=left, font=\scriptsize, inner sep=9pt
  },
]
  \node[scmain] (SC) at (0,0) {Slot Consensus\\ for slot $s$};

  \node[font=\tiny\itshape, gray] at (-3.1, 2.0) {input: proposals};
  \node[font=\tiny\itshape, gray] at (3.1, 2.0) {output: finalized block};

  \node[proposer] (P1) at (-5.2, 1.3) {proposer $p_1$};
  \node[proposer] (P2) at (-5.2, 0.4) {proposer $p_2$};
  \node[font=\scriptsize, gray] at (-5.2, -0.3) {$\vdots$};
  \node[proposer] (Pk) at (-5.2, -1.3) {proposer $p_k$};

  \draw[proparr] (P1.east) -- ([yshift=1.3cm]SC.west)
    node[midway, above, font=\tiny, violet!60!black] {propose$(\mathit{prop}_1)$};
  \draw[proparr] (P2.east) -- ([yshift=0.4cm]SC.west)
    node[midway, above, font=\tiny, violet!60!black] {propose$(\mathit{prop}_2)$};
  \draw[proparr] (Pk.east) -- ([yshift=-1.3cm]SC.west)
    node[midway, above, font=\tiny, violet!60!black] {propose$(\mathit{prop}_k)$};

  \node[validator] (V1) at (5.2, 1.3) {validator $v_1$};
  \node[validator] (V2) at (5.2, 0.4) {validator $v_2$};
  \node[font=\scriptsize, gray] at (5.2, -0.3) {$\vdots$};
  \node[validator] (Vn) at (5.2, -1.3) {validator $v_n$};

  \draw[finarr] ([yshift=1.3cm]SC.east) -- (V1.west)
    node[midway, above, font=\tiny, green!55!black] {finalize$(B_s)$};
  \draw[finarr] ([yshift=0.4cm]SC.east) -- (V2.west)
    node[midway, above, font=\tiny, green!55!black] {finalize$(B_s)$};
  \draw[finarr] ([yshift=-1.3cm]SC.east) -- (Vn.west)
    node[midway, above, font=\tiny, green!55!black] {finalize$(B_s)$};

  \node[propscard, anchor=north] (G) at (0, -2.3) {%
    \textbf{Guarantees}\\[3pt]
    \textcolor{orange!65!black}{\rule{1.3ex}{1.3ex}}~\textbf{Agreement:} honest validators never finalize conflicting blocks --- all finalize the same $B_s$.\\[2pt]
    \textcolor{orange!65!black}{\rule{1.3ex}{1.3ex}}~\textbf{Termination:} every honest validator eventually finalizes a block for the slot.\\[2pt]
    \textcolor{orange!65!black}{\rule{1.3ex}{1.3ex}}~\textbf{Proposal inclusion:} after $\mathrm{GST}$, an honest proposal submitted by the slot's starting time enters $B_s$.\\[2pt]
    \textcolor{orange!65!black}{\rule{1.3ex}{1.3ex}}~\textbf{Hiding:} proposals stay concealed until it is too late for a faulty proposer to react to them.
  };

\end{tikzpicture}%
}
\caption{The slot consensus primitive for a slot $s$: its interface and its guarantees.
The slot's proposers submit their proposals as input (\textcolor{violet!60!black}{violet arrows}), and every validator eventually finalizes a block for the slot as output (\textcolor{green!55!black}{green arrows}).
}
\label{fig:slot-consensus-overview}
\end{figure}

\subsection{Orchestrator}
\label{subsection:orchestrator-overview}
While slot consensus looks at one slot in isolation, the orchestrator spans all of them.
Its responsibility is to schedule the slots, that is, to determine their deadlines, and thereby the block interval. A smaller block interval means transactions wait less time before a slot is available to carry them, so the orchestrator directly governs how quickly a transaction can be picked up for inclusion. The slots are all the orchestrator knows about: proposals, blocks and other consensus-specific abstractions are entirely outside its view.

\subsubsection{Interface \& Guarantees.}
Like slot consensus before it, the orchestrator exposes a simple interface.
As input, a validator notifies the orchestrator that a slot is \kw{complete}, meaning that the validator's work within that slot is done; in our framework, as we show in \Cref{subsection:composition-overview}, this happens once the slot's block is finalized.
As output, the orchestrator \kw{schedules} new slots: it comes back to the validator with the deadlines of new slots.
With the interface in place, we can state the two properties that the orchestrator must satisfy:

\begin{compactitem}
    \item \emph{Boundedness.} The orchestrator should keep validators from running arbitrarily far ahead: it imposes a fixed limit on how many slots a correct validator may have \emph{underway} at any one time --- slots whose deadlines it has already received but has not yet completed.
    It enforces this limit through its output, withholding the deadlines of new slots and throttling the pace until the validator reports that enough earlier slots are complete.
    Crucially, this limit need not be finite: taking it to be \emph{infinite} imposes no constraint at all, recovering the most general orchestrator, which may let validators run arbitrarily far ahead.
    A finite limit is the stronger guarantee, and is what keeps each validator's resource footprint in check.


    \item \emph{Recovery.} Once the network stabilizes (some grace period after $\mathrm{GST}$), the orchestrator must schedule the slots ``properly'', and this requirement is twofold.
    First, the schedule must be steady: the deadlines of consecutive slots are spaced by exactly the fixed amount $\tau$ (the same $\tau$ with which we capture eventual stability in \Cref{subsection:mcp-overview}).
    Second, the schedule must be known in advance: every correct validator learns each deadline (i.e., receives it as the orchestrator's output) at least $\Delta$ time before the deadline itself.
    The first half is what keeps the slots coming without gaps; the second is what gives proposers room to act: in our framework, as we show in \Cref{subsection:composition-overview}, it ensures that every honest proposer issues its proposal on time, i.e., by the slot's starting time.
\end{compactitem}
\Cref{fig:orchestrator-overview} collects the interface and the guarantees of the orchestrator in one place.

\begin{figure}[h]
\centering
\scalebox{0.85}{%
\begin{tikzpicture}[
  orchmain/.style={
    draw=teal!60!black, line width=1.2pt, fill=teal!12,
    rounded corners=6pt, minimum width=3.8cm, minimum height=3.2cm,
    align=center, font=\small\bfseries
  },
  validator/.style={
    draw=teal!60!black, line width=1pt, fill=teal!10,
    rounded corners=4pt, minimum width=1.8cm, minimum height=0.6cm,
    align=center, font=\scriptsize
  },
  comparr/.style={->, >=stealth, violet!60!black, thick},
  schedarr/.style={->, >=stealth, green!55!black, thick},
  propscard/.style={
    draw=gray!55, line width=0.9pt, fill=gray!4,
    rounded corners=6pt, align=left, font=\scriptsize, inner sep=9pt
  },
]
  \node[orchmain] (O) at (0,0) {Orchestrator};

  \node[font=\tiny\itshape, gray] at (-3.1, 2.0) {input: complete notifications};
  \node[font=\tiny\itshape, gray] at (3.1, 2.0) {output: deadlines of new slots};

  \node[validator] (L1) at (-5.2, 1.3) {validator $v_1$};
  \node[validator] (L2) at (-5.2, 0.4) {validator $v_2$};
  \node[font=\scriptsize, gray] at (-5.2, -0.3) {$\vdots$};
  \node[validator] (Ln) at (-5.2, -1.3) {validator $v_n$};

  \draw[comparr] (L1.east) -- ([yshift=1.3cm]O.west)
    node[midway, above, font=\tiny, violet!60!black] {complete$(s)$};
  \draw[comparr] (L2.east) -- ([yshift=0.4cm]O.west)
    node[midway, above, font=\tiny, violet!60!black] {complete$(s)$};
  \draw[comparr] (Ln.east) -- ([yshift=-1.3cm]O.west)
    node[midway, above, font=\tiny, violet!60!black] {complete$(s)$};

  \node[validator] (R1) at (5.2, 1.3) {validator $v_1$};
  \node[validator] (R2) at (5.2, 0.4) {validator $v_2$};
  \node[font=\scriptsize, gray] at (5.2, -0.3) {$\vdots$};
  \node[validator] (Rn) at (5.2, -1.3) {validator $v_n$};

  \draw[schedarr] ([yshift=1.3cm]O.east) -- (R1.west)
    node[midway, above, font=\tiny, green!55!black] {schedule$(D_{s'})$};
  \draw[schedarr] ([yshift=0.4cm]O.east) -- (R2.west)
    node[midway, above, font=\tiny, green!55!black] {schedule$(D_{s'})$};
  \draw[schedarr] ([yshift=-1.3cm]O.east) -- (Rn.west)
    node[midway, above, font=\tiny, green!55!black] {schedule$(D_{s'})$};

  \node[propscard, anchor=north] (G) at (0, -2.3) {%
    \textbf{Guarantees}\\[3pt]
    \textcolor{teal!60!black}{\rule{1.3ex}{1.3ex}}~\textbf{Boundedness:} at any moment, each honest validator has only a bounded number of slots underway --- scheduled,\\
    \phantom{\rule{1.3ex}{1.3ex}}~but not yet complete.\\[2pt]
    \textcolor{teal!60!black}{\rule{1.3ex}{1.3ex}}~\textbf{Recovery:} after the network stabilizes (some ``short'' grace period after $\mathrm{GST}$), consecutive deadlines are spaced by\\
    \phantom{\rule{1.3ex}{1.3ex}}~exactly $\tau$, and every honest validator learns each deadline at least $\Delta$ time before the deadline itself.
  };

\end{tikzpicture}%
}
\caption{The orchestrator primitive: its interface and its guarantees.
Validators notify the orchestrator that slots are complete as input (\textcolor{violet!60!black}{violet arrows}), and the orchestrator schedules new slots by announcing their deadlines as output (\textcolor{green!55!black}{green arrows}).}
\label{fig:orchestrator-overview}
\end{figure}
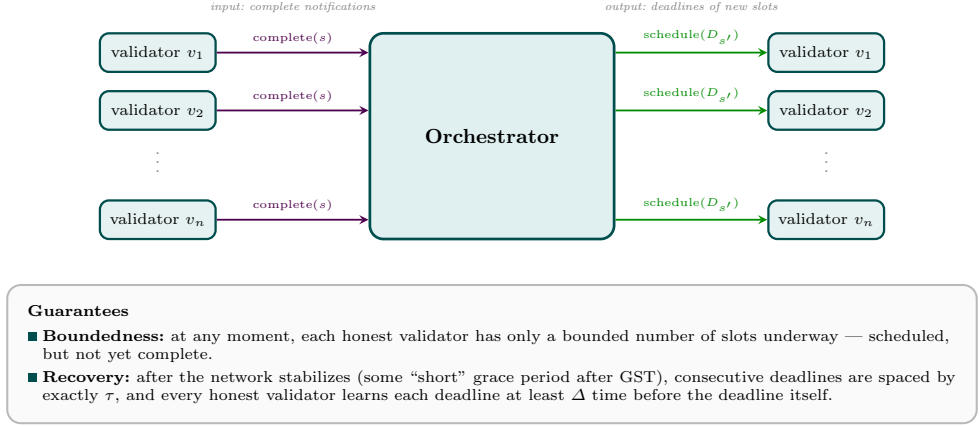

The two properties of the orchestrator address two different regimes of the network.
Boundedness is what protects the entire framework during an outage.
While the network is asynchronous, slot consensus instances may be unable to finalize, since finalization requires timely message delivery~\cite{fischer1985impossibility}; in our framework, where a validator reports a slot complete exactly when its block is finalized, this means that the complete notifications dry up.
If the orchestrator nevertheless kept opening fresh slots, the number of underway slots --- scheduled, but not yet complete --- would grow without bound for as long as the outage lasted, and with it the number of concurrently running slot consensus instances and the memory and computation each validator must devote to them.
Boundedness rules this out: deprived of complete notifications, the orchestrator withholds the deadlines of new slots, so that no matter how long the network stays asynchronous, at most a bounded number of slot consensus instances ever run at once.
Why do we care about bounding the number of concurrently running consensus instances at all?

The reason is memory: a bounded number of concurrent slots is necessary --- 
though not sufficient --- for bounded per-validator memory; one must also 
bound the memory each slot consensus instance uses (which in turn requires a 
primitive such as abortable broadcast~\cite{abortable-broadcast})\ifformalpart, 
and we return to it in \Cref{subsection:memory}\fi. This lies beyond our 
scope.


Recovery, in contrast, is what the orchestrator owes once the network has stabilized, and its two halves serve two different ends.
The first half, the steady $\tau$-spaced schedule, is precisely what establishes the MCP eventual stability property: eventual stability concerns the schedule rather than any individual slot, and the schedule is entirely the orchestrator's doing --- by no longer throttling and spacing the deadlines by exactly $\tau$, the orchestrator has the slots advance at the steady, ``gap-free'' cadence that eventual stability demands.
The second half, that every correct validator learns each deadline at least $\Delta$ in advance, plays a quieter but equally important role: it ensures that a proposer knows the deadline already by the slot's starting time, leaving it enough time to submit its proposal --- exactly what the proposal inclusion of slot consensus presupposes.

\subsection{Putting the Pieces Together}
\label{subsection:composition-overview}
Our framework composes these two primitives in the natural way; \Cref{fig:cadence-overview} illustrates the interaction.
Each validator runs its local instance of the orchestrator, which drives the protocol forward by setting each slot's deadline.
At each slot's starting time ($\Delta$ time before its deadline), or as soon as it learns the slot's deadline, if later, the validator begins participating in that slot's consensus instance and, if it is one of the slot's proposers, submits its proposal to that instance.
Whenever a slot consensus instance finalizes a block, the validator records the block, notifies the orchestrator that the slot is complete (which may in turn let the orchestrator schedule further slots), and eventually appends the block to its ledger.
Since slots may be finalized out of order, blocks are buffered and appended in slot-number order, so that each validator's ledger grows as a clean, contiguous sequence with no slot missing.

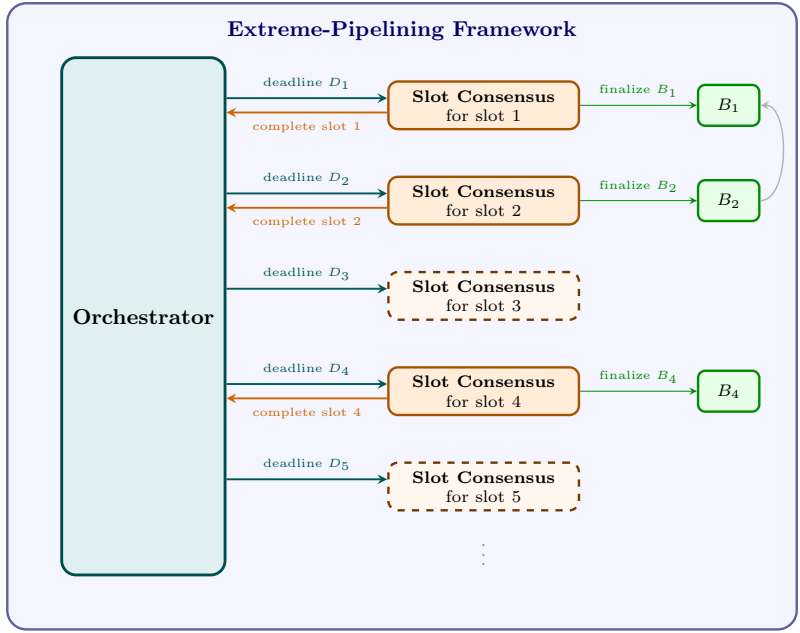
\begin{figure}[h]
\centering
\scalebox{0.9}{%
\begin{tikzpicture}[
  orch/.style={
    draw=teal!60!black, line width=1.2pt, fill=teal!12,
    rounded corners=6pt, minimum width=2.4cm, minimum height=7.6cm,
    align=center, font=\small\bfseries
  },
  scbox/.style={
    draw=orange!65!black, line width=1pt, fill=orange!15,
    rounded corners=4pt, minimum width=2.8cm, minimum height=0.7cm,
    align=center, font=\scriptsize
  },
  scprog/.style={
    draw=orange!45!black, line width=1pt, fill=orange!7, dashed,
    rounded corners=4pt, minimum width=2.8cm, minimum height=0.7cm,
    align=center, font=\scriptsize
  },
  blockbox/.style={
    draw=green!55!black, line width=1pt, fill=green!10,
    rounded corners=3pt, minimum width=0.9cm, minimum height=0.6cm,
    align=center, font=\scriptsize
  },
  blockgray/.style={
    draw=gray!40, line width=0.8pt, fill=gray!8, dashed,
    rounded corners=3pt, minimum width=0.9cm, minimum height=0.6cm,
    align=center, font=\scriptsize, text=gray!50
  },
  openarr/.style={->, >=stealth, teal!70!black, thick},
  skiparr/.style={->, >=stealth, teal!50!black, thick, dashed},
  complarr/.style={->, >=stealth, orange!80!black, thick},
  finarr/.style={->, >=stealth, green!55!black},
  parentarr/.style={->, >=stealth, gray!60, semithick},
  outer/.style={draw=blue!25!gray, line width=1pt, fill=blue!4, rounded corners=8pt},
]
  \node[outer, minimum width=11.6cm, minimum height=9.2cm]
    at (3.8, -0.8) {};
  \node[font=\small\bfseries, text=blue!40!black] at (3.8, 3.4) {Extreme-Pipelining Framework};

  \node[orch] (O) at (0, -0.8) {Orchestrator};

  \node[scbox]  (S1) at (5.0,  2.3) {\textbf{Slot Consensus} \\ for slot $1$};
  \node[scbox]  (S2) at (5.0,  0.9) {\textbf{Slot Consensus} \\ for slot $2$};
  \node[scprog] (S3) at (5.0, -0.5) {\textbf{Slot Consensus} \\ for slot $3$};
  \node[scbox]  (S4) at (5.0, -1.9) {\textbf{Slot Consensus} \\ for slot $4$};
  \node[scprog] (S5) at (5.0, -3.3) {\textbf{Slot Consensus} \\ for slot $5$};
  \node[font=\scriptsize, gray] at (5.0,-4.2) {$\vdots$};

  \draw[openarr] ([yshift=3pt]O.east |- S1.west) -- ([yshift=3pt]S1.west)
    node[midway, above, font=\tiny, teal!70!black] {deadline $D_1$};
  \draw[openarr] ([yshift=3pt]O.east |- S2.west) -- ([yshift=3pt]S2.west)
    node[midway, above, font=\tiny, teal!70!black] {deadline $D_2$};
  \draw[openarr] ([yshift=3pt]O.east |- S3.west) -- ([yshift=3pt]S3.west)
    node[midway, above, font=\tiny, teal!70!black] {deadline $D_3$};
  \draw[openarr] ([yshift=3pt]O.east |- S4.west) -- ([yshift=3pt]S4.west)
    node[midway, above, font=\tiny, teal!70!black] {deadline $D_4$};
  \draw[openarr] ([yshift=3pt]O.east |- S5.west) -- ([yshift=3pt]S5.west)
    node[midway, above, font=\tiny, teal!70!black] {deadline $D_5$};

  \draw[complarr] ([yshift=-3pt]S1.west) -- ([yshift=-3pt]O.east |- S1.west)
    node[midway, below, font=\tiny, orange!80!black] {complete slot $1$};
  \draw[complarr] ([yshift=-3pt]S2.west) -- ([yshift=-3pt]O.east |- S2.west)
    node[midway, below, font=\tiny, orange!80!black] {complete slot $2$};
  \draw[complarr] ([yshift=-3pt]S4.west) -- ([yshift=-3pt]O.east |- S4.west)
    node[midway, below, font=\tiny, orange!80!black] {complete slot $4$};

  \node[blockbox] (B1) at (8.6,  2.3) {$B_1$};
  \node[blockbox] (B2) at (8.6,  0.9) {$B_2$};
  \node[blockbox] (B4) at (8.6, -1.9) {$B_4$};

  \draw[finarr] (S1.east) -- (B1.west)
    node[midway, above, font=\tiny, green!55!black] {finalize $B_1$};
  \draw[finarr] (S2.east) -- (B2.west)
    node[midway, above, font=\tiny, green!55!black] {finalize $B_2$};
  \draw[finarr] (S4.east) -- (B4.west)
    node[midway, above, font=\tiny, green!55!black] {finalize $B_4$};

  \draw[parentarr] (B2.east) to[out=0, in=0, looseness=0.8] (B1.east);

\end{tikzpicture}%
}
\caption{Our framework as the composition of its two building blocks: a single orchestrator, and one slot consensus instance per slot.
The orchestrator schedules slots by setting their deadlines: it announces the deadline $D_s$ of each slot $s$ (\textcolor{teal!70!black}{teal arrows}), thereby spawning a slot consensus for it.
Each slot consensus independently \kw{finalizes} a block (\textcolor{green!55!black}{green arrows}).
The instances run concurrently and may finalize out of order: here, slot $4$ has already finalized its block $B_4$, even though the earlier slot $3$ (and the later slot $5$) is still in progress (dashed borders).
Upon a finalization, the validator notifies the orchestrator that the slot is \kw{complete} (\textcolor{orange!80!black}{orange arrows}), which may let the orchestrator schedule further slots; the finalized blocks are then appended to the ledger in slot-number order.}
\label{fig:cadence-overview}
\end{figure}

\subsubsection{Extreme Pipelining.}
With the entire framework in place, we can now explain how it achieves the extreme pipelining it is named after: a fixed block interval of $\tau$, regardless of how long consensus on any individual slot takes.
The key lies in both building blocks: the slots and their blocks are independent by construction (one slot consensus instance per slot), and the orchestrator is free to schedule them independently, constrained only by boundedness (it may keep opening new slots as long as the number of underway slots stays within the bound).
To appreciate this freedom, let us explain why it was out of reach for earlier designs.
In traditional BFT protocols, consecutive slots are chained to one another: a slot must wait for its predecessor to make enough progress before it can begin.
In unpipelined protocols, such as PBFT~\cite{Castro:1999:PBF:296806.296824}, basic HotStuff~\cite{HotStuff}, and SBFT~\cite{SBFT}, this chaining is the strictest possible: the next slot begins only after the current slot's block is fully finalized.
Pipelined protocols, such as pipelined HotStuff~\cite{HotStuff}, FastHotStuff~\cite{FastHotStuffTDSC}, and MonadBFT~\cite{jalalzai2026monadbftfastresponsiveforkresistant}, loosen this coupling, but do not remove it: the next slot begins earlier, once the current one produces a quorum certificate, yet it still requires that certificate as input.
In both cases, each slot passes a baton to the next, and the protocol can only advance as fast as the baton travels.
Our framework drops this dependency altogether: under favorable network conditions, a new slot opens every $\tau$ time, dictated solely by the global clock, and nothing from any earlier slot --- no artifact, no certificate, no event --- is needed for it to start.
Indeed, nothing in the orchestrator's interface ties a slot to its predecessor's progress: complete notifications may throttle the pace (boundedness), but they are never required as input for a new slot to open.
The slots no longer wait for one another; they simply follow the schedule.
The block interval thus becomes decoupled from the per-slot consensus latency: slots can be scheduled arbitrarily close together, and hence blocks can be produced at an arbitrarily high rate, no matter how long any individual slot takes to certify or finalize.


\subsubsection{Proof Sketch.}
Let us now explain why our framework solves the MCP problem as defined in \Cref{subsection:mcp-overview}, walking through its guarantees one by one.
Safety follows from the agreement property of slot consensus, applied slot by slot.
Agreement ensures that, for every slot, honest validators never finalize conflicting blocks; and since every validator appends the finalized blocks to its ledger in slot-number order, any two honest ledgers contain the same blocks in the same order, one being a prefix of the other.
In short, if no single slot can produce disagreement, the ledgers have no way to diverge.
Liveness follows from an alternation between the two primitives.
Consider the slots currently underway.
By the termination of slot consensus, every honest validator eventually finalizes a block for each of them and, having done so, reports these slots as complete to the orchestrator.
Once these completions arrive, the orchestrator's boundedness no longer stands in the way: the orchestrator schedules new slots, whose consensus instances again eventually terminate, and so on forever.
Hence, every slot is eventually scheduled and eventually contributes a block to every honest validator's ledger: no slot stalls forever, and the ledgers keep growing.

Hiding requires no work at all: the hiding of each individual slot consensus instance is precisely the global hiding property, read across all slots.
Censorship resistance, in contrast, needs the two primitives to cooperate.
After the grace period $\mathcal{G}$ following $\mathrm{GST}$, the recovery property of the orchestrator guarantees that every honest validator learns each upcoming deadline at least $\Delta$ in advance, that is, by the slot's starting time.
Every honest proposer therefore has enough time to submit its proposal by the starting time, and from there the proposal inclusion property of slot consensus takes over: the on-time proposal makes it into the slot's finalized block, and hence into the ledger.

Eventual stability is exactly the first half of the recovery property: once the network stabilizes (from time $\mathrm{GST} + \mathcal{G}$ onward), the orchestrator spaces consecutive deadlines by exactly $\tau$, so the slots advance at the steady cadence that eventual stability demands.
Finally, one practical concern remains; it is not an MCP guarantee, but, as discussed, a prerequisite for feasibility: the number of slot consensus instances a validator must attend to at any point in time.
This is where boundedness of the orchestrator comes in: at any moment, each honest validator has only a bounded number of slots underway, hence only a bounded number of slot consensus instances to run, keeping its memory and computation in check no matter how the network behaves.

\section{\chorus: Our Slot Consensus}
\label{section:chorus-overview}

We now present \chorus, the slot consensus protocol we use in \cadence.
We describe \chorus through two modes of operation.
The first, the \emph{fast path}, handles the good case: it is essentially a classical single-leader consensus generalized to accommodate several proposers at once, finalizing a block quickly whenever the network is synchronous and the proposers behave.
The second, the \emph{fallback path}, takes over when the fast path fails to finalize, ensuring that the slot reaches agreement nonetheless.
Note that the split into a fast and a fallback path is only expository: the two are interwoven into a single protocol (and not two modes that the protocol toggles between).

\subsection{Preliminaries}
\label{subsection:chorus-preliminaries}

Before presenting the protocol itself, we settle two preliminaries: the cryptographic primitives \chorus builds on, and what its validators actually agree upon (which, as we will see, is not blocks of transactions).

\subsubsection{Cryptography.}
\chorus relies on a handful of standard cryptographic primitives, which we use as follows.
\begin{compactitem}
    \item \emph{Collision-resistant hashing}, used to build \emph{Merkle commitments}: a short root that binds a whole sequence of values, while any single value can later be proven to be among them with a logarithmic-size proof.

    \item \emph{Digital signatures} that can moreover be \emph{aggregated}, so that many signatures on the same message collapse into a single short multi-signature.

    \item \emph{Erasure coding}, which encodes a piece of data into $n$ chunks such that any $f + 1$ of them suffice to reconstruct the original data.

    \item A per-slot \emph{threshold encryption} scheme, which encrypts a message so that it can be decrypted only once $f + 1$ validators release their decryption shares for that slot. The shares are bound to the slot itself, not to any individual ciphertext, so a single set of $f+1$ shares opens every message encrypted for that slot.
\end{compactitem}

\subsubsection{The Object of Agreement.}\label{subsubsection:ObjectofAgreement}
We now turn to the following question: what \chorus's validators actually agree on.
The answer is not a block, but \emph{short digests} of the proposals. 
Much as a single-leader protocol can agree on a hash of the leader's block rather than on the block itself, \chorus has its validators agree on a digest of each proposal rather than on the proposal itself.
Concretely, each proposal goes through three steps before agreement (\Cref{fig:proposal-encoding}): it is first encrypted; then erasure-coded into $n$ chunks; and finally committed to by a single Merkle root, the digest the validators agree on. Recall from \Cref{section:introduction} (Idea 4) that committing to
these digests rather than to the full proposals is exactly what lets
\emph{dissemination} run in parallel with consensus, following
AVID~\cite{AVID}, DispersedSimplex~\cite{DispersedSimplex}, and Deterministic
RaptorCast~\cite{mip10}. Each proposal is disseminated as erasure-coded chunks
committed by a Merkle root, and the deterministic encoding gives
\emph{consistency}, meaning all correct validators reach the same verdict on a
root, either all recovering the same proposal or all rejecting it as invalidly
encoded. Consistency is a property of the encoding alone. \emph{Availability},
that the proposal behind a root can be recovered, is achieved once \chorus
certifies the root, since the certificate guarantees that enough honest
validators hold their chunks for the data to be reconstructible. Any
erasure-coded dissemination scheme providing consistency composes with \chorus
in this way. \chorus builds its object of agreement on Deterministic
RaptorCast~\cite{mip10}, whose deterministic encoding makes a root's verdict
unique (\emph{consistency}), while \chorus's certificates over these roots
establish that the committed data is recoverable (\emph{availability}). Both
are established for \chorus in \Cref{section:slot_agreement}.
With several proposers in a slot, this object of agreement grows into a \emph{vector} of such roots: one entry per proposer, each committing to that proposer's proposal, or to a designated empty value when the proposer contributes nothing.

\begin{figure}[h]
\centering
\scalebox{0.9}{%
\begin{tikzpicture}[
  stage/.style={
    draw=violet!60!black, line width=1pt, fill=violet!8,
    rounded corners=4pt, minimum width=2.3cm, minimum height=1.0cm,
    align=center, font=\small
  },
  cipher/.style={
    draw=teal!60!black, line width=1pt, fill=teal!10,
    rounded corners=4pt, minimum width=2.3cm, minimum height=1.0cm,
    align=center, font=\small
  },
  chunk/.style={
    draw=orange!65!black, line width=0.9pt, fill=orange!12,
    rounded corners=2pt, minimum width=0.78cm, minimum height=0.62cm,
    align=center, font=\scriptsize
  },
  mnode/.style={
    draw=gray!60, line width=0.8pt, fill=gray!6,
    circle, minimum size=0.5cm, inner sep=1pt, font=\scriptsize
  },
  rootbox/.style={
    draw=green!45!black, line width=1.2pt, fill=green!10,
    rounded corners=3pt, minimum width=0.9cm, minimum height=0.62cm,
    align=center, font=\small\bfseries
  },
  steparr/.style={->, >=stealth, gray!60!black, thick},
  treeline/.style={gray!60, line width=0.8pt},
  steplab/.style={font=\scriptsize\itshape, gray!50!black},
]
  \node[stage] (M) at (0, 0) {plaintext\\ proposal $M$};

  \node[cipher] (C) at (4.3, 0) {encrypted\\ proposal $C$};
  \draw[steparr] (M.east) -- (C.west)
    node[midway, above, steplab] {(1) encrypt};

  \node[chunk] (c1) at (8.0, 0) {$c_1$};
  \node[chunk] (c2) at (8.9, 0) {$c_2$};
  \node[chunk] (c3) at (9.8, 0) {$c_3$};
  \node[chunk] (c4) at (10.7, 0) {$c_4$};
  \node[steplab] at (9.35, -0.65) {$n$ chunks, one per validator (any $f+1$ recover $C$)};
  \draw[steparr] (C.east) -- (7.5, 0)
    node[midway, above, steplab] {(2) encode};

  \node[mnode] (h12) at (8.45, 1.05) {$h_{12}$};
  \node[mnode] (h34) at (10.25, 1.05) {$h_{34}$};
  \node[rootbox] (r) at (9.35, 2.15) {$r$};
  \draw[treeline] (c1.north) -- (h12.south);
  \draw[treeline] (c2.north) -- (h12.south);
  \draw[treeline] (c3.north) -- (h34.south);
  \draw[treeline] (c4.north) -- (h34.south);
  \draw[treeline] (h12.north) -- (r.south);
  \draw[treeline] (h34.north) -- (r.south);
  \node[steplab, align=left, anchor=west] at (10.9, 1.6) {(3) build Merkle tree\\ \phantom{(3)} over the chunks};
  \node[font=\scriptsize, green!45!black, anchor=south] at (9.35, 2.55) {Merkle root};

\end{tikzpicture}%
}
\caption{From a plaintext proposal to its Merkle root, step by step.}
\label{fig:proposal-encoding}
\end{figure}
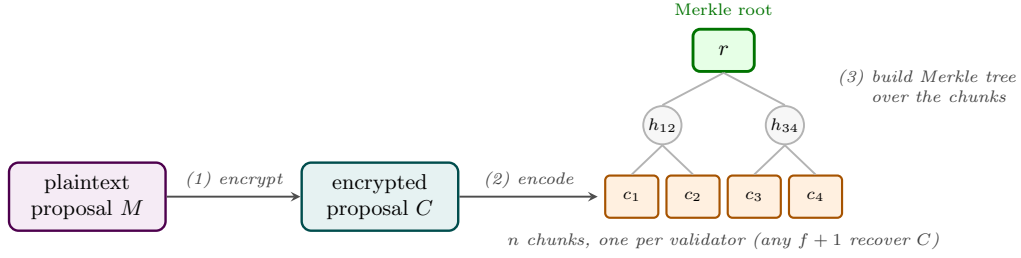

\noindent Crucially, an entry is never agreed upon on its own: it always comes with a \emph{certificate}.
\begin{compactitem}
    \item For a \emph{positive} entry (one carrying a Merkle root), the certificate is a \emph{proof of availability}, attesting that the data behind the root is available: every honest validator can retrieve it and arrive at the same outcome, either reconstructing the committed proposal in full, or deterministically concluding that its encoding is broken (the chunks are inconsistent); crucially, whichever of the two it is, all honest validators conclude the same.
    

    \item For a \emph{negative} entry, the certificate is instead a proof that the proposer's proposal is correctly excluded, for instance because the proposer failed to disseminate it on time.
\end{compactitem}
We call this vector of certified entries a \emph{meta-block} (\Cref{fig:meta-block}): one entry per proposer, each a Merkle root or the empty value $\bot$.
A meta-block \emph{includes} a proposer's proposal when its entry is positive --- the root commits to the proposal --- and \emph{excludes} it when the entry is negative.
From here on, then, \chorus's goal is simply to agree on a meta-block's \emph{entries} --- the Merkle roots and $\bot$ values themselves, not the certificates that justify them. The fast path and the fallback path are two mechanisms for doing exactly that, and we describe both purely as procedures for agreeing on the aforementioned entries.
How the finalized entries are turned into the slot's block is described in \Cref{subsection:meta-to-block}.


\begin{figure}[h]
\centering
\scalebox{0.9}{%
\begin{tikzpicture}[
  entry/.style={
    draw=teal!60!black, line width=1pt, fill=teal!10,
    rounded corners=3pt, text width=2.8cm, minimum height=1.4cm,
    inner sep=4pt, align=center, font=\small
  },
  empty/.style={
    draw=gray!50, line width=1pt, fill=gray!8, dashed,
    rounded corners=3pt, text width=2.8cm, minimum height=1.4cm,
    inner sep=4pt, align=center, font=\small, text=gray!55
  },
  blockbox/.style={
    draw=orange!65!black, line width=1.2pt, fill=orange!12,
    rounded corners=4pt, minimum width=2.6cm, minimum height=1.1cm,
    align=center, font=\small\bfseries
  },
  maparr/.style={->, >=stealth, gray!60!black, thick},
]
  \draw[draw=blue!25!gray, line width=1pt, fill=blue!4, rounded corners=8pt]
    (-2.05, -1.0) rectangle (11.65, 1.45);
  \node[font=\small\bfseries, text=blue!40!black] at (4.8, 1.95) {meta-block};

  \node[font=\scriptsize, gray!50!black] at (0.0, 1.0) {proposer $1$};
  \node[font=\scriptsize, gray!50!black] at (3.2, 1.0) {proposer $2$};
  \node[font=\scriptsize, gray!50!black] at (6.4, 1.0) {proposer $3$};
  \node[font=\scriptsize, gray!50!black] at (9.6, 1.0) {proposer $4$};

  \node[entry, draw=teal!60!black, fill=teal!10] (e1) at (0.0, 0.0) {$r_1$ \\[3pt] {\scriptsize\textcolor{green!55!black}{$\pi_1$: data available}}};
  \node[entry, draw=violet!60!black, fill=violet!10] (e2) at (3.2, 0.0) {$r_2$ \\[3pt] {\scriptsize\textcolor{green!55!black}{$\pi_2$: data available}}};
  \node[empty] (e3) at (6.4, 0.0) {$\bot$ \\[3pt] {\scriptsize $\pi_3$: correctly excluded}};
  \node[entry, draw=brown!70!black, fill=brown!12] (e4) at (9.6, 0.0) {$r_4$ \\[3pt] {\scriptsize\textcolor{green!55!black}{$\pi_4$: data available}}};

  \node[blockbox] (B) at (4.8, -2.7) {block $B$};
  \draw[maparr] (4.8, -1.0) -- (B.north)
    node[midway, right, font=\scriptsize, gray!55!black] {deterministic};

\end{tikzpicture}%
}
\caption{A \emph{meta-block}, the object \chorus agrees on.
Once it is agreed upon, the slot's block follows with no further agreement: each validator independently recovers the proposal behind every positive entry (discarding any invalidly encoded one) and assembles them into the same well-defined block (see \Cref{subsection:meta-to-block}).}
\label{fig:meta-block}
\end{figure}
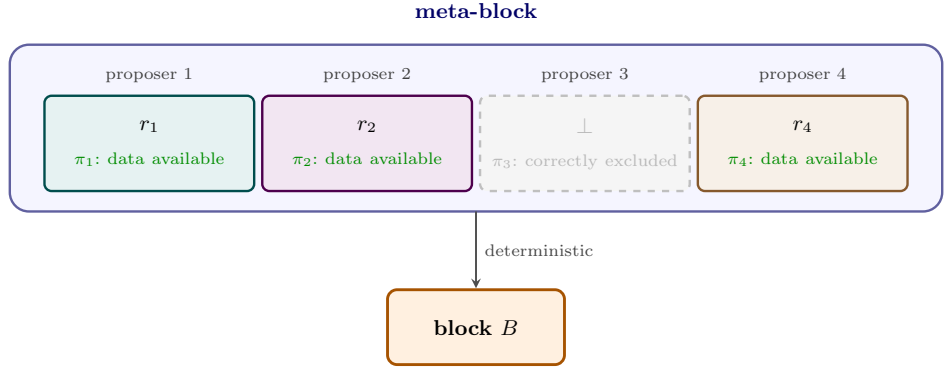

\subsection{Fast Path}
\label{subsection:fast-path-overview}

The fast path of \chorus follows the familiar structure of single-leader consensus, in which a leader disseminates its proposal and the validators then run two voting rounds before finalizing.
\chorus keeps this template and departs from it in only one respect: a slot has several proposers, so all of them act as leaders at once, each disseminating a proposal of its own.

\subsubsection{Description.}
Before dissemination, each proposer \emph{encrypts} its proposal using the threshold encryption scheme, producing a ciphertext that can be decrypted only once $f+1$ validators release their decryption shares for the slot.
The proposer then disseminates this encrypted proposal using the chunked encoding of \Cref{fig:proposal-encoding}, sending each validator that validator's chunk together with a Merkle proof against the root and the proposer's signature over the slot number, its own identity, and the root. The root commits to the encrypted proposal rather than the plaintext one.
An honest proposer disseminates by the slot's \emph{starting time} $\deadline_\slot - \Delta$, so that under synchrony its chunks reach every validator before the deadline.

The first voting round opens at the slot's deadline, and in it each validator casts one vote per proposer $p_j$: if it has received a valid chunk from $p_j$, carrying $p_j$'s Merkle root $r_j$, it broadcasts a positive vote
\[
    \langle \textsc{yes}, s, p_j, r_j, \sigma \rangle,
\]
where $\sigma$ is its signature on $\langle \textsc{yes}, s, p_j, r_j \rangle$, so that $2f+1$ such votes can later be combined into a single certificate on $r_j$.
Importantly, a validator also attaches to each positive vote the chunk it received \tobiaschange{(a positive vote is accepted only if it carries its sender's assigned chunk)}; it does so to support the fallback path, as we explain in \Cref{subsection:fallback-path-overview}.
If, instead, no valid chunk from $p_j$ arrived by the deadline,\footnote{Validators also ignore chunks that arrive too long before $\deadline_\slot$, so a Byzantine proposer cannot flood them with proposals for far-future slots.} the validator broadcasts a signed negative vote
\[
    \langle \textsc{no}, s, p_j, \sigma \rangle,
\]
where $\sigma$ is now its signature on $\langle \textsc{no}, s, p_j \rangle$.
Crucially, as already noted in \Cref{subsection:chorus-preliminaries}, these shares are tied to the slot itself rather than to any individual proposal, so $f+1$ of them suffice to decrypt \emph{every} proposal of the slot.

Once a validator has gathered, for \emph{every} proposer, a quorum of $2f+1$ matching votes, that is, $2f+1$ \textsc{yes} votes on the same root or $2f+1$ \textsc{no} votes, it assembles a \emph{fast meta-block}: for each proposer it records either the root $r_j$ with its certificate of $2f+1$ \textsc{yes} votes, or, when the quorum is negative, a negative entry with a certificate of $2f+1$ \textsc{no} votes.
(When such a quorum fails to form for some proposer, the fallback path takes over, as we describe in \Cref{subsection:fallback-path-overview}.)
These certificates are exactly the certificates that the entries carry.
A $2f+1$ \textsc{yes} certificate is a proof of availability: at least $f+1$ of its signers are honest, and each signed only after receiving and verifying its chunk, so at least $f+1$ honest validators hold their chunks, enough to recover the data behind $r_j$ and, as discussed above, to either reconstruct the proposal in full or unanimously deem it invalid (see \Cref{subsection:meta-to-block}).
A $2f+1$ \textsc{no} certificate, in turn, justifies excluding the proposer, since at least $f+1$ honest validators saw no valid chunk by the deadline.
Each validator that has assembled a fast meta-block casts a single \emph{fast vote} on it as a whole, broadcasting the meta-block along with the vote. As soon as it collects $2f+1$ fast votes on the same fast meta-block, it finalizes its entries.
\Cref{fig:fast-path} traces this communication pattern for a slot with two proposers.

\begin{figure}[h]
\centering
\scalebox{0.9}{%
\begin{tikzpicture}[
  >=stealth,
  vline/.style={gray!55, line width=0.8pt},
  sep/.style={gray!45, dashed, line width=0.6pt},
  dissarr/.style={->, orange!85!black, line width=0.9pt},
  nochunk/.style={->, gray!55, dashed, line width=0.6pt},
  votearr/.style={->, blue!70, dashed, line width=0.5pt, opacity=0.55},
  optarr/.style={->, green!60!black, line width=0.5pt, opacity=0.7},
  vdot/.style={circle, draw=blue!55, fill=white, line width=0.7pt, inner sep=1.6pt},
  pdot/.style={circle, draw=violet!60!black, fill=violet!35, line width=0.6pt, inner sep=2pt},
  fdot/.style={circle, draw=green!55!black, fill=green!35, line width=1pt, inner sep=2.6pt},
  propdot/.style={circle, draw=orange!75!black, fill=orange!45, line width=0.9pt, inner sep=2.6pt},
  sildot/.style={circle, draw=gray!60, fill=white, line width=0.9pt, inner sep=2.6pt},
  mcell/.style={draw=teal!60!black, line width=1pt, fill=teal!10, rounded corners=3pt, minimum width=2.5cm, minimum height=1.7cm, align=center, font=\small},
  ecell/.style={draw=gray!50, line width=1pt, fill=gray!8, dashed, rounded corners=3pt, minimum width=2.5cm, minimum height=1.7cm, align=center, font=\small, text=gray!55},
  plab/.style={font=\footnotesize, anchor=east, align=right},
  blab/.style={font=\scriptsize, align=center, text width=1.75cm, anchor=north, text=gray!45!black},
  leglab/.style={font=\scriptsize, anchor=west},
  chunkmini/.style={draw=orange!65!black, fill=orange!35, line width=0.6pt, minimum width=0.2cm, minimum height=0.2cm, inner sep=0pt},
  sharemini/.style={draw=teal!60!black, fill=teal!25, line width=0.6pt, minimum width=0.2cm, minimum height=0.2cm, inner sep=0pt, rotate=45},
]
  \foreach \y in {5,3.5,2,0.5}{ \draw[vline] (0.5,\y) -- (8.7,\y); }
  \node[plab, text=orange!70!black] at (0.35,5) {\textbf{Validator 1}\\ proposer $p_1$};
  \node[plab, text=gray!55!black] at (0.35,3.5) {\textbf{Validator 2}\\ proposer $p_2$ (silent)};
  \node[plab] at (0.35,2) {Validator 3};
  \node[plab] at (0.35,0.5) {Validator 4};

  \foreach \x in {3.3, 5.6, 7.9}{ \draw[sep] (\x,-0.55) -- (\x,5.7); }

  \foreach \yb in {5,3.5,2,0.5}{ \draw[dissarr] (1.4,5) -- (3.05,\yb); }
  \foreach \yb in {5,3.5,2,0.5}{ \draw[nochunk] (1.4,3.5) -- (3.05,\yb); }

  \foreach \ya in {5,3.5,2,0.5}{ \foreach \yb in {5,3.5,2,0.5}{ \draw[votearr] (3.3,\ya) -- (5.6,\yb); } }

  \foreach \ya in {5,3.5,2,0.5}{ \foreach \yb in {5,3.5,2,0.5}{ \draw[optarr] (5.6,\ya) -- (7.9,\yb); } }

  \node[propdot] at (1.4,5) {};
  \node[sildot]  at (1.4,3.5) {};
  \foreach \y in {5,3.5,2,0.5}{ \node[vdot] at (3.3,\y) {}; }
  \foreach \y in {5,3.5,2,0.5}{ \node[pdot] at (5.6,\y) {}; }
  \foreach \y in {5,3.5,2,0.5}{ \node[fdot] at (7.9,\y) {}; }

  \foreach \y in {5,3.5,2,0.5}{
    \node[chunkmini] at ({3.62},{\y+0.32}) {};
    \node[sharemini] at ({3.98},{\y+0.32}) {};
  }

  \draw[draw=blue!25!gray, line width=1pt, fill=blue!4, rounded corners=8pt] (2.7,5.65) rectangle (8.5,8.1);
  \node[font=\small\bfseries, text=blue!40!black] at (5.6,8.45) {fast meta-block};
  \node[font=\scriptsize, gray!50!black] at (4.2,7.7) {proposer $1$};
  \node[font=\scriptsize, gray!50!black] at (7.0,7.7) {proposer $2$};
  \node[mcell, draw=teal!60!black, fill=teal!10] (mc1) at (4.2,6.6) {$r_1$ \\[3pt] {\scriptsize\textcolor{green!55!black}{$\pi_1 = 2f{+}1$}} \\ {\scriptsize\textcolor{green!55!black}{signatures on $r_1$}}};
  \node[ecell] (mc2) at (7.0,6.6) {$\bot$ \\[3pt] {\scriptsize $2f{+}1$ \textsc{no} certificate}};
  \draw[gray!55, dashed, line width=0.5pt] (5.6,5.65) -- (5.6,5.05);

  \draw[draw=gray!45, rounded corners=4pt, fill=white, line width=0.7pt] (9.1,0.45) rectangle (15.7,5.1);
  \node[font=\small\bfseries, text=gray!55!black, anchor=west] at (9.3,4.75) {Legend};
  \draw[dissarr] (9.35,4.25) -- (9.95,4.25); \node[leglab] at (10.05,4.25) {chunk dissemination (encrypted proposal)};
  \draw[nochunk] (9.35,3.75) -- (9.95,3.75); \node[leglab] at (10.05,3.75) {silent proposer: no chunk received};
  \draw[votearr, opacity=1] (9.35,3.25) -- (9.95,3.25); \node[leglab] at (10.05,3.25) {round-1 vote (per proposer)};
  \node[chunkmini] at (9.65,2.75) {}; \node[leglab] at (10.05,2.75) {chunk attached to the round-1 vote};
  \node[sharemini] at (9.65,2.25) {}; \node[leglab] at (10.05,2.25) {decryption share attached to the vote};
  \draw[optarr, opacity=1] (9.35,1.75) -- (9.95,1.75); \node[leglab] at (10.05,1.75) {fast vote (round 2)};
  \node[pdot] at (9.65,1.25) {}; \node[leglab] at (10.05,1.25) {build fast meta-block, \emph{speculatively finalize}};
  \node[fdot] at (9.65,0.75) {}; \node[leglab] at (10.05,0.75) {\emph{finalize} on $2f{+}1$ fast votes};

  \node[blab] at (1.4,-0.75) {encrypt, then disseminate (before deadline $D$)};
  \node[blab] at (3.3,-0.75) {deadline $D$: round-1 vote (per proposer)};
  \node[blab] at (5.6,-0.75) {build meta-block; \emph{spec.\ finalize}; fast vote};
  \node[blab] at (7.9,-0.75) {collect $2f{+}1$ fast votes, then finalize};

\end{tikzpicture}%
}
\caption{The \chorus fast path for a slot with two proposers $p_1$ and $p_2$ ($n=4$, $f=1$).
Proposer $p_1$ disseminates its (encrypted) proposal as one chunk per validator (\textcolor{orange!85!black}{orange}), while $p_2$ stays silent (\textcolor{gray!60!black}{gray, dashed}).
At deadline $D$, each validator broadcasts one vote per proposer (\textcolor{blue!70}{blue, dashed}), $\langle\textsc{yes}, r_1\rangle$ for $p_1$ and $\langle\textsc{no}\rangle$ for $p_2$, carrying along the chunk it received (\textcolor{orange!65!black}{small squares}) and its decryption share (\textcolor{teal!60!black}{diamonds}).
On $2f+1$ matching votes per proposer, it builds the \emph{fast meta-block} ($r_1$ with its $2f+1$-\textsc{yes} certificate for $p_1$, $\bot$ with its $2f + 1$-\textsc{no} certificate for $p_2$), casts a fast vote (\textcolor{green!60!black}{green}), and, as we explain later in this subsection, can already \emph{speculatively finalize} the block (\textcolor{violet!60!black}{violet nodes}).
A second round of $2f+1$ fast votes then finalizes the meta-block's entries (\textcolor{green!55!black}{green nodes}).}
\label{fig:fast-path}
\end{figure}
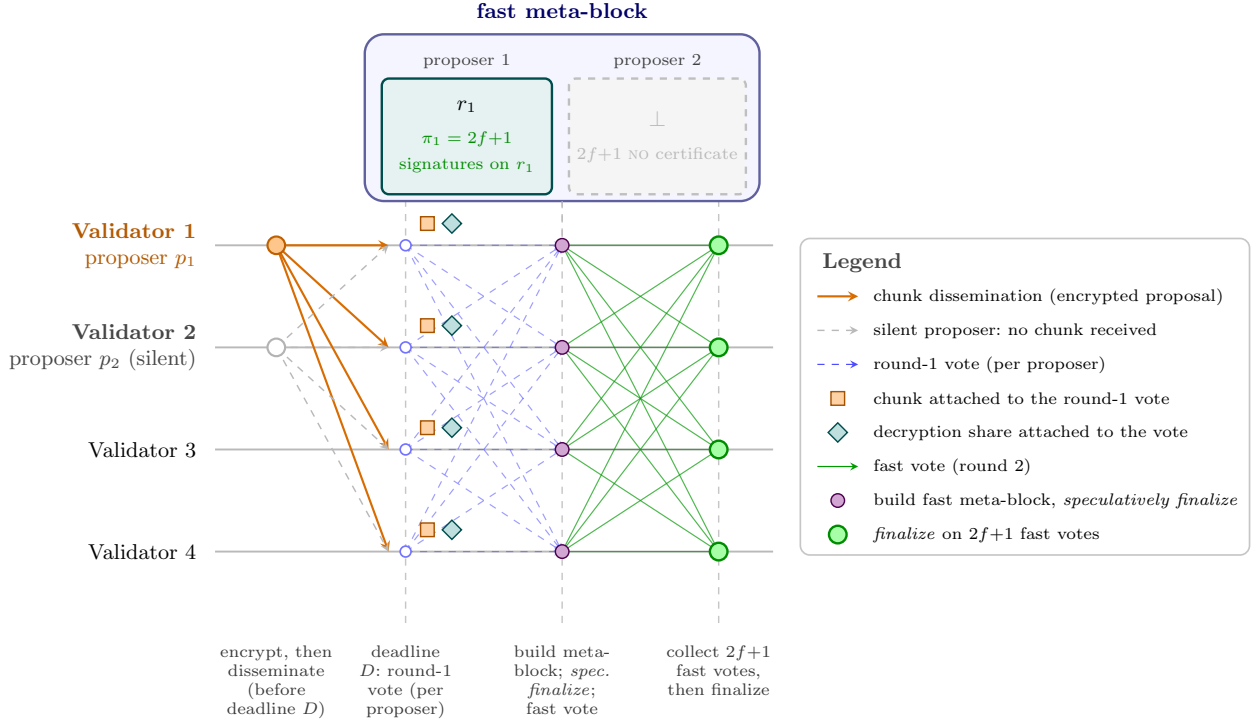

\subsubsection{How Fast is the Fast Path?}
\label{subsubsec:fast-path-speed}
As we emphasized in \Cref{section:introduction}, latency is a first-order concern, so we want the fast path to finalize as quickly as possible.
Under favorable conditions it does: provided the slot occurs after $\mathrm{GST}$, the fast path finalizes in just two voting rounds, by time
\[
D + 2\delta,
\]
where $D$ is the slot's deadline and $\delta \leq \Delta$ is the \emph{actual} message delay.
This is optimal: no protocol can finalize in fewer rounds~\cite{GoodCase-Latency-PODC}.
In addition, \chorus offers a faster but weaker notion of finalization, which applications and users can choose instead of the regular one when they prefer lower latency over the strongest guarantee.
A validator need not wait for both voting rounds: following prior work~\cite{jalalzai2026monadbftfastresponsiveforkresistant}, it may \emph{speculatively finalize} the entries after a single round of voting, as soon as it constructs the fast meta-block. A speculative finalization may ultimately be reverted, but, as we show once the fallback path is in place, only if some validator \emph{equivocated}, disseminating conflicting votes for the same proposer; such deliberate misbehavior is exceedingly rare in practice and can always be detected and punished.

The fast path succeeds whenever the network is synchronous and every proposer is either correct (sending the same Merkle root to all correct validators) or fully silent (sending nothing to any): all correct validators then vote the same way on each proposer (all \textsc{yes}, or all \textsc{no}), a quorum forms, and the slot commits on the fast path. Only partial dissemination or equivocation splits a proposer's votes and forces the fallback. By casting explicit \textsc{no} votes rather than staying silent, validators keep the slot on the fast path even when a proposer is fully offline.

\subsection{Fallback Path}
\label{subsection:fallback-path-overview}

We now turn to the \emph{fallback path}, which validators follow when the fast path fails.

\subsubsection{Description.}
At the end of the first voting round, once the \textsc{yes} and \textsc{no} votes have been disseminated, a correct validator that can assemble a fast meta-block (as in \Cref{subsection:fast-path-overview}) simply issues its fast vote carrying it; as we saw, this happens whenever the network is synchronous and every proposer is correct or silent. A validator that cannot assemble a fast meta-block instead takes the fallback path, in three steps.

\paragraph{First, it announces that it is abandoning the fast path.}
It broadcasts a signed \textsc{no-fast-vote} message, but only after a timeout has elapsed and it has heard first-round votes from at least $2f+1$ validators: by then it has given the fast path enough time and still cannot assemble a fast meta-block, so it times out on it.
It is a signed statement that the validator will not cast a fast vote.
By quorum intersection, a $2f+1$ no-fast-vote quorum and the $2f+1$ fast votes a fast-path commit needs cannot both form; so once $2f+1$ no-fast-votes exist, the fast path can no longer commit, and it is safe to enter the fallback path instead.

\paragraph{Second, it votes on each proposer separately.}
For a proposer $p_j$, it checks whether it has collected at least $f+1$ \textsc{yes} votes on the same Merkle root:
\begin{compactitem}
  \item \emph{If it has}, it can recover $p_j$'s proposal on its own: recall that each \textsc{yes} vote carries its sender's chunk, and any $f+1$ chunks suffice to reconstruct the proposal. (This is exactly why, in the first round, every positive vote carried its chunk.) To be sure the recovered data is genuine, the validator re-encodes it and checks that this reproduces the same Merkle root. If it does, it broadcasts a \textsc{fallback-yes} vote signing that root.
  \item \emph{Otherwise} (there are no $f+1$ \textsc{yes} votes on a common root, or the re-encoding yields a different root), it broadcasts a \textsc{fallback-no} vote for $p_j$.
\end{compactitem}

\paragraph{Third, it finalizes through a black-box agreement run at the end.}
In this second post-deadline round, the validator waits for $2f+1$ messages.
Each is either a fast vote, which carries a fast meta-block, or a bundle of fallback votes, namely the per-proposer \textsc{fallback-yes}/\textsc{fallback-no} votes together with the sender's single \textsc{no-fast-vote}.
Whichever way the round goes, the validator ends it by submitting a single meta-block as its proposal to the slot's \emph{fallback agreement} (see \Cref{fig:fallback-path}).
Here \chorus relies on an off-the-shelf agreement primitive: for each slot, \chorus runs a single instance of such a consensus protocol and treats it purely as a black box.
Any protocol will do, as long as it guarantees agreement (no two correct validators decide differently), termination (every correct validator eventually decides), and external validity (its decision always satisfies a fixed, predetermined predicate).
Which meta-block the validator submits depends on the messages it gathered this round:

\begin{compactitem}
  \item \emph{At least one fast vote is received in the second voting round.} In this case, the validator adopts the fast meta-block that vote carries and proposes it to the fallback agreement.
  (Note that any validator that built a fast meta-block at the end of the first voting round holds its own fast vote, so it always falls into this case and proposes that fast meta-block.)

  \item \emph{All $2f+1$ messages are fallback votes.} The validator instead assembles a \emph{fallback meta-block}, a second kind of meta-block alongside the fast one, and proposes it to the fallback agreement.
  We now describe how a fallback meta-block is built. 
  The validator first aggregates the $2f+1$ received \textsc{no-fast-vote}s into a single certificate; by quorum intersection, the mere existence of this certificate proves that no one could have finalized on the fast path.
  It then fills in one entry per proposer $p_j$:
  \begin{compactitem}
    \item[(i)] if it holds $f+1$ \textsc{fallback-yes} votes on the same root $r$, then $p_j$'s entry is that root $r$, certified by those $f+1$ \textsc{fallback-yes} votes;
    \item[(ii)] if it sees two \textsc{fallback-yes} votes on different roots, then $p_j$'s entry is $\bot$, certified by those two conflicting votes as an \emph{equivocation} proof;\footnote{For this to pin the equivocation on proposer $p_j$, each \textsc{fallback-yes} vote must carry the Merkle root \emph{as signed by $p_j$}; two such votes on different roots then amount to $p_j$'s own signatures on conflicting roots.}
    \item[(iii)] otherwise, $p_j$'s entry is $\bot$, certified by $f+1$ \textsc{fallback-no} votes.
  \end{compactitem}
  In short, the fallback meta-block bundles a single $2f+1$ \textsc{no-fast-vote} certificate with, for each proposer, one of three entries: an $f+1$ \textsc{fallback-yes} certificate on a root, an equivocation proof, or an $f+1$ \textsc{fallback-no} certificate.
\end{compactitem}
The proposed meta-block is either fast or fallback.
The external validity predicate we use in the fallback agreement is therefore that the agreement may decide only a well-formed fast or fallback meta-block for that slot.
The slot is then finalized by whatever meta-block the fallback agreement decides.
\Cref{fig:fallback-path} traces the whole fallback path end to end.

\begin{figure}[h]
\centering
\scalebox{0.78}{%
\begin{tikzpicture}[
  >=stealth,
  vline/.style={gray!55, line width=0.8pt},
  sep/.style={gray!45, dashed, line width=0.6pt},
  fbarr/.style={->, violet!65, dashed, line width=0.5pt, opacity=0.6},
  proparr/.style={->, orange!85!black, line width=0.9pt},
  decarr/.style={->, green!55!black, line width=1.1pt},
  dotE/.style={circle, draw=red!55!black, fill=red!10, line width=0.8pt, inner sep=2pt},
  pdot/.style={circle, draw=violet!60!black, fill=violet!35, line width=0.6pt, inner sep=2pt},
  cbox/.style={draw=blue!45!gray, line width=1.1pt, fill=blue!6, rounded corners=6pt, align=center, font=\small\bfseries, text=blue!40!black},
  decbox/.style={draw=green!55!black, fill=green!12, line width=1pt, rounded corners=4pt, minimum width=1.8cm, minimum height=0.9cm, align=center, font=\small\bfseries, text=green!45!black},
  mbox/.style={draw=blue!25!gray, line width=0.9pt, fill=blue!4, rounded corners=5pt, align=center, font=\scriptsize},
  plab/.style={font=\footnotesize, anchor=east},
  blab/.style={font=\scriptsize, align=center, text width=2.0cm, anchor=north, text=gray!45!black},
  leglab/.style={font=\scriptsize, anchor=west},
]
  \foreach \y in {5,3.5,2,0.5}{ \draw[vline] (0.5,\y) -- (6.25,\y); }
  \node[plab] at (0.35,5) {Validator 1};
  \node[plab] at (0.35,3.5) {Validator 2};
  \node[plab] at (0.35,2) {Validator 3};
  \node[plab] at (0.35,0.5) {Validator 4};

  \foreach \x in {1.9, 4.2}{ \draw[sep] (\x,-0.55) -- (\x,5.7); }

  \foreach \ya in {5,3.5,2,0.5}{ \foreach \yb in {5,3.5,2,0.5}{ \draw[fbarr] (1.9,\ya) -- (4.2,\yb); } }

  \node[cbox, minimum width=2.1cm, minimum height=5.6cm] (CB) at (7.4,2.75) {Black-box\\fallback\\agreement};
  \foreach \y in {5,3.5,2,0.5}{ \draw[proparr] (4.2,\y) -- (6.35,\y); }
  \node[decbox] (DEC) at (9.8,2.75) {FINALIZE};
  \draw[decarr] (CB.east) -- (DEC.west);

  \foreach \y in {5,3.5,2,0.5}{ \node[dotE] at (1.9,\y) {}; }
  \foreach \y in {5,3.5,2,0.5}{ \node[pdot] at (4.2,\y) {}; }

  \draw[draw=gray!45, rounded corners=4pt, fill=white, line width=0.7pt] (11.0,1.45) rectangle (16.9,4.1);
  \node[font=\small\bfseries, text=gray!55!black, anchor=west] at (11.2,3.75) {Legend};
  \node[dotE] at (11.45,3.25) {}; \node[leglab] at (11.8,3.25) {no fast meta-block built};
  \draw[fbarr, opacity=1] (11.25,2.75) -- (11.7,2.75); \node[leglab] at (11.8,2.75) {\textsc{no-fast-vote} $+$ fallback votes};
  \draw[proparr] (11.25,2.25) -- (11.7,2.25); \node[leglab] at (11.8,2.25) {propose to fallback agreement};
  \draw[decarr] (11.25,1.75) -- (11.7,1.75); \node[leglab] at (11.8,1.75) {fallback agreement finalizes};

  \node[blab] at (1.9,-0.75) {end of the first voting round};
  \node[blab] at (4.2,-0.75) {collect $2f{+}1$ second-round votes; obtain a fast or a fallback meta-block};
  \node[blab] at (7.4,-0.75) {propose to fallback agreement};

\end{tikzpicture}%
}
\caption{The \chorus fallback path ($n=4$, $f=1$).
Each validator broadcasts either a fast vote (see \Cref{subsection:fast-path-overview}) or, failing that, a \textsc{no-fast-vote} together with a per-proposer \textsc{fallback-yes}/\textsc{fallback-no} vote (\textcolor{violet!65}{violet, dashed}).
It then proposes a meta-block (fast or fallback; \textcolor{orange!85!black}{orange}) to the slot's single black-box fallback agreement, which finalizes one (\textcolor{green!55!black}{green}) as the slot's.}
\label{fig:fallback-path}
\end{figure}
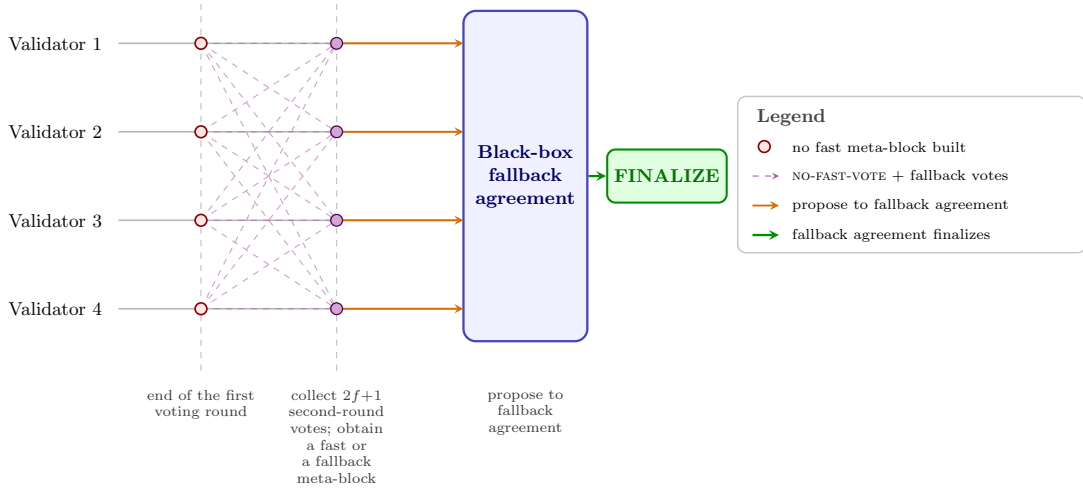

\subsection{From Meta-Blocks to Proposals (and Hence Blocks)}
\label{subsection:meta-to-block}
We have now seen both paths, but neither actually outputs a block of transactions: each finalizes only the entries of a meta-block, and an entry is just a digest (a Merkle root) of a proposal, not the proposal itself.
Producing the slot's block from these entries therefore takes two steps.
First, every correct validator must recover the proposals that the positive entries point to.
Second, it merges these recovered proposals into the slot's block by a deterministic transformation: discarding invalid proposals, deduplicating transactions, and ordering the rest.
This second step is a purely local computation that yields the same block at every correct validator, so the real work lies in the recovery, on which we now focus.

The recovery must ensure that, from a meta-block $B$ whose entries are finalized, every correct validator eventually obtains the same \emph{proposal vector}, one entry per proposer, each either that proposer's recovered proposal or $\bot$ (\Cref{fig:recovery}), subject to the following two guarantees:
\begin{compactitem}[(\emph{ii})]
\item [(\emph{i})] all correct validators recover the same proposal vector, each entry being
either the same recovered proposal at every validator or discarded at every
validator (the \emph{consistency} of \Cref{subsubsection:ObjectofAgreement});
    \item[(\emph{ii})] an honest proposer's proposal is never lost: if $B$ certifies the proposal $P_i$ of a correct proposer $p_i$ (i.e., $B$ contains the root committing to the encryption of $P_i$), then every correct validator recovers exactly $P_i$ in its proposal vector.
\end{compactitem}
Both follow from how a single positive entry of $B$, committing to a root $r$, is recovered.

\subsubsection{Recovery Procedure.}
If $B$ is a fast meta-block, then such an entry carries a $2f+1$ \textsc{yes} certificate (see \Cref{subsection:fast-path-overview}), and a validator casts a \textsc{yes} vote only while holding the chunk it received for $r$; hence, at least $f+1$ correct validators hold their chunk for $r$.
These honest validators re-disseminate their chunks, so every correct validator eventually collects $f+1$ chunks committed under $r$.
Collecting them does not by itself guarantee a recoverable proposal, since a faulty proposer may have committed chunks that do not form a valid erasure encoding; this shows up in one of two ways.
If the validator cannot decode the chunks at all, the encoding is certainly invalid and it discards the proposal.
Otherwise it decodes a candidate ciphertext $\mathcal{D}$ (the root commits to the ciphertext, not the plaintext), but a successful decode still does not guarantee a consistent encoding: the validator re-encodes $\mathcal{D}$ and checks that this reproduces exactly $r$.
If it does, the encoding is valid and the validator decrypts $\mathcal{D}$ with the slot's $f+1$ decryption shares (released at the slot's deadline) to obtain the plaintext proposal.\footnote{A ciphertext that does not decrypt to a well-formed proposal is discarded in the same way; this check, too, is deterministic, so all correct validators agree.}
If the re-encoding yields any other root, the encoding is invalid and the validator discards the proposal.

If instead $B$ is a fallback meta-block, such an entry carries an $f+1$ \textsc{fallback-yes} certificate (see \Cref{subsection:fallback-path-overview}).
At least one of its signers is honest, and an honest validator casts \textsc{fallback-yes} only after reconstructing the proposal and checking that it re-encodes to $r$; as part of casting that vote (\Cref{subsection:fallback-path-overview}) it sends each validator its assigned chunk, and the validators re-broadcast their chunks, so every correct validator obtains the $f+1$ chunks needed to reconstruct the encrypted proposal $\mathcal{D}$ behind $r$.
Recovery then proceeds exactly as in the fast case: re-encode $\mathcal{D}$ to check it against $r$ and, if it matches, decrypt it into the plaintext proposal, otherwise discarding it as invalid.

In both cases, once the root is certified, the certificate
guarantees that $f+1$ honest validators hold their chunks, so every correct
validator can collect $f+1$ chunks and reconstruct the data. This is the
\emph{availability} of \Cref{subsubsection:ObjectofAgreement}, achieved through
\chorus's certificate over the disseminated root.  Whether $B$ is fast or fallback, recovery provides both guarantees:
\begin{compactitem}[(\emph{ii})]
    \item[(\emph{i})] holds because whether a proposal is discarded depends only on $r$: it is discarded if and only if the chunks committed under $r$ do not form a valid erasure encoding, regardless of which chunks a validator happens to collect.
    When they do form one, every correct validator decodes the same ciphertext and, decryption being deterministic, recovers the same plaintext.
    When they do not, every correct validator discards the proposal: even one whose own chunks were individually valid finds, on re-encoding, that they fail to reproduce $r$, since validity is a property of the whole committed set, not of the chunks one happened to receive.

    \item[(\emph{ii})] holds because a correct proposer encrypts a valid proposal and encodes the ciphertext into a valid erasure encoding; the chunks committed under $r$ therefore re-encode to $r$ at every correct validator, so the proposal is always recovered and never discarded.
\end{compactitem}

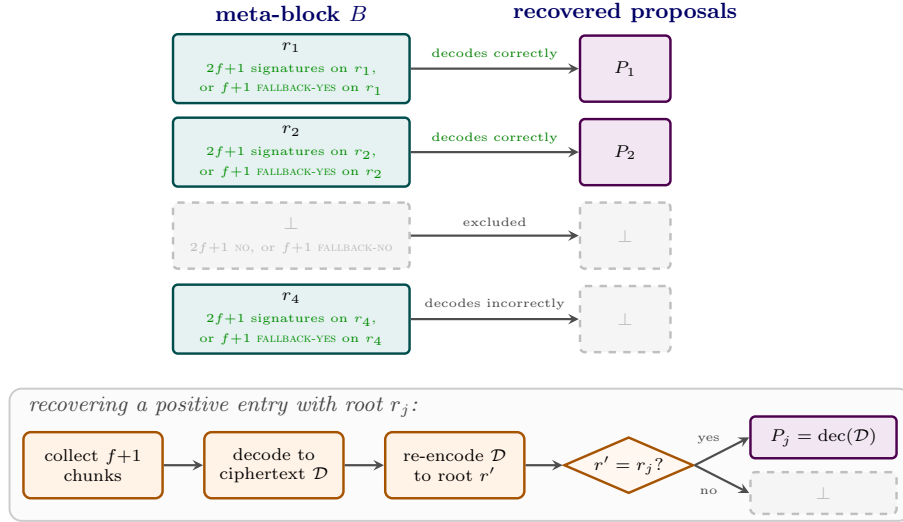
\begin{figure}[h]
\centering
\scalebox{0.92}{%
\begin{tikzpicture}[
  >=stealth,
  cell/.style={draw=teal!60!black, line width=1pt, fill=teal!10, rounded corners=2pt,
    minimum width=3.4cm, minimum height=0.95cm, align=center, font=\scriptsize},
  cellno/.style={draw=gray!50, line width=1pt, fill=gray!8, dashed, rounded corners=2pt,
    minimum width=3.4cm, minimum height=0.95cm, align=center, font=\scriptsize, text=gray!55},
  prop/.style={draw=violet!60!black, line width=1pt, fill=violet!10, rounded corners=2pt,
    minimum width=1.3cm, minimum height=0.95cm, align=center, font=\scriptsize},
  propno/.style={draw=gray!50, line width=1pt, fill=gray!8, dashed, rounded corners=2pt,
    minimum width=1.3cm, minimum height=0.95cm, align=center, font=\scriptsize, text=gray!55},
  flow/.style={draw=orange!65!black, line width=1pt, fill=orange!10, rounded corners=3pt,
    minimum width=2.0cm, minimum height=0.95cm, align=center, font=\scriptsize},
  dec/.style={draw=orange!65!black, line width=1pt, fill=orange!6, diamond, aspect=2.3,
    inner sep=1pt, align=center, font=\scriptsize},
  arr/.style={->, gray!60!black, line width=0.9pt},
  ttl/.style={font=\footnotesize\bfseries, text=blue!40!black},
]
  \node[ttl] at (1.95,2.75) {meta-block $B$};
  \node[cell]   (a1) at (1.95,1.95)  {$r_1$ \\[1pt] {\tiny\textcolor{green!55!black}{$2f{+}1$ signatures on $r_1$,}} \\ {\tiny\textcolor{green!55!black}{or $f{+}1$ \textsc{fallback-yes} on $r_1$}}};
  \node[cell]   (a2) at (1.95,0.75)  {$r_2$ \\[1pt] {\tiny\textcolor{green!55!black}{$2f{+}1$ signatures on $r_2$,}} \\ {\tiny\textcolor{green!55!black}{or $f{+}1$ \textsc{fallback-yes} on $r_2$}}};
  \node[cellno] (a3) at (1.95,-0.45) {$\bot$ \\[1pt] {\tiny $2f{+}1$ \textsc{no}, or $f{+}1$ \textsc{fallback-no}}};
  \node[cell]   (a4) at (1.95,-1.65) {$r_4$ \\[1pt] {\tiny\textcolor{green!55!black}{$2f{+}1$ signatures on $r_4$,}} \\ {\tiny\textcolor{green!55!black}{or $f{+}1$ \textsc{fallback-yes} on $r_4$}}};

  \node[ttl] at (6.75,2.75) {recovered proposals};
  \node[prop]   (b1) at (6.75,1.95)  {$P_1$};
  \node[prop]   (b2) at (6.75,0.75)  {$P_2$};
  \node[propno] (b3) at (6.75,-0.45) {$\bot$};
  \node[propno] (b4) at (6.75,-1.65) {$\bot$};

  \draw[arr] (a1.east) -- (b1.west) node[midway, above, font=\tiny, green!50!black] {decodes correctly};
  \draw[arr] (a2.east) -- (b2.west) node[midway, above, font=\tiny, green!50!black] {decodes correctly};
  \draw[arr] (a3.east) -- (b3.west) node[midway, above, font=\tiny, gray!55!black] {excluded};
  \draw[arr] (a4.east) -- (b4.west) node[midway, above, font=\tiny, gray!55!black] {decodes incorrectly};

  \draw[draw=gray!40, line width=0.8pt, rounded corners=5pt, fill=gray!3]
    (-2.1,-4.55) rectangle (10.85,-2.65);
  \node[font=\footnotesize\itshape, gray!55!black, anchor=west] at (-1.95,-2.9)
    {recovering a positive entry with root $r_j$:};
  \node[flow] (s1) at (-0.9,-3.75) {collect $f{+}1$\\ chunks};
  \node[flow] (s2) at (1.7,-3.75)  {decode to\\ ciphertext $\mathcal{D}$};
  \node[flow] (s3) at (4.3,-3.75)  {re-encode $\mathcal{D}$\\ to root $r'$};
  \node[dec]  (s4) at (6.8,-3.75)  {$r' = r_j$?};
  \draw[arr] (s1) -- (s2);
  \draw[arr] (s2) -- (s3);
  \draw[arr] (s3) -- (s4);
  \node[prop, minimum width=2.1cm, minimum height=0.6cm]   (oy) at (9.6,-3.35) {$P_j = \mathrm{dec}(\mathcal{D})$};
  \node[propno, minimum width=2.1cm, minimum height=0.6cm] (on) at (9.6,-4.15) {$\bot$};
  \draw[arr] (s4.east) -- (oy.west);
  \draw[arr] (s4.east) -- (on.west);
  \node[font=\tiny, gray!55!black] at (7.95,-3.38) {yes};
  \node[font=\tiny, gray!55!black] at (7.95,-4.12) {no};
\end{tikzpicture}%
}
\caption{Recovering the proposals from a (fast or fallback) meta-block.\lioba{Some text in the figure is very small}}
\label{fig:recovery}
\end{figure}

\subsubsection{Proposal Recovery Adds No Latency.} \jovan{I again think it is fine to just focus on the fact that the fast path exhibits no additional latency. All of these things we do are extremely standard and I do not think we should be spelling them out (might add just more meat and confusion).}
Proposal recovery runs in parallel with voting, so it adds no latency.
The latency claims of \Cref{subsection:fast-path-overview} (speculative finalization after one round, full finalization after two) are about the fast path and were stated for the meta-block; they carry over to the full block of transactions.
On the fast path, every first-round positive vote already carries the sender's chunk, and every validator bundles its slot decryption share into the same first-round message, so a validator that assembles a fast meta-block at the end of the first round holds, behind each positive entry, $2f+1$ chunks (more than the $f+1$ needed to reconstruct the ciphertext) and $2f+1$ decryption shares (more than the $f+1$ needed to decrypt it), and recovers the proposal with no further communication.
The fallback path adds no extra latency from recovery either. A positive entry certified by $2f+1$ \textsc{yes} votes is recovered exactly as on the fast path. For one certified by $f+1$ \textsc{fallback-yes} votes, a fallback-yes voter sent each validator its assigned chunk when voting (\Cref{subsection:fallback-path-overview}), so every correct validator holds its chunk for $r$; once the fallback meta-block forms, each broadcasts that chunk, and one $\Delta$ later every correct validator can decode the proposal, before the fallback path finalizes.
Hiding adds no latency either: the slot's $f+1$ decryption shares are released at the deadline, bundled with the first-round votes, and open every ciphertext for the slot, so a validator decrypts a proposal as soon as it has decoded it.

\subsection{Proof Sketch}
\label{subsection:chorus-proof}
Finally, we sketch why \chorus satisfies the four properties of slot consensus (see \Cref{fig:slot-consensus-overview}).

\paragraph{Agreement.}
We must show that no two correct validators finalize different blocks. By recovery's consistency guarantee~(\emph{i}), agreeing on the same entries yields the same block, so it suffices to show that no two correct validators finalize different \emph{entries}.
On the fast path this is immediate: each entry carries a certificate of $2f+1$ signatures, so, by quorum intersection, no proposer can have two conflicting certified entries.
On the fallback path it follows from the fallback agreement, whose agreement property makes all correct validators decide the same meta-block, and hence the same entries.
It remains to rule out a split in which one validator finalizes on the fast path and another finalizes different entries on the fallback path.
But if any validator finalizes on the fast path, it holds a quorum of $2f+1$ fast votes, and, by quorum intersection, no $2f+1$ \textsc{no-fast-vote} certificate can then exist; hence no fallback meta-block can be formed, and the fast-path entries are the only ones anyone can finalize.

\paragraph{Termination.}
Every correct validator eventually finalizes, and it suffices to show that every correct validator proposes a meta-block to the fallback agreement.
We distinguish two cases.
If some correct validator assembles a fast meta-block at the end of the first round (and broadcasts its fast vote), then every correct validator eventually receives it, adopts the (unique) fast meta-block, and proposes it.
Otherwise, every correct validator eventually collects $2f+1$ fallback bundles, builds a fallback meta-block, and proposes that.
In either case all correct validators propose to the consensus, so, by its termination, it decides a meta-block, by which the slot is finalized.

\paragraph{Proposal Inclusion.}
Consider an honest proposer that submits its proposal on time ($\Delta$ time before the slot's deadline), while the network is synchronous, and let the proposal commit to root $r$.
Every correct validator then receives a valid chunk by the deadline and votes \textsc{yes} on root $r$, so no correct validator votes \textsc{no} for this proposer.
A fast meta-block can therefore only record $r$ for this proposer, never excluding the proposal.
The fallback path includes it just as well: every correct validator holds $f+1$ \textsc{yes} votes on $r$, reconstructs the proposal, and re-encodes it successfully (an honest proposal is correctly encoded), so it sends a \textsc{fallback-yes} on $r$ and never a \textsc{fallback-no}.
So no $f+1$ \textsc{fallback-no} votes can ever form against it either.
In both cases, then, the finalized entry for this proposer certifies the proposal under root $r$ rather than excluding it; by recovery's guarantee~(\emph{ii}), every correct validator therefore recovers the proposal into the slot's block, which ensures proposal inclusion.

\paragraph{Hiding.}
Proposals are revealed only at the deadline, by which point it is too late for the adversary to craft a competing proposal and still have it included, so it cannot make its own proposal depend on the honest proposers' proposals for the slot.

\paragraph{Safety of speculative finalization.}
We claimed in \Cref{subsection:fast-path-overview} that a speculatively finalized block can be reverted only if some validator equivocated; here is why.
Suppose a correct validator speculatively finalizes the fast meta-block $B$, but the fallback consensus ultimately decides a different meta-block $B'$, so that the two assign some proposer $p_j$ a different proposal in the slot's block.
Since $B$ is a fast meta-block, it backs $p_j$'s entry with a first-round quorum: $2f+1$ \textsc{yes} votes on one root $r$, or $2f+1$ \textsc{no} votes.
There are two ways $B'$ can disagree:
\begin{compactitem}
  \item \emph{Conflicting roots.} Say $B$ records a Merkle root $r$ for $p_j$ while $B'$ records a different root $r'$.
A positive entry always rests on first-round \textsc{yes} votes for its root --- $2f+1$ of them in a fast meta-block, and the $f+1$ that every \textsc{fallback-yes} requires its caster to have collected.
Since $(2f+1) + (f+1) > n = 3f+1$, the \textsc{yes}-voters for $r$ and those for $r'$ must overlap; as a correct validator \textsc{yes}-votes a single root, the validator in the overlap signed both --- an equivocation.

\item \emph{Inclusion vs. exclusion.} Say instead that $B$ includes $p_j$ under a root $r$ while $B'$ leaves $p_j$ out.
Including $p_j$ rests on $2f+1$ first-round \textsc{yes} votes on $r$, since $B$ is a fast meta-block.
Excluding $p_j$ requires a negative entry carrying $f+1$ \textsc{fallback-no} votes, at least one of them from an honest validator --- and an honest \textsc{fallback-no} is cast after its sender has collected $f+1$ first-round \textsc{no} votes.
Since $(2f+1) + (f+1) > n = 3f+1$, the two sets overlap, and the validator in the overlap cast both a \textsc{yes} on $r$ and a \textsc{no} for $p_j$ in the first round --- contradictory votes for the same proposer, once again an equivocation.
(There are two further ways an honest validator casts a \textsc{fallback-no} for $p_j$: it gathers $f+1$ \textsc{yes} votes for a root $r$ whose chunks fail to re-encode to $r$, or it receives conflicting \textsc{yes} votes and never gathers $f+1$ on any single root. Either way the proposer is the culprit --- committing to an invalidly encoded root or disseminating several distinct proposals --- misbehavior all the same.)

\end{compactitem}

\section{\conductor: Our Orchestrator}
\label{section:conductor-overview}

This section introduces \conductor, the orchestrator \cadence employs.
It fixes each window's deadlines in advance, so block production needs no per-slot coordination; its one real decision is when to open the next window, which it makes by having the validators agree, once per window and off the critical path, on the next window's first deadline, from which its remaining deadlines follow $\tau$ apart, where $\tau$ is the block interval.

\subsection{Protocol Description}

In order to bound the number of open slots, \conductor groups slots into windows of $W$ consecutive slots: window $1$ is slots $1, \ldots, W$, window $2$ is slots $W + 1, \ldots, 2W$, and so on.
It schedules these windows one at a time, and the schedule of each depends on the progress of the earlier ones: window $1$ is scheduled at \emph{genesis} (global time $0$), and window $\omega + 1$ is scheduled only once correct validators report that \emph{every} slot of windows $1, \ldots, \omega - 1$ and the first $p$ slots of window $\omega$ are complete. This is what keeps the number of open slots bounded as window $\omega+1$ is withheld until the first $p$ slots of window $\omega$ are complete.
Recall the two parameters of \conductor: the window size $W$ and the threshold $p \in \{0, \dots, W-1\}$.

We now describe the protocol from the perspective of a single correct validator $p_i$; recall that the validators' clocks are synchronized, so they all share a single global timeline.
At genesis, $p_i$ schedules window $1$ by issuing its slots' $\tau$-spaced deadlines: slot $1$ at $\Delta$ (so its starting time is $\Delta - \Delta = 0$), slot $2$ at $\Delta + \tau$, and so on up to slot $W$ at $\Delta + (W-1)\tau$.
From then on, the only input that drives \conductor is the \emph{completion} of slots.
(In \cadence, as discussed in \Cref{subsection:composition-overview}, $p_i$ marks a slot complete once it has finalized that slot's block.)

\paragraph{Scheduling the next window.}
This is the crucial and most technical part of \conductor.
We explain it for the transition from window $1$ to window $2$; every later transition, from a window $\omega$ to $\omega + 1$, works the same way (up to one minor change that we flag below).

As soon as $p_i$ has completed all of the first $p$ slots, it reads the current time $\mathcal{T}$ and computes a \emph{proposed deadline} for the \emph{first} slot of window $2$; the remaining slots of window $2$ are then simply $\tau$-spaced after it, so this single deadline fixes the whole window.
This is only $p_i$'s own proposal. Since correct validators complete the first $p$ slots at different times, they compute different proposed deadlines and then reconcile them into a single agreed deadline through a sub-protocol described below; we first explain how $p_i$ forms its proposal.
Validator $p_i$ forms its proposed deadline as follows (see \Cref{fig:conductor-schedule}):
\begin{compactitem}
  \item If $\mathcal{T}$ is earlier than the deadline of window $1$'s last slot (slot $W$) --- i.e., $p_i$ completed the first $p$ slots before slot $W$'s deadline had even arrived --- then the system is keeping up, and $p_i$ proposes the deadline of slot $W+1$ to be exactly $\tau$ after the deadline of slot $W$.
  In other words, $p_i$ proposes that the global cadence runs on unbroken, with no gap between the two windows.
  \item Otherwise, the first $p$ slots completed only after slot $W$'s deadline had already passed; $p_i$ reads this as the system falling behind, and its proposed deadline for slot $W+1$ is the current time $\mathcal{T}$ --- a proposal that, being later than the cadence point, leaves a gap for the lagging slots to catch up.
\end{compactitem}
These two cases capture exactly what \conductor wants to achieve: it preserves the steady $\tau$-cadence whenever the system is ``healthy'' (\conductor's recovery property, which yields the eventual stability of \cadence), and stretches the schedule when it is not (\conductor's boundedness property, which caps the slots in flight for each validator in \cadence).

\begin{figure}[h]
\centering
\scalebox{0.86}{%
\begin{tikzpicture}[
  >=stealth,
  slot/.style={draw=teal!60!black, line width=1pt, fill=teal!10, rounded corners=2pt,
    minimum width=0.9cm, minimum height=0.55cm, align=center, font=\scriptsize},
  slotp/.style={draw=orange!65!black, line width=1pt, fill=orange!12, rounded corners=2pt,
    minimum width=0.9cm, minimum height=0.55cm, align=center, font=\scriptsize},
  win/.style={draw=blue!25!gray, line width=1pt, fill=blue!4, rounded corners=6pt},
  axis/.style={->, gray!55, line width=0.9pt},
  meas/.style={<->, gray!60!black, line width=0.6pt},
  tline/.style={violet!60!black, dashed, line width=0.7pt},
  wlab/.style={font=\scriptsize\bfseries, text=blue!40!black},
  rlab/.style={font=\footnotesize\bfseries, text=gray!45!black, anchor=east},
  dlab/.style={font=\tiny, gray!55!black},
  note/.style={font=\scriptsize\itshape, gray!50!black},
]
\begin{scope}[yshift=2.7cm]
  \draw[win] (0.35,-0.4) rectangle (4.65,0.4);
  \draw[win] (4.85,-0.4) rectangle (9.15,0.4);
  \node[wlab] at (2.5,0.95) {window $1$};
  \node[wlab] at (7.0,0.95) {window $2$};
  \draw[axis] (0.0,0) -- (9.7,0);
  \node[note,anchor=west] at (9.45,0) {time};
  \node[slotp] (a1) at (1.0,0) {$s_1$};
  \node[slotp] (a2) at (2.5,0) {$s_2$};
  \node[slot]  (a3) at (4.0,0) {$s_3$};
  \node[slot]  (a4) at (5.5,0) {$s_4$};
  \node[slot]  (a5) at (7.0,0) {$s_5$};
  \node[slot]  (a6) at (8.5,0) {$s_6$};
  \node[dlab] at (1.0,-0.6) {$\Delta$};
  \node[dlab] at (2.5,-0.6) {$\Delta{+}\tau$};
  \node[dlab] at (4.0,-0.6) {$\Delta{+}2\tau$};
  \node[dlab] at (5.5,-0.6) {$\Delta{+}3\tau$};
  \node[dlab] at (7.0,-0.6) {$\Delta{+}4\tau$};
  \node[dlab] at (8.5,-0.6) {$\Delta{+}5\tau$};
  \draw[tline] (3.25,0.55) -- (3.25,-0.78);
  \node[note,text=violet!60!black] at (3.25,0.7) {$\mathcal{T}$};
  \draw[meas] (4.0,-0.92) -- (5.5,-0.92);
  \node[note] at (4.75,-1.12) {$\tau$: no gap};
  \node[rlab] at (-0.05,0) {healthy};
\end{scope}
\begin{scope}
  \draw[win] (0.35,-0.4) rectangle (4.65,0.4);
  \draw[win] (5.75,-0.4) rectangle (10.05,0.4);
  \node[wlab] at (2.5,0.95) {window $1$};
  \node[wlab] at (7.9,0.95) {window $2$};
  \draw[axis] (0.0,0) -- (10.6,0);
  \node[note,anchor=west] at (10.35,0) {time};
  \draw[tline] (6.4,0.55) -- (6.4,-0.35);
  \node[slotp] (b1) at (1.0,0) {$s_1$};
  \node[slotp] (b2) at (2.5,0) {$s_2$};
  \node[slot]  (b3) at (4.0,0) {$s_3$};
  \node[slot]  (b4) at (6.4,0) {$s_4$};
  \node[slot]  (b5) at (7.9,0) {$s_5$};
  \node[slot]  (b6) at (9.4,0) {$s_6$};
  \node[dlab] at (1.0,-0.6) {$\Delta$};
  \node[dlab] at (2.5,-0.6) {$\Delta{+}\tau$};
  \node[dlab] at (4.0,-0.6) {$\Delta{+}2\tau$};
  \node[dlab] at (6.4,-0.6) {$\mathcal{T}$};
  \node[dlab] at (7.9,-0.6) {$\mathcal{T}{+}\tau$};
  \node[dlab] at (9.4,-0.6) {$\mathcal{T}{+}2\tau$};
  \draw[meas] (4.0,-0.92) -- (6.4,-0.92);
  \node[note] at (5.2,-1.12) {gap};
  \node[rlab] at (-0.05,0) {lagging};
\end{scope}
\end{tikzpicture}%
}
\caption{How a single validator $p_i$ computes its \emph{proposed deadline} for window $2$'s first slot ($W=3$, $p=2$).
\emph{Healthy} (top): no gap; \emph{lagging} (bottom): gap exists.}
\label{fig:conductor-schedule}
\end{figure}
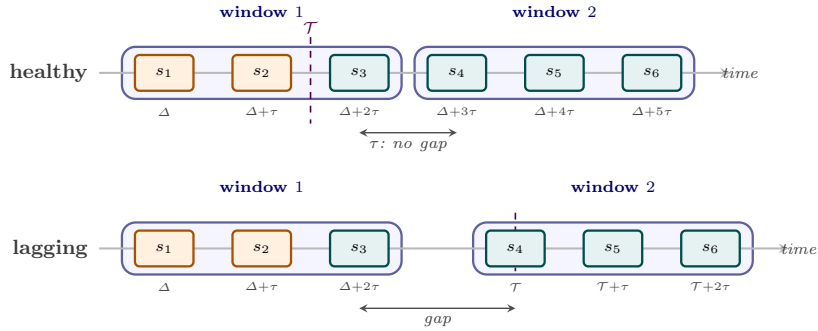

Note that different correct validators may complete the first $p$ slots at slightly different times, and so arrive at different proposed deadlines for window $2$; they must nonetheless agree on a single one.
To this end, the validators feed their proposed deadlines into this off-the-shelf ACS primitive, which outputs a common vector of $2f+1$ proposed deadlines --- at least $f+1$ of them from honest validators --- and guarantees agreement (all correct validators obtain the same vector) and termination (every correct validator eventually obtains one).
Each validator then takes the \emph{median} of the vector's entries as the agreed deadline of slot $W+1$, and opens window $2$ accordingly.
Taking the median is what tames Byzantine influence: among $2f+1$ values of which at most $f$ are Byzantine, the median is guaranteed to lie between the smallest and the largest honest proposed deadline, so faulty validators can never drag the deadline outside the range the honest ones proposed (\Cref{fig:conductor-acs}).

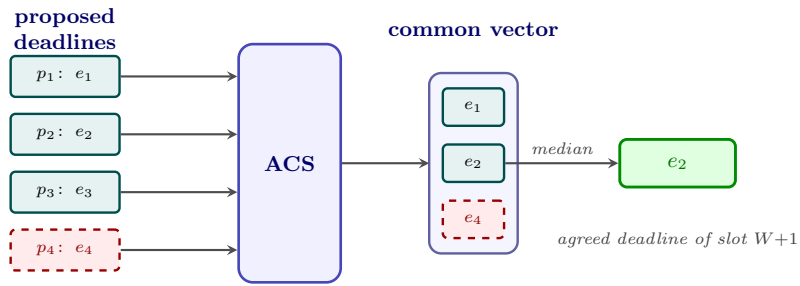
\begin{figure}[h]
\centering
\scalebox{0.9}{%
\begin{tikzpicture}[
  >=stealth,
  est/.style={draw=teal!60!black, line width=1pt, fill=teal!10, rounded corners=2pt,
    minimum width=1.6cm, minimum height=0.6cm, align=center, font=\scriptsize},
  estb/.style={draw=red!60!black, line width=1pt, fill=red!8, dashed, rounded corners=2pt,
    minimum width=1.6cm, minimum height=0.6cm, align=center, font=\scriptsize, text=red!55!black},
  acs/.style={draw=blue!45!gray, line width=1.1pt, fill=blue!6, rounded corners=6pt,
    minimum width=1.5cm, minimum height=3.5cm, align=center, font=\small\bfseries, text=blue!40!black},
  vcon/.style={draw=blue!25!gray, line width=1pt, fill=blue!4, rounded corners=5pt},
  vc/.style={draw=teal!60!black, line width=1pt, fill=teal!10, rounded corners=2pt,
    minimum width=0.9cm, minimum height=0.55cm, align=center, font=\scriptsize},
  vcb/.style={draw=red!60!black, line width=1pt, fill=red!8, dashed, rounded corners=2pt,
    minimum width=0.9cm, minimum height=0.55cm, align=center, font=\scriptsize, text=red!55!black},
  med/.style={draw=green!55!black, line width=1.2pt, fill=green!12, rounded corners=3pt,
    minimum width=1.7cm, minimum height=0.7cm, align=center, font=\small\bfseries, text=green!45!black},
  arr/.style={->, gray!60!black, line width=0.9pt},
  ttl/.style={font=\footnotesize\bfseries, text=blue!40!black},
  note/.style={font=\scriptsize\itshape, gray!50!black},
]
  \node[ttl,align=center] at (0,3.25) {proposed\\deadlines};
  \node[est]  (e1) at (0,2.55) {$p_1\!:\ e_1$};
  \node[est]  (e2) at (0,1.70) {$p_2\!:\ e_2$};
  \node[est]  (e3) at (0,0.85) {$p_3\!:\ e_3$};
  \node[estb] (e4) at (0,0.00) {$p_4\!:\ e_4$};
  \node[acs] (acs) at (3.3,1.275) {ACS};
  \foreach \i in {1,2,3,4}{ \draw[arr] (e\i.east) -- (acs.west |- e\i); }
  \draw[vcon] (5.35,-0.05) rectangle (6.65,2.6);
  \node[ttl] at (6.0,3.25) {common vector};
  \node[vc]  (o1) at (6.0,2.1)  {$e_1$};
  \node[vc]  (o2) at (6.0,1.275){$e_2$};
  \node[vcb] (o3) at (6.0,0.45) {$e_4$};
  \draw[arr] (acs.east) -- (5.35,1.275);
  \node[med] (D) at (9.0,1.275) {$e_2$};
  \draw[arr] (o2.east) -- (D.west) node[midway,above,note] {median};
  \node[note,align=center,text width=3.6cm] at (9.0,0.15) {agreed deadline of slot $W{+}1$};
\end{tikzpicture}%
}
\caption{How validators agree on the deadline ($n=4$, $f=1$; $p_4$ \textcolor{red!60!black}{Byzantine}).
An ACS instance delivers all correct validators the \emph{same} vector of $2f{+}1$ proposed deadlines, of which each correct validator takes the \emph{median} as the agreed deadline $e_2$.
Since at most $f$ entries are Byzantine, the median always lies between the smallest and largest honest proposed deadline in the common vector.}
\label{fig:conductor-acs}
\end{figure}

Every later window is handled the same way, with the one change flagged above: $p_i$ must \emph{also} wait for all earlier windows to be complete (a requirement that is vacuous for the transition from window $1$ to window $2$). Concretely, once window $\omega$'s prerequisites are met --- the slots of all earlier windows and the first $p$ slots of window $\omega$ are complete --- validator $p_i$ runs the same two-case computation (now against window $\omega$'s last deadline) and agrees on it through ACS and the median.
This finally yields the bound promised above: when window $\omega+1$ opens, only the unfinished part of window $\omega$ and the fresh window $\omega+1$ can be in flight, so at most $(W - p) + W = 2W - p$ slots are underway at any time.

\subsection{Proof Sketch}

We now argue that \conductor achieves the two properties specified in \Cref{subsection:orchestrator-overview}.

\paragraph{Boundedness.}
We already did the counting at the end of the previous subsection: a correct validator opens window $\omega+1$ only once every slot of windows $1, \ldots, \omega-1$ and the first $p$ slots of window $\omega$ are complete, so at most $2W - p$ slots are ever underway at once.

\paragraph{Recovery.}
Recovery asks that, within a bounded grace period after the network stabilizes, the schedule ``heals'': deadlines become $\tau$-spaced again, with no gaps, and every correct validator learns each deadline at least $\Delta$ in advance. We sketch why this happens within bounded time after $\mathrm{GST}$.

Everything hinges on one question: once the network is running smoothly (after $\mathrm{GST}$), does such a ``good'' window leave enough time to schedule its successor before its own cadence ends?
A window lasts $W\tau$, and within it two things must happen (\Cref{fig:conductor-recovery}).
To begin, the first $p$ slots must \emph{complete}.
\chorus finalizes a slot within some latency $\ell_{\textsc{Chorus}}$ after $\mathrm{GST}$, so every correct validator finalizes all of the first $p$ slots by time $\mathcal{T}_p + \ell_{\textsc{Chorus}}$ --- the $p$-th slot's deadline $\mathcal{T}_p$ plus that latency --- and by then \emph{all} of them have proposed to the window's ACS instance.
Second, since every correct validator has proposed, the ACS decides within a further $\ell_{\textsc{acs}}$ time, fixing the next window's first deadline.
We pick the parameters so that this whole pipeline fits inside the window, i.e., $\mathcal{T}_p + \ell_{\textsc{Chorus}} + \ell_{\textsc{acs}} \le W\tau$.
Thus, in a good window, every correct validator opens the first $p$ slots (outputs their deadlines) by $\mathcal{T}_p$ and they all settle the next window's deadline by $\mathcal{T}_p + \ell_{\textsc{Chorus}} + \ell_{\textsc{acs}} \le W\tau$ --- with no gap and well ahead of time. This is exactly what we needed: from the first good window onward, the schedule stays ``smooth'' forever.


\begin{figure}[h]
\centering
\scalebox{0.9}{%
\begin{tikzpicture}[
  >=stealth,
  slot/.style={draw=teal!60!black, line width=1pt, fill=teal!10, rounded corners=2pt,
    minimum width=0.7cm, minimum height=0.5cm},
  slotp/.style={draw=orange!65!black, line width=1pt, fill=orange!14, rounded corners=2pt,
    minimum width=0.7cm, minimum height=0.5cm},
  win/.style={draw=blue!25!gray, line width=1pt, fill=blue!4, rounded corners=6pt},
  axis/.style={->, gray!55, line width=0.9pt},
  cbar/.style={draw=orange!65!black, line width=1pt, fill=orange!16, rounded corners=2pt,
    minimum height=0.5cm, align=center, font=\scriptsize, anchor=west},
  abar/.style={draw=blue!45!gray, line width=1pt, fill=blue!12, rounded corners=2pt,
    minimum height=0.5cm, align=center, font=\scriptsize, anchor=west},
  tline/.style={violet!60!black, dashed, line width=0.7pt},
  meas/.style={<->, gray!60!black, line width=0.6pt},
  note/.style={font=\scriptsize\itshape, gray!50!black},
  wlab/.style={font=\scriptsize\bfseries, text=blue!40!black},
]
  \draw[win] (-0.45,-0.4) rectangle (7.45,0.4);
  \node[wlab] at (3.5,1.25) {window $\omega$};
  \draw[axis] (-0.7,0) -- (10.5,0);
  \node[note,anchor=west] at (10.65,0) {time};
  \foreach \x in {0,1,2,3}{ \node[slotp] at (\x,0) {}; }
  \foreach \x in {4,5,6,7}{ \node[slot]  at (\x,0) {}; }
  \draw[meas] (0,0.6) -- (3,0.6);
  \node[note] at (1.5,0.79) {first $p$ slots};
  \node[slot, fill=teal!22] at (8,0) {};
  \node[wlab] at (8.0,1.25) {window $\omega{+}1$};
  \node[note, text width=1.9cm, align=center] at (9.9,0.95) {gapless: $\tau$-cadence continues};
  \draw[tline] (3,0.4) -- (3,-1.6);
  \node[note,text=violet!60!black,anchor=north] at (3,-1.62) {$\mathcal{T}_p$};
  \node[cbar, minimum width=2cm] (cb) at (3,-1.0) {$\ell_{\textsc{Chorus}}$};
  \node[abar, minimum width=3cm] (ab) at (5,-1.0) {$\ell_{\textsc{acs}}$};
  \draw[tline] (5,-0.75) -- (5,-1.6);
  \node[note,text=violet!60!black,anchor=north,align=center,text width=3.0cm] at (5,-1.62) {all complete the first $p$ slots \\ $\Rightarrow$ all propose to ACS};
  \draw[tline] (8,0.4) -- (8,-1.6);
  \node[note,text=violet!60!black,anchor=north] at (8,-1.62) {next deadline decided};
\end{tikzpicture}%
}
\caption{Scheduling the next window inside a ``good'' (post-$\mathrm{GST}$) window.
By time $\mathcal{T}_p + \ell_{\textsc{Chorus}}$ --- the $p$-th slot's deadline $\mathcal{T}_p$ plus \chorus's latency (\textcolor{orange!70!black}{orange}) --- every correct validator has completed the first $p$ slots and proposed to the window's ACS, which then decides within a further $\ell_{\textsc{acs}}$ (\textcolor{blue!45!gray}{blue}), fixing the next window's first deadline.
Whenever $\mathcal{T}_p + \ell_{\textsc{Chorus}} + \ell_{\textsc{acs}} \le W\tau$ (shown here at the tight extreme), that deadline is fixed before the cadence ends, so window $\omega{+}1$ continues the $\tau$-cadence with no gap.}
\label{fig:conductor-recovery}
\end{figure}
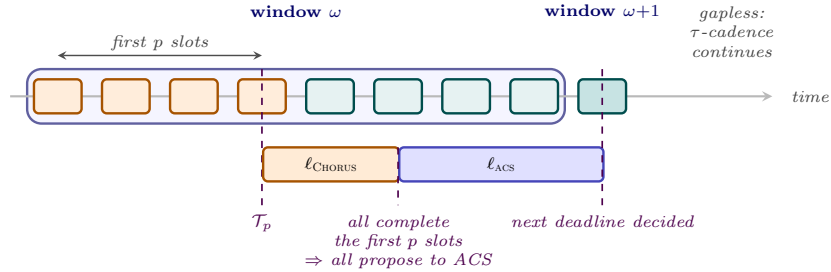

%% file: figures/overview-tikz.tex
\begin{tikzpicture}[
  >=stealth,
  propbox/.style={draw=violet!60!black, line width=1pt, fill=violet!8,
    rounded corners=3pt, minimum width=1.8cm, minimum height=0.85cm,
    align=center, font=\scriptsize},
  byzbox/.style={draw=red!60!black, line width=1pt, fill=red!8,
    rounded corners=3pt, minimum width=1.8cm, minimum height=0.85cm,
    align=center, font=\scriptsize},
  inclbox/.style={draw=violet!60!black, line width=0.8pt, fill=violet!10,
    rounded corners=2pt, minimum width=1.55cm, minimum height=0.45cm,
    align=center, font=\tiny},
  ledgerblock/.style={draw=orange!65!black, line width=0.8pt, fill=orange!10,
    rounded corners=2pt, minimum width=1.1cm, minimum height=0.5cm,
    align=center, font=\scriptsize},
  pendingblock/.style={draw=gray!50, line width=0.8pt, fill=gray!5, dashed,
    rounded corners=2pt, minimum width=1.1cm, minimum height=0.5cm,
    align=center, font=\scriptsize, text=gray!55},
  proparr/.style={->, violet!60!black, thick},
  ptr/.style={gray!55, dashed, line width=0.6pt},
  tick/.style={gray!70!black, line width=0.9pt},
  propscard/.style={draw=gray!55, line width=0.9pt, fill=gray!4,
    rounded corners=6pt, align=left, font=\scriptsize, inner sep=9pt,
    text width=11.9cm},
  lock/.pic={
    \draw[teal!60!black, fill=teal!25, line width=0.9pt] (-0.11,-0.11) rectangle (0.11,0.06);
    \draw[teal!60!black, line width=1pt] (-0.065,0.06) -- (-0.065,0.125) arc (180:0:0.065) -- (0.065,0.06);
  },
]
  \draw[draw=blue!25!gray, line width=1pt, fill=blue!4, rounded corners=8pt]
    (0.6,1.45) rectangle (6.1,5.0);
  \node[font=\small\bfseries, text=blue!40!black, anchor=south west] at (0.75,5.05) {slot $s$};
  \node[propbox] (p1) at (1.85,4.3) {proposer $p_1$};
  \node[propbox] (p2) at (1.85,3.2) {proposer $p_2$};
  \node[byzbox]  (p3) at (1.85,2.1) {proposer $p_3$\\ (Byzantine)};
  \draw[draw=orange!65!black, line width=1.2pt, fill=orange!8, rounded corners=4pt]
    (4.0,2.45) rectangle (5.9,4.1);
  \node[font=\scriptsize\bfseries, text=orange!65!black] at (4.95,3.85) {block $B_s$};
  \node[inclbox] (q1) at (4.95,3.4) {proposal of $p_1$};
  \node[inclbox] (q2) at (4.95,2.8) {proposal of $p_2$};
  \draw[proparr] (p1.east) -- (4.0,3.69);
  \pic at (3.38,3.99) {lock};
  \draw[proparr] (p2.east) -- (4.0,3.28);
  \pic at (3.38,3.24) {lock};
  \draw[proparr, red!60!black] (p3.east) -- (4.0,2.86);
  \pic at (3.38,2.48) {lock};

  \foreach \i/\y in {1/1.95, 2/2.62, 3/3.29, 4/3.96}{
    \node[ledgerblock] at (7.7,\y) {$B_\i$};
  }
  \node[font=\scriptsize, gray!40!black, anchor=north, align=center] at (7.7,1.43) {ledger of\\ validator 1};
  \foreach \i/\y in {1/1.95, 2/2.62, 3/3.29}{
    \node[ledgerblock] at (9.5,\y) {$B_\i$};
  }
  \node[pendingblock] at (9.5,3.96) {$\cdots$};
  \node[font=\scriptsize, gray!40!black, anchor=north, align=center] at (9.5,1.43) {ledger of\\ validator 2};
  \draw[teal!60!black, dashed, line width=0.8pt, rounded corners=4pt]
    (7.0,1.55) rectangle (10.2,3.6);
  \node[font=\scriptsize, teal!50!black, anchor=west, align=left] at (10.3,2.6) {common\\ prefix};
  \draw[->, orange!70!black, line width=0.9pt] (7.7,4.27) -- (7.7,4.9);
  \draw[->, orange!70!black, line width=0.9pt] (9.5,4.27) -- (9.5,4.9);
  \node[font=\scriptsize, gray!45!black] at (8.6,5.1) {one block per slot};

  \draw[->, gray!70!black, line width=1pt] (0.3,-0.6) -- (10.7,-0.6);
  \node[font=\scriptsize, gray!60!black, anchor=west] at (10.75,-0.6) {time};
  \draw[tick] (1.0,-0.74) -- (1.0,-0.46);
  \node[font=\scriptsize, gray!40!black, anchor=north] at (1.0,-0.8) {$\mathrm{GST}$};
  \fill[gray!22] (1.0,-0.7) rectangle (2.8,-0.5);
  \node[font=\scriptsize\itshape, gray!50!black, anchor=south] at (1.9,-0.42) {grace period $\mathcal{G}$};
  \foreach \x/\l in {3.8/$D_{s-1}$, 5.6/$D_s$, 7.4/$D_{s+1}$}{
    \draw[tick] (\x,-0.74) -- (\x,-0.46);
    \node[font=\scriptsize, gray!40!black, anchor=north] at (\x,-0.8) {\l};
  }
  \draw[<->, gray!55, line width=0.7pt] (3.8,-0.18) -- (5.6,-0.18)
    node[midway, above, font=\scriptsize, gray!45!black] {$\tau$};
  \draw[<->, gray!55, line width=0.7pt] (5.6,-0.18) -- (7.4,-0.18)
    node[midway, above, font=\scriptsize, gray!45!black] {$\tau$};
  \draw[ptr] (5.6,-0.44) -- (5.6,1.0) -- (5.0,1.43);

  \node[propscard, anchor=north] at (5.6,-1.55) {%
    \textbf{Guarantees}\\[3pt]
    \textcolor{gray!55!black}{\rule{1.2ex}{1.2ex}}~\textbf{Safety:} the ledgers never fork --- at any time, one is a prefix of the other (dashed box).\\[2pt]
    \textcolor{gray!55!black}{\rule{1.2ex}{1.2ex}}~\textbf{Liveness:} every slot eventually adds its block to every honest ledger (upward arrows).\\[2pt]
    \textcolor{gray!55!black}{\rule{1.2ex}{1.2ex}}~\textbf{Censorship resistance:} for slots past the grace period $\mathcal{G}$, an honest proposal always\\
    \phantom{\rule{1.2ex}{1.2ex}}~enters the block (see $B_s$).\\[2pt]
    \textcolor{gray!55!black}{\rule{1.2ex}{1.2ex}}~\textbf{Hiding:} proposals stay concealed until it is too late to react to them (locks).\\[2pt]
    \textcolor{gray!55!black}{\rule{1.2ex}{1.2ex}}~\textbf{Eventual stability:} after the grace period $\mathcal{G}$, deadlines come ``without gaps''\\
    \phantom{\rule{1.2ex}{1.2ex}}~(in \cadence: consecutive deadlines spaced by exactly $\tau$).%
  };

\end{tikzpicture}%

%% file: p1_deploying.tex

\cadence is a consensus protocol, but in the context of a blockchain it provides a different interface than blockchains usually expect from their consensus layer, departing in three ways at once: consensus commits only a digest of each proposer's encrypted proposal, not the transactions it contains, and without executing them; these digests carry no reference to a parent; and each proposer builds its proposal without seeing the current executed state or the other proposals its own will be merged with. This gives up machinery a blockchain normally relies on: linking blocks into a chain, committing to executed state, and validating transactions before they are included. We show how to recover each, so that the result behaves like an ordinary blockchain: \tobiaschange{we build valid proposals under this uncertainty (\Cref{sec:deploy-building}), turn the committed entries into an execution block (\Cref{sec:deploy-translation}), and certify the resulting history and state (\Cref{sec:deploy-certifying})}. All of this is what it takes to behave like an ordinary blockchain; the further consequences and refinements of running with multiple proposers, which a deployment may want but does not need for that, we collect in \Cref{sec:practical}.

Users submit each transaction to one or more of a slot's proposers rather than to a single leader. We treat which proposers a user submits to, and the fee mechanism that would price this choice, as outside our scope. \fatima{feels like it came out of nowhere/ doesn't flow easily from prior paragraph}

We now follow one slot. A proposer selects transactions and assembles a proposal, validating each against a settled view of the executed state. It encrypts and erasure-codes the proposal and disseminates the chunks. The slot then commits the entries of a \emph{meta-block}, one per proposer: the digest of its proposal if included, or empty if not. In parallel with consensus, the proposals behind the committed digests are recovered from their chunks and decrypted, and a deterministic step we call \emph{translation} turns the committed meta-block into the slot's \execblock, the ordered list of transactions to execute. Execution runs asynchronously, trailing consensus by up to $\xi$ slots, the \emph{execution lag}, and the resulting state, together with the committed prefix, is certified for external verifiers.

\subsection{Proposal Construction and Validity}
\label{sec:deploy-building}
\label{subsection:uncertainty}

A proposer builds its proposal under two kinds of uncertainty. First, execution is asynchronous~\cite{yin-sep-exec,monad-async-exec}: it trails consensus, so the proposer does not know the state its transactions will execute against. Second, it cannot see the proposals of the other proposers in its slot (the \emph{peer} blind spot, from running multiple proposers) and those of recent slots due to extreme pipelining 
(the \emph{temporal} blind spot). 

A traditional blockchain validates each transaction before including it (i.e., synchronous execution), against the exact state the block builds upon.
\cadence has no such state to give a proposer, so an implementation may instead fix a common reference for validity, such as the executed state after slot $s-\xi$, a settled state that all validators share. A transaction is admitted only if it is valid against this reference; one already invalid against it, for instance with a reused nonce, is rejected. Because the reference is shared, the decision is the same for everyone: an honest proposer applies it when building, and translation applies it deterministically (\Cref{sec:deploy-translation}), so a transaction invalid against this reference is kept out of the \execblock, even one a faulty proposer includes. A proposer also validates against its own recent proposals. Others cannot see them yet, in the temporal blind spot, but it can, so it builds on its own not-yet-settled transactions. By the time the slot is translated (\Cref{sec:deploy-translation}), these are revealed, and translation re-applies the same check.

Since the state reference is stale by $\xi$ slots, a transaction admitted against it could be unpayable by the time it executes. To address this, a reserve-balance scheme~\cite{monad-reserve-balance} adapted to the multi-proposer setting can be used.

To prove that it validated its proposal against the right reference, a proposer attaches an \emph{execution certificate} for slot $s-\xi$ to its proposal, its \emph{ticket}. This certificate pins the executed state it validated against (\Cref{sec:deploy-certifying}). A proposer can form it only once $2f+1$ validators certify that they have executed slot $s-\xi$, so it cannot make a valid proposal while execution lags more than $\xi$ slots behind: there is then no ticket to attach. This keeps proposals within $\xi$ slots of certified execution, preventing fresh proposals from running too far ahead of it. This bounds proposal validity, not how many slot instances may be open at once; keeping that number bounded is instead \conductor's role. Because proposals are encrypted \tobiaschange{and encoded, and votes are cast on their digests}, the ticket cannot be checked at voting time: validators may vote on and finalize a proposal whose ticket is missing or invalid, even forcing the slot onto the fallback path. Translation then filters such a proposal out when assembling the \execblock (\Cref{sec:deploy-translation}), so its transactions never execute, and it is itself evidence of misbehavior.

\subsection{Translation: From Meta-Block to Execution Block}
\label{sec:deploy-translation}

The consensus commitment fixes which proposals belong to slot $s$, but not yet a block of transactions to run. \emph{Translation} produces the \execblock, and it must be deterministic: from the same committed data, every validator must compute the same \execblock. Translation runs serially across slots: to translate slot $s$, that slot must be committed, the proposals behind its committed digests must have been recovered and decrypted, and slots $1$ through $s-1$ must already be translated. Although the slots' consensus instances run in parallel, translation therefore builds each \execblock on a fixed prefix. 
Over that prefix, and against the certified executed state at slot $s-\xi$, it merges the slot's proposals into the \execblock. Consensus thus agrees on the entries alone, and the \execblock is what they deterministically yield.

The merge runs in stages. Translation first removes every invalid inclusion, leaving only valid transactions. It then deduplicates them, keeping a single copy of each. Validity already keeps a proposer from re-including a transaction it carried itself or one settled before the blind spot, so any surviving duplication is across proposers: building blind, several may each carry the same transaction, within a slot (the peer blind spot) or across recent slots (the temporal one). It finally orders what remains, for example by priority fee. Execution then runs each \execblock asynchronously, once it is available.

However widely a transaction is duplicated, only one copy executes; the rest cost only bandwidth and block space. A deployment can shrink that residual: keeping the proposer set semi-stable, rotated slowly rather than replaced each slot, lets a proposer skip its own recent inclusions, and a transaction-fee mechanism can charge more for each additional inclusion while letting a user choose which proposers may carry a transaction, so a user pays for more inclusions only when it wants them for censorship resistance.

\subsection{Certifying State and History}
\label{sec:deploy-certifying}
\label{sec:chain-certification}

Validators that follow every slot already hold the committed history and, in time, the executed state. Other parties do not: light clients, bridges, and validators catching up after asynchrony or newly joining need to verify both from a small proof, rather than by replaying the chain. In a traditional blockchain every block carries a pointer to its parent, so the certificate finalizing slot $s$ implicitly certifies the entire chain back to genesis, and each block commits to its executed state. Under extreme pipelining, neither holds. Because slots run concurrently, the proposer of slot $s$, when it proposes, cannot in general know the outcome of slot $s-1$: that slot may not yet have terminated, and even its proposals may not yet have been received. A slot therefore cannot point to its predecessor, and the certificate finalizing it certifies only that slot, carrying no executed state. We recover both with a mechanism that runs off the critical path, so the core protocol never waits on it.

A \emph{prefix certificate} aggregates $2f+1$ validator signatures\footnote{For safety, $f+1$ signatures would suffice, since at least one is then from a correct validator; we use $2f+1$ so that a supermajority signs, keeping the network in sync.} attesting a predicate over a prefix of committed slots; since at least $f+1$ of the signers are correct, the predicate holds. We instantiate it twice. A \emph{chain certificate} attests that slots $1, \dots, s$ are all committed; an \emph{execution certificate} attests that slots $1, \dots, s-\xi$ have all been committed and executed, with its $2f+1$ signers agreeing on the resulting state root at slot $s-\xi$. The execution certificate is a certificate on the Ethereum-style header that execution outputs, which carries the executed-state root and, through its parent hash, the committed prefix up to that slot; either certificate can also certify the recursive hash that links the committed blocks back to genesis, giving recursive verification of the whole prefix. A proposer attaches the slot-$(s-\xi)$ execution certificate to its proposal as its ticket (\Cref{sec:deploy-building}). Either certificate may also be used for other purposes, such as external verification by light clients or bridges. Together, the two recover what a traditional block certificate covers at once: the committed history and the executed state.

To form these prefix certificates, a validator signs the corresponding attestation once it has observed the relevant slots committed (and, for an execution certificate, executed); $2f+1$ such signatures make up the certificate. The proof trails finalization slightly and is fetched separately rather than carried in the slot's commit.

%% file: p1_practical_considerations.tex

This section collects consequences and refinements of \cadence's design that matter for deployment but are not required to follow the core protocol. We discuss how a proposer that breaks the fast path can be held accountable (\S\ref{pc:accountability}), questions of latency, scheduling, and economic ticks (\S\ref{pc:latency}), and transaction privacy (\S\ref{pc:privacy}).

\subsection{Accountability}
\label{pc:accountability}

As we specified in \S\ref{subsection:fast-path-overview}, \chorus enjoys a fast path that finalizes a slot after only two voting rounds.
That same section also tells us \emph{when} this fast path is available: it triggers precisely when every proposer is either \emph{correct}, disseminating its proposal to all correct validators, or \emph{fully silent}, disseminating it to none.
The difficulty lies in the intermediate case.
A proposer that \emph{partially disseminates}, delivering its proposal to some correct validators but not all, splits the validators' views, and such a split can derail the fast path, forcing \chorus onto the slower fallback path (\S\ref{subsection:fallback-path-overview}).\footnote{The fast path can also be derailed by \emph{equivocation}, a proposer sending conflicting proposals to different validators; but two such signed proposals already form an equivocation certificate, so that case carries its own proof. We therefore set it aside and focus on partial dissemination, assuming below that no proposer equivocates.}

If we must pay this latency, we would at least like to know which proposer caused it.
We therefore equip \chorus with an \emph{accountability} mechanism: whenever a proposer forces the fallback path, we single it out and hold it responsible for its partial dissemination, leaving correct validators a transferable proof of which proposer broke the fast path.
This guarantee comes neither for free nor always: it holds only when the number of Byzantine validators is at most $f/2$.\footnote{In practice, this should be ``good enough'': the actual number of faults is typically far below the worst-case bound $f$ the protocol provisions for.}
Under this assumption (and after GST, since detecting a missed fast path is ultimately a matter of liveness), every fallback leaves correct validators holding enough signed evidence to reliably name every proposer that partially disseminated.
Let us see why.

\paragraph{First-round votes \& fast meta-blocks.}
Recall the votes of the first round.
At the deadline, each validator casts, for every proposer $p$, a signed \emph{positive} vote if it received $p$'s proposal in time, and a signed \emph{negative} vote otherwise.
A fast meta-block can be assembled only when, for \emph{every} proposer, these votes converge: $p$ must gather either a quorum of $2f+1$ positive votes, or a quorum of $2f+1$ negative votes.
When neither quorum forms for some proposer, the fast path cannot proceed, and the slot falls back to the slower path.

\paragraph{From a split vote to a blamed proposer.}
Now suppose at most $f/2$ validators are Byzantine, and consider a correct validator that fails to build the fast meta-block.
After GST, under synchrony, this validator hears from every validator except the (at most $f/2$) faulty ones, that is, from $n - f/2 = 2f+1+f/2$ validators, and hence collects at least that many votes for each proposer.
Failing to build the meta-block means that, for some proposer $p$, neither quorum has formed: among the votes it holds for $p$, fewer than $2f+1$ are positive \emph{and} fewer than $2f+1$ are negative.
Thus $p$ drew at most $2f$ positive votes and at most $2f$ negative votes; and since the two counts together sum to at least $2f+1+f/2$, each of them is, in turn, at least $f/2+1$.

This split is precisely the fingerprint of partial dissemination, and, crucially, one we can prove.
Among the $f/2+1$ positive votes, at most $f/2$ can come from Byzantine validators, so at least one comes from a \emph{correct} validator: one that genuinely received $p$'s proposal by the deadline, witnessing that $p$ \emph{did} disseminate to someone.\footnote{Since each positive vote carries the proposal's Merkle root signed by the proposer, a \emph{single} positive vote already proves that the proposer issued a proposal.}
Among the $f/2+1$ negative votes, by the very same counting, at least one comes from a correct validator that genuinely did not receive it, witnessing that $p$ did \emph{not} disseminate to everyone.
Because all votes are signed, these two witnesses together form a transferable certificate: anyone, knowing that at most $f/2$ validators are faulty, can read off from $f/2+1$ signed positive and $f/2+1$ signed negative votes that some correct validator received $p$'s proposal while some other correct validator did not.
That is exactly partial dissemination, so correct validators can hold $p$ accountable.

The split-vote certificate is always a sound, transferable proof that $p$ partially disseminated; what depends on synchrony is only reading a fallback as $p$'s fault, since under asynchrony even a correct proposer can miss the fast path. A single fallback is therefore no cause to act, but a proposer that provably and repeatedly partially disseminates can be penalized.

\subsection{Latency, Scheduling, and Economic Ticks}
\label{pc:latency}

\paragraph{Proposal timing.}
An honest proposer disseminates its chunks early enough that, under synchrony, every chunk reaches every validator before the deadline. The worst-case message delay $\Delta$ is a conservative bound; a well-connected proposer reaches all validators within a smaller delay $\delta < \Delta$ of its own, so it can disseminate as late as $\delta$ before the deadline and still have its chunks arrive in time. Disseminating later lets a proposal include fresher transactions and so improves end-to-end latency for users, so each proposer sets its own dissemination time according to its connectivity. Disseminating too late---while potentially lucrative for rational proposers~\cite{schwarzschilling2023time,alpturer2025timing}---risks reaching only some validators in time, which can push the slot onto the slower fallback path; our accountability mechanism (\S\ref{pc:accountability}) deters this by leaving evidence of partial dissemination. A proposer that is entirely silent, by contrast, forces no fallback and leaves no such split-vote evidence: every validator simply records a negative entry and votes against it. Such a proposer needs no dedicated liveness incentive, since forgoing its proposal reward is penalty enough.

\paragraph{Early voting.}\label{subsection:early_voting}
By default, a validator broadcasts its proposal vote at the slot deadline. As an optional optimization, which we call \emph{early voting}, a validator may vote earlier in slot $\slot$ once every entry for $\slot$ is positive. So that this does not weaken hiding, a validator releases its decryption share only at the deadline, decoupled from its vote, so proposals are revealed no earlier than before. Recovering the proposals does not delay finalization: their chunks are disseminated together with the votes, so reconstruction needs no extra round, and only decryption waits for the shares released at the deadline. In our experiments, gathering the $f+1$ shares completes around the time a slot is speculatively finalized, well before it is fully finalized, so a slot's proposals are recovered in time even when its votes were cast early.

Early voting reduces finalization latency. A slot can finalize as soon as its proposals have arrived and gathered commit votes, rather than waiting for the deadline. How much it saves depends on all of the slot's proposers together: a vote can be cast early only once every entry is positive, so the more conservatively they disseminate, leaving a margin before the deadline, the further ahead the slot can finalize.


\paragraph{Cycles.}\label{subsection:slices}
Some blockchains decouple the on-chain economic tick from the consensus tick with \emph{cycles}: on-chain activity advances in shorter cycles while consensus still runs once per slot, as in Constellation~\cite{Constellation}. Cycles let activity such as auctions and oracle updates tick at the finer granularity, but they do not lower end-to-end latency, since a transaction is still finalized only once per slot. They are thus no substitute for extreme pipelining, which lowers the block interval itself. Their one advantage is message complexity, since consensus does not run more often. \chorus could support cycles as a straightforward extension: each slot carries $K$ cycles, each proposer disseminates one part of its proposal before each cycle's deadline, and a validator casts a single proposal vote per slot, positive for a proposer only if it received that proposer's assigned chunk of every part in time. The positive vote then carries the combined Merkle root over the proposer's $K$ parts, so $2f+1$ such votes certify its whole proposal exactly as without cycles.

\subsection{Transaction Privacy}
\label{pc:privacy}

Hiding conceals each proposal from everyone until the deadline, but a proposer a user submits to sees the transaction in plaintext. So a user trades off how many proposers to submit to: more improve inclusion and censorship resistance, while fewer reduce the chance that one of them leaks or front-runs the transaction. Client-side \emph{transaction encryption}, where the user encrypts the transaction before submitting it, as in encrypted mempools~\cite{agarwal2025weighted,agarwal2026btx,bormet2025beast}, removes this trade-off, since the user can submit to many proposers without revealing the contents to any of them. It does not replace MCP: on its own it gives no short-term censorship resistance~\cite{encrypted-mempool-limits}, since a proposer can censor a ciphertext without reading it. Nor does it replace proposal encryption: in a practical deployment, where proposers compete (say, for transaction fees), hiding protects each proposer's payload from its rivals, which transaction encryption does not. The two can compose. For either scheme, both the timing of share release and the decryption threshold ($f+1$ or more) are protocol choices that trade off latency and availability against the strength of hiding.


%% file: p1_latency_evaluation.tex
We estimate \cadence's fast-path network latency, the common case in which every proposer is correct or fully silent, on the $200$ stake-weighted, globally distributed validators of Monad mainnet. We define end-to-end latency as the time from when a transaction reaches a proposer until it is finalized. It has two parts: the finalization latency, the time from the transaction entering a proposal until it is finalized, and the inclusion latency, the wait until it enters a proposal. The finalization latency is, on average, $219$\,ms to full finality and $167$\,ms to speculative finalization after a single voting round. The inclusion latency is, on average, half a block interval, so, for example, $50$\,ms for a block interval of $\blockint = 100$\,ms. Composing the two, the end-to-end latency averages $269$\,ms to full finality and $217$\,ms to speculative finalization. The early-voting optimization (\Cref{subsection:early_voting}) improves these figures by about $30$\,ms: the finalization latency drops to $191$\,ms (full) and $137$\,ms (speculative), and the end-to-end latency to $241$\,ms and $187$\,ms. In the early-voting case, a slot's decryption shares arrive on average $5$\,ms after its speculative finalization, well before full finalization. The user-to-proposer hop these figures omit tends to be smaller in MCP than in single-leader protocols, since a user is more likely to have a nearby proposer, which improves overall latency.

We obtain these figures from a simulation calibrated to RIPE Atlas round-trip measurements among the 83 (city, ASN) groups in which these validators run as of June 2026, with five concurrent proposers per slot. We set each proposer's lead time, the interval by which it broadcasts ahead of the deadline, so that in $99\%$ of trials at least $90\%$ of validators receive its chunk before the deadline. Across proposers, this lead time averages $104$\,ms.

The simulated propagation delays are the dominant part of this end-to-end latency. The contributions it omits are small: local processing at validators, where execution and most validity checks occur asynchronously (\Cref{subsection:uncertainty}), and the bandwidth-induced queueing delay, which is at most a few milliseconds.\footnote{We assume that, in practice, chunk dissemination is decoupled from the more latency-critical vote dissemination: votes are sent first, before chunk re-dissemination begins.}

%% file: related_work.tex
\section{Related Work}
\label{sec:related}

\subsection{Pipelining}

Motivated by the goal of driving inter-proposal time down, much of the community's effort has centered on pipelining. Protocols such as the HotStuff family~\cite{kang2024hotstuff1linearconsensusonephase,HotStuff,HotStuff-2:2023/397}, MonadBFT~\cite{jalalzai2026monadbftfastresponsiveforkresistant}, and Jolteon~\cite{gelashvili2021jolteon} accomplish this by chaining proposals through quorum certificates (QCs): a leader is free to issue its next proposal as soon as it collects a QC for the previous one, which brings the inter-proposal time down to $2\Delta$.
Pushing this idea further, optimistic designs such as Moonshot~\cite{doidge2024moonshotoptimizingchainbasedrotating} and Hydrangea++~\cite{hydrangeapp} assume consecutive honest leaders and let a leader propose even before the prior proposal has been quorum-certified, bringing the inter-proposal time down to the network delay $\Delta$ itself, though not below it.
We go below this bound with extreme pipelining, which achieves sub-$\Delta$
block intervals by running independent, non-chained single-shot consensus
instances rather than chaining proposals at all.

A similar line of
work reaches sub-$\Delta$ intervals by staggering several multi-shot consensus instances in
time and merging their outputs into one log: Shoal++~\cite{shoalpp} interleaves
three DAG protocols one message delay apart, Raptr~\cite{raptr} generalizes this
to $K$ instances offset by $2/K$ message delays, and, concurrently with our
work, Gatling~\cite{gatling} casts it as \emph{parallel composition} over any
atomic broadcast protocol. Where we run one sequence of independent single-shot
slots, these interleave several instances to the same effect.

\tobiaslater{Compare our protocol with Gatling along these lines: it fixes a constant number $K$ of always-running instances, whereas we run a dynamic number of slot-consensus instances up to a bound $\mathcal{B}$, using only those in flight; and Gatling's instances are multi-shot, whereas ours are single-shot, one per slot.}

Constellation~\cite{Constellation} introduces \emph{cycles}, arbitrarily short economic ticks; complementary to extreme pipelining, they shorten the economic tick, while their underlying consensus produces blocks at a rate bounded by the network delay. We discuss cycles, and how they can be incorporated into \cadence, in \Cref{subsection:slices}.

\subsection{Multiple Concurrent Proposers}

Garimidi et al.~\cite{mcp-why-and-how} set out the goals of multiple concurrent proposers, including censorship resistance and hiding; these goals have since been pursued, through different mechanisms, by Prefix Consensus~\cite{prefix-consensus}, Constellation~\cite{Constellation}, and AMP~\cite{amp}. All of these designs provide some notion of censorship resistance. Only Garimidi et al.~\cite{mcp-why-and-how} incorporate hiding into their protocol design, though AMP notes that hiding is orthogonal to its core protocol and could be added. To our knowledge, \chorus is the first MCP design to offer both short-term censorship resistance (for every block after GST plus a grace period) and an optimal three-round good-case latency.

\paragraph{Prefix Consensus.}
Prefix Consensus~\cite{prefix-consensus}, a leaderless multi-proposer SMR, achieves good-case latency comparable to that of our fast path, while guaranteeing \emph{$f$-censorship-resistance}: after $\mathrm{GST}$, there may be up to $f$ slots for which censorship resistance does not hold.
Crucially, and in contrast to our work, these $f$ slots are adversarially chosen; they need not be the first $f$ slots after $\mathrm{GST}$, but may be placed at any point.
In \chorus, by contrast, censorship resistance holds for \emph{all} slots beyond a bounded interval following $\mathrm{GST}$: only slots falling within this bounded grace period may lack censorship resistance, while every subsequent slot is guaranteed it, and the adversary can do nothing to violate it.


\paragraph{Constellation.}
Constellation~\cite{Constellation} is a multiple-proposer protocol layered on Solana's Alpenglow~\cite{Alpenglow}.
Proposers disseminate their transactions to a layer of attesters whose attestations constrain which proposals the leader must include, giving it \emph{selective censorship resistance}.
Because this inclusion layer sits on top of Alpenglow's consensus, finalization takes two more communication rounds than the consensus alone.

\paragraph{AMP.}
AMP~\cite{amp} is a multiple-proposer layer on the Tendermint consensus algorithm. Like Constellation, it separates dissemination from agreement, which costs two extra communication rounds: proposers broadcast transaction payloads, and validators agree only on their identifiers, carried in vote extensions that constrain which payloads the leader must include.
Its central guarantee is \emph{bounded inclusion}: a payload attested by all correct validators at one height is forced into the block at the next.
AMP inherits Tendermint's safety and liveness and its $f < n/3$ threshold.

\paragraph{DAG-based protocols.}
DAG-based protocols~\cite{narwhal,bullshark,dag-rider,babel2024mysticetireachinglimitslatency} are aimed at high throughput (and, in recent designs, low latency) rather than at the guarantees we require, and although several provide some notion of censorship resistance, they do not target the strict short-term censorship-resistance guarantee we require.
Since every validator proposes each round, they are multiple-concurrent-proposer in a broad sense, but as Garimidi et al.~\cite{mcp-why-and-how} observe, amending them to provide short-term censorship resistance is not straightforward, as it appears to require giving up the optimistic responsiveness, which lets them commit at network speed in the good case (synchrony and an honest leader).

\tobiaslater{\paragraph{Comparison.}
As noted above, no prior MCP design provides both censorship resistance and hiding at \chorus's three-round latency: Constellation and AMP each add two communication rounds to their consensus for the inclusion guarantee; the protocol of Garimidi et al.~\cite{mcp-why-and-how} finalizes in five rounds including the consensus (four with an optimization), the only one that also hides; and the fast DAG-based designs do not provide hiding or the strict short-term censorship-resistance guarantee considered here.}

\tobiaslater{\paragraph{The latency cost of censorship resistance.}
\tobias{This is my best understanding; it should be checked independently. I still find this confusing.}
Abraham et al.~\cite{AbrahamER25} show that guaranteeing a correct client's transaction in the next block, so that no leader can exclude it, costs at least five communication rounds under partial synchrony at optimal resilience, even when the leader is correct, where plain consensus needs three.
\chorus's three-round \fastpath does not contradict this bound, because a single Byzantine proposer can force the fallback path. The bound applies when censorship resistance must hold against an adversary that corrupts proposers up to the threshold consensus safety tolerates. However, there is no a priori reason to tie the two thresholds: for censorship resistance, \chorus needs only one of a slot's $\numprop$ proposers to be honest.}

%% file: p2_problem_definition.tex
\section{Multiple Concurrent Proposers (MCP): Formal Problem Definition} \label{section:formal_problem_definition}

This section gives a precise, formal statement of the MCP problem that \Cref{subsection:mcp-overview} introduced informally.

\subsection{Preliminaries}
\label{subsection:mcp-preliminaries}
The MCP problem involves several abstract entities; we begin by introducing each of them precisely.

\subsubsection{Slots.}
Each \emph{slot} $s$ carries three fields:
\begin{compactitem}
  \item a number $s.\fnumber$, which is a positive integer;
  \item a deadline $s.\fdeadline$, which is a point in global time; and
  \item a set of proposers $s.\fproposers$, which is a subset of the entire set of validators.
\end{compactitem}
We emphasize that all three fields are \emph{read-only}: a slot's number, deadline, and proposers are fixed and never change.
This is a deliberate, minor departure from \Cref{subsection:mcp-overview}, where we described the deadline as something the protocol determines.
Treating deadlines as given --- rather than as a protocol \emph{output} --- keeps the problem statement clean: otherwise the MCP problem definition would have to carry extra guarantees about how deadlines are produced.
The protocol still adapts its timing not by moving deadlines but by skipping slots, as we make precise below.

For each $i \in \mathbb{N}_{\geq 1}$, there exists exactly one slot $s$ with $s.\fnumber = i$.
Moreover, deadlines respect the slot order: for any two slots $s_1, s_2$ with $s_1.\fnumber < s_2.\fnumber$, we have $s_1.\fdeadline \leq s_2.\fdeadline$.
(For our extreme-pipelining framework, we additionally assume that consecutive deadlines are a fixed, known interval $\tau$ apart, with $\tau > 0$: that is, $s.\fdeadline = s'.\fdeadline + \tau$ for every slot $s$ with $s.\fnumber > 1$, where $s'$ is the slot with $s'.\fnumber = s.\fnumber - 1$.)
We sometimes refer to the \emph{starting time} of a slot $s$, by which we mean $s.\fdeadline - \Delta$.
Moreover, we require the first slot (the slot $s$ with $s.\fnumber = 1$) to have a deadline of at least $\Delta$, so its starting time is non-negative; and since deadlines only grow with the slot number, every slot then has a non-negative starting time.
(For our extreme-pipelining framework, we assume that the starting time of the first slot is $0$.)

\subsubsection{Proposals \& Proposal Vectors.}
A \emph{proposal} $P$ carries three fields:
\begin{compactitem}
    \item a slot $P.\fslot$;
    \item a proposer $P.\fproposer$, which must be a proposer of the proposal's slot ($P.\fproposer \in P.\fslot.\fproposers$); and
    \item a payload $P.\fpayload$, an arbitrary string carrying the proposal's content (e.g., its transactions).
\end{compactitem}
We also introduce \emph{proposal vectors}, which bundle one proposal per proposer of a slot.
Concretely, a proposal vector $V$ carries:
\begin{compactitem}
    \item a slot $V.\fslot$; and
    \item a mapping that assigns to each proposer $p \in V.\fslot.\fproposers$ either a proposal or the special value $\bot$ (meaning ``no proposal''); we write $V[p]$ for the proposal (or $\bot$) assigned to $p$.
\end{compactitem}
We impose one constraint on the aforementioned mapping: whenever $V[p] \neq \bot$, the proposal $V[p]$ must be $p$'s own proposal for that slot, i.e., $V[p].\fslot = V.\fslot$ and $V[p].\fproposer = p$.
Two proposal vectors are \emph{identical} if and only if their slots are identical and their mappings are identical: $V_1 = V_2$ exactly when $V_1.\fslot = V_2.\fslot$ and $V_1[p] = V_2[p]$ for every proposer $p \in V_1.\fslot.\fproposers$.
For our purposes a proposal vector is interchangeable with a block: a block can be obtained from a proposal vector by a fixed, deterministic transformation (ordering the proposals' transactions, discarding invalid ones, and so on).
Since this transformation is an application-level concern, irrelevant to consensus, and since proposal vectors make our later definitions cleaner, we phrase everything in terms of proposal vectors.

\subsubsection{Logs.}
A \emph{log} is a (possibly empty) ordered list of proposal vectors with strictly increasing slot numbers: if $V_1$ precedes $V_2$ in the log, then $V_1.\fslot.\fnumber < V_2.\fslot.\fnumber$.
For a log $\mathcal{L}$, we write $\mathcal{L}.\length$ to denote its number of proposal vectors, and $\mathcal{L}[i]$ for its $i$-th proposal vector, where $i \in [1, \mathcal{L}.\length]$.
Notably, a log need not contain a proposal vector for every slot: some slots may be ``skipped'', contributing no proposal vector to the log.
In particular, $\mathcal{L}.\length$ may be strictly less than $\mathcal{L}[\mathcal{L}.\length].\fslot.\fnumber$.
(Permitting such skips is the formal counterpart of the informal view in \Cref{subsection:mcp-overview}, where no slot is skipped and the schedule instead simply spaces consecutive slot deadlines apart in time. These two views --- skipping slots here versus spacing deadlines there --- are equivalent reformulations of the same underlying ledger, and we make the correspondence between them precise in \Cref{subsection:building_blocks_framework}.)
For a proposal vector $V$ and a log $\mathcal{L}$, we write ``$V \in \mathcal{L}$'' to denote that $V$ appears in $\mathcal{L}$.
Furthermore, we say that two logs $\mathcal{L}_1$ and $\mathcal{L}_2$ are \emph{consistent} if and only if $\mathcal{L}_1[i] = \mathcal{L}_2[i]$, for every $i \in [1, \min(\mathcal{L}_1.\length, \mathcal{L}_2.\length)]$; otherwise, they are \emph{inconsistent}.
(Any log is consistent with the empty log.)

\subsection{Interface}

The interface of the MCP problem is simple:
\begin{compactitem}
    \item \emph{Input.} For each slot $s$, every proposer $p \in s.\fproposers$ holds a single proposal $P$ for that slot (with $P.\fslot = s$ and $P.\fproposer = p$).
    \item \emph{Output.} Each validator $\proc_i$ maintains an append-only local log, denoted $\locallog(\proc_i)$. Every correct validator initializes its local log to the empty log $[]$ at time $0$, and we write $\locallog(\proc_i, t)$ for the local log of $\proc_i$ at time $t$.
\end{compactitem}

\subsection{Guarantees}
We now specify the correctness properties that any MCP protocol must satisfy, defining each in turn.
The MCP problem inherits the standard guarantees of consensus --- safety and liveness --- and adds two more on top: censorship resistance and hiding.

\subsubsection{Safety.}
As in any consensus protocol, the foundational property is \emph{safety}: the logs of all correct validators must be mutually consistent at all times.

\begin{definition}[Safety] \label{def:safety}
A protocol satisfies \emph{safety} if and only if, for any two correct validators $\proc_i, \proc_j$ and any two times $t_1, t_2$, logs $\locallog(\proc_i, t_1)$ and $\locallog(\proc_j, t_2)$ are consistent.
\end{definition}

\subsubsection{Liveness.}
As noted in \Cref{subsection:mcp-preliminaries}, logs may inherently skip slots: not every slot need contribute a proposal vector to a validator's local log.
However, liveness requires that this skipping cannot persist indefinitely after $\mathrm{GST}$: there exists a post-$\mathrm{GST}$ grace period after which every slot contributes a proposal vector to the local log of every correct validator.
This precisely corresponds to the eventual stability guarantee we described informally in \Cref{subsection:mcp-overview}.

\begin{definition}[$\ell$-Liveness] \label{def:liveness}
A protocol satisfies \emph{$\ell$-liveness}, for some time duration $\ell$, if and only if, for every slot $s$ with $s.\fdeadline - \Delta \geq \mathrm{GST} + \ell$ and every correct validator $\proc_i$, there exist a proposal vector $V$ and a time $t$ such that (1) $V.\fslot = s$, and (2) $V \in \locallog(\proc_i, t)$.
\end{definition}

\noindent We underline that the parameter $\ell$ is not a fixed universal constant but rather a measure of the quality of a protocol: it captures how quickly a protocol resumes producing proposal vectors after $\mathrm{GST}$, and we naturally prefer protocols that achieve a smaller $\ell$.
Accordingly, $\ell$ need not be a constant and may depend on the parameters of the model and the execution, such as $n$, $f$, $\Delta$, and $\mathrm{GST}$.
For instance, a protocol whose grace period grows the longer the network remains asynchronous (e.g., under a linear or exponential backoff) is captured by an $\ell$ that increases with $\mathrm{GST}$.

\subsubsection{Short-Term Censorship Resistance.}

Liveness ensures that slots are not skipped indefinitely after $\mathrm{GST}$, but says nothing about the \emph{content} of the finalized proposal vectors: a protocol that always finalizes empty proposal vectors (assigning $\bot$ to every proposer) is live yet useless.
We therefore additionally require \emph{censorship resistance}: the proposal of a correct proposer for a slot cannot be suppressed by Byzantine validators, but must be included in the proposal vector finalized for that slot.
Concretely, there exists a grace period $c$ such that, for every slot $s$ whose starting time is at least $c$ after $\mathrm{GST}$, the proposal vector finalized for $s$ includes the proposals of all correct proposers of $s$.

\begin{definition}[$c$-Censorship Resistance] \label{def:censorship-resistance}
A protocol satisfies \emph{$c$-censorship resistance}, for some time duration $c$, if and only if, for every slot $s$ with $s.\fdeadline - \Delta \geq \mathrm{GST} + c$, every correct validator $\proc_i$, and every correct proposer $\proc_j \in s.\fproposers$, there exist a proposal vector $V$ and a time $t$ such that (1) $V.\fslot = s$, (2) $V[\proc_j] = P$, where $P$ is the proposal of proposer $\proc_j$ for slot $s$, and (3) $V \in \locallog(\proc_i, t)$.
\end{definition} 

\noindent As with the parameter $\ell$ of $\ell$-liveness, the parameter $c$ measures the quality of a protocol: it bounds how quickly the proposals of correct proposers are included after $\mathrm{GST}$ (smaller is preferable), and may likewise depend on the model and execution parameters.
We clarify the relationship between liveness and censorship resistance once more.
The two differ in what they assume about the proposers of a slot.
Liveness guarantees that, after the grace period, every slot eventually contributes a proposal vector to the log, regardless of whether its proposers are correct or faulty.
Censorship resistance, in contrast, concerns the content of that proposal vector, and only constrains slots that have at least one correct proposer (since it speaks about proposals held by correct proposers).
The two properties are thus complementary: liveness ensures that a slot contributes a proposal vector, while censorship resistance, for a slot with a correct proposer, ensures that the proposer's proposal is included by that slot.

\subsubsection{Hiding.}
In defining the hiding property, the intuition we want to capture is that faulty proposers cannot submit proposals that depend on the contents of proposals submitted by honest proposers.
However, there are some nuances that need to be addressed:
\begin{compactitem}
    \item At the very least, such a property should imply ``semantic security'',
    which would just say that the adversary does not learn the contents of the honest proposals
    before the deadline for submitting proposals passes.

    \item Ideally, we want a stronger notion of security that captures ``non-malleability'', which would say that the adversary cannot submit a proposal whose contents depend in some non-trivial way on an honest proposal.
    For example, suppose an honest proposal $x$ is encrypted as $y$.
    If the encryption scheme is malleable, an adversary may be able to compute an encryption $y'$ of $x+1$ without ever learning what $x$ is.
    Our definition of security will capture this guarantee.

    \item The most ambitious definition of security would say that each faulty proposer gets to submit (at most) one proposal (which does not depend on the honest proposals), and does not learn the contents of the honest proposals until the set of included proposals is fully decided.
    While such a definition could be realized, it would add significant latency to the protocol.
    For that reason, we propose a somewhat weaker definition that gives the adversary the ability to (1) have faulty proposers effectively submit multiple proposals (which do not depend on honest proposals), and (2) use the contents of the honest proposals
    to influence the decided set of included proposals.
\end{compactitem}
To formalize the above, we make a simulation-based definition.
This definition is in the spirit of Universal Composability (UC)~\cite{DBLP:conf/focs/Canetti01}, but we do not require the full machinery of the UC framework. 
Just as
in the UC framework, we have a ``real world'' and an ``ideal world''.
In the real world, validators interact with the actual protocol.
In the ideal world, there is an ideal functionality $\mathcal{F}$
and a simulator (a.k.a.\ ideal-world adversary) $\mathcal{S}$.
Both worlds are driven by an \emph{environment} $\mathcal{Z}$ --- an arbitrary efficient (i.e., probabilistic polynomial-time) entity that supplies the inputs to the honest validators and observes their outputs.
The environment also plays the role of the adversary: in the real world it directly controls the faulty validators, whereas in the ideal world those faulty validators are instead controlled by the simulator $\mathcal{S}$, with which $\mathcal{Z}$ may interact.
We assume \emph{static} corruption: the set of faulty validators is fixed before the execution begins.
The environment's goal is to tell the two worlds apart.
In the ideal world, all inputs from and outputs to honest validators pass directly between $\mathcal{Z}$ and $\mathcal{F}$, while inputs from and outputs to faulty validators pass between $\mathcal{S}$ and $\mathcal{F}$.
Critically, the interface seen by the environment is the same in both worlds.
It thus suffices to define the ideal functionality $\mathcal{F}$ corresponding to
our hiding property, which we do next.
We describe $\mathcal{F}$ for a single slot; the full functionality applies it to every slot.
\begin{compactitem}
\item
Proposers, both honest and faulty, can input proposals to $\mathcal{F}$,
{\em but only before the slot's deadline}.
\item
Unlike an honest proposer, a faulty proposer may input
several proposals (and not all at the same time).
\item
When an honest proposer $p_i$ inputs a proposal to $\mathcal{F}$,
$\mathcal{F}$ informs $\mathcal{S}$ that $p_i$ has submitted a proposal,
but divulges no information about the proposal to $\mathcal{S}$
beyond the identity of $p_i$ and the length of its proposal.
\item
When the deadline passes, $\mathcal{F}$ gives to $\mathcal{S}$
the contents of all of the honest proposals.
\item
Later, $\mathcal{S}$ chooses the actual set $S$ of proposals
to be included in the proposal vector, and sends this to $\mathcal{F}$.
The set $S$ must be a subset of the proposals that were input to the ideal
functionality, and can include at most one proposal per proposer.
\item
After sending $S$ to $\mathcal{F}$,
$\mathcal{S}$ may then instruct $\mathcal{F}$ to output
the proposal vector determined by $S$ to individual validators (one at a time, in the order of its choosing).
\end{compactitem}
Having specified $\mathcal{F}$, we now state our hiding property: a protocol is hiding precisely when its real-world execution cannot be told apart from this ideal world.

\begin{definition}[Hiding] \label{def:hiding}
A protocol satisfies \emph{hiding} if and only if there exists a simulator $\mathcal{S}$ such that no efficient environment can distinguish the real world from the ideal world, where the ideal world is defined by the ideal functionality $\mathcal{F}$ described above, applied to every slot.
\end{definition}

%% file: p2_framework.tex
\section{Our Extreme-Pipelining Framework: Formal Exposition} \label{section:framework}

We now give the formal account of \emph{the extreme-pipelining framework}, the framework we introduced informally in \Cref{subsection:cadence-overview}.
Recall that \cadence, our concrete MCP protocol, is the specific instantiation of the extreme-pipelining framework obtained by taking \chorus as the slot consensus and \conductor as the orchestrator.

\subsection{Building Blocks: Slot Consensus \& Orchestrator} \label{subsection:building_blocks_framework}



We now formally define the two abstract building blocks on which our extreme-pipelining framework relies: slot consensus and the orchestrator, introduced informally in \Cref{subsection:slot-consensus-overview} and \Cref{subsection:orchestrator-overview}, respectively.

\subsubsection{Slot Consensus.}
The slot consensus primitive is specified formally in \Cref{mod:slotconsensus}; each of its instances is parameterized by a slot $s$.
The interface reflects the two roles a validator plays: a participant in the consensus process and, for designated slots, a proposer.
On the participation side, validators explicitly signal when they start and stop contributing to an instance; between these two events we say that the validator is \emph{actively contributing}.
(We track this explicitly --- though consensus protocols usually leave it implicit --- because our analysis of the extreme-pipelining framework aims to bound how many instances a correct validator actively contributes to at any one time.)
On the proposer side, designated validators additionally submit their proposals.
Together, these participation signals and proposals form a validator's \emph{inputs} to the instance, while its sole \emph{output} is a finalization: a validator may finalize a single proposal vector for slot $s$.


The slot consensus primitive guarantees six correctness properties.
Agreement and termination are the usual safety and liveness: correct validators never finalize conflicting proposal vectors, and if all correct validators start participating, then every correct validator eventually finalizes.
Importantly, agreement constrains each validator individually as well: no correct validator finalizes two different proposal vectors, even on separate occasions.
Slot safety requires that the finalized proposal vector $V$ carry the correct slot identifier $s$.
(Slot safety is a property we omitted from the informal overview in \Cref{subsection:slot-consensus-overview}, as it is a self-evident condition that reveals nothing conceptually interesting about the primitive.)
Proposal inclusion guarantees that every proposal submitted by a correct proposer by the slot's starting time appears in $V$ (under synchrony).
Hiding is the slot-level counterpart of the MCP hiding guarantee (\Cref{def:hiding}), specialized to the instance's slot $s$.
Quiescence, finally, confines a correct validator's communication to its participation window: it sends no protocol message for the instance before it starts participating or after it stops doing so. 
(This property, together with the bounded number of slot instances underway in \cadence at any time --- which we show later --- keeps a validator's work bounded: it stays silent on the slots it has not yet opened or has already completed.)


\begin{module}[h]
\caption{Slot Consensus}
\label{mod:slotconsensus}
\footnotesize
\begin{algorithmic}[1]

\Statex \textbf{Parameters:}
\begin{compactitem}[-]
    \item $s \in \Slot$
\end{compactitem}

\smallskip
\Statex \textbf{Interface:}
\begin{compactitem}[-]
    \item input $\mathsf{participate}()$: a validator starts participating.
    
    \item input $\mathsf{abandon}()$: a validator stops participating.

    \item input $\mathsf{propose}(P \in \mathsf{Proposal})$: a proposer submits its proposal $P$.
    
    \item output $\mathsf{finalize}(V \in \mathsf{PVector})$: a validator finalizes proposal vector $V$.
\end{compactitem}

    



\smallskip
\Statex \textbf{Properties:} 
\begin{compactitem}[-]
    \item \emph{Agreement:} If a correct validator $p_i$ finalizes a proposal vector $V_i$ and a correct validator $p_j$ finalizes a proposal vector $V_j$, then $V_i = V_j$.
    
    \item \emph{Termination:}
    If every correct validator starts participating, then every correct validator eventually finalizes a proposal vector.

    \item \emph{Slot safety:} If a correct validator finalizes a proposal vector $V$, then $V.\fslot = s$.

    \item \emph{Proposal inclusion:} If $s.\fdeadline - \Delta \geq \mathrm{GST}$, a correct proposer $\proc_j \in s.\fproposers$ proposes its proposal $P$ at time $s.\fdeadline - \Delta$, and a correct validator finalizes a proposal vector $V$, then $V[\proc_j] = P$.

    \item \emph{Hiding:} Hiding as in \Cref{def:hiding}, specialized to this instance's single slot $s$.
    
    \item \emph{Quiescence:} No correct validator sends any protocol message before it starts participating or after it stops participating.
\end{compactitem}
\end{algorithmic}
\end{module}


\subsubsection{Orchestrator.}
\Cref{mod:orchestrator_2} provides the orchestrator's full specification.
Validators supply a single input: upon completing their work on a slot $s$, a validator notifies the orchestrator by completing $s$.
The orchestrator, in turn, produces a single output: for an upcoming slot $s'$, it may instruct a validator to open $s'$; there is no separate ``skip'' output.
Instead, we informally say that a slot a validator never opens is \emph{skipped}.
It is worth reconciling this interface with the informal overview (\Cref{subsection:orchestrator-overview}), where we described the orchestrator's job as \emph{determining the deadlines of slots}.
In the formal model, by contrast, slots --- and all information about them, including their deadlines --- are fixed and read-only (\Cref{subsection:mcp-preliminaries}).
The two views are equivalent: instead of \emph{delaying} a slot's deadline, as the informal orchestrator would, the formal orchestrator simply \emph{skips} the intervening slots and opens the slot whose (fixed) deadline already coincides with the desired, later time.
Skipping slots is thus the formal counterpart of pushing deadlines back.
Throughout the rest of the paper, we adopt the $\tau$-spaced deadline structure foreshadowed in \Cref{subsection:mcp-preliminaries}: consecutive slots' deadlines are a fixed, known interval $\tau$ apart.


The orchestrator must satisfy five correctness properties.
The first three establish the correctness baseline.
Totality ensures that an opening propagates to all correct validators: if any correct validator opens a slot $s$, then every correct validator eventually opens $s$ as well.
Integrity ensures that no correct validator opens the same slot more than once, and that no slot is opened before its starting time (though it may be opened after that time).
Monotonicity ensures that validators open slots in strictly increasing order.
Let us underline a stronger consequence of these three properties together: whenever a correct validator $\proc_i$ opens a slot $s$, every correct validator opens exactly the same set of slots with number at most $s.\fnumber$.
To see this, consider any slot $s'$ with $s'.\fnumber \leq s.\fnumber$.
First, note that $s'$ can be opened at most once (due to the integrity property).
If $\proc_i$ opens $s'$, then by totality every correct validator opens $s'$ as well.
Otherwise, $\proc_i$ skips $s'$; that is, it opens $s$ without ever opening $s'$ (here, $s'.\fnumber < s.\fnumber$).
We argue that then no correct validator opens $s'$ either: suppose some correct validator $\proc_j$ did: by totality, $\proc_i$ would open $s'$ too, and by monotonicity --- as $s'.\fnumber < s.\fnumber$ --- it would have to do so before opening $s$, which is impossible, since $\proc_i$ opened $s$ without ever opening $s'$.
In both cases, $\proc_i$ and every other correct validator agree on whether $s'$ is opened.


The remaining two properties are the most involved and the most distinctive: rather than establishing a correctness baseline, as the three above do, they measure the quality of an orchestrator. 

\paragraph{Memory consumption in our framework.}
\label{subsection:memory}
In our extreme-pipelining framework (as we detail in \Cref{subsection:composition_cadence}), a validator opens a fresh slot consensus instance for every slot it does not skip, so the more slots stay open, the more memory it must devote to them; we therefore want to keep each validator's memory bounded.
Formally reasoning about memory, however, is delicate.
To make a memory bound meaningful, one must work in a model with an \emph{unreliable} network: under a reliable network, undelivered messages must be buffered until delivery, which is fundamentally at odds with bounding memory.
Yet adopting an unreliable network would introduce considerable additional machinery (retransmission, garbage collection, checkpointing, and similar mechanisms), whose techniques are, at least in theory, standard and largely orthogonal to our contribution.
As memory management is not where we innovate, we retain the reliable-network assumption throughout, and consequently do not reason about memory consumption directly.
Instead, we capture the same concern through a clean abstraction: the number of slot consensus instances in which a validator is \emph{actively participating} at any given time, namely those it has started but not yet abandoned.
(Recall that, by the quiescence property of slot consensus, a validator sends messages only for the instances it is actively participating in --- so this count captures exactly the instances it is actively communicating for at any given time.)
Bounding this number therefore bounds the ongoing work a validator sustains at once; we adopt it as our proxy for memory.
We call this property \emph{bounded concurrency}. 
Reassuringly, it is not a guarantee we need to establish on its own: a validator participates in a slot's consensus instance exactly while that slot is open but not yet completed, so bounded concurrency follows directly from the $\mathcal{B}$-boundedness property we define next --- and with the very same bound $\mathcal{B}$.

\paragraph{$\mathcal{B}$-Boundedness.}
As discussed in the previous paragraph, bounded concurrency reduces to $\mathcal{B}$-boundedness, which we now define.
This property limits how far ahead a validator may open slots without completing earlier ones: at any point in time, only the last $\mathcal{B}$ slots opened (by slot number) may still be pending completion, while all previously opened slots must already have been completed.
In the context of our extreme-pipelining framework, since completing a slot means its proposal vector has been finalized, $\mathcal{B}$-boundedness implies that every correct validator has already finalized all but the last $\mathcal{B}$ opened slots, yielding a complete log prefix that trails the frontier by at most $\mathcal{B}$ slots.
Here, $\mathcal{B} \in \mathbb{N}_{\geq 0} \cup \{\infty\}$: a finite $\mathcal{B}$ bounds the number of simultaneously open slots, whereas $\mathcal{B} = \infty$ imposes no such bound and thus captures orchestrators that open slots arbitrarily far ahead without completing earlier ones.
We underline once more that $\mathcal{B}$ measures the quality of a protocol (smaller is preferable) and need not be a constant: it may depend on the model and execution parameters, such as $n$, $f$, $\Delta$, and $\mathrm{GST}$.

\paragraph{$\mathcal{R}$-Recovery.}
This property captures what we require once the network stabilizes after a period of asynchrony.
Two demands must be met jointly.
First, validators must stop skipping slots: once the network recovers, all upcoming slots are opened rather than skipped.
Second, and more subtly, validators must open slots ``on time'', namely at exactly their starting time.
Opening a slot on time is a natural expectation: validators may have time-sensitive work to perform upon opening, making the precise timing of this event crucial.
To see why this matters, we illustrate with the concrete example of our extreme-pipelining framework.
There, the proposal-inclusion property of slot consensus guarantees that every proposal submitted by a correct proposer \emph{at the slot's starting time} appears in the finalized proposal vector.
If the slot is opened late, the proposer misses the starting time, proposal inclusion no longer applies, and Byzantine validators may suppress correct proposers' proposals entirely.
This would cause slot consensus to finalize a proposal vector devoid of any correct proposer's proposals --- potentially an empty one --- which, in the blockchain setting we target, leaves the application state unchanged and is therefore no more useful than a skipped slot, yet far costlier to produce.
An orchestrator that merely guarantees eventual opening, without the timing condition, would therefore be insufficient.
Combining both demands, we say the orchestrator satisfies $\mathcal{R}$-recovery if both conditions hold for every slot whose starting time is at least $\mathcal{R}$ after $\mathrm{GST}$.
Unlike $\mathcal{B}$, $\mathcal{R}$ must be finite; but, like $\mathcal{B}$, it captures protocol quality (smaller is preferable) and may be any function of the model and execution parameters.



\begin{module}[h]
\caption{Orchestrator}
\label{mod:orchestrator_2}
\footnotesize
\begin{algorithmic}[1]

    

\Statex \textbf{Interface:}
\begin{compactitem}[-]
    \item input $\mathsf{complete}(s \in \mathsf{Slot})$: a validator completes slot $s$.
    \item output $\mathsf{open}(s \in \mathsf{Slot})$: a validator opens slot $s$.
\end{compactitem}

    

\smallskip
\Statex \textbf{Properties:} 
\begin{compactitem}[-]
    \item \emph{Totality:} If some correct validator opens any slot $s$, then every correct validator eventually opens $s$.

    \item \emph{Integrity:} No correct validator opens the same slot $s$ more than once.
    Moreover, no correct validator opens slot $s$ before time $s.\fdeadline - \Delta$.

    \item \emph{Monotonicity:} If a correct validator opens any two slots $s$ and $s'$ with $s.\fnumber < s'.\fnumber$, then it opens $s$ before opening $s'$.
    
    

    \item \emph{$\mathcal{B}$-Boundedness:}
    Let $\mathcal{B} \in \mathbb{N}_{\geq 0} \cup \{\infty\}$. For every correct validator $\proc_i$ and every time $t$, if $\proc_i$ has opened $k$ slots by time $t$ (ordered by slot number as $s_1, \ldots, s_k$), then every $s_j$ with $j \leq k - \mathcal{B}$ has already been completed by $\proc_i$.
    
    \item \emph{$\mathcal{R}$-Recovery:} Let $\mathcal{R}$ be a time duration.
    For every slot $s$ with $s.\fdeadline - \Delta \geq \mathrm{GST} + \mathcal{R}$, every correct validator opens slot $s$ and it does so at time $s.\fdeadline - \Delta$.

\end{compactitem}


\end{algorithmic}
\end{module}

\subsection{Composing the Building Blocks} \label{subsection:composition_cadence}

We now describe how the slot consensus and the orchestrator compose into the extreme-pipelining framework.
Each correct validator maintains a single orchestrator instance $\mathcal{O}$ and one slot consensus instance $\mathcal{S}[s]$ per slot $s$, and the framework is simply the glue that wires these components together.
Throughout, $\mathcal{O}$ is assumed to achieve $\mathcal{B}$-boundedness and $\mathcal{R}$-recovery for some $\mathcal{B} \in \mathbb{N}_{\geq 0} \cup \{\infty\}$ and some time duration $\mathcal{R}$.
The pseudocode is given in \Cref{algorithm:cadence}, traced from the perspective of a correct validator $\proc_i$, with a complementary visual overview in \Cref{fig:cadence-arch}.

\paragraph{Protocol description.}
Upon starting, $\proc_i$ begins executing $\mathcal{O}$ (line~\ref{line:start-orch}), which drives the protocol forward by deciding, for each slot $s$, whether to open it; recall that the orchestrator has no explicit skip output, so a slot it never opens is implicitly skipped.
When $\mathcal{O}$ opens slot $s$ (line~\ref{line:upon-open}), $\proc_i$ records $s$ as opened and starts participating in $\mathcal{S}[s]$ (lines~\ref{line:opened-update}--\ref{line:participate}); if $\proc_i$ is a designated proposer (i.e., $\proc_i \in s.\fproposers$), it additionally submits its proposal $P_s$ for slot $s$ to $\mathcal{S}[s]$ (lines~\ref{line:proposer-check}--\ref{line:propose}). 
Whereas opening is signaled explicitly by $\mathcal{O}$, skipping is not: when $\proc_i$ opens $s$, it records every smaller-numbered slot it has not opened as skipped (line~\ref{line:implicit-skip}).
For a skipped slot, no consensus instance is spawned and no proposal vector is produced.

Crucially, multiple slot consensus instances may run concurrently and may finalize out of order.
Upon $\mathcal{S}[s]$ finalizing a proposal vector $V$, for some slot $s$ (line~\ref{line:upon-finalize}), $\proc_i$ places $V$ in a pending set (line~\ref{line:pending-add}), notifies $\mathcal{O}$ that slot $s$ has been completed (line~\ref{line:complete}), and stops participating in $\mathcal{S}[s]$ (line~\ref{line:abandon}).
Pending proposal vectors are then appended to $\proc_i$'s local log in slot-number order: a pending proposal vector $V'$ is appended (line~\ref{line:append}) as soon as every slot with a strictly smaller number than $V'.\fslot.\fnumber$ has already been resolved, that is, either recorded as skipped or finalized and appended to $\mathit{log}_i$ (line~\ref{line:func-ready-to-append-return}).
This ordering condition ensures that $\mathit{log}_i$ is indeed a log in the sense of \Cref{subsection:mcp-preliminaries}: proposal vectors appear in strictly increasing order of slot number.


\begin{algorithm}[h]
\caption{Extreme-Pipelining Framework: Pseudocode (for validator $p_i$)}
\label{algorithm:cadence}
\begin{algorithmic}[1]
\footnotesize


\State \textbf{Uses:}
\State \hskip2em Slot consensus, \textbf{instances} $\mathcal{S}[s]$, for every $s \in \Slot$, parameterized by slot $s$ \label{line:uses-sc}
\State \hskip2em Orchestrator ($\mathcal{B}$-boundedness, $\mathcal{R}$-recovery), \textbf{instance} $\mathcal{O}$ \label{line:uses-orch}

\medskip
\State \textbf{Local variables:}
\State \hskip2em $\mathsf{Log}$ $\mathit{log}_i \gets []$ \BlueComment{append-only local log of $\proc_i$} \label{line:var-log}
\State \hskip2em $\mathsf{Set}(\Slot)$ $\mathit{opened}_i \gets \emptyset$ \BlueComment{slots opened by the orchestrator} \label{line:var-opened}
\State \hskip2em $\mathsf{Set}(\Slot)$ $\mathit{skipped}_i \gets \emptyset$ \BlueComment{slots skipped by the orchestrator} \label{line:var-skipped}
\State \hskip2em $\mathsf{Set}(\mathsf{PVector})$ $\mathit{pending}_i \gets \emptyset$ \BlueComment{finalized proposal vectors not yet appended to $\mathit{log}_i$} \label{line:var-pending}


\medskip
\State \textbf{Local functions:}
\State \label{line:func-ready-to-append} \hskip2em \textbf{function} $\mathsf{ready\_to\_append}(V \in \mathsf{PVector}) \to \mathsf{Bool}$:
\State \hskip4em \textbf{return} $\mathsf{true}$ if and only if every slot $s'$ with $s'.\fnumber < V.\fslot.\fnumber$ satisfies one of: \label{line:func-ready-to-append-return}
\Statex \hskip6em (1) $s' \in \mathit{skipped}_i$, or
\Statex \hskip6em (2) some proposal vector $V' \in \mathit{log}_i$ has $V'.\fslot = s'$

\medskip
\State \textbf{upon} starting the protocol: \label{line:startup}
\State \hskip2em start executing $\mathcal{O}$ \label{line:start-orch}

\medskip
\State \textbf{upon} $\mathcal{O}.\mathsf{open}(s \in \Slot)$: \label{line:upon-open}
\State \hskip2em $\mathit{opened}_i \gets \mathit{opened}_i \cup \{ s \}$ \label{line:opened-update}
\State \hskip2em $\mathit{skipped}_i \gets \mathit{skipped}_i \cup \{ s' \in \Slot : s'.\fnumber < s.\fnumber \text{ and } s' \notin \mathit{opened}_i \}$ \BlueComment{implicitly skipped slots} \label{line:implicit-skip}
\State \hskip2em \textbf{invoke} $\mathcal{S}[s].\mathsf{participate()}$ \label{line:participate}
\State \hskip2em \textbf{if} $\proc_i \in s.\fproposers$: \label{line:proposer-check}
\State \hskip4em \textbf{invoke} $\mathcal{S}[s].\mathsf{propose}(P_s)$, where $P_s$ is $p_i$'s proposal for slot $s$ \label{line:propose}


\medskip
\State \textbf{upon} $\mathcal{S}[s].\mathsf{finalize}(V \in \mathsf{PVector})$, for some $s \in \mathit{opened}_i$: \BlueComment{fires once $s$ is opened; ``early'' finalizations buffered} \label{line:upon-finalize}
\State \hskip2em $\mathit{pending}_i \gets \mathit{pending}_i \cup \{ V \}$ \label{line:pending-add}
\State \hskip2em \textbf{invoke} $\mathcal{O}.\mathsf{complete}(V.\fslot)$ \label{line:complete}
\State \hskip2em \textbf{invoke} $\mathcal{S}[s].\mathsf{abandon}()$ \label{line:abandon}

\medskip
\State \textbf{upon} $\mathsf{ready\_to\_append}(V' \in \mathsf{PVector})$, for some proposal vector $V' \in \mathit{pending}_i$: \label{line:upon-ready-to-append}
\State \hskip2em \textbf{append} $V'$ to $\mathit{log}_i$ \label{line:append}
\State \hskip2em $\mathit{pending}_i \gets \mathit{pending}_i \setminus \{ V' \}$ \label{line:pending-remove}
\end{algorithmic}
\end{algorithm}


\begin{figure}[h]
\centering
\begin{tikzpicture}[
  orch/.style={
    draw=teal!60!black, line width=1.2pt, fill=teal!12,
    rounded corners=6pt, minimum width=2.4cm, minimum height=7.6cm,
    align=center, font=\small\bfseries
  },
  scbox/.style={
    draw=orange!65!black, line width=1pt, fill=orange!15,
    rounded corners=4pt, minimum width=2.8cm, minimum height=0.7cm,
    align=center, font=\scriptsize
  },
  scprog/.style={
    draw=orange!45!black, line width=1pt, fill=orange!7, dashed,
    rounded corners=4pt, minimum width=2.8cm, minimum height=0.7cm,
    align=center, font=\scriptsize
  },
  blockbox/.style={
    draw=green!55!black, line width=1pt, fill=green!10,
    rounded corners=3pt, minimum width=0.9cm, minimum height=0.6cm,
    align=center, font=\scriptsize
  },
  blockgray/.style={
    draw=gray!40, line width=0.8pt, fill=gray!8, dashed,
    rounded corners=3pt, minimum width=0.9cm, minimum height=0.6cm,
    align=center, font=\scriptsize, text=gray!50
  },
  openarr/.style={->, >=stealth, teal!70!black, thick},
  skiparr/.style={->, >=stealth, teal!50!black, thick, dashed},
  complarr/.style={->, >=stealth, orange!80!black, thick},
  finarr/.style={->, >=stealth, green!55!black},
  parentarr/.style={->, >=stealth, gray!60, semithick},
  outer/.style={draw=blue!25!gray, line width=1pt, fill=blue!4, rounded corners=8pt},
]
  \node[outer, minimum width=11.6cm, minimum height=9.2cm]
    at (3.8, -0.8) {};
  \node[font=\small\bfseries, text=blue!40!black] at (3.8, 3.4) {Extreme pipelining};

  \node[orch] (O) at (0, -0.8) {Orchestrator \\ $\mathcal{O}$};

  \node[scbox]  (S1) at (5.0,  2.3) {$\mathcal{S}[s_1]$: Slot Consensus};
  \node[scbox]  (S2) at (5.0,  0.9) {$\mathcal{S}[s_2]$: Slot Consensus};
  \coordinate   (skip3) at (0, -0.35);
  \node[blockgray] (B3empty) at (8.6, -0.35) {$\emptyset$};
  \node[scbox]  (S4) at (5.0, -1.6) {$\mathcal{S}[s_4]$: Slot Consensus};
  \node[scprog] (S5) at (5.0, -3.0) {$\mathcal{S}[s_5]$: Slot Consensus};
  \node[font=\scriptsize, gray] at (5.0,-3.9) {$\vdots$};

  \draw[openarr] ([yshift=3pt]O.east |- S1.west) -- ([yshift=3pt]S1.west)
    node[midway, above, font=\tiny, teal!70!black] {$\mathsf{open}(s_1)$};
  \draw[openarr] ([yshift=3pt]O.east |- S2.west) -- ([yshift=3pt]S2.west)
    node[midway, above, font=\tiny, teal!70!black] {$\mathsf{open}(s_2)$};
  \node[font=\tiny\itshape, gray!60, align=center] (skiplbl) at (2.7, -0.35) {$s_3$ skipped:\\no $\mathsf{open}(s_3)$ issued};
  \draw[openarr] ([yshift=3pt]O.east |- S4.west) -- ([yshift=3pt]S4.west)
    node[midway, above, font=\tiny, teal!70!black] {$\mathsf{open}(s_4)$};
  \draw[openarr] ([yshift=3pt]O.east |- S5.west) -- ([yshift=3pt]S5.west)
    node[midway, above, font=\tiny, teal!70!black] {$\mathsf{open}(s_5)$};

  \draw[complarr] ([yshift=-3pt]S1.west) -- ([yshift=-3pt]O.east |- S1.west)
    node[midway, below, font=\tiny, orange!80!black] {$\mathsf{complete}(s_1)$};
  \draw[complarr] ([yshift=-3pt]S2.west) -- ([yshift=-3pt]O.east |- S2.west)
    node[midway, below, font=\tiny, orange!80!black] {$\mathsf{complete}(s_2)$};
  \draw[complarr] ([yshift=-3pt]S4.west) -- ([yshift=-3pt]O.east |- S4.west)
    node[midway, below, font=\tiny, orange!80!black] {$\mathsf{complete}(s_4)$};

  \node[blockbox] (B1) at (8.6,  2.3) {$V_1$};
  \node[blockbox] (B2) at (8.6,  0.9) {$V_2$};
  \node[blockbox] (B4) at (8.6, -1.6) {$V_4$};

  \draw[finarr] (S1.east) -- (B1.west)
    node[midway, above, font=\tiny, green!55!black] {$\mathsf{finalize}(V_1)$};
  \draw[finarr] (S2.east) -- (B2.west)
    node[midway, above, font=\tiny, green!55!black] {$\mathsf{finalize}(V_2)$};
  \draw[finarr] (S4.east) -- (B4.west)
    node[midway, above, font=\tiny, green!55!black] {$\mathsf{finalize}(V_4)$};

  \draw[parentarr] (B4.east) to[out=0, in=0, looseness=0.5] (B2.east);
  \draw[parentarr] (B2.east) to[out=0, in=0, looseness=0.7] (B1.east);

\end{tikzpicture}
\caption{Interaction between the orchestrator $\mathcal{O}$ and slot consensus instances $\mathcal{S}[\cdot]$ in the extreme-pipelining framework (\Cref{algorithm:cadence}).
\textcolor{teal!70!black}{Teal solid arrows}: $\mathcal{O}$ outputs $\mathsf{open}(s)$, upon which the framework invokes $\mathsf{participate}()$ on $\mathcal{S}[s]$ (and $\mathsf{propose}(\cdot)$ if $\proc_i \in s.\fproposers$).
\emph{Implicit skip} (slot $s_3$): the orchestrator has no ``skip'' output, so a slot is skipped simply by never being opened --- no $\mathsf{open}(s_3)$ is issued, no slot consensus instance is spawned, and no proposal vector is produced ($\emptyset$).
\textcolor{orange!80!black}{Orange arrows}: upon $\mathsf{finalize}(V)$, the framework invokes $\mathsf{complete}(s)$ on $\mathcal{O}$ and $\mathsf{abandon}()$ on $\mathcal{S}[s]$.
\textcolor{green!55!black}{Green arrows}: $\mathcal{S}[s]$ outputs $\mathsf{finalize}(V)$; the proposal vector is eventually appended to the local log.
Slot $s_5$ (dashed border) is still in progress.
}
\label{fig:cadence-arch}
\end{figure}
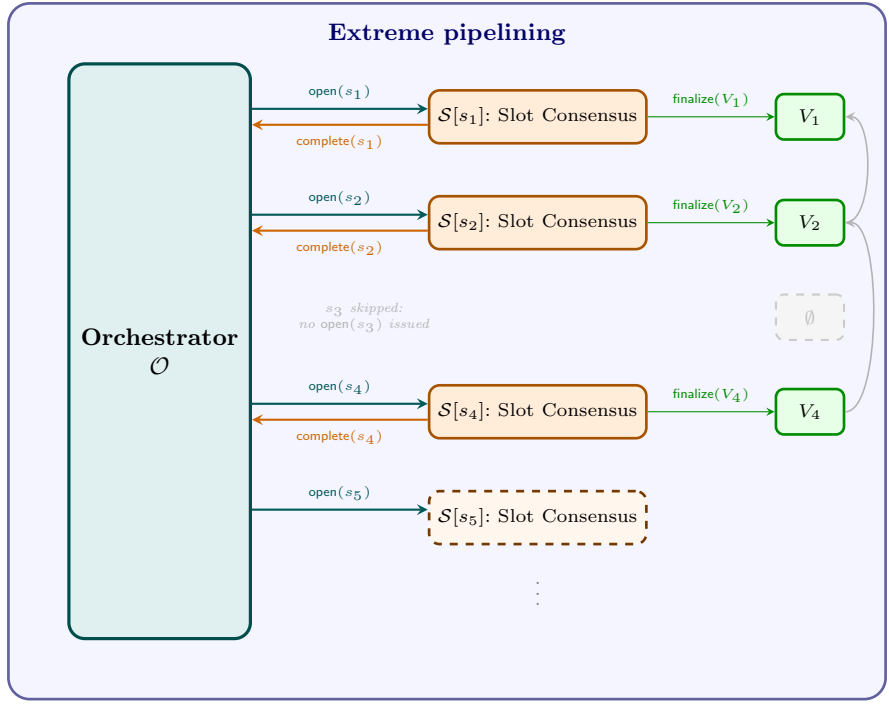

\paragraph{On the generality of the framework.}
We wish to highlight that extreme pipelining is intentionally presented as a generic framework, not merely as a vehicle for our specific instantiation \cadence.
The framework is deliberately minimal: it prescribes nothing more than the interaction between a slot consensus instance and an orchestrator, leaving both components entirely abstract and subject only to the properties stated in \cref{mod:slotconsensus,mod:orchestrator_2}.
As a direct consequence, every property of the framework (including its resilience to Byzantine faults) is inherited from the properties of the underlying components rather than hard-coded into the framework itself.
In this paper, we indeed focus on the standard Byzantine fault model with $n = 3f + 1$ validators.
However, the framework imposes no such constraint: to obtain a different protocol, it suffices to supply a slot consensus and an orchestrator satisfying the respective module specifications, and the correctness of the resulting protocol follows immediately from their properties, with no additional argument needed.
This flexibility extends well beyond the resilience threshold.
For instance, one may trade resilience for another guarantee such as improved performance, simply by instantiating the two components accordingly.
One may also depart from the multiple-concurrent-proposers setting entirely: instantiating the slot consensus with a single designated proposer per slot recovers a classical ``single-leader'' blockchain.
Likewise, one is free to plug in a linear-communication slot consensus, and nothing else changes.
In every case, the framework itself is untouched; only the two plugged-in components change.

\subsection{Proof} \label{subsection:correctness_cadence}

In this subsection, we prove that the extreme-pipelining framework introduced above, composed of the slot consensus and the orchestrator, solves the MCP problem.

\paragraph{Safety.}
We first establish that the safety property is satisfied.

\begin{lemma}[Safety]
\label{lemma:cadence-safety}
\Cref{algorithm:cadence} satisfies safety.
\end{lemma}
\begin{proof}
Suppose, for contradiction, that two correct validators $\proc_i, \proc_j$ and two times $t_1, t_2$ exist such that $\locallog(\proc_i, t_1)$ and $\locallog(\proc_j, t_2)$ are inconsistent.
Let $k \in \mathbb{N}_{\geq 1}$ be the smallest index at which the two logs differ, i.e., $\locallog(\proc_i, t_1)[k] \neq \locallog(\proc_j, t_2)[k]$; such an index exists by the definition of inconsistency.
Let $V_1 = \locallog(\proc_i, t_1)[k]$ and $V_2 = \locallog(\proc_j, t_2)[k]$.
We distinguish two possible cases:
\begin{compactitem}
  \item Let $V_1.\fslot \neq V_2.\fslot$.
  The $\mathsf{ready\_to\_append}$ condition (line~\ref{line:func-ready-to-append}) ensures that a validator appends a proposal vector for a slot only after every smaller-numbered slot is either skipped or already covered by a proposal vector in its log; hence the slot at position $k$ is determined by how each validator resolved all preceding slots.
  Since both logs agree at positions $1, \ldots, k-1$, the two validators resolved all those slots identically.
  Assume, without loss of generality, that $V_1.\fslot.\fnumber < V_2.\fslot.\fnumber$.
  Then, $\proc_i$ appended a proposal vector for slot $V_1.\fslot$, and so opened that slot.
  Validator $\proc_j$, in contrast, appended a proposal vector for the higher-numbered slot $V_2.\fslot$ at position $k$ without appending one for $V_1.\fslot$, and so did not open $V_1.\fslot$.
  By the monotonicity property, having already opened the higher-numbered slot $V_2.\fslot$, $\proc_j$ never opens $V_1.\fslot$.
  This contradicts the totality property of the orchestrator $\mathcal{O}$: since the correct validator $\proc_i$ opened $V_1.\fslot$, every correct validator --- including $\proc_j$ --- must eventually open it.

  \item Let $V_1.\fslot = V_2.\fslot$.
  Since $V_1 \neq V_2$ while $V_1.\fslot = V_2.\fslot$, the two proposal vectors must differ in their mappings.
  Note that a validator appends a proposal vector only upon finalizing it (line~\ref{line:upon-finalize}), and by the slot safety property of slot consensus a proposal vector finalized for slot $V_1.\fslot = V_2.\fslot$ comes from the instance $\mathcal{S}[V_1.\fslot = V_2.\fslot]$.
  Thus $\proc_i$ and $\proc_j$ finalized $V_1$ and $V_2$, respectively, from the same instance $\mathcal{S}[V_1.\fslot = V_2.\fslot]$, contradicting its agreement property. \qed
\end{compactitem}
\end{proof}

\paragraph{Liveness.}
Next, we prove that our framework satisfies $\mathcal{R}$-liveness, where $\mathcal{R}$ is the orchestrator's recovery parameter.

\begin{lemma} [Liveness]
\label{lemma:cadence-liveness}
\Cref{algorithm:cadence} satisfies $\mathcal{R}$-liveness.
\end{lemma}
\begin{proof}
Fix any slot $s$ with $s.\fdeadline - \Delta \geq \mathrm{GST} + \mathcal{R}$.
By the $\mathcal{R}$-recovery property of the orchestrator $\mathcal{O}$, every correct validator eventually opens $s$ (line~\ref{line:upon-open}).
By the monotonicity property of $\mathcal{O}$, before opening $s$, every correct validator has already resolved --- opened or (implicitly) skipped --- every slot $s'$ with $s'.\fnumber < s.\fnumber$.
Moreover, by the totality property of $\mathcal{O}$, for every slot $s^{\star}$ with $s^{\star}.\fnumber \in [1, s.\fnumber]$, either (1) every correct validator opens $s^{\star}$ and begins participating in the slot consensus $\mathcal{S}[s^{\star}]$ (line~\ref{line:participate}), or (2) no correct validator opens $s^{\star}$, i.e., every correct validator skips it (line~\ref{line:implicit-skip}).
In case (1), since every correct validator starts participating in $\mathcal{S}[s^{\star}]$, the termination property of $\mathcal{S}[s^{\star}]$ guarantees that every correct validator eventually finalizes a proposal vector for $s^{\star}$ (line~\ref{line:upon-finalize}).
Inductively, starting from slot number $1$, each correct validator resolves slot $s^{\star}$ (by either skipping it or finalizing and appending a proposal vector for it) before proceeding to slot $s^{\star} + 1$: once slot $s^{\star}$ is resolved, the $\mathsf{ready\_to\_append}$ condition (line~\ref{line:func-ready-to-append-return}) for a proposal vector associated with slot $s^{\star} + 1$ is satisfied, and the next proposal vector can be appended (line~\ref{line:append}).
Applying this argument inductively up to slot number $s.\fnumber - 1$, every correct validator eventually resolves all preceding slots, at which point the $\mathsf{ready\_to\_append}$ condition for a proposal vector associated with slot $s$ is satisfied.
Therefore, every correct validator eventually appends a proposal vector $V$ with $V.\fslot = s$ (line~\ref{line:append}), establishing $\mathcal{R}$-liveness. \qed
\end{proof}

\noindent As noted earlier, safety and liveness are standard properties, not specific to the MCP problem.
The fact that the extreme-pipelining framework satisfies both underscores its applicability beyond the MCP setting.
We now turn to the MCP-specific properties.

\paragraph{Censorship resistance.}
We next prove the $\mathcal{R}$-censorship resistance property.

\begin{lemma} [Censorship resistance]
\Cref{algorithm:cadence} satisfies $\mathcal{R}$-censorship resistance.
\end{lemma}
\begin{proof}
Fix any slot $s$ with $s.\fdeadline - \Delta \geq \mathrm{GST} + \mathcal{R}$, any correct validator $\proc_i$, and any correct proposer $\proc_j \in s.\fproposers$; let $P_s$ denote $\proc_j$'s proposal for slot $s$.
To prove the lemma, we must exhibit a proposal vector $V$ and a time $t$ such that $V.\fslot = s$, $V[\proc_j] = P_s$, and $V \in \locallog(\proc_i, t)$.

By the $\mathcal{R}$-liveness of the extreme-pipelining framework (\Cref{lemma:cadence-liveness}), $\proc_i$ eventually finalizes (line~\ref{line:upon-finalize}) and appends (line~\ref{line:append}) a proposal vector $V$ with $V.\fslot = s$ to its local log, at some time $t$.
By the $\mathcal{R}$-recovery property of the orchestrator $\mathcal{O}$, $\proc_j$ opens slot $s$ at time $s.\fdeadline - \Delta$ (line~\ref{line:upon-open}) and, being a correct proposer, submits its proposal $P_s$ to $\mathcal{S}[s]$ at that same time (line~\ref{line:propose}).
Since $s.\fdeadline - \Delta \geq \mathrm{GST}$ and $\proc_j$ proposes by the slot's starting time, the proposal-inclusion property of $\mathcal{S}[s]$ ensures that the proposal vector $V$ finalized by $\proc_i$ satisfies $V[\proc_j] = P_s$, which concludes the proof. \qed
\end{proof}

\paragraph{Hiding.}
Next, we prove hiding, which reduces to the hiding of each slot consensus instance.

\begin{lemma} [Hiding]
\Cref{algorithm:cadence} satisfies hiding.
\end{lemma}
\begin{proof}
A correct proposer submits its proposal for slot $s$ solely to $\mathcal{S}[s]$ (line~\ref{line:propose}), and a correct validator appends a proposal vector for slot $s$ only upon $\mathcal{S}[s]$ finalizing it (lines~\ref{line:upon-finalize} and~\ref{line:append}). Proposal contents are thus observable only through the per-slot instances, so hiding follows by composing the hiding of each. \qed
\end{proof}

\paragraph{Bounded concurrency.}
Recall from \Cref{subsection:memory} that we do not reason about memory directly, but use the number of slot consensus instances a validator actively participates in as a proxy for its memory consumption. This number stays bounded:

\begin{lemma} [Bounded concurrency]
\label{lemma:cadence-bounded-concurrency}
\Cref{algorithm:cadence} satisfies $\mathcal{B}$-bounded concurrency.
That is, for every correct validator $\proc_i$ and every time $t$, at most $\mathcal{B}$ slot consensus instances $\mathcal{S}[s]$ exist such that $\proc_i$ has started participating in $\mathcal{S}[s]$ by time $t$ but has not yet abandoned $\mathcal{S}[s]$ by time $t$.
\end{lemma}
\begin{proof}
For any slot $s$, validator $\proc_i$ starts participating in $\mathcal{S}[s]$ upon opening $s$ (line~\ref{line:upon-open}--\ref{line:participate}), and abandons $\mathcal{S}[s]$ immediately after completing $s$ (line~\ref{line:abandon}).
Hence, $\mathcal{S}[s]$ is active for $\proc_i$ at time $t$ if and only if $\proc_i$ has opened $s$ but not yet completed it by time $t$.
The bound of $\mathcal{B}$ therefore follows directly from the $\mathcal{B}$-boundedness property of the orchestrator $\mathcal{O}$.
\hfill$\square$
\end{proof}


%% file: p2_chorus.tex
\section{\chorus: Our Slot Consensus for \cadence}
\label{section:slot_agreement}

This section presents \chorus, the protocol that instantiates the slot-consensus primitive of \cadence.
We introduced \chorus informally, and in some detail, in \Cref{section:chorus-overview}; here we give its full specification.

\subsection{Cryptographic Primitives} \label{appendix:crypto}

\chorus builds on several standard cryptographic primitives, whose notation we fix here.

\paragraph{Hash function.}
We assume a collision-resistant cryptographic hash function
\[
    H : \{0,1\}^* \rightarrow \{0,1\}^\lambda .
\]

\paragraph{Digital signatures.}
Each validator $\proc_i$ holds a signing keypair $(sk_i,pk_i)$. We use the following operations:
\begin{compactitem}
    \item $\Sign_i(m)$: On input message $m$, outputs a digital signature $\sigma$ computed using $sk_i$.
    \item $\Verify_{pk_i}(m,\sigma)$: Outputs $\top$ if $\sigma$ is a valid signature on $m$ under public key $pk_i$, and outputs $\bot$ otherwise.
\end{compactitem}
A set of signatures $\{\sigma_i\}_{i \in \mathcal{I}}$ on a common message $m$ can be aggregated into a short multi-signature $\Sigma$ verifiable against $\{pk_i\}_{i \in \mathcal{I}}$.

\paragraph{Erasure coding.}
The encrypted proposal is erasure-coded to enable distributed availability. 
Throughout the section, we rely on the following operations:
\begin{compactitem}
    \item $\Encode(c)$: On input ciphertext $c$, outputs $n$ fragments $(\cdata_1,\ldots,\cdata_n)$ such that any subset of at least $f+1$ fragments is sufficient to reconstruct $c$.

    \item $\Decode(\{\cdata_i\})$: On input a set of fragments $\{\cdata_i\}$ with $|\{\cdata_i\}| \ge f+1$, outputs the reconstructed ciphertext $c$, or $\bot$ if decoding fails.
\end{compactitem}

\paragraph{Merkle trees.}
Merkle trees are used to commit to encoded payload chunks. We assume the following interface:
\begin{compactitem}
    \item $\MerkleRoot(c_1,\ldots,c_n)$: Outputs the root $\roothash$ of the Merkle tree over leaves $c_1,\ldots,c_n$.
    \item $\MerkleProof(i)$: Outputs a Merkle authentication path $\pi_i$ for chunk $c_i$.
    \item $\VerifyMerkle(\roothash,c_i,\pi_i)$: Outputs $\top$ if $\pi_i$ is a valid proof that $c_i$ is a leaf in the Merkle tree with root $\roothash$, and outputs $\bot$ otherwise.
\end{compactitem}

\paragraph{Random-oracle convention.}
For the simulation argument in~\S\ref{appendix:encryption}, we model several hash functions as distinct random oracles, including the identity hash $H_{\mathsf{id}}$, the pad-derivation hash $H_{\mathsf{pad}}$, and the Merkle hash $H_{\mathsf{Merkle}}$.  The simulator programs $H_{\mathsf{pad}}$; it does not program $H_{\mathsf{Merkle}}$, but it inspects the adversary's Merkle-oracle queries to recover the leaves committed by a submitted Merkle root.  In an implementation, all such functions should be realized through a standard domain-separated hash API, using unambiguous encodings of all inputs.  Untagged calls to a raw hash function, or ad hoc concatenation of values before hashing, should be avoided.  The Merkle hash domain should be separate from the identity-hash and pad-derivation domains; ideally, the Merkle implementation should also distinguish leaf hashing from internal-node hashing.

\subsection{Proposal Encryption and Slot-Key Release} \label{appendix:encryption}

This subsection describes the cryptographic mechanism used to hide proposal contents until the slot's deadline.  The important point is that the threshold operation is not proposal-specific decryption.  Instead, validators release shares of the \emph{slot secret key}: the identity secret key for an identity derived from the slot~$\slot$, with standard domain separation for the protocol instance.  Once this slot key is reconstructed, anyone can open any well-formed proposal ciphertext for that slot.

\paragraph{Identity-based KEM with distributed extraction.}
We use an identity-based key encapsulation mechanism (IB-KEM) whose private-key extraction algorithm is implemented by a distributed private-key generator.  The ordinary IB-KEM syntax is
\[
    (\mathit{mpk},\mathit{msk}) \leftarrow \mathsf{Setup}(1^\lambda),
    \qquad
    \mathit{sk}_{\mathit{id}} \leftarrow \mathsf{Extract}(\mathit{msk},\mathit{id}),
\]
\[
    (Z,C_{\mathsf{kem}}) \leftarrow \mathsf{Encap}(\mathit{mpk},\mathit{id}),
    \qquad
    Z \leftarrow \mathsf{Decap}(\mathit{sk}_{\mathit{id}},\mathit{id},C_{\mathsf{kem}}).
\]
Here, $Z$ is the KEM's raw key and $C_{\mathsf{kem}}$ is the public encapsulation.  The raw key need not itself be a bit string; in the pairing-based instantiation below it is a target-group element.  Bit strings used by the protocol are derived by hashing $Z$ with appropriate domain separation.

The distributed version replaces only $\mathsf{Extract}$.  A $(t,n)$-distributed extraction mechanism consists of algorithms
\[
    (\mathit{mpk},\mathit{vk},\mathit{st}_1,\ldots,\mathit{st}_n)
        \leftarrow \mathsf{DSetup}(1^\lambda,n,t),
\]
\[
    \sigma_i \leftarrow \mathsf{ShareExtract}(\mathit{st}_i,\mathit{id}),
    \qquad
    b \leftarrow \mathsf{VerifyShare}(\mathit{mpk},\mathit{vk},\mathit{id},i,\sigma_i),
\]
\[
    \mathit{sk}_{\mathit{id}} \leftarrow
    \mathsf{CombineExtract}(\mathit{mpk},\mathit{vk},\mathit{id},\{(i,\sigma_i)\}_{i\in I}),
    \qquad |I|\ge t .
\]
After combining extraction shares to obtain $\mathit{sk}_{\mathit{id}}$, decapsulation is the ordinary non-threshold operation.  Thus extraction shares depend only on the identity $\mathit{id}$, not on any particular proposal ciphertext.  Correctness requires that honestly generated and verified shares combine to the same $\mathit{sk}_{\mathit{id}}$ as the ideal extraction algorithm.  Robustness requires that adversarial shares that pass verification contribute correctly, so that any verified qualified set either combines to the unique correct slot key or is rejected.

The security property we need is weaker than full KEM key indistinguishability.  We require \emph{unpredictability}: before enough valid extraction shares have been released for an identity $\mathit{id}$, no efficient adversary corrupting fewer than the threshold number of extraction servers can compute the raw key $Z$ belonging to a fresh encapsulation under $\mathit{id}$, except with negligible probability.  Equivalently, if a later layer derives a pad as a random-oracle value $H(Z,\ldots)$, then any adversary that learns this pad before release must have queried the random oracle at the hidden point involving~$Z$.

A concrete instantiation is obtained from the Boneh--Franklin/BLS pairing construction.  Let $G$ be a generator, let $a$ be the master secret, and let the master public key contain $aG$.  For identity $\mathit{id}$, let $Q_{\mathit{id}}=H_{\mathsf{id}}(\mathit{id})$ and $\mathit{sk}_{\mathit{id}}=aQ_{\mathit{id}}$.  A distributed private-key generator can Shamir-share $a$: server $i$ returns the verified extraction share $\sigma_i=a_iQ_{\mathit{id}}$, and interpolation reconstructs $aQ_{\mathit{id}}$.  Encapsulation chooses $r$ at random, publishes $C_{\mathsf{kem}}=U=rG$, and computes
\[
    Z=e(Q_{\mathit{id}},aG)^r .
\]
Decapsulation with the reconstructed identity secret key computes
\[
    Z=e(\mathit{sk}_{\mathit{id}},U)=e(aQ_{\mathit{id}},rG).
\]
Unpredictability of this raw key follows, in the random-oracle model for the identity hash, from the standard computational-BDH assumption and the usual Boneh--Franklin proof strategy.  The protocol never uses the target-group element $Z$ directly as a symmetric key; it hashes $Z$ with appropriate domain separation.

\paragraph{Encrypting proposals.}
For a slot $\slot$, let $\mathit{id}_\slot$ be the canonical encoding of $\slot$ as an IB-KEM identity.  A proposer $P$ encapsulates to this identity, obtaining
\[
    (Z,C_{\mathsf{kem}}) \leftarrow \mathsf{Encap}(\mathit{mpk},\mathit{id}_\slot).
\]
It then derives a pad from the raw key, the slot, the proposer identity, and any additional public context bound to this encryption:
\[
    K = H_{\mathsf{pad}}(Z,\slot,P,\mathsf{ctx}).
\]
The inputs $Z$, $\slot$, and $P$ are essential.  The value $\mathsf{ctx}$ denotes any remaining public context for this encryption, and may include $C_{\mathsf{kem}}$, protocol-version information, and deployment or chain identifiers.
If the proposal payload is $M$, the encrypted payload is the xor mask
\[
    C_{\mathsf{pay}} = K \oplus M,
\]
where $H_{\mathsf{pad}}$ is viewed as a hash-to-bytes or XOF-style oracle whose output is expanded as needed to match the payload length.  We deliberately use this xor-with-random-oracle form, rather than an authenticated symmetric encryption scheme, because it gives the simulator the equivocation needed in the hiding proof: an honestly generated ciphertext can initially be a random bit string and later be made to open to the ideal-world proposal by programming $H_{\mathsf{pad}}$ at the appropriate point.

The object encoded for data availability contains at least $(C_{\mathsf{kem}},C_{\mathsf{pay}})$ and the necessary public context.  This object is erasure-coded, the fragments are committed by a Merkle tree, and the proposer signs the resulting Merkle root.  Thus the submitted encrypted proposal is fixed by a signed Merkle root, not by a plaintext.  A proposal is considered submitted before the deadline only if at least one honest validator receives, before $\slot.\fdeadline$, a valid signature by the claimed proposer on the corresponding Merkle root.  Signature unforgeability ensures that, except with negligible probability, the adversary cannot create a submitted proposal that appears to come from an honest proposer.

\paragraph{Releasing the slot key.}
At the slot's deadline, validators publish extraction shares
\[
    \sigma_i \leftarrow \mathsf{ShareExtract}(\mathit{st}_i,\mathit{id}_\slot)
\]
for the slot identity.  Once enough verified shares have been collected, anyone computes
\[
    \mathit{sk}_{\mathit{id}_\slot} \leftarrow
    \mathsf{CombineExtract}(\mathit{mpk},\mathit{vk},\mathit{id}_\slot,\{(i,\sigma_i)\}_{i\in I}).
\]
The same reconstructed slot key opens all well-formed proposal encapsulations for slot~$\slot$: for each decoded proposal ciphertext, parties compute
\[
    Z \leftarrow \mathsf{Decap}(\mathit{sk}_{\mathit{id}_\slot},\mathit{id}_\slot,C_{\mathsf{kem}}),
    \qquad
    M = C_{\mathsf{pay}} \oplus H_{\mathsf{pad}}(Z,\slot,P,\mathsf{ctx}).
\]
The released values are therefore extraction shares for the slot identity, not decryption shares tied to individual proposals.

\paragraph{Simulation argument.}
We sketch why this realizes the hiding property (\Cref{def:hiding}).  The proof is best viewed as an ordering argument at the deadline of a fixed slot~$\slot$; no infinitesimal notion of time is needed.  At the deadline transition, corrupt proposals that were already submitted are fixed first.  Only after those corrupt inputs have been determined does the simulator obtain the honest proposals from the ideal functionality.

Before the deadline, the simulator handles honest proposers as follows.  For each honest proposer, it chooses the public KEM encapsulation as in the real protocol, but it chooses the encrypted payload $C_{\mathsf{pay}}$ as a uniformly random bit string of the appropriate length.  It erasure-codes the resulting encrypted object, computes the Merkle root, signs the root as the honest proposer, and delivers fragments as prescribed by the protocol.  At this point the simulator does not yet know the honest plaintext proposal.

At the deadline transition, the simulator first identifies the corrupt submissions.  A corrupt submission is eligible only if some honest validator received, before $\slot.\fdeadline$, a valid signature on a Merkle root under the verification key of the claimed corrupt proposer.  The simulator then walks back from this root through the transcript of queries to the Merkle random oracle $H_{\mathsf{Merkle}}$.  It must recover all leaves under the signed root.  These leaves must be valid erasure-code fragments and must form a complete codeword encoding a well-formed encrypted proposal.  If any leaves are missing, if the leaves do not form a codeword, or if the decoded object is malformed, the simulator does not treat the root as a properly submitted proposal.  The simulator uses the Merkle oracle only in this inspection mode; unlike $H_{\mathsf{pad}}$, it is not programmed.

For each well-formed corrupt ciphertext submitted by a corrupt proposer $P'$, the simulator obtains the released slot key $\mathit{sk}_{\mathit{id}_\slot}$, decapsulates $C_{\mathsf{kem}}$, and obtains the raw key~$Z$.  It then considers the pad-oracle point
\[
    q=(Z,\slot,P',\mathsf{ctx}).
\]
If $H_{\mathsf{pad}}(q)$ has already been defined, the simulator uses that existing value.  If it has not been defined, the simulator programs it to a fresh random string.  In either case it sets
\[
    M' = C_{\mathsf{pay}} \oplus H_{\mathsf{pad}}(q)
\]
and submits this $M'$ to the ideal functionality as the corrupt proposer's input for slot~$\slot$.  This is the only possible behavior consistent with the already fixed ciphertext and the random-oracle transcript: if the adversary queried the point, the answer is fixed; if it did not, a fresh random answer has exactly the distribution of a random oracle answer.

Only after all such corrupt inputs have been submitted does the simulator receive the honest proposals from the ideal functionality.  For each honest proposer $P$, the simulator now knows both the previously fixed random ciphertext $C_{\mathsf{pay}}$ and the desired plaintext $M$.  It programs the corresponding pad-oracle point by setting
\[
    H_{\mathsf{pad}}(Z,\slot,P,\mathsf{ctx})
        := C_{\mathsf{pay}} \oplus M .
\]
This makes the already disseminated encrypted payload decrypt to the honest proposal supplied by the ideal functionality.

The simulation can fail only if the simulator is forced to program an oracle point that was already defined inconsistently, or if the adversary creates an ambiguity in what was submitted.  For corrupt proposals, the simulator explicitly avoids inconsistency by using any existing pad-oracle value.  For honest proposals, unpredictability of the IB-KEM raw key prevents the adversary from querying the honest pad-oracle point before the release transition, except with negligible probability; the deadline transition is ordered so that there is no intervening adversarial oracle query between release of the slot key and the simulator's programming of honest points.  Including the proposer identity in the pad-oracle input separates honest-proposer points from corrupt-proposer points, and signature unforgeability prevents the adversary from submitting roots under honest proposer identities.  Finally, the Merkle/random-oracle and erasure-code checks ensure that a signed root determines a unique well-formed encrypted proposal; otherwise the root is treated as malformed.  Conditioned on the complement of these bad events, the simulated transcript is distributed as in the real protocol, while corrupt proposals are fixed before honest proposal contents are revealed.

\subsection{Protocol} \label{subsection:chorus-protocol-overview}

We now present the \chorus protocol. It is given as several algorithms, split across the subsubsections below only to separate its concerns; they form a single protocol and should be read as one. 
The \chorus protocol proceeds in three phases: proposers disseminate their proposals (Phase~I), validators broadcast a single proposal vote at the slot's deadline (or when they enter the slot, if later) (Phase~II), and a \metablock's entries are committed either via the fast path or via the fallback path (Phase~III). \Cref{alg:proposer-dissemination} implements Phase~I and runs at proposers only. \Cref{alg:voting} implements Phase~II at every validator. Phase~III is implemented by two modules running in parallel: \Cref{alg:fast-path-certification} aggregates proposal votes into $\FastQC$s and runs the two-round fast commit on a fast \metablock, while \Cref{alg:fallback} drives the fallback path when a fast \metablock cannot be formed: it harvests the strongest available evidence per proposer into a fallback \metablock, feeds it to a \emph{multi-valued validated Byzantine agreement} (MVBA, \Cref{mod:mvba}) --- a primitive in which validators agree on a single externally-valid value among those they propose --- and, after one final round of commit votes, commits the entries of the MVBA output. Cutting across all phases, the data-availability module (\Cref{alg:da}) provides chunk validation, ciphertext reconstruction, decryption, and proposal recovery. 

The different pseudocodes communicate through the shared local state declared in their ``\textbf{Local variables}'' blocks (e.g., $\Ev(\pid)$ and $\mathit{pathVote}$ are shared between \Cref{alg:fast-path-certification} and \Cref{alg:fallback}).
Throughout the pseudocode, we say that $\proc_i$ is \emph{actively participating} in $\slot$ if it has invoked the slot-consensus input $\mathsf{participate}()$ but has not yet invoked $\mathsf{abandon}()$.
As a standing convention, \emph{no rule of \chorus's pseudocode that sends a message is activated unless $\proc_i$ is actively participating}; this covers rules that send directly as well as rules that trigger sending by invoking a sub-protocol input (in particular, the rules proposing to $\mathsf{MVBA}[\slot]$). We state the convention once here, and it applies to every such rule of \Cref{alg:proposer-dissemination,alg:voting,alg:fast-path-certification,alg:fallback,alg:da} without being repeated in the individual rules.
Rules that merely process received messages and update local state are exempt: they may fire even while $\proc_i$ is not actively participating (so, e.g., chunks are ingested on arrival).
A message whose rule is blocked by this convention is not lost: it remains in $\proc_i$'s buffer, and the rule may fire once $\proc_i$ starts participating.

\subsubsection{Dissemination \& Voting} \label{subsection:dissemination}

Each slot $s$ is bounded by a deadline $s.\fdeadline$. Each honest proposer in $s.\fproposers$ (1) constructs a proposal, (2) encrypts it, (3) erasure-codes the ciphertext into \emph{chunks} (one per validator), and (4) disseminates them by time $s.\fdeadline - \Delta$ (in the style of AVID~\cite{AVID} and DispersedSimplex~\cite{DispersedSimplex})\footnote{In practice, the protocol instantiates this step using
Deterministic RaptorCast~\cite{mip10}.}. Under synchrony, every honest validator receives a chunk from every honest proposer by the deadline. Validators record which chunks arrived in time and, at the deadline (or when they enter the slot, if later), broadcast a \emph{proposal vote}: a single message containing a signed positive or negative entry per proposer.

\paragraph{Proposal.}
A \emph{proposal} by proposer $\proc_\pid$ in slot $\slot$ is
\[
    P = \langle s, \pid, \payload, \sigma \rangle,
\]
where $\slot$ is the slot, $\pid$ is the proposer index, $\payload$ is the proposal's transaction payload, and $\sigma$ is $\proc_\pid$'s signature over $\langle s, \pid, \Hash(\payload) \rangle$. The triple $\langle s, \pid, \payload \rangle$ is exactly a \emph{proposal} in the sense of the problem definition (\Cref{subsection:mcp-preliminaries}); we simply spell out its fields here rather than carry it as a single object.

To preserve hiding, the proposer does not disseminate $P$ directly: it encrypts $P$ to produce an \emph{encrypted proposal} $EP$, decryptable only once $f+1$ validators have released their decryption shares for the slot (after the deadline). We disseminate $EP$ in place of $P$.

\paragraph{Chunks.}
To distribute the dissemination load, the proposer erasure-codes the encrypted proposal $EP$ into $n$ fragments $\cdata_1, \ldots, \cdata_n$ (one per validator) and builds a Merkle tree over the leaves $\Hash(1, \cdata_1), \ldots, \Hash(n, \cdata_n)$, yielding a root~$\mroot$. The reconstruction threshold is chosen so that any $f+1$ fragments suffice to reconstruct $EP$. The proposer forms a \emph{chunk header} common to all recipients,
\[
    \mathit{ChunkHeader} = \langle s, \pid, \mroot, \sigma \rangle,
\]
where $\sigma$ is $\proc_\pid$'s signature over $(s, \pid, \mroot)$, and sends a \emph{chunk}
\[
    \langle \textsc{Chunk}, \mathit{ChunkHeader}, r, \cdata_r, \pi_r \rangle
\]
to each validator $\proc_r$, where $\pi_r$ is a Merkle proof for $\Hash(r, \cdata_r)$ under root $\mroot$. We refer to this message as the chunk \emph{assigned to} $\proc_r$.

\paragraph{Entries.}
Each validator $\proc_i$ collects incoming chunks until the slot's deadline $s.\fdeadline$. For each proposer $\proc_\pid$, $\proc_i$ records an \emph{entry}
\[
    E = \langle s, \pid, \rhobot \rangle,
\]
where $\rhobot \in \mathsf{MerkleRoot} \cup \{\bot\}$ is the Merkle root of the chunk \emph{assigned to} $\proc_i$ that it received from $\proc_\pid$ (a \emph{positive entry}), or $\bot$ if no valid chunk arrived by the deadline (a \emph{negative entry}). Validators vote on entries by signing them.

\paragraph{Proposal vote.}
At the deadline $s.\fdeadline$ (or when $\proc_i$ enters the slot, if later), each validator $\proc_i$ reports what it received for the slot by broadcasting a single \emph{proposal vote}. This message does three things at once: it carries $\proc_i$'s signed positive/negative entry for each proposer, supplies the chunks backing the positive entries, and releases $\proc_i$'s decryption share for the slot. Concretely:
\[
    \langle \textsc{Vote}, s, \mathit{SignedEntries}, \mathit{Chunks}, \mathit{DecryptShare} \rangle,
\]
where $\mathit{SignedEntries}$ contains a \emph{signed entry} $\langle E, \sigma\rangle$ per proposer in $s.\fproposers$ (with $\proc_i$'s signature $\sigma$ on $\langle\votetag, E\rangle$), $\mathit{Chunks}$ contains the chunk messages $\proc_i$ received and validated for each positive entry, and $\mathit{DecryptShare}$ is $\proc_i$'s share toward decrypting slot-$\slot$ proposals.

\paragraph{$\FastQC$.}
A \emph{fast-path quorum certificate} ($\FastQC$) for an entry $E$ aggregates $2f+1$ signatures over $\langle\votetag, E\rangle$:
\[
    \FastQC = \langle \votetag, E, \Sigma \rangle.
\]
A $\FastQC$ for a positive entry $E = \langle s, \pid, \mroot\rangle$ certifies that $2f+1$ validators received and validated a chunk from $\proc_\pid$ under root $\mroot$, so at least $f+1$ honest validators hold chunks; this is what guarantees data availability for $\proc_\pid$'s proposal.

\input{alg_proposer}
\input{alg_voting}

\subsubsection{Fast Path.} \label{subsection:fast_path}

The fast path commits the slot in two voting rounds after the deadline, without needing a leader to drive progress, provided a $\FastQC$ can be formed for every proposer (a \emph{fast \metablock}). Recall that a $\FastQC$ aggregates $2f+1$ votes on the same entry. If a proposer disseminates early enough, every honest validator receives its assigned chunk by the deadline and records the same positive entry, so $2f+1$ matching positive votes are cast. If a proposer does not propose at all, every honest validator records a negative entry, so $2f+1$ negative votes are cast. If instead a proposer disseminates too late, only some validators receive a chunk by the deadline (\emph{partial dissemination}): these vote positively and the rest negatively, so there may be neither $2f+1$ positive nor $2f+1$ negative votes, and no $\FastQC$ forms. In that case, validators take the fallback path.

\paragraph{Fast \metablock.}
A \emph{fast \metablock} for slot $\slot$ is
\[
    B = \langle CE_{j_1}, \ldots, CE_{j_k} \rangle,
\]
where each $CE_\pid$ is a $\FastQC$ for proposer $\proc_\pid \in s.\fproposers$. We write $\entries(B)$ for the sequence of entries $E$ embedded in the $\FastQC$s.

\paragraph{Fast commit vote and certificate.}
Once a validator $\proc_i$ has a fast \metablock $B$, it broadcasts a \emph{fast commit vote}
\[
    \langle \CommitVote, s, \entries(B), \sigma_i \rangle,
\]
where $\sigma_i$ signs $\langle \CommitMsg, s, \entries(B)\rangle$. A $(2f{+}1)$-aggregate of such votes forms a \emph{fast commit certificate}
\[
    \commitQC = \langle \CommitMsg, s, \entries(B), \Sigma \rangle,
\]
which finalizes the entries $\entries(B)$: any validator that holds the $\commitQC$ recovers the underlying proposals via $\modcall{alg:da}.\recoverProposals(\entries(B))$ and commits them. The commit fixes the entries, not the certificates that $B$ carries: across validators a slot's entries agree, whereas the certificates justifying them may differ.

\paragraph{Speculative commit.}
Once a validator has $\FastQC$s for every proposer, it can speculatively commit the proposal vector before waiting for the $\commitQC$.

\input{alg_fast}

\subsubsection{Fallback Path.} \label{subsection:fallback_path}

If a $\FastQC$ fails to form for some proposer, validators cannot assemble a fast \metablock and the fast path stalls. A $\FastQC$ fails to form when the votes split between positive and negative, so that neither a positive nor a negative quorum forms (partial dissemination). This can result from the proposer disseminating too late or to too few validators, or from network asynchrony. A $\FastQC$ can also fail to form under equivocation, when the proposer sends different proposals so that no single entry gathers a quorum. From time $s.\fdeadline + \Delta$ onwards, each validator that has received at least $2f+1$ proposal votes but has not been able to form a fast \metablock casts a \emph{fallback vote} in place of a fast commit vote; the two are mutually exclusive, so each validator casts at most one. On the fallback path, the validator builds a \emph{fallback \metablock} by collecting, and where necessary contributing to, per-proposer evidence, by decreasing strength: a $\FastQC$ when one exists; otherwise, in adversarial cases, an $\EquivCert$ that pins a proposer to equivocation; otherwise a $\FallbackQC$ (a weaker certificate that the validators help assemble in this phase). It then feeds the fallback \metablock to a Multi-Value Byzantine Agreement (MVBA) primitive. The \metablock output by the MVBA is confirmed in one final round of \emph{fallback commit votes} --- which doubles as a data-availability certificate --- and becomes the \metablock committed by the slot consensus.

\paragraph{Fallback signed entries \& $\FallbackQC$.}
A \emph{fallback signed entry} from validator $\proc_r$ for proposer $\proc_\pid$ is
\[
    \langle E, \sigma_v\rangle \text{ (negative, } E.\rhobot = \bot\text{)} \quad \text{or} \quad \langle E, \sigma_v, \sigma_p\rangle \text{ (positive, } E.\rhobot = \mroot\text{)},
\]
where $\sigma_v$ is $\proc_r$'s signature over $\langle\textsc{fb}, E\rangle$ and, for the positive case, $\sigma_p$ is the proposer's signature on $\langle s, \pid, \mroot\rangle$ carried over from the chunk header. A validator constructs its own fallback signed entry for each proposer for which it holds no evidence yet: it sets $\rhobot \gets \mroot$ if it has accumulated $f+1$ valid fast-path positive votes for the same root $\mroot$ and the data is locally available ($\modcall{alg:da}.\isDecoded(\mroot)$ returns true), and $\rhobot \gets \bot$ otherwise.

A \emph{fallback-path quorum certificate} ($\FallbackQC$) for entry $E$ aggregates $f+1$ signatures over $\langle\textsc{fb}, E\rangle$:
\[
    \FallbackQC = \langle \textsc{fb}, E, \Sigma \rangle.
\]
The $\votetag$/$\textsc{fb}$ tags domain-separate the universal proposal-vote signatures from the fallback-only signatures.

\paragraph{$\EquivCert$.}
An \emph{equivocation certificate} for proposer $\proc_\pid$ is
\[
    \EquivCert = \langle \textsc{equiv}, s, j, \mroot_1, \sigma_{p,1}, \mroot_2, \sigma_{p,2} \rangle,
\]
where $\mroot_1 \neq \mroot_2$ and each $\sigma_{p,k}$ is $\proc_\pid$'s signature over $\langle s, j, \mroot_k\rangle$. The $\EquivCert$ exists for liveness: a Byzantine proposer can split its chunks across multiple Merkle roots so that no single fallback signed entry attracts $f+1$ matching votes, leaving a $\FallbackQC$ unformable; pinning the proposer to two contradictory signed roots still certifies its exclusion.

\paragraph{Fallback vote message.}
Once time has reached $s.\fdeadline + \Delta$ and $\proc_r$ has received at least $2f+1$ proposal votes (the inputs needed to assemble the message), and provided it has not entered the fast commit path, $\proc_r$ broadcasts a \emph{fallback vote}:
\[
    \langle \textsc{FallbackVote}, s, \sigma_r, \Ev_{j_1}, \ldots, \Ev_{j_k} \rangle,
\]
where $\sigma_r = \mathsf{Sign}_r(\langle\textsc{fallback}, s\rangle)$ is $\proc_r$'s commitment to entering the fallback path, and each $\Ev_\pid$ is a $\FastQC$ for $\proc_\pid$ if $\proc_r$ holds one, and $\proc_r$'s own fallback signed entry otherwise.

\paragraph{Fallback certificate.}
Once $2f+1$ fallback votes are collected, their $\sigma_r$ fields aggregate into a \emph{fallback certificate}
\[
    \FBCert_\slot = \langle \textsc{fallback}, s, \Sigma \rangle,
\]
certifying that the fast path can no longer commit slot $\slot$.

\paragraph{Fallback \metablock.}
From $2f+1$ fallback votes, each validator builds a \emph{fallback \metablock}
\[
    B = \langle CE_{j_1}, \ldots, CE_{j_k}, \FBCert_\slot \rangle,
\]
where each $CE_\pid$ is a \emph{certified entry} for $\proc_\pid$, chosen by decreasing strength: a $\FastQC$ when one is available; otherwise an $\EquivCert$, formed from two positive fallback signed entries with conflicting roots; otherwise a $\FallbackQC$, aggregated from $f+1$ matching fallback signed entries. A counting argument over the $2f+1$ votes shows that one of the three always exists: if all evidence for $\proc_\pid$ consists of bare signed entries and no two positive entries conflict, the entries span at most two distinct values --- a single root and $\bot$ --- so one of them gathers $f+1$ matching copies. Extending the convention from \Cref{subsection:fast_path}, $\entries(B)$ contributes the embedded $E$ for a $\FastQC$ or $\FallbackQC$, and the negative entry $\langle s, \pid, \bot\rangle$ for an $\EquivCert$ for $\proc_\pid$ (equivocation excludes the proposer); $\FBCert_\slot$ is not part of $\entries(B)$. We call a \metablock \emph{valid} if it is either a fast \metablock or a fallback \metablock in this sense.

\paragraph{Invocation to the MVBA primitive.}
The validator then proposes its meta-block $B$ to $\mathsf{MVBA}[\slot]$, the slot's multi-valued Byzantine agreement instance, invoking it at most once; $B$ is valid by construction.
Proposing doubles as starting to participate in the instance; a validator stops participating by invoking $\mathsf{abandon}()$, which it does upon finalizing the slot.
The MVBA we use (\Cref{mod:mvba}) is standard: all correct validators decide the same \metablock (agreement), at most once (integrity), any decided \metablock is valid (external validity), every correct validator decides within $\ell_{\mathsf{MVBA}}$ time once all correct validators have proposed, provided none abandons before that time (termination), and, mirroring the slot consensus itself, no correct validator sends any MVBA message outside its participation window (quiescence).

\input{p2_mvba}

\paragraph{Fallback commit.}
The MVBA's decision alone does not yet finalize the slot: for a $\FallbackQC$-backed entry of the decided \metablock $B'$, only \emph{one} honest validator is guaranteed to hold the underlying data. \chorus therefore follows the decision with one final round of commit votes, which doubles as an availability certificate. Upon deciding $B'$, a validator first waits, for every positive $\FallbackQC$-backed entry of $B'$, until it holds its own assigned chunk under that entry's root --- guaranteed to arrive, since an honest contributor to the $\FallbackQC$ has sent every validator its assigned chunk --- re-broadcasts those chunks, and only then broadcasts a signed \emph{fallback commit vote} for $\entries(B')$. An aggregate of $2f+1$ such votes forms the \emph{fallback commit certificate}, the fallback path's transferable commitment proof, mirroring the fast path's $\commitQC$. Its mere existence certifies availability: among its $2f+1$ signers, at least $f+1$ are correct, and each re-broadcast its (distinct) assigned chunk before signing, so the chunks needed to recover every entry of $B'$ are already in flight.
A validator finalizes upon receiving a valid fallback commit certificate, recovering the underlying proposals via $\recoverProposals(\entries(B'))$ and committing them, and re-broadcasts the certificate. Even after invoking the MVBA, a validator still finalizes on the fast path if a $\commitQC$ arrives; whenever it finalizes, by either path, it broadcasts the commitment proof and concludes all activity for the slot, abandoning its MVBA invocation ($\mathsf{MVBA}[\slot].\mathsf{abandon}()$).

\input{alg_fallback}

\subsubsection{Data Availability.} \label{subsection:da}

The data-availability (DA) module is responsible for recovering the proposals originally embedded by each proposer in chunks. 

\paragraph{Chunk reception.}
On receiving a chunk, the DA module verifies the proposer's signature and the chunk's Merkle proof against its root $\mroot$. Two things then follow. First, if the chunk is the one assigned to this validator, the module notifies the voting module, so the validator can record a positive entry for $\mroot$ if the chunk arrives by the deadline (the chunk itself travels onward inside the validator's proposal vote and, on the fallback path, in the commit round). Second, once $f+1$ valid chunks for the same root $\mroot$ accumulate, the module decodes them into a candidate encrypted proposal $c$ and re-encodes $c$, comparing the resulting root against $\mroot$: if they match, $c$ is confirmed; otherwise $\mroot$ is marked invalid, catching a proposer that disseminated inconsistent chunks under one root.

\paragraph{Decryption.}
Because proposals are disseminated encrypted to preserve hiding, reconstructing $c$ does not yet reveal the proposal. Each validator releases a decryption share for the slot together with its vote; once the module holds $f+1$ valid shares, it decrypts $c$ into the proposal, or marks $\mroot$ invalid if decryption fails.

\paragraph{Proposal recovery.}
At commit time (speculative or formal), a validator calls $\recoverProposals(\entries(B))$ on the DA module. The call blocks until every positive entry $\langle s, \pid, \mroot\rangle \in \entries(B)$ has either been recovered to a proposal or marked invalid, then returns the recovered proposals (ordering their transactions into a block is a subsequent, application-level step). Both commit routes guarantee that this call eventually returns: a $\FastQC$ for a positive entry implies that $f+1$ honest validators contributed chunks (enough to decode by erasure coding); a $\FallbackQC$ for a positive entry appears only behind a fallback commit certificate, whose $f+1$ honest signers each re-broadcast their assigned chunk before commit-voting (\modcall{alg:fallback}) --- again enough to decode. These two facts are the formal counterpart of the \emph{availability} of
\S\ref{subsubsection:ObjectofAgreement}. Its companion \emph{consistency} corresponds
to the determinism of the re-encode-and-compare check in \textsc{tryIngestChunk}
(\modcall{alg:da}, line~\ref{line:da-reencode}), where the verdict on a root, whether to accept and decode it
or mark it invalid, depends only on the chunks committed under it and not on
which $f+1$ a validator happens to collect. Every correct validator therefore
reaches the same verdict, and on acceptance decodes the same ciphertext.

\input{alg_da}

\subsection{Proof} \label{subsection:proof_sketches}

This subsection proves that \chorus, when run within \cadence, realizes the slot consensus primitive (\Cref{mod:slotconsensus}); we establish each of its properties --- quiescence, hiding, slot safety, agreement, proposal inclusion, and termination --- in turn.
All properties hold unconditionally, with one exception: termination is proven under a precondition on how validators start participating.
As we show later, this precondition is exactly what \conductor provides when it runs alongside \chorus within \cadence.

\paragraph{Quiescence.}
We begin with the quiescence property, which is immediate from the pseudocode's standing convention.

\begin{lemma}[Quiescence]
\label{lemma:chorus-quiescence}
\chorus (\Cref{alg:proposer-dissemination,alg:voting,alg:fast-path-certification,alg:fallback,alg:da}) satisfies quiescence.
\end{lemma}
\begin{proof}
A correct validator sends protocol messages either from within the rules of \Cref{alg:proposer-dissemination,alg:voting,alg:fast-path-certification,alg:fallback,alg:da} or from within the underlying MVBA instance.
For the former, by the standing convention of \S\ref{subsection:chorus-protocol-overview}, every rule that sends a message is activated only while the validator is actively participating, i.e., after it has invoked $\mathsf{participate}()$ and before it has invoked $\mathsf{abandon}()$; the remaining rules only process received messages and send nothing.
For the latter, the quiescence property of the MVBA (\Cref{mod:mvba}) confines its messages to the window between $\mathsf{MVBA}[\slot].\mathsf{propose}$ --- invoked from rules the standing convention explicitly gates (\Cref{alg:fallback}, lines~\ref{line:fb-mvba-propose-fast} and~\ref{line:fb-mvba-propose}), hence while actively participating --- and $\mathsf{MVBA}[\slot].\mathsf{abandon}$, which \chorus invokes exactly when it receives the slot-consensus $\mathsf{abandon}()$ input (\Cref{alg:fallback}, line~\ref{line:fb-abandon}), i.e., when the validator stops participating.
Hence no correct validator sends any protocol message before it starts participating or after it stops. \qed
\end{proof}

\paragraph{Hiding.}
We continue with the hiding property.

\begin{lemma}[Hiding]
\label{lemma:chorus-hiding}
\chorus (\Cref{alg:proposer-dissemination,alg:voting,alg:fast-path-certification,alg:fallback,alg:da}) satisfies hiding.
\end{lemma}
\begin{proof}
This is established by the simulation argument of \S\ref{appendix:encryption}.
\end{proof}

\paragraph{Slot safety.}
Next, we prove the slot safety property.

\begin{lemma}[Slot safety]
\label{lemma:chorus-slot-safety}
\chorus (\Cref{alg:proposer-dissemination,alg:voting,alg:fast-path-certification,alg:fallback,alg:da}) satisfies slot safety.
\end{lemma}
\begin{proof}
Follows directly from line~\ref{line:da-recover-slot} of \Cref{alg:da}: a validator finalizes only the proposal vector returned by $\recoverProposals$, which is a proposal vector for slot $s$.
\end{proof}

\paragraph{Agreement.}
We now prove that \chorus satisfies agreement.
Recall that a validator finalizes only when it holds a \emph{commitment proof} for a \metablock's entries --- a fast commit certificate $\commitQC$ over $\entries(B)$, for a fast \metablock $B$ (\Cref{alg:fast-path-certification}, line~\ref{line:fast-recv-commitqc}), or a fallback commit certificate over $\entries(B')$, for a \metablock $B'$ decided by $\mathsf{MVBA}[\slot]$ (\Cref{alg:fallback}, line~\ref{line:fb-recv-commit}) --- whereupon it triggers $\mathsf{finalize}(\cdot)$: in the fast path on a valid $\commitQC$ (\Cref{alg:fast-path-certification}, line~\ref{line:fast-finalize}), and in the fallback path on a valid fallback commit certificate (\Cref{alg:fallback}, line~\ref{line:fb-finalize}).
We argue in two steps: any two commitment proofs carry the same entries, and the same entries yield the same proposal vector at every correct validator.

\begin{proposition}
\label{prop:agreement-entries}
Consider any two correct validators $\proc_i$ and $\proc_j$ such that $\proc_i$ obtains (i.e., stores in its local memory) a commitment proof for the entries of a meta-block $B_i$ and $\proc_j$ obtains a commitment proof for the entries of a meta-block $B_j$.
Then, $\entries(B_i) = \entries(B_j)$.
\end{proposition}
\begin{proof}
A commitment proof is either a fast commit certificate $\commitQC$ (\Cref{alg:fast-path-certification}) or a fallback commit certificate (\Cref{alg:fallback}), so we distinguish three cases.

\begin{compactitem}
    \item \emph{Both $\proc_i$ and $\proc_j$ hold a $\commitQC$.}
    Then $B_i$ and $B_j$ are fast meta-blocks (\Cref{alg:fast-path-certification}, line~\ref{line:fast-metablock}), each of whose certified entries is a $\FastQC$ (\Cref{alg:fast-path-certification}, line~\ref{line:fast-formqc}) --- an aggregate of $2f+1$ matching signed entries. Fix a proposer; the two supporting sets of $2f+1$ signed entries (one from $B_i$, one from $B_j$) intersect, among the $n = 3f+1$ validators, in at least $f+1$, hence in a correct validator. Since a correct validator signs a single entry per proposer (\Cref{alg:voting}, line~\ref{line:vote-broadcast}), the two entries for that proposer coincide. As this holds for every proposer, $\entries(B_i) = \entries(B_j)$.

    \item \emph{Both hold a fallback commit certificate.}
    Each certificate aggregates $2f+1$ commit-vote signatures (\Cref{alg:fallback}, line~\ref{line:fb-formcommitqc}); the two supporting sets intersect, among the $n = 3f+1$ validators, in at least $f+1$, hence in a correct validator. A correct validator sends a single fallback commit vote --- one per MVBA decision (\Cref{alg:fallback}, line~\ref{line:fb-commitvote}), and the MVBA decides at most once by its integrity property (\Cref{mod:mvba}).
    Therefore, $\entries(B_i) = \entries(B_j)$.

    \item \emph{One holds a $\commitQC$ and the other a fallback commit certificate.}
    Without loss of generality, assume that $\proc_i$ holds a $\commitQC$ for $B_i$ and $\proc_j$ a fallback commit certificate for $B_j$. A $\commitQC$ aggregates $2f+1$ commit votes (\Cref{alg:fast-path-certification}, line~\ref{line:fast-collect-commit}), so at least $f+1$ correct validators broadcast a commit vote and thereby set $\mathit{pathVote} = \textsc{fast}$ (\Cref{alg:fast-path-certification}, line~\ref{line:fast-pathvote}). A correct validator issues a fallback vote only while $\mathit{pathVote} = \textsc{none}$ (\Cref{alg:fallback}, line~\ref{line:fb-pathvote-guard}), so these $f+1$ correct validators never do; fewer than $2f+1$ fallback votes can form, and hence no fallback meta-block --- which requires a $2f+1$-fallback-vote certificate --- is valid.
    Moreover, $\proc_j$'s certificate carries the signature of at least one correct validator, which commit-votes only for the meta-block it decided from $\mathsf{MVBA}[\slot]$ (\Cref{alg:fallback}, line~\ref{line:fb-commitvote}); by the external validity of the MVBA (\Cref{mod:mvba}), that meta-block is valid. Since no fallback meta-block is valid, the decided meta-block is a fast one, so $B_j$ is a fast meta-block.
    As $B_i$ is a fast meta-block as well, $\entries(B_i) = \entries(B_j)$ due to the quorum intersection. \qed
\end{compactitem}
\end{proof}

\begin{proposition}
\label{prop:recovery-consistency}
For any set of entries $E$, any two correct validators that recover proposal vectors $V$ and $V'$ from $\recoverProposals(E)$ satisfy $V = V'$.
\end{proposition}
\begin{proof}
We show that whenever two correct validators both return from $\recoverProposals(E)$, they return the same vector. 
Negative entries recover to $\bot$ at every validator, so fix a positive entry $\langle s, \pid, \mroot\rangle \in E$; we show that any two correct validators that resolve it recover the same value for $\proc_\pid$.
A correct validator returns a value for $\langle s, \pid, \mroot\rangle$ only once $\status(\mroot) \in \{\langle\mathsf{plain}, \cdot\rangle,\, \mathsf{invalid}\}$ (\Cref{alg:da}, line~\ref{line:da-wait}).
The verdict on $\mroot$ --- accept and decode, or mark invalid --- is fixed by the re-encode-and-compare check of $\tryIngestChunk$ (\Cref{alg:da}, line~\ref{line:da-reencode}): a validator decodes its $f+1$ chunks to a candidate $c$ and accepts $\mroot$ only if $\MerkleRoot(\Encode(c)) = \mroot$. The chunks admissible under $\mroot$ are pinned down by the root itself (collision resistance of the Merkle hash), and $\Decode$ and $\Encode$ are deterministic; hence the verdict depends only on $\mroot$ and not on which $f+1$ chunks a validator happens to collect. Every correct validator that resolves $\mroot$ thus reaches the same verdict and, on acceptance, the same ciphertext $c$. Decryption is likewise deterministic once the slot key is reconstructed, so all such validators recover the same proposal $P$ for $\proc_\pid$, or all mark $\mroot$ invalid (recovering $\bot$).
As this holds for every proposer, the two validators recover the same proposal vector. \qed
\end{proof}

We are ready to prove the agreement property of \chorus.

\begin{lemma}[Agreement]
\label{lemma:chorus-agreement}
\chorus (\Cref{alg:proposer-dissemination,alg:voting,alg:fast-path-certification,alg:fallback,alg:da}) satisfies agreement.
\end{lemma}
\begin{proof}
A correct validator finalizes only the proposal vector $\recoverProposals(\entries(B))$ (\Cref{alg:da}, line~\ref{line:da-recover-slot}) for a meta-block $B$ whose entries are backed by a commitment proof it holds. Let $V_i$ and $V_j$ be proposal vectors finalized by correct validators, backed by commitment proofs for the entries of meta-blocks $B_i$ and $B_j$; thus $V_i = \recoverProposals(\entries(B_i))$ and $V_j = \recoverProposals(\entries(B_j))$. By \Cref{prop:agreement-entries}, $\entries(B_i) = \entries(B_j)$, and by \Cref{prop:recovery-consistency}, applying $\recoverProposals$ to these equal entries yields the same proposal vector at both validators. Hence, $V_i = V_j$. \qed
\end{proof}

\paragraph{Proposal inclusion.}
We now prove the proposal inclusion property.
We first show that an on-time honest proposer can never be excluded: in any valid meta-block, its entry is the positive entry carrying its proposal's root.

\begin{proposition}
\label{prop:honest-positive-entry}
Suppose $s.\fdeadline - \Delta \geq \mathrm{GST}$ and a correct proposer $\proc_\pid \in s.\fproposers$ disseminates its proposal $P$, with Merkle root $\mroot_P$, at time $s.\fdeadline - \Delta$.
Then, in every valid meta-block, $\proc_\pid$'s entry is the positive entry $\langle s, \pid, \mroot_P\rangle$.
\end{proposition}
\begin{proof}
Since $s.\fdeadline - \Delta \geq \mathrm{GST}$ and $\proc_\pid$ disseminates at the slot's starting time, every correct validator receives and validates its assigned chunk under $\mroot_P$ by the deadline --- chunk ingestion only processes the received message, so it requires no active participation (\S\ref{subsection:chorus-protocol-overview}) --- setting its entry for $\proc_\pid$ to the positive entry $\langle s, \pid, \mroot_P\rangle$ (\Cref{alg:voting}, line~\ref{line:vote-positive}).
Hence no correct validator ever signs a negative entry $\langle s, \pid, \bot\rangle$, nor an entry on a root $\neq \mroot_P$, for $\proc_\pid$.

A $\FastQC$ for $\proc_\pid$ is an aggregate of $2f+1$ matching signed entries (\Cref{alg:fast-path-certification}, line~\ref{line:fast-formqc}); as at most $f$ validators are Byzantine, no $2f+1$ matching signed entries can carry $\bot$ or a root $\neq \mroot_P$, so any $\FastQC$ for $\proc_\pid$ is the positive entry $\langle s, \pid, \mroot_P\rangle$.

The fallback path yields the same conclusion. A correct validator casts a fallback signed entry for $\proc_\pid$ only after collecting $2f+1$ first-round votes (\Cref{alg:fallback}, line~\ref{line:fb-pathvote-guard}); at least $f+1$ of these come from correct validators, each carrying the positive entry $\langle s, \pid, \mroot_P\rangle$ for $\proc_\pid$ together with its chunk. The validator therefore holds $f+1$ positive entries on $\mroot_P$ and can decode $\mroot_P$, so it casts a \emph{positive} fallback signed entry on $\mroot_P$ --- never a negative one, nor one on another root (\Cref{alg:fallback}, line~\ref{line:fb-positive-entry}). Since a $\FallbackQC$ aggregates $f+1$ matching fallback signed entries (\Cref{alg:fallback}, line~\ref{line:fb-formqc}) and at most $f$ validators are Byzantine, any $\FallbackQC$ for $\proc_\pid$ is likewise positive on $\mroot_P$, and no negative $\FallbackQC$ can form; an $\EquivCert$ would require $\proc_\pid$ to sign two different roots, impossible for a correct proposer.

Thus every piece of evidence ever formed for $\proc_\pid$ --- a $\FastQC$ or a $\FallbackQC$ --- is the positive entry $\langle s, \pid, \mroot_P\rangle$, which is therefore its entry in every valid meta-block. \qed
\end{proof}

With the proposer's entry pinned down by \Cref{prop:honest-positive-entry}, we can now establish proposal inclusion.

\begin{lemma}[Proposal inclusion]
\label{lemma:chorus-proposal-inclusion}
\chorus (\Cref{alg:proposer-dissemination,alg:voting,alg:fast-path-certification,alg:fallback,alg:da}) satisfies proposal inclusion.
\end{lemma}
\begin{proof}
Let $\mroot_P$ be the Merkle root of $\proc_\pid$'s encrypted proposal, and let $B$ be the \metablock whose entries the validator finalizes, so that $V = \recoverProposals(\entries(B))$; note that $B$ is valid --- for a $\commitQC$, $B$ is a fast \metablock, and for a fallback commit certificate, $B$ is the \metablock decided by $\mathsf{MVBA}[\slot]$, valid by its external validity (\Cref{mod:mvba}).
By \Cref{prop:honest-positive-entry}, $\proc_\pid$'s entry in $B$ is the positive entry $\langle s, \pid, \mroot_P\rangle$ --- a $\FastQC$ or a $\FallbackQC$ on $\mroot_P$.
Since the validator finalizes $V$, the call $\recoverProposals(\entries(B))$ has returned, so the validator has settled the status of $\mroot_P$ as either a recovered proposal or $\mathsf{invalid}$ (\Cref{alg:da}, line~\ref{line:da-wait}).
Because $\proc_\pid$ is correct, the settlement is necessarily the proposal $P$: the chunks admissible under $\mroot_P$ are pinned down by the root itself (collision resistance of the Merkle hash), so the $f+1$ chunks the validator decoded reconstruct the ciphertext committed by $\mroot_P$, which re-encodes to $\mroot_P$ (\Cref{alg:da}, line~\ref{line:da-reencode}) and, under the slot key, decrypts to $P$ (\S\ref{appendix:encryption}); moreover, the recovered proposal's embedded proposer is $\proc_\pid$, matching the entry's, so the entry is not marked invalid (\Cref{alg:da}, line~\ref{line:da-proposer-check}).
Hence, $\recoverProposals(\cdot)$ recovers $P$ for $\proc_\pid$, i.e., $V[\proc_\pid] = P$. \qed
\end{proof}

\paragraph{Termination.}
We now prove that \chorus satisfies termination.
Unlike the previous properties, termination is not unconditional: it holds when \chorus is run within \cadence, and under the additional guarantee that participation is \emph{$\Delta$-synchronized}.

\begin{definition}[$\Delta$-synchronized participation]
\label{def:delta-synchronized-participation}
Participation is \emph{$\Delta$-synchronized} if and only if the following holds: if a correct validator starts participating at some time $T$, then every correct validator starts participating by time $\max(T, \mathrm{GST}) + \Delta$.
\end{definition}

\noindent Conditioning termination on $\Delta$-synchronized participation places no extra burden on \cadence.
Within \cadence, a validator starts participating in a slot exactly when \conductor opens it, so the condition asks precisely that the openings of each slot be synchronized within $\Delta$ across correct validators --- and this is exactly \conductor's totality guarantee.
As we show later (\Cref{lemma:conductor-totality}), \conductor indeed satisfies this guarantee when run within \cadence; the precondition is therefore always discharged in the composed system, and \chorus's termination holds there.

With the precondition in place, we can turn to the proof itself.
We begin with \emph{totality}: once one correct validator finalizes, every correct validator finalizes within $\Delta$ time (assuming synchrony).
Totality serves a dual purpose --- it is the key step in \chorus's own termination, and \conductor relies on it in its own right --- so we record it as a standalone property.

\begin{proposition}[Totality]
\label{prop:chorus-totality}
Suppose participation is $\Delta$-synchronized.
When run within \cadence, the following holds for \chorus: if a correct validator finalizes a proposal vector at some time $t$, then every correct validator does so by time $\max(t, \mathrm{GST}) + \Delta$.
\end{proposition}
\begin{proof}
On finalizing a meta-block $B$, a correct validator broadcasts a transferable commitment proof for $\entries(B)$ --- a fast commit certificate $\commitQC$ (\Cref{alg:fast-path-certification}, line~\ref{line:fast-rebroadcast-commitqc}) or a fallback commit certificate (\Cref{alg:fallback}, line~\ref{line:fb-commit-rebroadcast}). A correct validator finalizes upon receiving any valid commitment proof (\Cref{alg:fast-path-certification}, line~\ref{line:fast-recv-commitqc}; \Cref{alg:fallback}, line~\ref{line:fb-recv-commit}), by recovering and committing $\recoverProposals(\entries(B))$ (\Cref{alg:da}, line~\ref{line:da-recover-slot}).
Suppose a correct validator finalizes at time $t$ with such a proof for $\entries(B)$. After $\mathrm{GST}$ its broadcast reaches every correct validator within $\Delta$, so the proof arrives at every correct validator by $\max(t, \mathrm{GST}) + \Delta$.
We first argue that every correct validator also \emph{processes} it by that time, despite the finalization rules --- which re-broadcast the proof --- being gated on active participation.
Consider a correct validator $\proc_j$. If $\proc_j$ has stopped participating, it has already finalized: within \cadence, a validator abandons the instance only after finalizing (\Cref{algorithm:cadence}, line~\ref{line:abandon}).
Otherwise, since the finalizer started participating at some time $\leq t$, $\Delta$-synchronized participation guarantees that $\proc_j$ starts participating by $\max(t, \mathrm{GST}) + \Delta$; any message that arrives earlier stays buffered, and its rule fires as soon as $\proc_j$ participates (\S\ref{subsection:chorus-protocol-overview}).
Either way, $\proc_j$ processes every message that reaches it by $\max(t, \mathrm{GST}) + \Delta$ no later than that time.
It remains to show that recovery succeeds by the same time, i.e., that by $\max(t, \mathrm{GST}) + \Delta$ every correct validator holds $f+1$ valid chunks under each positive entry's root and $f+1$ decryption shares.
For a $\FastQC$ entry, its $2f+1$ positive signed entries include at least $f+1$ from honest validators, each of which included the chunk assigned to it in its first-round vote (\Cref{alg:voting}, line~\ref{line:vote-broadcast}), broadcast to every validator no later than time $t$ (the $\FastQC$ is part of $B$); these chunks reach every correct validator by $\max(t,\mathrm{GST})+\Delta$.
For a $\FallbackQC$ entry, the commitment proof is a fallback commit certificate aggregating $2f+1$ commit votes (\Cref{alg:fallback}, line~\ref{line:fb-formcommitqc}) cast by time $t$; at least $f+1$ of the voters are correct, and each received and re-broadcast its assigned chunk under the entry's root before casting its vote (\Cref{alg:fallback}, line~\ref{line:fb-commit-wait}), hence by time $t$. These $f+1$ distinct chunks reach every correct validator by $\max(t,\mathrm{GST})+\Delta$.
Finally, if $B$ has any positive entry, its certificate rests on $2f+1$ first-round votes (directly for a $\FastQC$; via each fallback caster's $2f+1$ received votes for a $\FallbackQC$, \Cref{alg:fallback}, line~\ref{line:fb-pathvote-guard}), so at least $f+1$ correct validators broadcast their first-round votes --- each carrying a decryption share --- by time $t$, and every correct validator holds $f+1$ valid shares by $\max(t,\mathrm{GST})+\Delta$.
Hence by $\max(t, \mathrm{GST}) + \Delta$ every correct validator holds a commitment proof for $\entries(B)$ together with all the chunks and shares that $\recoverProposals(\entries(B))$ needs, so it recovers the proposal vector (\Cref{alg:da}, line~\ref{line:da-recover-slot}) and finalizes it. \qed
\end{proof}

With totality in hand, we bound when every correct validator finalizes assuming no correct validator stops participating.

\begin{proposition}
\label{prop:chorus-finalization-time}
Suppose 
participation is $\Delta$-synchronized, 
all correct validators start participating by some time $t$, and no correct validator stops participating before time $T = \max(t, \mathrm{GST}) + 4\Delta + \ell_{\mathsf{MVBA}}$. 
When run within \cadence, the following holds for \chorus: every correct validator finalizes a proposal vector by time $T$.
\end{proposition}
\begin{proof}
Throughout the proof, let $M = \max(t, \mathrm{GST})$; after time $M$, every message between correct validators is delivered within $\Delta$.
We first record one observation: if some correct validator finalizes a proposal vector by time $T - \Delta$, then totality (\Cref{prop:chorus-totality}) carries every correct validator to finalization within a further $\Delta$, i.e., by $T$.
For the rest of the proof, we therefore assume that no correct validator finalizes by time $T - \Delta$, and we trace the protocol on a single timeline.

\smallskip
\noindent \emph{By $M + \Delta$: first-round votes.} Within \cadence, a correct validator starts participating in $\slot$ no earlier than $s.\fdeadline - \Delta$ (by the integrity of \conductor)\footnote{\conductor satisfies integrity unconditionally, without relying on any guarantee of \chorus; hence, there is no circularity.}, so the deadline $s.\fdeadline$ lies at most $\Delta$ after any correct validator's entry, and hence at most $\Delta$ after time $t$.
Each correct validator issues its first-round vote upon reaching the deadline --- or as soon as it is past, if the validator entered the slot late (\Cref{alg:voting}, line~\ref{line:vote-broadcast}) --- so every correct validator issues it by time $M + \Delta$.
Each such vote carries the voter's signed entries, its assigned chunks for its positive entries, and its decryption share.

\smallskip
\noindent \emph{By $M + 2\Delta$: shares and second-round votes.} The first-round votes of all $n - f \geq 2f+1$ correct validators reach every correct validator by $M + 2\Delta$.
Consequently, by $M + 2\Delta$, every correct validator holds at least $f+1$ valid decryption shares --- enough to reconstruct the slot key --- and casts its second-round vote: a commit vote if it holds a fast \metablock (\Cref{alg:fast-path-certification}, line~\ref{line:fast-commitvote}), and a fallback vote otherwise, the guard at line~\ref{line:fb-pathvote-guard} of \Cref{alg:fallback} being met since $s.\fdeadline + \Delta \leq M + 2\Delta$ and $2f+1$ valid first-round votes have arrived.

\smallskip
\noindent \emph{By $M + 3\Delta$: MVBA proposals and fallback chunks.} The second-round votes reach every correct validator by $M + 3\Delta$; we argue that every correct validator $\proc_i$ proposes to $\mathsf{MVBA}[\slot]$ by then.
If $\proc_i$ holds $2f+1$ fallback votes (each entry of a correct fallback vote is a $\FastQC$ or the sender's own fallback signed entry), the rule at line~\ref{line:fb-mvba-propose} of \Cref{alg:fallback} fires by $M + 3\Delta$; the per-proposer selection it performs (lines~\ref{line:fb-build-entry}--\ref{line:fb-formqc}) always succeeds --- for every proposer, if no $\FastQC$ is available and no two collected positive entries conflict, the $2f+1$ collected entries span at most two distinct values, so $f+1$ of them match and a $\FallbackQC$ forms (if two collected positive entries do conflict, their verified proposer signatures directly yield an $\EquivCert$, line~\ref{line:fb-build-equiv}) --- and the assembled \metablock is valid by construction.
Otherwise, some correct validator cast a commit vote instead of a fallback vote; that validator held a fast \metablock and broadcast it by $M + 2\Delta$ (\Cref{alg:fast-path-certification}, line~\ref{line:fast-metablock}), so $\proc_i$ receives it and adopts a $\FastQC$ for every proposer by $M + 3\Delta$; as local time has passed $s.\fdeadline + 2\Delta \leq M + 3\Delta$, the rule at line~\ref{line:fb-mvba-propose-fast} fires by $M + 3\Delta$.
Moreover, by $M + 3\Delta$, every correct validator holds its assigned chunk under the root of every positive $\FallbackQC$ that ever forms: such a $\FallbackQC$ aggregates $f+1$ fallback signed entries, at least one cast by a correct validator --- by $M + 2\Delta$, as just argued --- which, upon casting it, sent every validator its assigned chunk (\Cref{alg:fallback}, line~\ref{line:fb-redisseminate}).

\smallskip
\noindent \emph{By $T - \Delta = M + 3\Delta + \ell_{\mathsf{MVBA}}$: MVBA decision and commit votes.} Since no correct validator finalizes by $T - \Delta$ and a validator abandons its MVBA invocation only upon finalizing, no correct validator abandons $\mathsf{MVBA}[\slot]$ by $T - \Delta$.
Both premises of the MVBA's $\ell_{\mathsf{MVBA}}$-termination (\Cref{mod:mvba}) thus hold --- all correct validators propose by $M + 3\Delta$, and none abandons before $M + 3\Delta + \ell_{\mathsf{MVBA}}$ --- so every correct validator decides by $T - \Delta$, and, by the MVBA's agreement and external validity, they all decide the same valid \metablock $B'$.
Upon deciding, each correct validator immediately casts its fallback commit vote (\Cref{alg:fallback}, line~\ref{line:fb-commitvote}): the wait at line~\ref{line:fb-commit-wait} has already cleared, as each holds its assigned chunk for every positive $\FallbackQC$ of $B'$ since $M + 3\Delta$, and casting re-broadcasts those chunks.

\smallskip
\noindent \emph{By $T$: certificates and finalization.} The $n - f \geq 2f+1$ correct commit votes --- all carrying the same $\entries(B')$ --- reach every correct validator by $T$, so each assembles the fallback commit certificate (\Cref{alg:fallback}, line~\ref{line:fb-formcommitqc}), broadcasts it, and finalizes upon its receipt (lines~\ref{line:fb-recv-commit}--\ref{line:fb-finalize}).
The recovery performed at finalization (\Cref{alg:da}, line~\ref{line:da-recover-slot}) returns immediately: for every $\FastQC$ entry of $B'$, the chunks of its $\geq f+1$ correct signers arrived with their first-round votes by $M + 2\Delta$; for every positive $\FallbackQC$ entry, the $f+1$ correct commit voters' re-broadcast chunks arrive by $T$; and the $f+1$ decryption shares arrived by $M + 2\Delta$.
Hence $\recoverProposals(\entries(B'))$ does not block (\Cref{alg:da}, line~\ref{line:da-wait}), and every correct validator finalizes a proposal vector for $s$ by time $T$. \qed
\end{proof}

We are finally ready to establish \chorus's termination.

\begin{lemma}[Termination]
\label{lemma:chorus-termination}
When run within \cadence, and provided participation is $\Delta$-synchronized, \chorus (\Cref{alg:proposer-dissemination,alg:voting,alg:fast-path-certification,alg:fallback,alg:da}) satisfies $\ell$-termination with $\ell = 5\Delta + \ell_{\mathsf{MVBA}}$: if all correct validators start participating by some time $t$, then every correct validator finalizes a proposal vector by time $\max(t, \mathrm{GST}) + 5\Delta + \ell_{\mathsf{MVBA}}$.
\end{lemma}
\begin{proof}
Within \cadence a correct validator invokes $\mathsf{abandon}()$ only after it has finalized (\Cref{algorithm:cadence}, line~\ref{line:abandon}); equivalently, once a correct validator starts participating it keeps participating until it finalizes.
Let $T_0 = \max(t, \mathrm{GST}) + 4\Delta + \ell_{\mathsf{MVBA}}$ and distinguish two cases.
\begin{compactitem}
    \item \emph{Some correct validator finalizes by $T_0$.} By totality (\Cref{prop:chorus-totality}), every correct validator then finalizes within a further $\Delta$, hence by $T_0 + \Delta = \max(t, \mathrm{GST}) + 5\Delta + \ell_{\mathsf{MVBA}}$.

    \item \emph{No correct validator finalizes by $T_0$.} Then no correct validator has abandoned its instance by $T_0$, so all correct validators remain participating throughout $[t, T_0]$; by \Cref{prop:chorus-finalization-time}, every correct validator finalizes by $T_0 \leq \max(t, \mathrm{GST}) + 5\Delta + \ell_{\mathsf{MVBA}}$.
\end{compactitem}
In either case, every correct validator finalizes by $\max(t, \mathrm{GST}) + 5\Delta + \ell_{\mathsf{MVBA}}$. \qed
\end{proof}

\paragraph{Epilogue.}
We conclude by consolidating the results of this subsection into a single statement.

\begin{theorem}[Correctness of \chorus]
\label{thm:chorus-correctness}
When run within \cadence, and provided participation is $\Delta$-synchronized, \chorus (\Cref{alg:proposer-dissemination,alg:voting,alg:fast-path-certification,alg:fallback,alg:da}) is a correct implementation of the slot consensus primitive.
\end{theorem}
\begin{proof}
Quiescence, hiding, slot safety, agreement, and proposal inclusion are satisfied unconditionally (\Cref{lemma:chorus-quiescence,lemma:chorus-hiding,lemma:chorus-slot-safety,lemma:chorus-agreement,lemma:chorus-proposal-inclusion}, respectively), whereas termination is satisfied exactly under the two assumptions of the theorem --- \chorus running within \cadence, with $\Delta$-synchronized participation (\Cref{lemma:chorus-termination}). \qed
\end{proof}



%% file: alg_proposer.tex
\begin{algorithm}[tbp]
\caption{Proposer Dissemination for Slot $\slot$ (for validator $\proc_i$)}
\label{alg:proposer-dissemination}
\begin{algorithmic}[1]
\footnotesize

\State \textbf{Rule} (for \Cref{alg:proposer-dissemination,alg:voting,alg:fast-path-certification,alg:fallback,alg:da}):
\statelabel{- A message-sending rule --- including one that invokes $\mathsf{MVBA}[\slot].\mathsf{propose}$ --- is activated only while $\proc_i$ is actively participating; rules that only process received messages may fire at any time (\S\ref{subsection:chorus-protocol-overview}).}

\medskip
\State \textbf{upon} $\mathsf{propose}(\mathit{proposal} \in \mathsf{Proposal})$:
\State \hskip2em parse $\mathit{proposal} = \langle \slot, i, \payload\rangle$ \BlueComment{$\slot = \mathit{proposal}.\fslot$, $i = \mathit{proposal}.\fproposer$, $\payload = \mathit{proposal}.\fpayload$}
\State \hskip2em \textbf{if} local time $= s.\fdeadline - \Delta$:
\State \hskip4em Let $P \gets \langle \slot, i, \payload, \mathsf{Sign}_i(\langle\slot, i, \Hash(\payload)\rangle) \rangle$
\State \hskip4em Let $EP \gets \mathsf{TIBE}.\Enc(\mpk, \slot, P)$ \BlueComment{hide proposal contents}
\State \hskip4em Let $(\cdata_1,\ldots,\cdata_n) \gets \Encode(EP)$ \BlueComment{erasure-code into $n$ chunks; any $f+1$ suffice to decode}
\State \hskip4em Let $\mroot \gets \MerkleRoot(\Hash(1, \cdata_1),\ldots,\Hash(n, \cdata_n))$ \BlueComment{bind each chunk to its position}
\State \hskip4em Let $\mathit{ChunkHeader} \gets \langle \slot, i, \mroot, \mathsf{Sign}_i(\langle\slot, i, \mroot\rangle) \rangle$
\State \hskip4em \textbf{for each} validator $\proc_r$:
\State \hskip6em \textbf{send} $\langle \textsc{Chunk}, \mathit{ChunkHeader}, r, \cdata_r, \pi_r \rangle$ to $\proc_r$

\end{algorithmic}
\end{algorithm}

%% file: alg_voting.tex
\begin{algorithm}[tbp]
\caption{Voting for Slot $\slot$ (for validator $\proc_i$)}
\label{alg:voting}
\begin{algorithmic}[1]
\footnotesize

\State \textbf{Local variables:}
\statelabel{$\mathsf{Map}(\mathsf{Proposer} \to \mathsf{Entry})$ $\Entry(\pid)$ for each proposer $\proc_\pid \in s.\fproposers$, initially $\bot$; the entry $\proc_i$ will sign for $\proc_\pid$ (positive iff $\rhobot \neq \bot$).}

\medskip
\State \textbf{function} \textsf{onChunkValidated}($\pid, \mroot$): \BlueComment{invoked by \modcall{alg:da} when a valid chunk has been received from $\proc_\pid$}
\State \hskip2em \textbf{if} local time $\leq s.\fdeadline$ and $\Entry(\pid) = \bot$:
\State \hskip4em $\Entry(\pid) \gets \langle s, \pid, \mroot\rangle$ \label{line:vote-positive} \BlueComment{positive entry; chunk received by the deadline}

\medskip
\State \textbf{upon} local time $\geq s.\fdeadline$ for the first time:
\State \hskip2em \textbf{for each} proposer $\proc_\pid \in s.\fproposers$ with $\Entry(\pid)=\bot$:
\State \hskip4em $\Entry(\pid) \gets \langle s, \pid, \bot\rangle$ \BlueComment{negative entry: no valid chunk by the deadline}
\State \hskip2em Let $\mathit{SignedEntries} \gets \{\langle \Entry(\pid), \sigma_\pid\rangle\}_{\proc_\pid \in s.\fproposers}$ where $\sigma_\pid = \mathsf{Sign}_i(\langle\votetag, \Entry(\pid)\rangle)$
\State \hskip2em Let $\mathit{Chunks} \gets \emptyset$
\State \hskip2em \textbf{for each} $\proc_\pid \in s.\fproposers$ with $\Entry(\pid)$ positive:
\State \hskip4em add to $\mathit{Chunks}$ the chunk assigned to $\proc_i$ that it received and validated from $\proc_\pid$
\State \hskip2em Let $\mathit{DecryptShare} \gets \TIBE.\mathsf{KeyShare}(s, \mathit{MSK}_i)$
\State \hskip2em \textbf{broadcast} $\langle \textsc{Vote}, s, \mathit{SignedEntries}, \mathit{Chunks}, \mathit{DecryptShare}\rangle$ \label{line:vote-broadcast}

\end{algorithmic}
\end{algorithm}

%% file: alg_fast.tex
\begin{algorithm}[tbp]
\caption{Fast-Path for Slot $\slot$ (for validator $\proc_i$)}
\label{alg:fast-path-certification}
\begin{algorithmic}[1]
\footnotesize

\State \textbf{Local variables:}
\statelabel{$\mathsf{Map}(\mathsf{Proposer} \times (\mathsf{MerkleRoot} \cup \{\bot\}) \to \mathsf{Set}(\mathsf{SignedEntry}))$ $\mathit{Votes}(\pid, \rhobot)$ for each $\proc_\pid \in s.\fproposers$ and $\rhobot \in \mathsf{MerkleRoot} \cup \{\bot\}$, initially $\emptyset$; collects the fast-path signed entries received for $(\proc_\pid, \rhobot)$.}
\statelabel{$\mathsf{Map}(\mathsf{Proposer} \to \mathsf{Evidence})$ $\Ev(\pid)$ for each proposer $\proc_\pid \in s.\fproposers$, initially $\bot$; the strongest evidence about $\proc_\pid$, shared with \modcall{alg:fallback}. Evolves monotonically along $\bot \to$ fallback signed entry $\to \FastQC$.}
\statelabel{$\mathsf{Enum}$ $\mathit{pathVote} \in \{\textsc{none}, \textsc{fast}, \textsc{fallback}\}$, initially $\textsc{none}$; $\proc_i$'s second-phase vote (fast-commit or fallback) for slot $\slot$, shared with \modcall{alg:fallback}.}

\medskip
\State \textbf{upon} receiving $\langle \textsc{Vote}, s, \mathit{SignedEntries}, \mathit{Chunks}, \mathit{DecryptShare}\rangle$ from validator $\proc_r$:
\State \hskip2em \textbf{for each} proposer $\proc_\pid \in s.\fproposers$:
\State \hskip4em \textbf{if} $\mathit{SignedEntries}$ has no entry for $\proc_\pid$, or its signature fails verification under $\proc_r$:
\State \hskip6em \textbf{return}
\State \hskip4em Let $\langle\langle s, \pid, \rhobot\rangle, \sigma\rangle$ be the signed entry for $\proc_\pid$
\State \hskip4em \textbf{if} $\rhobot \neq \bot$, and either no $m \in \mathit{Chunks}$ has \tobiaschange{(chunk index $r$ and} matching $\rhobot$\tobiaschange{),} or $\modcall{alg:da}.\tryIngestChunk(m) = \false$:
\State \hskip6em \textbf{return}
\State \hskip2em \textbf{if} $\modcall{alg:da}.\tryIngestShare(r, \mathit{DecryptShare}) = \false$:
\State \hskip4em \textbf{return}
\State \hskip2em \textbf{for each} signed entry $\mathit{pv} = \langle E, \sigma\rangle \in \mathit{SignedEntries}$ with $E = \langle s, \pid, \rhobot\rangle$:
\State \hskip4em $\mathit{Votes}(\pid, \rhobot) \gets \mathit{Votes}(\pid, \rhobot) \cup \{\mathit{pv}\}$

\medskip
\State \textbf{upon} $|\mathit{Votes}(\pid, \rhobot)| \geq 2f+1$ for some $\proc_\pid$ and $\rhobot$:
\State \hskip2em Let $\Sigma$ be the aggregate of the signatures
\State \hskip2em $\Ev(\pid) \gets \langle \votetag, \langle s, \pid, \rhobot\rangle, \Sigma\rangle$ \label{line:fast-formqc} \BlueComment{$\FastQC$}

\medskip
\State \textbf{upon} first time $\Ev(\pid)$ is a $\FastQC$ for every proposer $\proc_\pid \in s.\fproposers$:
\State \hskip2em Let $B \gets \{\Ev(\pid)\}_{\proc_\pid \in s.\fproposers}$ \label{line:fast-metablock} \BlueComment{fast \metablock}
\State \hskip2em \textbf{broadcast} $\langle \textsc{FastBlock}, s, B\rangle$ \BlueComment{disseminate so peers can adopt; broadcast unconditionally}
\State \hskip2em \textbf{if} $\mathit{pathVote} = \textsc{none}$:
\State \hskip4em Let $\sigma_i \gets \mathsf{Sign}_i(\langle \CommitMsg, s, \entries(B)\rangle)$
\State \hskip4em \textbf{broadcast} commit vote $\langle \CommitVote, s, \entries(B), \sigma_i\rangle$ \label{line:fast-commitvote}
\State \hskip4em $\mathit{pathVote} \gets \textsc{fast}$ \label{line:fast-pathvote} \BlueComment{commit vote and fallback vote are mutually exclusive}
\State \hskip2em Let $V \gets \modcall{alg:da}.\recoverProposals(\entries(B))$
\State \hskip2em \textbf{speculatively commit} $V$

\medskip
\State \textbf{upon} receiving $\langle \textsc{FastBlock}, s, B\rangle$ where $B$ is a valid fast \metablock:
\State \hskip2em \textbf{for each} proposer $\proc_\pid \in s.\fproposers$:
\State \hskip4em $\Ev(\pid) \gets B(\pid)$

\medskip
\State \textbf{upon} $\proc_i$ has collected $2f+1$ fast commit votes carrying the same entries $\entries(B)$: \label{line:fast-collect-commit}
\State \hskip2em Let $\Sigma$ be the aggregate of the $2f+1$ signatures, and let $\commitQC \gets \langle \CommitMsg, s, \entries(B), \Sigma\rangle$
\State \hskip2em \textbf{broadcast} $\commitQC$ \label{line:fast-broadcast-commitqc}

\medskip
\State \textbf{upon} receiving a valid fast $\commitQC = \langle \CommitMsg, s, \entries(B), \Sigma\rangle$: \label{line:fast-recv-commitqc}
\State \hskip2em \textbf{broadcast} $\commitQC$ \label{line:fast-rebroadcast-commitqc}
\State \hskip2em \textbf{trigger} $\mathsf{finalize}(\modcall{alg:da}.\recoverProposals(\entries(B)))$ \label{line:fast-finalize}

\end{algorithmic}
\end{algorithm}

%% file: p2_mvba.tex
\begin{module}[h]
\caption{Multi-Valued Byzantine Agreement (MVBA)}
\label{mod:mvba}
\footnotesize
\begin{algorithmic}[1]

\Statex \textbf{Parameters:}
\begin{compactitem}[-]
    \item $\mathit{id}$: a unique identifier.
\end{compactitem}

\smallskip
\Statex \textbf{Interface:}
\begin{compactitem}[-]
    \item input $\mathsf{propose}(B)$: a validator proposes a valid meta-block $B$, thereby starting to participate.
    \item input $\mathsf{abandon}()$: a validator stops participating.
    \item output $\mathsf{decide}(B)$: a validator decides a meta-block $B$.
\end{compactitem}

\smallskip
\Statex \textbf{Properties:}
\begin{compactitem}[-]
    \item \emph{Agreement:} If a correct validator decides a meta-block $B$ and a correct validator decides a meta-block $B'$, then $B = B'$.

    \item \emph{Integrity:} Every correct validator decides at most once.

    \item \emph{External validity:} If a correct validator decides a meta-block $B$, then $B$ is a valid meta-block.

    \item \emph{$\ell_{\mathsf{MVBA}}$-Termination:} If all correct validators propose by time $t$ and no correct validator abandons before time $\max(t, \mathrm{GST}) + \ell_{\mathsf{MVBA}}$, then every correct validator decides by time $\max(t, \mathrm{GST}) + \ell_{\mathsf{MVBA}}$.

    \item \emph{Quiescence:} No correct validator sends any protocol message before it proposes or after it abandons.
\end{compactitem}
\end{algorithmic}
\end{module}

%% file: alg_fallback.tex
\begin{algorithm}[tbp]
\caption{Transition to Fallback Path for Slot $\slot$ (for validator $\proc_i$)}
\label{alg:fallback}
\begin{algorithmic}[1]
\footnotesize
\State \textbf{Local variables:}
\statelabel{$\mathsf{Map}(\mathsf{Proposer} \to \mathsf{Evidence})$ $\Ev(\pid)$ for each $\proc_\pid \in s.\fproposers$, shared with \modcall{alg:fast-path-certification}.}
\statelabel{$\mathsf{Set}(\mathsf{FallbackVote})$ $M_i$, initially $\emptyset$; the \textsc{FallbackVote} messages received for slot $\slot$.}
\statelabel{$\mathsf{Enum}$ $\mathit{pathVote}$, shared with \modcall{alg:fast-path-certification}.}
\statelabel{$\mathsf{Map}(\mathsf{Proposer} \times (\mathsf{MerkleRoot} \cup \{\bot\}) \to \mathsf{Set}(\mathsf{SignedEntry}))$ $\mathit{Votes}(\pid,\rhobot)$, shared with \modcall{alg:fast-path-certification}.}
\statelabel{$\mathsf{Bool}$ $\mathit{mvbaInvoked}$, initially $\false$.}

\smallskip
\rulecmt{If no fast \metablock can be formed one timeout after the deadline, we enter the fallback path.}
\State \textbf{upon} local time reaches $s.\fdeadline+\Delta$, at least $2f+1$ valid \textsc{Vote} messages have been received, and $\mathit{pathVote} = \textsc{none}$: \label{line:fb-pathvote-guard}
\State \hskip2em \textbf{for each} proposer $\proc_\pid \in s.\fproposers$ with $\Ev(\pid)=\bot$: \label{line:fb-cast-entry}
\State \hskip4em \textbf{if} collected $f+1$ valid positive votes for $(\proc_\pid, \mroot)$ for some $\mroot$ and $\modcall{alg:da}.\isDecoded(\mroot)$:
\State \hskip6em Let $E \gets \langle s, \pid, \mroot\rangle$, and let $\sigma_p$ be the proposer's signature on $E$ from any chunk validated under~$\mroot$
\State \hskip6em $\Ev(\pid) \gets \langle E, \mathsf{Sign}_i(\langle\textsc{fb}, E\rangle), \sigma_p\rangle$ \label{line:fb-positive-entry} \BlueComment{positive fallback signed entry}
\State \hskip6em re-encode proposal; \textbf{send} each validator its assigned chunk for $\mroot$ \label{line:fb-redisseminate} 
\State \hskip4em \textbf{else}: $\Ev(\pid) \gets \langle E, \mathsf{Sign}_i(\langle\textsc{fb}, E\rangle)\rangle$ where $E = \langle s, \pid, \bot\rangle$ \BlueComment{negative fallback signed entry}
\State \hskip2em Let $\sigma_i \gets \mathsf{Sign}_i(\langle \textsc{fallback}, s\rangle)$ \BlueComment{$\proc_i$'s fallback vote for slot $s$}
\State \hskip2em \textbf{broadcast} $\langle \textsc{FallbackVote}, s, \sigma_i, \{\Ev(\pid)\}_{\proc_\pid \in s.\fproposers}\rangle$
\State \hskip2em $\mathit{pathVote} \gets \textsc{fallback}$

\smallskip
\State \textbf{upon} receiving the first $m = \langle \textsc{FallbackVote}, \slot, \sigma, \{R_\pid\}_{\proc_\pid \in s.\fproposers} \rangle$ from a validator $\proc_r$ where $\sigma$ verifies as $\proc_r$'s signature over $\langle \textsc{fallback}, s\rangle$:
\State \hskip2em \textbf{if} some $R_\pid$ is neither a valid $\FastQC$ for $(\slot, \pid)$ nor a valid fallback signed entry from $\proc_r$ for $(\slot, \pid)$: \textbf{return} \label{line:fb-accept}
\State \hskip2em $M_i \gets M_i \cup \{m\}$
\State \hskip2em \textbf{for each} proposer $\proc_\pid$ with $R_\pid$ a $\FastQC$: $\Ev(\pid) \gets R_\pid$ \label{line:fb-harvest} \BlueComment{harvest $\FastQC$s}

\smallskip
\rulecmt{MVBA proposal, case 1: a fast \metablock is held; delayed to $s.\fdeadline + 2\Delta$ so the fast path has a chance to commit first.}
\State \textbf{upon} first time $\Ev(\pid)$ is a $\FastQC$ for every proposer $\proc_\pid \in s.\fproposers$, local time reaches $s.\fdeadline + 2\Delta$, and $\mathit{mvbaInvoked} = \false$:
\State \hskip2em $\mathit{mvbaInvoked} \gets \true$
\State \hskip2em \textbf{invoke} $\mathsf{MVBA}[\slot].\mathsf{propose}(\{\Ev(\pid)\}_{\proc_\pid \in s.\fproposers})$ \label{line:fb-mvba-propose-fast} \BlueComment{fast \metablock}

\smallskip
\rulecmt{MVBA proposal, case 2: $2f+1$ fallback votes; build the \metablock from the strongest per-proposer evidence in $M_i$.}
\State \textbf{upon} first time $|M_i| \geq 2f+1$ and $\mathit{mvbaInvoked} = \false$:
\State \hskip2em $\mathit{mvbaInvoked} \gets \true$
\State \hskip2em \textbf{for each} proposer $\proc_\pid \in s.\fproposers$: \label{line:fb-build-entry} \BlueComment{one of the three cases always applies, by counting}
\State \hskip4em \textbf{if} $\Ev(\pid)$ is a $\FastQC$: $B[\pid] \gets \Ev(\pid)$ \label{line:fb-build-fast}
\State \hskip4em \textbf{else if} two messages in $M_i$ carry positive fallback signed entries for $\proc_\pid$ with roots $\mroot_1 \neq \mroot_2$ and proposer signatures $\sigma_{p,1}, \sigma_{p,2}$:
\State \hskip6em $B[\pid] \gets$ $\langle \textsc{equiv}, s, \pid, \mroot_1, \sigma_{p,1}, \mroot_2, \sigma_{p,2}\rangle$ built from the two entries \label{line:fb-build-equiv} \BlueComment{$\EquivCert$}
\State \hskip4em \textbf{else}: $B[\pid] \gets$ a $\FallbackQC$ for $\proc_\pid$: $\langle \textsc{fb}, E, \Sigma_\pid\rangle$ with $\Sigma_\pid$ aggregating $f+1$ matching fallback signed entries $\langle E, \cdot\rangle$ in $M_i$ with $E.\pid = \pid$ \label{line:fb-formqc}
\State \hskip2em Let $B \gets \{B[\pid]\}_{\proc_\pid \in s.\fproposers}$
\State \hskip2em \textbf{if} some $B[\pid]$ is not a $\FastQC$:
\State \hskip4em Let $\Sigma_{\textsc{fb}}$ be the aggregate of $2f+1$ fallback votes (the $\sigma$ fields of messages in $M_i$)
\State \hskip4em Let $\FBCert_\slot \gets \langle \textsc{fallback}, s, \Sigma_{\textsc{fb}} \rangle$
\State \hskip4em append $\FBCert_\slot$ to $B$ \BlueComment{fallback \metablock}
\State \hskip2em \textbf{invoke} $\mathsf{MVBA}[\slot].\mathsf{propose}(B)$ \label{line:fb-mvba-propose}

\smallskip
\State \textbf{upon} $\mathsf{MVBA}[\slot].\mathsf{decide}(B')$: \label{line:fb-mvba-decide}
\State \hskip2em \textbf{for each} $\FallbackQC$ in $B'$ with a positive entry $\langle s, \pid, \mroot\rangle$: \label{line:fb-commit-foreach}
\State \hskip4em \textbf{wait until} $\proc_i$ has received and validated its assigned chunk for $\mroot$, then \textbf{broadcast} that chunk \label{line:fb-commit-wait}
\State \hskip2em Let $\sigma_i \gets \mathsf{Sign}_i(\langle \textsc{FallbackCommit}, s, \entries(B')\rangle)$
\State \hskip2em \textbf{broadcast} $\langle \textsc{FallbackCommitVote}, s, \entries(B'), \sigma_i\rangle$ \label{line:fb-commitvote}

\smallskip
\State \textbf{upon} $\proc_i$ has collected $2f+1$ \textsc{FallbackCommitVote} messages carrying the same entries $E$: \label{line:fb-collect-commit}
\State \hskip2em Let $\Sigma$ be the aggregate of the $2f+1$ signatures, and let $\mathit{fbCommitQC} \gets \langle \textsc{FallbackCommit}, s, E, \Sigma\rangle$ \label{line:fb-formcommitqc}
\State \hskip2em \textbf{broadcast} $\mathit{fbCommitQC}$ \label{line:fb-commit-broadcast} \BlueComment{transferable commitment proof}

\smallskip
\State \textbf{upon} receiving a valid $\mathit{fbCommitQC} = \langle \textsc{FallbackCommit}, s, E, \Sigma\rangle$: \label{line:fb-recv-commit}
\State \hskip2em \textbf{broadcast} $\mathit{fbCommitQC}$ \label{line:fb-commit-rebroadcast}
\State \hskip2em \textbf{trigger} $\mathsf{finalize}(\modcall{alg:da}.\recoverProposals(E))$ \label{line:fb-finalize}

\smallskip
\State \textbf{upon} $\mathsf{abandon}()$: \textbf{if} $\mathit{mvbaInvoked}$, \textbf{invoke} $\mathsf{MVBA}[\slot].\mathsf{abandon}()$ \label{line:fb-abandon} \BlueComment{forward the abandonment to the MVBA}

\end{algorithmic}
\end{algorithm}

%% file: alg_da.tex
\begin{algorithm}[hp]
\caption{Data Availability for Slot $\slot$ (for validator $\proc_i$)}
\label{alg:da}
\begin{algorithmic}[1]
\footnotesize

\State \textbf{Local variables:}
\statelabel{$\mathsf{Map}(\mathsf{MerkleRoot} \to \mathsf{Status})$ $\status(\mroot)$ for each $\mroot \in \mathsf{MerkleRoot}$, initially $\bot$, with possible values:}
\substatelabel{$\bot$: pending; insufficient chunks received for decoding.}
\substatelabel{$\langle\mathsf{cipher}, c\rangle$ for $c \in \mathsf{Ciphertext}$: ciphertext reconstructed and Merkle-verified, awaiting decryption key.}
\substatelabel{$\langle\mathsf{plain}, P\rangle$ for a decrypted proposal $P = \langle \slot, \pid, \payload, \sigma\rangle$: the recovered proposal, retaining its embedded proposer identity $\pid$ and signature $\sigma$.}
\substatelabel{$\mathsf{invalid}$: settled as not recoverable.}
\statelabel{$\mathsf{Map}(\mathsf{MerkleRoot} \to \mathsf{Set}(\mathsf{Symbol}))$ $\Data(\mroot)$ for each $\mroot \in \mathsf{MerkleRoot}$, initially maps to $\emptyset$; a set of validated erasure-code symbols $\cdata_r$.}
\statelabel{$\mathsf{Set}$ $\Shares$ initially $\emptyset$; a set of (validator-index, decryption-share) pairs.}

\medskip
\State \textbf{function} \textsf{isDecoded}($\mroot$):
\State \hskip2em \textbf{return} $(\status(\mroot) \in \{\langle\mathsf{cipher}, \cdot\rangle,\, \langle\mathsf{plain}, \cdot\rangle\})$

\medskip
\State \textbf{function} \textsf{recoverProposals}($\mathit{entries}$):
\State \hskip2em wait until, for every positive entry $\langle s, \pid, \rhobot\rangle \in \mathit{entries}$, $\status(\rhobot) \in \{\langle\mathsf{plain}, \cdot\rangle,\, \mathsf{invalid}\}$ \label{line:da-wait}
\State \hskip2em the entry is invalid unless $\status(\rhobot) = \langle\mathsf{plain}, P\rangle$ with $P$'s proposer equal to $\pid$ \label{line:da-proposer-check}
\State \hskip2em \textbf{return} the proposal vector $V$ for slot $s$ assigning each proposer its recovered proposal (or $\bot$) \label{line:da-recover-slot} 

\medskip
\State \textbf{function} \textsf{tryIngestChunk}($\mathit{chunkMsg}$):
\State \hskip2em parse $\mathit{chunkMsg} = \langle\textsc{Chunk}, \mathit{ChunkHeader}, r, \cdata_r, \pi_r\rangle$ with $\mathit{ChunkHeader} = \langle\slot, j, \mroot, \sigma\rangle$
\State \hskip2em \textbf{if} $j \notin s.\fproposers$, $\sigma$ does not verify, or $\VerifyMerkle(\mroot, \cdata_r, \pi_r) \neq \top$:
\State \hskip4em \textbf{return} $\false$
\State \hskip2em \textbf{if} $r = i$:
\State \hskip4em call $\modcall{alg:voting}.\textsc{onChunkValidated}(j, \mroot)$ \BlueComment{records a positive entry if received by the deadline}
\State \hskip2em $\Data(\mroot) \gets \Data(\mroot) \cup \{\cdata_r\}$
\State \hskip2em \textbf{if} $\status(\mroot) = \bot$ and $|\Data(\mroot)| \geq f+1$:
\State \hskip4em Let $c \gets \Decode(\Data(\mroot))$
\State \hskip4em \textbf{if} $c = \bot$ or $\MerkleRoot(\Encode(c)) \neq \mroot$: \label{line:da-reencode}
\State \hskip6em $\status(\mroot) \gets \mathsf{invalid}$
\State \hskip4em \textbf{else}:
\State \hskip6em $\status(\mroot) \gets \langle\mathsf{cipher}, c\rangle$
\State \hskip2em \textbf{return} $\true$

\medskip
\State \textbf{function} \textsf{tryIngestShare}($r, \dsk_r$):
\State \hskip2em \textbf{if} $\TIBE.\mathsf{VerifyShare}(\mathit{MPK}, \slot, r, \dsk_r) \neq \top$:
\State \hskip4em \textbf{return} $\false$
\State \hskip2em $\Shares \gets \Shares \cup \{(r, \dsk_r)\}$
\State \hskip2em \textbf{return} $\true$

\medskip
\State \textbf{upon} receiving $\mathit{chunkMsg} = \langle \textsc{Chunk}, \ldots\rangle$:
\State \hskip2em \tryIngestChunk$(\mathit{chunkMsg})$

\medskip
\State \textbf{upon} $\exists\, \mroot$ such that $\status(\mroot) = \langle\mathsf{cipher},\, c\rangle$ for some $c$, and $|\Shares| \geq f+1$:
\State \hskip2em Let $m \gets \TIBE.\Dec(\mathit{MPK}, c, \slot, \Shares)$
\State \hskip2em \textbf{if} $m = \bot$, or $m$ does not parse as a proposal $P = \langle \slot, \pid, \payload, \sigma\rangle$ whose signature $\sigma$ verifies under $\proc_\pid$:
\State \hskip4em $\status(\mroot) \gets \mathsf{invalid}$
\State \hskip2em \textbf{else}:
\State \hskip4em $\status(\mroot) \gets \langle\mathsf{plain},\, P\rangle$

\end{algorithmic}
\end{algorithm}

%% file: p2_conductor_proofs.tex
\section{\conductor: Our Window-Based Orchestrator for \cadence}\label{section:conductor-formal}

In this section, we present \conductor, a protocol that serves as the orchestrator primitive within \cadence.
Recall that we described \conductor informally in \Cref{section:conductor-overview}.

\subsection{Protocol}

Here, we introduce \conductor (\Cref{algorithm:conductor}).
Before presenting it, we describe the building block it relies on, \emph{agreement on a core set} (ACS).

\subsubsection{Agreement on a Core Set (ACS).}
The full specification of ACS is given in \Cref{mod:acs}.
In a nutshell, each of the $n$ validators submits a proposal, and ACS outputs a single set, agreed upon by all correct validators, that contains at least $n - f$ of the submitted proposals.
Proposing doubles as starting to participate in the instance, a validator stops participating by invoking $\mathsf{abandon}()$, and the instance is quiescent outside this window.
We emphasize that only the two timing properties --- $\ell$-termination and $\Delta$-totality --- are conditioned on the module's two assumptions; all other properties hold unconditionally.

\begin{module}[h]
\caption{Agreement on a Core Set (ACS)}
\label{mod:acs}
\footnotesize
\begin{algorithmic}[1]

\Statex \textbf{Parameters:}
\begin{compactitem}[-]
    \item $\mathit{id} \in \mathbb{N}_{\geq 1}$: a unique identifier.
\end{compactitem}

\smallskip
\Statex \textbf{Interface:}
\begin{compactitem}[-]
    \item input $\mathsf{propose}(s \in \mathsf{Slot})$: a validator proposes slot $s$, thereby starting to participate.
    \item input $\mathsf{abandon}()$: a validator stops participating.
    \item output $\mathsf{decide}(\mathit{set} \in \mathsf{Set}(\mathsf{Validator} \times \mathsf{Slot}))$: a validator decides a set $\mathit{set}$ of validator-slot pairs.
\end{compactitem}

\smallskip
\Statex \textbf{Assumptions:}
\begin{compactitem}[-]
    \item \emph{$\Delta$-synchronized proposals:} If a correct validator proposes at some time $t$, then every correct validator proposes by time $\max(t, \mathrm{GST}) + \Delta$.
    \item \emph{No premature abandonment:} If a correct validator proposes, it does not abandon before deciding.
\end{compactitem}

\smallskip
\Statex \textbf{Properties:}
\begin{compactitem}[-]
    \item \emph{Agreement:} No two correct validators decide different sets.

    \item \emph{$\ell$-Termination:} Under the above assumptions: if all correct validators propose by some time $t$, then every correct validator decides by time $\max(t, \mathrm{GST}) + \ell$.

    \item \emph{Validity:} If a correct validator decides a set $\mathit{set}$, then $|\mathit{set}| \geq 2f + 1$, and for every validator-slot pair $(p_i, s_i) \in \mathit{set}$ such that $p_i$ is a correct validator, $p_i$ proposed slot $s_i$.

    \item \emph{Integrity:} If a correct validator decides, then it has previously proposed.

    \item \emph{$\Delta$-Totality:} Under the above assumptions: if a correct validator decides at some time $t$, then all correct validators decide by time $\max(t, \mathrm{GST}) + \Delta$.

    \item \emph{Quiescence:} No correct validator sends any protocol message before it proposes or after it abandons.
\end{compactitem}
\end{algorithmic}
\end{module}

\subsubsection{Protocol Description.}
\Cref{algorithm:conductor} gives \conductor's pseudocode, written from the perspective of a single correct validator $\proc_i$; recall that validators' clocks are synchronized, so they share one global timeline.
As described informally in \Cref{section:conductor-overview}, \conductor groups slots into \emph{windows} of $W$ consecutive slots and schedules them one at a time, opening a window only once enough of the earlier ones have completed.
Window $1$ consists of slots $1, \ldots, W$, and each subsequent window again spans $W$ consecutive slots.
Unlike in the informal description, however, consecutive windows need not be adjacent: a window's first slot is chosen dynamically and may lie strictly beyond the previous window's last slot, leaving a gap of intervening slots that are simply never opened.
(This stems from the slight difference in the problem definition between the two parts: here, slots have fixed, read-only deadlines, whereas in the informal part the protocol determines them. The orchestrator therefore cannot freely place the next deadline; it can only select an existing slot whose fixed deadline matches, skipping the intervening ones.)
Throughout, $\proc_i$ tracks the window it is currently scheduling in $\mathit{current\_window}_i$ (initially $1$, line~\ref{line:current_window_init}), the slots it has opened in $\mathit{opened}_i$, the slots reported complete in $\mathit{completed}_i$, the last (highest-numbered) slot of each window in $\mathit{last}_i$, and the windows for which it has already submitted an ACS proposal in $\mathit{proposed}_i$.

\smallskip
\noindent\emph{Parameters \& constants.}
\conductor takes two tunable parameters: the window size $W \in \mathbb{N}_{\geq 1}$ and the readiness threshold $p \in \{0, \dots, W-1\}$, which controls how much of a window must complete before the next one is opened.
Its timing further involves several constants: $\ell$, the termination latency of the underlying ACS primitive (line~\ref{line:acs-instances}); $\ell_{\textsc{chorus}}$, the termination latency of \chorus, the slot-consensus protocol whose finalizations complete \conductor's slots within \cadence; $\dtot = \Delta$ (\Cref{prop:chorus-totality}), the \emph{totality latency} of \chorus --- the time within which one correct validator's completion of a slot propagates to all correct validators; $\Delta$, the post-$\mathrm{GST}$ bound on message delays; and $\tau$, the spacing between consecutive slot deadlines.
From these we derive the \emph{open-to-complete delay} $\Phi_{oc} = \ell_{\textsc{chorus}} + \dtot$, which bounds the time between a correct validator opening a slot and completing it (\Cref{prop:conductor-open-to-complete}).
For \conductor to achieve its recovery property, $W$ and $p$ must be chosen so that the following four assumptions hold (lines~\ref{line:assumption-one}--\ref{line:assumption-four}):
\begin{compactitem}
    \item $(p - 1)\tau + \Phi_{oc} + \ell \leq W\tau$;
    \item $(p - 1)\tau + \Phi_{oc} \leq (W - 1)\tau$;
    \item $\Delta < \ell$; and
    \item $\dtot + \ell \leq (p - 1)\tau$.
\end{compactitem}
Intuitively, they ensure that all the work tied to a window --- proposing to the next ACS instance, deciding, and completing the window's slots --- fits within its time span of $W\tau$, so that, once the network is synchronous, one window follows the next with no gaps.
These assumptions are crucial for recovery, as we show in \Cref{subsection:conductor-proof}.

\smallskip
\noindent\emph{Startup.}
Upon starting (line~\ref{line:conductor-startup}), $\proc_i$ enters window $1$ (line~\ref{line:enter_window_1}) and schedules the opening of its $W$ slots in increasing order (lines~\ref{line:startup-foreach}--\ref{line:startup-opened-update}), recording slot $W$ as that window's last slot (line~\ref{line:startup-last}).
Opening is asynchronous: $\mathsf{schedule\_opening}(s)$ waits until the slot's starting time $s.\fdeadline - \Delta$ --- or fires at once if it has already passed --- and then outputs $\mathsf{open}(s)$ (lines~\ref{line:conductor-wait-for-open}--\ref{line:trigger-open}).
Since deadlines are $\tau$-spaced, window $1$'s slots thus open at times $0, \tau, \ldots, (W-1)\tau$.

\smallskip
\noindent\emph{Completions.}
After startup, the only input to \conductor is the completion of slots: upon $\mathsf{completed}(s)$ (line~\ref{line:upon-completed}), $\proc_i$ adds $s$ to $\mathit{completed}_i$ (line~\ref{line:completed-update}).
Within \cadence, $\proc_i$ completes a slot precisely when it finalizes that slot's proposal vector.

\smallskip
\noindent\emph{Estimating the next window's deadline.}
\conductor opens window $\mathit{current\_window}_i + 1$ only once enough of the current schedule has completed, as captured by $\mathsf{ready\_for\_next\_window}()$ (line~\ref{line:ready-check}): writing $s_1 < \cdots < s_k$ for the opened slots, it returns $\mathsf{true}$ once all but at most the last $W - p$ of them are complete --- equivalently, once every slot of the earlier windows and the first $p$ slots of the current window are done.
When this first holds for the next window (line~\ref{line:ready}), $\proc_i$ estimates that window's first slot: it reads the current time $t_{\mathit{cur}}$ (line~\ref{line:ready-time}) and lets $s^\star$ be the earliest slot whose starting time has not yet passed (line~\ref{line:sstar-compute}); if that slot still falls within the current window, it advances $s^\star$ to the first slot beyond it (lines~\ref{line:sstar-guard}--\ref{line:sstar-update}).
These two cases are exactly the estimate of \Cref{fig:conductor-schedule}: when the system is keeping up, $s^\star$ is the slot immediately after the current window, so the cadence continues with no gap; when it is lagging, $s^\star$ is pushed to $t_{\mathit{cur}}$, leaving a gap for the backlog to drain.
Validator $\proc_i$ then proposes $s^\star$ to the next window's agreement instance $\mathcal{ACS}[\mathit{current\_window}_i + 1]$ (line~\ref{line:acs-propose}) and records that it has done so (line~\ref{line:proposed-update}), so that it proposes at most once per window.

\smallskip
\noindent\emph{Agreeing on and opening the next window.}
Different correct validators may complete the first $p$ slots at different times and thus estimate different first slots for the next window, so they agree on one through ACS.
Once $\mathcal{ACS}[\mathit{current\_window}_i + 1]$ decides and $\proc_i$ is ready (line~\ref{line:acs-decide}), $\proc_i$ abandons the instance, which has served its purpose (line~\ref{line:acs-abandon}); as abandonment thus never precedes decision, the no-premature-abandonment assumption of \Cref{mod:acs} is satisfied by construction (we record this formally in \Cref{prop:acs-no-premature-abandonment}).
Validator $\proc_i$ then advances to the next window (lines~\ref{line:window-increment}--\ref{line:enter_window_omega}) and takes its first slot $s^\star$ to be the \emph{median} of the decided estimates (line~\ref{line:median-compute}); since at most $f$ of the $2f+1$ decided values are faulty, the median lies between the smallest and largest honest estimate, bounding Byzantine influence.
It then schedules the opening of the window's $W$ slots $s^\star, \ldots, s^\star + W - 1$ in increasing order (lines~\ref{line:open-foreach}--\ref{line:acs-opened-update}) and records $s^\star + W - 1$ as the window's last slot (line~\ref{line:last-update}).

\begin{algorithm}[p]
\caption{\conductor: Pseudocode (for validator $p_i$)}
\label{algorithm:conductor}
\begin{algorithmic}[1]
\footnotesize

\State \textbf{Parameters:}
\State \hskip2em $\mathsf{Integer}$ $W \in \mathbb{N}_{\geq 1}$ \BlueComment{window size}
\State \hskip2em $\mathsf{Integer}$ $p \in \{0, \dots, W-1\}$ \BlueComment{readiness threshold}
\State \hskip2em $\Time$ $\Phi_{oc} = \ell_{\textsc{chorus}} + \dtot$ \BlueComment{open-to-complete delay}
\State \hskip2em $\Time$ $\ell$ \BlueComment{latency of the utilized ACS}

\medskip
\State \textbf{Assumptions on the parameters:}

\State \hskip2em (1) $(p - 1)\tau + \Phi_{oc} + \ell \leq W\tau$ \label{line:assumption-one}
\State \hskip2em (2) $(p - 1)\tau + \Phi_{oc} \leq (W - 1)\tau$ \label{line:assumption-two}
\State \hskip2em (3) $\Delta < \ell$ \label{line:assumption-three}
\State \hskip2em (4) $\dtot + \ell \leq (p - 1)\tau$ \label{line:assumption-four}

\medskip
\State \textbf{Uses:}
\State \hskip2em ACS ($\ell$-termination), \textbf{instances} $\mathcal{ACS}[i]$, for every $i \in \mathbb{N}_{\geq 2}$, parameterized by $i$ \label{line:acs-instances}

\medskip
\State \textbf{Local variables:}
\State \hskip2em $\mathsf{Integer}$ $\mathit{current\_window}_i \gets 1$ \label{line:current_window_init}
\State \hskip2em $\mathsf{Set}(\mathsf{Slot})$ $\mathit{opened}_i \gets \emptyset$
\State \hskip2em $\mathsf{Set}(\mathsf{Slot})$ $\mathit{completed}_i \gets \emptyset$
\State \hskip2em $\mathsf{Map}(\mathsf{Integer} \to \mathsf{Slot})$ $\mathit{last}_i$, with $\mathit{last}_i[j] \gets \bot$ for all $j \in \mathbb{N}_{\geq 1}$
\State \hskip2em $\mathsf{Set}(\mathsf{Integer})$ $\mathit{proposed}_i \gets \emptyset$ \BlueComment{windows for which an ACS proposal has been submitted}

\medskip
\State \textbf{Local functions:}
\State \hskip2em \textbf{function} $\mathsf{ready\_for\_next\_window}() \to \mathsf{Bool}$:
\State \hskip4em Let $s_1 < s_2 < \cdots < s_k$ be the slots in $\mathit{opened}_i$
\State \hskip4em Let $j$ be the largest index such that $s_1, \ldots, s_j \in \mathit{completed}_i$ \hspace{0.5em} ($j = 0$ if none)
\State \hskip4em \textbf{return} $k - j \leq W - p$ \label{line:ready-check}

\smallskip
\State \hskip2em \textcolor{blue}{$\triangleright$ asynchronous: returns immediately; opening is triggered in the background at the appropriate time}
\State \hskip2em \textbf{function} $\mathsf{schedule\_opening}(s \in \mathsf{Slot}) \to \mathsf{Void}$:
\State \hskip4em \textcolor{blue}{$\triangleright$ open at $s.\mathsf{deadline} - \Delta$, or immediately if that time has already passed}
\State \hskip4em \textbf{wait until} the first time $\geq \max(\text{current local time},\ s.\mathsf{deadline} - \Delta)$ \label{line:conductor-wait-for-open}
\State \hskip4em \textbf{trigger} $\mathsf{open}(s)$ \BlueComment{output of \conductor} \label{line:trigger-open}

\medskip
\State \textbf{upon} starting the protocol: \label{line:conductor-startup}
\State \hskip2em \textbf{enter} window $1$ \BlueComment{marked for the analysis} \label{line:enter_window_1}
\State \hskip2em \textbf{for each} $s \in \mathsf{Slot}$ such that $s.\mathsf{number} \in [1, W]$: \BlueComment{process slots in increasing order} \label{line:startup-foreach}
\State \hskip4em \textbf{invoke} $\mathsf{schedule\_opening}(s)$ \label{line:startup-open}
\State \hskip4em $\mathit{opened}_i \gets \mathit{opened}_i \cup \{ s \}$ \label{line:startup-opened-update}
\State \hskip2em $\mathit{last}_i[1] \gets $ the slot whose number is $W$ \label{line:startup-last}

\medskip
\State \textbf{upon} $\mathsf{completed}(s \in \mathsf{Slot})$: \label{line:upon-completed}
\State \hskip2em $\mathit{completed}_i \gets \mathit{completed}_i \cup \{ s \}$ \label{line:completed-update}

\medskip
\State \textbf{upon} $\mathsf{ready\_for\_next\_window}() = \true$ \textbf{and} $(\mathit{current\_window}_i + 1) \notin \mathit{proposed}_i$: \label{line:ready}
\State \hskip2em Let $t_{\mathit{cur}}$ be the current local time \label{line:ready-time}
\State \hskip2em Let $s^{\star}$ be the slot with the smallest number whose starting time (i.e., $s^{\star}.\mathsf{deadline} - \Delta$) is $\geq t_{\mathit{cur}}$ \label{line:sstar-compute}
\State \hskip2em \textbf{if} $s^{\star}.\mathsf{number} \leq \mathit{last}_i[\mathit{current\_window}_i].\mathsf{number}$: \label{line:sstar-guard}
\State \hskip4em $s^{\star} \gets$ the slot whose number is $\mathit{last}_i[\mathit{current\_window}_i].\mathsf{number} + 1$ \label{line:sstar-update}
\State \hskip2em \textbf{invoke} $\mathcal{ACS}[\mathit{current\_window}_i + 1].\mathsf{propose}(s^{\star})$ \label{line:acs-propose}
\State \hskip2em $\mathit{proposed}_i \gets \mathit{proposed}_i \cup \{ \mathit{current\_window}_i + 1 \}$ \label{line:proposed-update}

\medskip
\State \textbf{upon} $\mathcal{ACS}[\mathit{current\_window}_i + 1].\mathsf{decide}(\mathit{set} \in \mathsf{Set}(\mathsf{Validator} \times \mathsf{Slot}))$ \textbf{and} $\mathsf{ready\_for\_next\_window}() = \mathsf{true}$: \label{line:acs-decide}
\State \hskip2em \textbf{invoke} $\mathcal{ACS}[\mathit{current\_window}_i + 1].\mathsf{abandon}()$ \BlueComment{decided; the instance is no longer needed} \label{line:acs-abandon}
\State \hskip2em $\mathit{current\_window}_i \gets \mathit{current\_window}_i + 1$ \label{line:window-increment}
\State \hskip2em \textbf{enter} window $\mathit{current\_window}_i$ \BlueComment{marked for the analysis} \label{line:enter_window_omega}
\State \hskip2em Let $s^{\star}$ be the slot with $s^{\star}.\mathsf{number} = \mathsf{median}(\{ s.\mathsf{number} : (\cdot,\, s) \in \mathit{set} \})$ \label{line:median-compute}

\State \hskip2em \textbf{for each} $s \in \mathsf{Slot}$ with $s.\mathsf{number} \in [s^{\star}.\mathsf{number},\, s^{\star}.\mathsf{number} + W)$: \BlueComment{process slots in increasing order} \label{line:open-foreach}
\State \hskip4em \textbf{invoke} $\mathsf{schedule\_opening}(s)$ \label{line:acs-decide-open}
\State \hskip4em $\mathit{opened}_i \gets \mathit{opened}_i \cup \{ s \}$ \label{line:acs-opened-update}

\State \hskip2em $\mathit{last}_i[\mathit{current\_window}_i] \gets $ the slot whose number is $s^{\star}.\mathsf{number} + W - 1$ \label{line:last-update}

\end{algorithmic}
\end{algorithm}

\subsection{Proof}
\label{subsection:conductor-proof}

This section proves that, when run within \cadence, \conductor satisfies the orchestrator specification of \Cref{mod:orchestrator_2}.
Before proving the individual properties, we first establish a number of structural facts about \conductor's windows; these underpin the entire subsequent analysis of the specific orchestrator properties.

\paragraph{Structural facts about windows.}

We first establish that windows are entered at most once and in strictly increasing order.

\begin{proposition}
\label{lemma:window-entry}
For every correct validator $\proc_i$ and every window $\omega \in \mathbb{N}_{\geq 1}$, $\proc_i$ enters window $\omega$ at most once.
Moreover, if $\omega \geq 2$, then $\proc_i$ enters window $\omega$ only after having previously entered window $\omega - 1$.
\end{proposition}
\begin{proof}
Let us first prove that $\proc_i$ enters window $\omega$ at most once.
First, note that the variable $\mathit{current\_window}_i$ of a correct validator never decreases: it is modified only at line~\ref{line:window-increment}, where it is incremented by one.
Moreover, any two successive enterings are separated by such an increment.
Indeed, only the first entering occurs at startup (line~\ref{line:enter_window_1}), and this happens exactly once since the startup handler fires only once; every later entering occurs at line~\ref{line:enter_window_omega}, immediately after $\mathit{current\_window}_i$ is incremented at line~\ref{line:window-increment}.
Hence each entering corresponds to a distinct, strictly increasing value of $\mathit{current\_window}_i$, so no window is entered more than once.

We now turn to the second claim.
The handler that enters window $\omega \geq 2$ (line~\ref{line:acs-decide}) fires only when $\mathit{current\_window}_i$ equals $\omega - 1$; it then increments $\mathit{current\_window}_i$ to $\omega$ (line~\ref{line:window-increment}) and enters window $\omega$ (line~\ref{line:enter_window_omega}).
If $\omega - 1 = 1$, then $\mathit{current\_window}_i = 1$ was set at initialization (line~\ref{line:current_window_init}), which coincides with $\proc_i$ entering window $1$ at startup (line~\ref{line:enter_window_1}).
If $\omega - 1 \geq 2$, then $\mathit{current\_window}_i$ reached $\omega - 1$ exclusively by being incremented at line~\ref{line:window-increment}; this increment executes in the same handler that enters window $\omega - 1$ (line~\ref{line:enter_window_omega}).
In both cases, $\proc_i$ must have previously entered window $\omega - 1$. \qed
\end{proof}

For every correct validator $\proc_i$ and every window $\omega \in \mathbb{N}_{\geq 1}$ that $\proc_i$ enters, we define the following set of slots:
\[
\mathsf{slots}_i(\omega) \equiv \{ s \in \Slot : \proc_i \text{ schedules the opening of } s \text{ upon entering window } \omega \}.
\]
Concretely, $\mathsf{slots}_i(1)$ consists of the $W$ slots scheduled to be opened at startup (lines~\ref{line:startup-foreach}--\ref{line:startup-opened-update}).
For $\omega \geq 2$, $\mathsf{slots}_i(\omega)$ consists of the $W$ slots scheduled to be opened in the loop at lines~\ref{line:open-foreach}--\ref{line:acs-opened-update}.

Throughout the rest of the proof, we adopt the following notational convention.
Recall that, in \conductor, each validator invokes the ACS primitive to decide a vector of validator-slot pairs (line~\ref{line:acs-decide}) and then extracts the slot whose number is the median over all slot numbers in the vector (line~\ref{line:median-compute}).
For brevity, we say that the ACS primitive \emph{decides slot $s$} to refer to the outcome of this combined decision and extraction step, where $s$ is the resulting slot.
Crucially, the slot number of $s$ lies between the minimum and maximum slot numbers proposed by correct validators (since the decided vector contains at least $f + 1$ pairs contributed by correct validators).
We next establish that the slot sets of distinct windows are pairwise disjoint.
To this end, we first establish that the ACS primitives never decide overlapping slots.

\begin{proposition}
\label{prop:acs-nonoverlap}
For every correct validator $p_i$ and every window $\omega \in \mathbb{N}_{\geq 2}$, the following holds:
\begin{compactitem}
    \item If $\omega = 2$ and $p_i$ decides slot $s$ from $\mathcal{ACS}[2]$, then $s.\mathsf{number} > W$.
    \item If $\omega > 2$ and $p_i$ decides slot $s$ from $\mathcal{ACS}[\omega]$ and slot $s'$ from $\mathcal{ACS}[\omega - 1]$, then $s.\mathsf{number} \geq s'.\mathsf{number} + W$.
\end{compactitem}
\end{proposition}
\begin{proof}
We first address the case $\omega = 2$.
No correct validator proposes to $\mathcal{ACS}[2]$ a slot with number $\leq W$. 
To see this, observe that a correct validator $p_k$ proposes to $\mathcal{ACS}[2]$ only via the rule at line~\ref{line:ready}, which fires when $\mathsf{ready\_for\_next\_window()}$ returns $\mathsf{true}$.
At that point, $s^{\star}$ is computed as the first slot whose starting time is at least the current local time (line~\ref{line:sstar-compute}).
If $s^{\star}.\mathsf{number} \leq \mathit{last}_k[1].\mathsf{number}$ (line~\ref{line:sstar-guard}), the guard at line~\ref{line:sstar-update} reassigns $s^{\star}$ to the slot with number $\mathit{last}_k[1].\mathsf{number} + 1$.
Crucially, $\mathit{last}_k[1].\mathsf{number} = W$ always holds, since $\mathit{last}_k[1]$ is set exactly once (upon entering window $1$), at line~\ref{line:startup-last}, to the slot with number $W$.
Hence, every correct proposal to $\mathcal{ACS}[2]$ carries a slot number of at least $W + 1$.
Since the decided slot number is at least the minimum slot number among all correct proposals, we obtain $s.\mathsf{number} \geq W + 1$, establishing the first point.

We now address the case $\omega > 2$.
First, we establish that no correct validator proposes to $\mathcal{ACS}[\omega]$ a slot with number $< s'.\mathsf{number} + W$.
A correct validator $p_k$ proposes to $\mathcal{ACS}[\omega]$ only via the rule at line~\ref{line:ready}, which fires when $\mathsf{ready\_for\_next\_window()}$ returns $\mathsf{true}$.
At this point, $\mathit{current\_window}_k + 1 = \omega$ (line~\ref{line:acs-propose}), so $\mathit{current\_window}_k = \omega - 1 \geq 2$.
Since $\mathit{current\_window}_k$ is incremented only at line~\ref{line:window-increment}, $p_k$ must have previously decided from $\mathcal{ACS}[\omega - 1]$ (line~\ref{line:acs-decide}).
By the agreement property of $\mathcal{ACS}[\omega-1]$, every correct validator that decides from $\mathcal{ACS}[\omega - 1]$ does so with slot $s'$; hence, line~\ref{line:last-update} sets $\mathit{last}_k[\omega-1]$ to the slot with number $s'.\mathsf{number} + W - 1$.
Thus, when $p_k$ computes $s^{\star}$ at line~\ref{line:sstar-compute}, if $s^{\star}.\mathsf{number} \leq \mathit{last}_k[\omega - 1].\mathsf{number} = s'.\mathsf{number} + W - 1$ (line~\ref{line:sstar-guard}), line~\ref{line:sstar-update} reassigns $s^{\star}$ to the slot with number $s'.\mathsf{number} + W$.
Hence, every correct proposal to $\mathcal{ACS}[\omega]$ carries a slot number of at least $s'.\mathsf{number} + W$.
Since the decided slot number is at least the minimum slot number among all correct proposals, we obtain $s.\mathsf{number} \geq s'.\mathsf{number} + W$, establishing the second point. \qed
\end{proof}

We now precisely characterize, for each window, the exact set of slots that a correct validator opens upon entering it, and establish that consecutive windows cover non-overlapping ranges of $W$ slots.

\begin{proposition}
\label{prop:acs-fate-range}
For every correct validator $\proc_i$ and every window $\omega \in \mathbb{N}_{\geq 1}$ that $\proc_i$ enters, the following holds:
\begin{compactitem}
    \item If $\omega = 1$, then $\mathsf{slots}_i(1) = \{ s \in \Slot : s.\fnumber \leq W \}$.

    \item If $\omega \geq 2$, let $s_\omega$ denote the slot decided by $\proc_i$ from $\mathcal{ACS}[\omega]$, and let $s_{\omega-1}$ denote the slot with the largest number in $\mathsf{slots}_i(\omega - 1)$.
    Then, $s_\omega.\fnumber > s_{\omega-1}.\fnumber$, and moreover:
    \[
    \mathsf{slots}_i(\omega) = \{ s \in \Slot : s.\fnumber \in [s_\omega.\fnumber,\ s_\omega.\fnumber + W - 1] \}.
    \]
\end{compactitem}
\end{proposition}
\begin{proof}
The first point follows directly from the startup handler (lines~\ref{line:startup-foreach}--\ref{line:startup-open}), which schedules the opening of exactly the slots with number in $[1, W]$.

We now prove the second point.
Fix any window $\omega \geq 2$ that $\proc_i$ enters.
We first observe that $\mathit{last}_i[\omega - 1]$ equals $s_{\omega - 1}$, the largest slot in $\mathsf{slots}_i(\omega - 1)$.
Indeed, if $\omega - 1 = 1$, then line~\ref{line:startup-last} sets $\mathit{last}_i[1]$ to the slot with number $W$, which is exactly the largest slot in $\mathsf{slots}_i(1) = \{ s : s.\fnumber \leq W \}$.
If $\omega - 1 \geq 2$, then line~\ref{line:last-update} sets $\mathit{last}_i[\omega - 1]$ to the largest of the $W$ slots opened upon entering window $\omega - 1$ (lines~\ref{line:open-foreach}--\ref{line:acs-opened-update}), which is exactly the largest slot in $\mathsf{slots}_i(\omega - 1)$.
Upon activating line~\ref{line:acs-decide} for $\mathcal{ACS}[\omega]$, $\proc_i$ opens every slot with number in $[s_\omega.\fnumber, s_\omega.\fnumber + W - 1]$ (line~\ref{line:open-foreach}).
By \Cref{prop:acs-nonoverlap}, $s_\omega.\fnumber > s_{\omega-1}.\fnumber$, establishing the ordering claim. \qed
\end{proof}

For each correct validator $\proc_i$ and each slot $s$, we define $\mathsf{window}_i(s)$ to be the window $\omega$ to which $s$ belongs, that is, the window $\omega$ with
\[
s \in \mathsf{slots}_i(\omega).
\]
If no such window exists, we set $\mathsf{window}_i(s) = \bot$.
By \Cref{prop:acs-fate-range}, whenever such a window exists, it is unique; hence, $\mathsf{window}_i(s)$ is well-defined.
We now establish that no two correct validators disagree on the window to which a slot belongs.

\begin{proposition}
\label{prop:window-agreement}
For every two correct validators $p_i$ and $p_j$ and every slot $s$, if $\mathsf{window}_i(s) \neq \bot$ and $\mathsf{window}_j(s) \neq \bot$, then $\mathsf{window}_i(s) = \mathsf{window}_j(s)$.
\end{proposition}
\begin{proof}
If $s.\mathsf{number} \leq W$, both windows equal $1$, so the claim holds trivially.
Thus, assume $s.\mathsf{number} > W$.
By \Cref{lemma:window-entry}, correct validators enter windows in increasing order.
Therefore, validator $p_i$ (resp., $p_j$) activates the rule at line~\ref{line:acs-decide} for $\mathcal{ACS}[2], \ldots, \mathcal{ACS}[\mathsf{window}_i(s)]$ (resp., $\mathcal{ACS}[2], \ldots, \mathcal{ACS}[\mathsf{window}_j(s)]$) in increasing order.
The agreement property of each ACS primitive together with \Cref{prop:acs-fate-range} then gives $\mathsf{window}_i(s) = \mathsf{window}_j(s)$. \qed
\end{proof}

\paragraph{Integrity.}
We now prove the integrity property specified in \Cref{mod:orchestrator_2}, which ensures that each correct validator opens each slot at most once, and that no slot is opened before its starting time.
We emphasize that integrity is unconditional: it holds for \conductor in isolation.

\begin{lemma} [Integrity]
\label{lemma:conductor-integrity}
\conductor (\Cref{algorithm:conductor}) satisfies integrity.
\end{lemma}
\begin{proof}
Fix a correct validator $\proc_i$ and a slot $s$.
By construction of the $\mathsf{schedule\_opening}$ function (line~\ref{line:conductor-wait-for-open}), $\proc_i$ does not open $s$ before time $s.\fdeadline - \Delta$.
It remains to show that $\proc_i$ opens $s$ at most once.
We distinguish two cases:
\begin{compactitem}
    \item If $\mathsf{window}_i(s) = \bot$, then $\proc_i$ never opens $s$, and the claim holds trivially.
    \item Otherwise, let $\omega = \mathsf{window}_i(s)$, so that $s \in \mathsf{slots}_i(\omega)$. By \Cref{lemma:window-entry}, $\proc_i$ enters window $\omega$ at most once.
    Hence, $p_i$ opens $s$ exactly once in this case. \qed
\end{compactitem}
\end{proof}

\paragraph{Monotonicity.}
We now establish that each correct validator schedules the opening of slots in strictly increasing order of slot number, which we leverage to prove monotonicity.
The following proposition states this formally: whenever a correct validator schedules the openings of two slots, it schedules the smaller-numbered one first.

\begin{proposition}
\label{prop:fate-order}
For every correct validator $\proc_i$ and every two slots $s, s'$ with $s.\fnumber < s'.\fnumber$, if $\proc_i$ schedules the opening of both $s$ and $s'$, then it schedules the opening of $s'$ only after having previously scheduled the opening of $s$.
\end{proposition}
\begin{proof}
Since $\proc_i$ schedules the opening of both $s$ and $s'$, we have $\mathsf{window}_i(s) \neq \bot$ and $\mathsf{window}_i(s') \neq \bot$.
Moreover, since $s.\fnumber < s'.\fnumber$, \Cref{prop:acs-fate-range} gives $\mathsf{window}_i(s) \leq \mathsf{window}_i(s')$.
We distinguish two cases:
\begin{compactitem}
    \item $\mathsf{window}_i(s) = \mathsf{window}_i(s')$.
    Both slots are scheduled to be opened upon entering this common window, within the loop that processes slots in strictly increasing order of slot number (line~\ref{line:startup-foreach} for window $1$, line~\ref{line:open-foreach} otherwise); since $s.\fnumber < s'.\fnumber$, $\proc_i$ schedules the opening of $s$ before that of $s'$.

    \item $\mathsf{window}_i(s) < \mathsf{window}_i(s')$.
    By \Cref{lemma:window-entry}, $\proc_i$ enters window $\mathsf{window}_i(s)$ before window $\mathsf{window}_i(s')$; as $s$ is scheduled upon entering the former and $s'$ upon entering the latter, $\proc_i$ schedules the opening of $s$ before that of $s'$. \qed
\end{compactitem}

\end{proof}

We are now ready to prove that \conductor satisfies monotonicity (see \Cref{mod:orchestrator_2}).
We underline that, like integrity, \conductor satisfies monotonicity unconditionally.

\begin{lemma} [Monotonicity]
\label{lemma:conductor-monotonicity}
\conductor (\Cref{algorithm:conductor}) satisfies monotonicity.
\end{lemma}
\begin{proof}
Fix any correct validator $p_i$ and any two slots $s$ and $s'$ that $p_i$ opens.
Without loss of generality, let $s'.\fnumber > s.\fnumber$.
Since $s'.\mathsf{number} > s.\mathsf{number}$, the definition of slots implies $s'.\mathsf{deadline} > s.\mathsf{deadline}$ (see \Cref{section:formal_problem_definition}).
By \Cref{prop:fate-order}, $p_i$ schedules the opening of $s$ prior to scheduling the opening of $s'$.
We distinguish two cases.
\begin{compactitem}
    \item Suppose $p_i$ opens $s$ at the moment of scheduling its opening (i.e., it does not wait for $s$'s starting time).
    In this case, when opening $s'$ (which cannot occur before the scheduling of its own opening), $p_i$ has already opened $s$.

    \item Suppose $p_i$ opens $s$ after scheduling its opening.
    Hence, $p_i$ opens $s$ at time $s.\mathsf{deadline} - \Delta$ (line~\ref{line:conductor-wait-for-open}).
    In this case, when opening $s'$, which cannot occur before time $s'.\mathsf{deadline} - \Delta > s.\mathsf{deadline} - \Delta$ (line~\ref{line:conductor-wait-for-open}), $p_i$ has already opened $s$. \qed
\end{compactitem}
\end{proof}

\paragraph{Boundedness.}
We say that a correct validator $\proc_i$ is \emph{in window $\omega \in \mathbb{N}_{\geq 1}$} from the moment it enters window $\omega$ until the moment it enters the next window (which is $\omega + 1$, by \Cref{lemma:window-entry}).
The following proposition relates the window in which a correct validator currently resides to the number of slots it has accumulated in its local variable $\mathit{opened}_i$.

\begin{proposition}
\label{prop:open-count-window}
For every correct validator $\proc_i$ and every window $\omega \in \mathbb{N}_{\geq 1}$, while $\proc_i$ is in window $\omega$, the local variable $\mathit{opened}_i$ contains exactly $\omega \cdot W$ slots.
\end{proposition}
\begin{proof}
We proceed by induction on $\omega$.

\smallskip
\noindent \emph{Base case} ($\omega = 1$).
Validator $\proc_i$ enters window $1$ at startup, upon which the startup handler adds exactly the slots with number in $[1, W]$ to $\mathit{opened}_i$ (lines~\ref{line:startup-foreach}--\ref{line:startup-opened-update}).
As $\mathit{opened}_i$ is modified only upon entering a new window, the count remains $1 \cdot W$ throughout the time $\proc_i$ is in window $1$.

\smallskip
\noindent \emph{Inductive step} ($\omega \geq 2$).
By the inductive hypothesis, $|\mathit{opened}_i| = (\omega - 1) \cdot W$ while $\proc_i$ is in window $\omega - 1$.
Validator $\proc_i$ enters window $\omega$ upon activating the rule at line~\ref{line:acs-decide}, where the loop at lines~\ref{line:open-foreach}--\ref{line:acs-opened-update} adds exactly $W$ fresh slots to $\mathit{opened}_i$, yielding $|\mathit{opened}_i| = (\omega - 1) \cdot W + W = \omega \cdot W$.
Since $\mathit{opened}_i$ is modified only upon entering a new window, the count remains $\omega \cdot W$ throughout the time $\proc_i$ is in window $\omega$. \qed
\end{proof}

We are now ready to prove that \conductor satisfies $(2W - p)$-boundedness (see \Cref{mod:orchestrator_2}).
As with the previous two properties, \conductor satisfies the boundedness property unconditionally.

\begin{lemma}[Boundedness]
\label{lem:boundedness}
\conductor (\Cref{algorithm:conductor}) satisfies $(2W - p)$-boundedness.
\end{lemma}
\begin{proof}
Fix a correct validator $\proc_i$ and a time $t$.
Let $s_1 < s_2 < \cdots < s_k$ be the slots in $\mathit{opened}_i$ at time $t$, ordered by slot number, and let $j^{\star}$ be the largest index such that $\proc_i$ has completed each of $s_1, \ldots, s_{j^{\star}}$ by time $t$ (with $j^{\star} = 0$ if no such index exists).
It suffices to show that $k - j^{\star} \leq 2W - p$, i.e., that at most $2W - p$ of the slots $s_1, \ldots, s_k$ remain uncompleted by $\proc_i$ at time $t$ (since every slot $\proc_i$ has opened has been previously recorded in $\mathit{opened}_i$).

Let $\omega = \mathsf{window}_i(s_k)$ be the window of slot $s_k$.
We know that $\proc_i$ is in window $\omega$ at time $t$.
Let us now distinguish two cases:
\begin{compactitem}
    \item Suppose $\omega = 1$.
    Then, $s_k.\fnumber \leq W$, so $k \leq W$, and the bound $k - j^{\star} \leq k \leq W \leq 2W - p$ holds immediately.

    \item Suppose $\omega \geq 2$.
    Validator $\proc_i$ enters window $\omega$ upon activating the rule at line~\ref{line:acs-decide} for $\mathcal{ACS}[\omega]$, which fires only when $\mathsf{ready\_for\_next\_window}()$ returns $\mathsf{true}$.
    By the readiness condition (line~\ref{line:ready-check}), at the moment $\proc_i$ enters window $\omega$, at most the last $W - p$ slots in $\mathit{opened}_i$ (by slot number) remain uncompleted.
    Just before entering window $\omega$, validator $\proc_i$ was in window $\omega - 1$, so $|\mathit{opened}_i| = (\omega - 1) W$ by \Cref{prop:open-count-window}.
    Upon entering window $\omega$, validator $\proc_i$ adds the $W$ slots of window $\omega$ to $\mathit{opened}_i$ (lines~\ref{line:open-foreach}--\ref{line:acs-opened-update}).
    Hence, while $\proc_i$ is in window $\omega$, at most $(W - p) + W = 2W - p$ of the slots in $\mathit{opened}_i$ remain uncompleted, and therefore $k - j^{\star} \leq 2W - p$, establishing the lemma in this case. \qed
\end{compactitem}
\end{proof}

\paragraph{Totality.}
The previous three properties were established for \conductor in isolation: each constrains only how \conductor schedules and opens slots, and therefore holds regardless of how --- or even whether --- slots are subsequently completed.
The remaining two, totality and recovery, are of a different nature.
Both hinge on slots actually completing, and on doing so in a timely, coordinated fashion --- something \conductor only reacts to, but does not itself bring about.
We therefore establish totality and recovery only for \conductor run \emph{within \cadence}, alongside \chorus (\Cref{algorithm:cadence}), whose finalization supplies precisely the completion guarantees that these two properties rely on.

We begin with totality.
In fact, we establish a stronger, quantitative form of it --- \emph{$\dtot$-totality}: if a correct validator opens a slot $s$ at some time $t$, then every correct validator opens $s$ by time $\max(t, \mathrm{GST}) + \dtot$.
Here, $\dtot = \Delta$ is \chorus's totality latency (by \Cref{prop:chorus-totality}), which enters the argument because completions within \cadence are \chorus finalizations.
This is more than the orchestrator's totality property demands --- that property asks only that an opening \emph{eventually} reach every correct validator --- but we prove the sharper $\dtot$ bound because the recovery proof later relies on it, not merely on eventual agreement about which slots are opened.

Before the formal argument, we sketch its shape.
The difficulty is a feedback loop between opening and completing, tightened by the fact that \chorus's totality is \emph{conditional}: it guarantees synchronized finalizations for a slot only if participation in it is already $\Delta$-synchronized (\Cref{prop:chorus-totality}) --- and, within \cadence, starting to participate in a slot is exactly opening it.
A correct validator opens the slots of window $\omega$ only once enough of the earlier windows have completed, so the openings of window $\omega$ are synchronized only if the completions of the earlier windows are; and completions, by the above, are synchronized only for slots whose openings already were.
The ACS instances are conditional in the same way: their timing guarantees hold only if the proposals of correct validators are $\Delta$-synchronized (\Cref{mod:acs}) --- and those proposals are triggered by completions, closing a second loop of the same shape.
No window can thus be treated in isolation: the synchrony that \chorus's totality and the ACS's guarantees require at window $\omega$ is produced by the very guarantees we are proving for windows $1, \ldots, \omega - 1$.
We untangle this by a single strong induction over the windows, whose invariant bundles four synchronization guarantees per window --- for its entry, its openings, its completions, and its proposals to the next window's agreement instance.
The induction closes precisely because \chorus's totality latency does not exceed the synchronization tolerance its condition grants (both equal $\Delta = \dtot$): each window's openings are synchronized as tightly as the earlier windows' completions, and its completions as tightly as its own openings, so the invariant ratchets across the windows undegraded.

Let $\mathcal{W}$ denote the set of windows entered by at least one correct validator.
By \Cref{prop:acs-fate-range} and the agreement property of the underlying ACS instances, any two correct validators that both enter a window $\omega$ agree on its slot set, i.e., $\mathsf{slots}_i(\omega) = \mathsf{slots}_j(\omega)$.
We may therefore drop the subscript and write $\mathsf{slots}(\omega)$ for this common set throughout the rest of the proof.

We name the per-window invariant explicitly.

\begin{definition}[Synchronized window]
\label{def:window-synchronized}
A window $\omega \in \mathcal{W}$ is \emph{synchronized} if and only if the following four conditions hold:
\begin{compactenum}
    \item \emph{Entry:} if a correct validator enters $\omega$ at some time $t$, then every correct validator enters $\omega$ by time $\max(t, \mathrm{GST}) + \dtot$;
    \item \emph{Openings:} for every slot $s \in \mathsf{slots}(\omega)$, if a correct validator opens $s$ at some time $t$, then every correct validator opens $s$ by time $\max(t, \mathrm{GST}) + \dtot$;
    \item \emph{Completions:} for every slot $s \in \mathsf{slots}(\omega)$, if a correct validator completes $s$ at some time $t$, then every correct validator completes $s$ by time $\max(t, \mathrm{GST}) + \dtot$;
    \item \emph{Proposals:} if a correct validator proposes to $\mathcal{ACS}[\omega + 1]$ at some time $t$, then every correct validator proposes to $\mathcal{ACS}[\omega + 1]$ by time $\max(t, \mathrm{GST}) + \dtot$.
\end{compactenum}
\end{definition}

Recall that the timing guarantees of the ACS module (\Cref{mod:acs}) hold only under its two assumptions.
We first record that the no-premature-abandonment assumption is always satisfied.

\begin{proposition}[No premature abandonment]
\label{prop:acs-no-premature-abandonment}
No correct validator abandons an ACS instance before deciding from it.
\end{proposition}
\begin{proof}
A correct validator abandons an ACS instance only at line~\ref{line:acs-abandon} of \Cref{algorithm:conductor}, within the rule at line~\ref{line:acs-decide}, which activates only once the validator has decided from that very instance. \qed
\end{proof}

The timing guarantees of the ACS instances thus hinge only on their remaining assumption, the $\Delta$-synchronization of proposals.
The following proposition is the heart of the totality argument: it establishes that every window is synchronized.

\begin{proposition}[Window synchronization]
\label{prop:window-synchronization}
When \conductor is run within \cadence, every window $\omega \in \mathcal{W}$ is synchronized.
\end{proposition}
\begin{proof}
We proceed by strong induction on $\omega$.

\smallskip
\noindent \emph{Base case} ($\omega = 1$).
Every correct validator enters window $1$ at time $0$, upon starting the protocol (line~\ref{line:enter_window_1}), and thereupon schedules the opening of each slot $s \in \mathsf{slots}(1)$ (line~\ref{line:startup-open}); hence all correct validators open $s$ at the same time $s.\fdeadline - \Delta$ (line~\ref{line:conductor-wait-for-open}).
The entry and opening conditions thus hold with perfect synchrony.

For the completion condition, consider a slot $s \in \mathsf{slots}(1)$ and a correct validator $\proc_i$ that completes $s$ at some time $t$.
Within \cadence, a validator records a completion only upon finalizing the corresponding slot consensus instance, and the handler that does so fires only for slots it has already opened (line~\ref{line:upon-finalize} of \Cref{algorithm:cadence}); hence $\proc_i$ finalizes $\mathcal{S}[s]$ at time $t$.
As all correct validators open $s$ --- and thus start participating in $\mathcal{S}[s]$ (line~\ref{line:participate} of \Cref{algorithm:cadence}) --- at the same time, participation in $\mathcal{S}[s]$ is trivially $\Delta$-synchronized, so the totality of \chorus (\Cref{prop:chorus-totality}) applies: every correct validator finalizes $\mathcal{S}[s]$ by time $\max(t, \mathrm{GST}) + \dtot$.
Since each has by then already opened $s$, the guard on \cadence's finalization handler (line~\ref{line:upon-finalize} of \Cref{algorithm:cadence}) is satisfied, so each completes $s$ by time $\max(t, \mathrm{GST}) + \dtot$.

Lastly, for the proposal condition, let $\proc_i$ be a correct validator that proposes to $\mathcal{ACS}[2]$ at some time $t$ (line~\ref{line:acs-propose}), write $T^{\star} = \max(t, \mathrm{GST}) + \dtot$, and consider any correct validator $\proc_j$.
When proposing, $\mathsf{ready\_for\_next\_window}()$ returns $\mathsf{true}$ at $\proc_i$ (line~\ref{line:ready}): of the $W$ slots in $\mathit{opened}_i$, $\proc_i$ completed the first $p$ by time $t$ (line~\ref{line:ready-check}).
By the completion condition just established, $\proc_j$ completes those $p$ slots by time $T^{\star}$; as $\mathit{opened}_j$ consists of the same $W$ slots while $\proc_j$ is in window $1$, and as $\proc_j$ can enter window $2$ only at a moment when readiness already holds (line~\ref{line:acs-decide}), there is a first time $t_r \leq T^{\star}$ at which $\proc_j$ is in window $1$ and $\mathsf{ready\_for\_next\_window}()$ returns $\mathsf{true}$ at $\proc_j$.
Since $\proc_j$ proposes to $\mathcal{ACS}[2]$ only under those same two conditions (line~\ref{line:ready}), it has not proposed before time $t_r$; hence $2 \notin \mathit{proposed}_j$ at time $t_r$, so the rule at line~\ref{line:ready} fires and $\proc_j$ proposes to $\mathcal{ACS}[2]$ at time $t_r \leq T^{\star}$ (line~\ref{line:acs-propose}).

\smallskip
\noindent \emph{Inductive step} ($\omega > 1$).
We prove that $\omega$ is synchronized assuming that every window $\omega' < \omega$ is, establishing the four conditions of \Cref{def:window-synchronized} in order; each rests on the previous ones.

\smallskip
\noindent \emph{Entry.}
Let $\proc_i$ be a correct validator that enters $\omega$ at some time $t$, and write $T^{\star} = \max(t, \mathrm{GST}) + \dtot$.
Consider any correct validator $\proc_j$.
A validator enters $\omega$ upon activating the rule at line~\ref{line:acs-decide} for $\mathcal{ACS}[\omega]$, which requires that it is in window $\omega - 1$, has decided from $\mathcal{ACS}[\omega]$, and $\mathsf{ready\_for\_next\_window}()$ returns $\mathsf{true}$.
We first collect three facts.
First, $\proc_i$ entered window $\omega - 1$ at some time $t^- \leq t$ (\Cref{lemma:window-entry}), so, by the entry condition of the induction hypothesis for $\omega - 1$, $\proc_j$ enters $\omega - 1$ by time $\max(t^-, \mathrm{GST}) + \dtot \leq T^{\star}$.
Second, $\proc_i$ decided from $\mathcal{ACS}[\omega]$ by time $t$, so, by the $\Delta$-totality of $\mathcal{ACS}[\omega]$ (\Cref{mod:acs}) --- whose $\Delta$-synchronized-proposals assumption is exactly the proposal condition of the induction hypothesis for window $\omega - 1$, and whose no-premature-abandonment assumption holds by \Cref{prop:acs-no-premature-abandonment} --- $\proc_j$ decides from $\mathcal{ACS}[\omega]$ by time $\max(t, \mathrm{GST}) + \Delta \leq T^{\star}$.
Third, when $\proc_i$ entered $\omega$, the readiness condition held at $\proc_i$ (line~\ref{line:ready-check}): of the $(\omega - 1)W$ slots in $\mathit{opened}_i$ (\Cref{prop:open-count-window}), all but at most the last $W - p$ --- that is, the first $(\omega - 2)W + p$, ordered by slot number --- were completed by $\proc_i$ by time $t$.
These slots lie in windows $1, \ldots, \omega - 1$ (\Cref{prop:acs-fate-range}), so, by the completion conditions of the induction hypothesis, $\proc_j$ completes each of them by time $\max(t, \mathrm{GST}) + \dtot = T^{\star}$; since $\mathit{opened}_j$ consists of exactly those $(\omega - 1)W$ slots while $\proc_j$ is in window $\omega - 1$ (\Cref{prop:acs-fate-range}), $\mathsf{ready\_for\_next\_window}()$ then returns $\mathsf{true}$ at $\proc_j$.
Now, as $\proc_j$ leaves window $\omega - 1$ only by entering $\omega$ --- correct validators enter windows in increasing order (\Cref{lemma:window-entry}) --- and as entering $\omega$ requires precisely the three requirements above, the three facts guarantee a first time $t_e \leq T^{\star}$ at which all three requirements hold simultaneously at $\proc_j$.
At that moment, the rule at line~\ref{line:acs-decide} for $\mathcal{ACS}[\omega]$ fires, and $\proc_j$ enters $\omega$ at time $t_e \leq T^{\star}$ (line~\ref{line:enter_window_omega}).

\smallskip
\noindent \emph{Openings.}
Let $s \in \mathsf{slots}(\omega)$, and let $\proc_i$ be a correct validator that opens $s$ at some time $t$; note that $t \geq s.\fdeadline - \Delta$ (line~\ref{line:conductor-wait-for-open}).
Validator $\proc_i$ scheduled the opening of $s$ upon entering $\omega$, at some time $t' \leq t$.
By the entry condition just established, every correct validator $\proc_j$ enters $\omega$ --- and thereby schedules the opening of $s$ (line~\ref{line:acs-decide-open}) --- at some time $t^{\star} \leq \max(t', \mathrm{GST}) + \dtot \leq \max(t, \mathrm{GST}) + \dtot$.
If $t^{\star} \geq s.\fdeadline - \Delta$, then $\proc_j$ opens $s$ at time $t^{\star} \leq \max(t, \mathrm{GST}) + \dtot$; otherwise, $\proc_j$ opens $s$ at time $s.\fdeadline - \Delta \leq t$ (line~\ref{line:conductor-wait-for-open}).
Either way, $\proc_j$ opens $s$ by time $\max(t, \mathrm{GST}) + \dtot$.

\smallskip
\noindent \emph{Completions.}
Let $s \in \mathsf{slots}(\omega)$, and let $\proc_i$ be a correct validator that completes $s$ at some time $t$.
Within \cadence, a validator records a completion only upon finalizing the corresponding slot consensus instance, and the handler that does so fires only for slots it has already opened (line~\ref{line:upon-finalize} of \Cref{algorithm:cadence}); hence $\proc_i$ finalizes $\mathcal{S}[s]$ at time $t$, having opened $s$ at some time $t_o \leq t$.
By the opening condition just established, every correct validator opens $s$ --- and thus starts participating in $\mathcal{S}[s]$ (line~\ref{line:participate} of \Cref{algorithm:cadence}) --- by time $\max(t_o, \mathrm{GST}) + \dtot \leq \max(t, \mathrm{GST}) + \dtot$.
As $\dtot = \Delta$, this says precisely that participation in $\mathcal{S}[s]$ is $\Delta$-synchronized, so the totality of \chorus (\Cref{prop:chorus-totality}) applies: every correct validator finalizes $\mathcal{S}[s]$ by time $\max(t, \mathrm{GST}) + \dtot$.
Since each has by then already opened $s$, the guard on \cadence's finalization handler (line~\ref{line:upon-finalize} of \Cref{algorithm:cadence}) is satisfied, so each completes $s$ by time $\max(t, \mathrm{GST}) + \dtot$.

\smallskip
\noindent \emph{Proposals.}
Let $\proc_i$ be a correct validator that proposes to $\mathcal{ACS}[\omega + 1]$ at some time $t$ (line~\ref{line:acs-propose}), write $T^{\star} = \max(t, \mathrm{GST}) + \dtot$, and consider any correct validator $\proc_j$.
When proposing, $\proc_i$ is in window $\omega$ and $\mathsf{ready\_for\_next\_window}()$ returns $\mathsf{true}$ (line~\ref{line:ready}): $\proc_i$ entered $\omega$ at some time $t^- \leq t$, and, of the $\omega W$ slots in $\mathit{opened}_i$ (\Cref{prop:open-count-window}), it completed all but at most the last $W - p$ --- that is, the first $(\omega - 1)W + p$, ordered by slot number --- by time $t$.
By the entry condition established above, $\proc_j$ enters $\omega$ by time $\max(t^-, \mathrm{GST}) + \dtot \leq T^{\star}$.
The completed slots lie in windows $1, \ldots, \omega$ (\Cref{prop:acs-fate-range}), so, by the completion conditions of the induction hypothesis for the windows $\omega' < \omega$ and the completion condition just established for $\omega$ itself, $\proc_j$ completes each of them by time $T^{\star}$; since $\mathit{opened}_j$ consists of the same $\omega W$ slots while $\proc_j$ is in window $\omega$ (\Cref{prop:acs-fate-range}), $\mathsf{ready\_for\_next\_window}()$ then returns $\mathsf{true}$ at $\proc_j$.
As $\proc_j$ can enter window $\omega + 1$ only at a moment when readiness already holds (line~\ref{line:acs-decide}), there is thus a first time $t_r \leq T^{\star}$ at which $\proc_j$ is in window $\omega$ and $\mathsf{ready\_for\_next\_window}()$ returns $\mathsf{true}$ at $\proc_j$.
Since $\proc_j$ proposes to $\mathcal{ACS}[\omega + 1]$ only under those same two conditions (line~\ref{line:ready}), it has not proposed before time $t_r$; hence $\omega + 1 \notin \mathit{proposed}_j$ at time $t_r$, so the rule at line~\ref{line:ready} fires and $\proc_j$ proposes to $\mathcal{ACS}[\omega + 1]$ at time $t_r \leq T^{\star}$ (line~\ref{line:acs-propose}). \qed
\end{proof}

The remainder of the analysis consumes \Cref{prop:window-synchronization} through four immediate consequences, one per condition.
We record the entry condition first, as the recovery proof relies on it directly.

\begin{corollary}[Entry synchronization]
\label{cor:entry-synchronization}
When \conductor is run within \cadence, the following holds for every window $\omega \in \mathcal{W}$: if a correct validator enters $\omega$ at some time $t$, then every correct validator enters $\omega$ by time $\max(t, \mathrm{GST}) + \dtot$.
\end{corollary}
\begin{proof}
Immediate from the entry condition of \Cref{prop:window-synchronization}. \qed
\end{proof}

The proposal condition, in turn, discharges the assumption of the ACS module (\Cref{mod:acs}) once and for all: within \cadence, the proposals of correct validators to every ACS instance are $\Delta$-synchronized.

\begin{corollary}[Proposal synchronization]
\label{cor:proposal-synchronization}
When \conductor is run within \cadence, the following holds for every window $\omega \geq 2$: if a correct validator proposes to $\mathcal{ACS}[\omega]$ at some time $t$, then every correct validator proposes to $\mathcal{ACS}[\omega]$ by time $\max(t, \mathrm{GST}) + \dtot$.
\end{corollary}
\begin{proof}
A correct validator proposes to $\mathcal{ACS}[\omega]$ only from window $\omega - 1$ (line~\ref{line:ready}), so $\omega - 1 \in \mathcal{W}$, and the claim is the proposal condition of \Cref{prop:window-synchronization} applied to window $\omega - 1$. \qed
\end{proof}

We are now ready to prove that \conductor satisfies $\dtot$-totality when run within \cadence.

\begin{lemma} [Totality]
\label{lemma:conductor-totality}
When run within \cadence, \conductor (\Cref{algorithm:conductor}) satisfies totality.
More specifically, for every slot $s$, if a correct validator opens $s$ at some time $t$, then every correct validator opens $s$ by time $\max(t, \mathrm{GST}) + \dtot$.
\end{lemma}
\begin{proof}
Consider any slot $s$ that a correct validator $\proc_i$ opens; then $s \in \mathsf{slots}(\omega)$ for the window $\omega = \mathsf{window}_i(s) \in \mathcal{W}$.
By \Cref{prop:window-synchronization}, $\omega$ is synchronized, and its opening condition (\Cref{def:window-synchronized}) is precisely the claim. \qed

\end{proof}

The completions inherit the same bound.

\begin{corollary}[Completion synchronization]
\label{cor:completion-totality}
When run within \cadence, \conductor satisfies $\dtot$-totality of completions: for every slot $s$, if a correct validator completes $s$ at some time $t$, then every correct validator completes $s$ by time $\max(t, \mathrm{GST}) + \dtot$.
\end{corollary}
\begin{proof}
A correct validator completes a slot only after opening it (the finalization handler fires only for opened slots; line~\ref{line:upon-finalize} of \Cref{algorithm:cadence}), so a completed slot $s$ satisfies $s \in \mathsf{slots}(\omega)$ for some window $\omega \in \mathcal{W}$.
By \Cref{prop:window-synchronization}, $\omega$ is synchronized, and its completion condition (\Cref{def:window-synchronized}) is precisely the claim. \qed
\end{proof}

\paragraph{Recovery.}
Finally, we turn to \conductor's recovery property, which we also show \conductor satisfies when run within \cadence.
To this end, we first bound how long a correct validator takes to complete a slot it opens, in terms of \chorus's termination latency $\ell_{\textsc{chorus}}$ --- the parameter for which, if all correct validators start participating by some time $t'$, then all finalize by time $\max(t', \mathrm{GST}) + \ell_{\textsc{chorus}}$.

\begin{proposition}
\label{prop:conductor-open-to-complete}
When run within \cadence, the following holds for every slot $s$: if a correct validator opens $s$ at some time $t$, then it completes $s$ by time $\max(t, \mathrm{GST}) + \dtot + \ell_{\textsc{chorus}}$.
\end{proposition}
\begin{proof}
Let $\proc_i$ be a correct validator that opens slot $s$ at time $t$; upon doing so, it starts participating in $\mathcal{S}[s]$ (line~\ref{line:participate} of \Cref{algorithm:cadence}).
By the $\dtot$-totality of \conductor (\Cref{lemma:conductor-totality}), every correct validator opens $s$ --- and hence starts participating in $\mathcal{S}[s]$ --- by time $\max(t, \mathrm{GST}) + \dtot$.
The $\ell_{\textsc{chorus}}$-termination of \chorus (\Cref{lemma:chorus-termination}) --- whose $\Delta$-synchronized-participation condition holds by the $\dtot$-totality of \conductor (\Cref{lemma:conductor-totality}), as opening a slot is exactly starting to participate in its instance and $\dtot = \Delta$ --- therefore guarantees that every correct validator finalizes $\mathcal{S}[s]$ --- and thus completes $s$ (line~\ref{line:complete} of \Cref{algorithm:cadence}) --- by time $\max(t, \mathrm{GST}) + \dtot + \ell_{\textsc{chorus}}$. \qed
\end{proof}

Recall that $\Phi_{oc} = \ell_{\textsc{chorus}} + \dtot$, which is exactly the open-to-complete bound established just above (\Cref{prop:conductor-open-to-complete}).
Throughout the remainder of this part of the proof, we focus exclusively on \conductor when run within \cadence.
To avoid clutter, we leave this qualification implicit rather than restating it at every step.
We now prove that every correct validator eventually enters every window.

\begin{proposition}
\label{prop:enters-every-window}
Every correct validator eventually enters every window $\omega \in \mathbb{N}_{\geq 1}$.
\end{proposition}
\begin{proof}
We prove the proposition by induction on the window number $\omega$.

\smallskip
\noindent \emph{Base case} ($\omega = 1$).
This holds trivially: every correct validator enters window $1$ upon starting the protocol (line~\ref{line:enter_window_1}).

\smallskip
\noindent \emph{Inductive step} ($\omega - 1 \to \omega$).
By the inductive hypothesis, every correct validator eventually enters window $\omega - 1$.
We first show that every correct validator eventually decides from $\mathcal{ACS}[\omega]$, distinguishing two cases:
\begin{compactitem}
    \item \emph{Some correct validator decides from $\mathcal{ACS}[\omega]$.}
    By the totality property of $\mathcal{ACS}[\omega]$ (\Cref{mod:acs}; its assumptions hold by \Cref{cor:proposal-synchronization} and \Cref{prop:acs-no-premature-abandonment}), every correct validator then decides from $\mathcal{ACS}[\omega]$.

    \item \emph{No correct validator decides from $\mathcal{ACS}[\omega]$.}
    Then no correct validator ever enters a window greater than $\omega - 1$: entering window $\omega$ requires activating the rule at line~\ref{line:acs-decide} for $\mathcal{ACS}[\omega]$, which cannot fire if no correct validator decides from $\mathcal{ACS}[\omega]$; and, by \Cref{lemma:window-entry}, entering any larger window requires first entering $\omega$.
    Fix any correct validator $\proc_j$.
    By the inductive hypothesis it enters window $\omega - 1$, and as it never advances beyond it, $\mathit{current\_window}_j = \omega - 1$ from then on; by \Cref{prop:open-count-window}, $\mathit{opened}_j$ then contains exactly $(\omega - 1) W$ slots.
    Each such slot was opened by $\proc_j$ (line~\ref{line:startup-open} or line~\ref{line:acs-decide-open}), so by \Cref{prop:conductor-open-to-complete}, $\proc_j$ eventually completes it.
    In particular, $\proc_j$ eventually completes the first $(\omega - 2)W + p$ slots of $\mathit{opened}_j$ (ordered by slot number), at which point $\mathsf{ready\_for\_next\_window}()$ returns $\mathsf{true}$ (line~\ref{line:ready-check}).
    Once it does, $\proc_j$ proposes to $\mathcal{ACS}[\omega]$ (lines~\ref{line:ready}--\ref{line:acs-propose}), if it has not already.
    Hence every correct validator proposes to $\mathcal{ACS}[\omega]$, so by the termination property of $\mathcal{ACS}[\omega]$ (\Cref{mod:acs}; its assumptions hold by \Cref{cor:proposal-synchronization} and \Cref{prop:acs-no-premature-abandonment}) every correct validator decides from it.
\end{compactitem}
In either case, every correct validator eventually decides from $\mathcal{ACS}[\omega]$.

It remains to show that every correct validator $\proc_i$ eventually enters window $\omega$.
While $\proc_i$ is in window $\omega - 1$ (i.e., $\mathit{current\_window}_i = \omega - 1$), it enters window $\omega$ upon activating the rule at line~\ref{line:acs-decide} for $\mathcal{ACS}[\omega]$ (line~\ref{line:enter_window_omega}), which fires once its two conditions hold simultaneously:
(1) $\proc_i$ decides from $\mathcal{ACS}[\omega]$, which it eventually does, as just shown; and
(2) $\mathsf{ready\_for\_next\_window}()$ returns $\mathsf{true}$ (line~\ref{line:ready-check}), which holds eventually --- by the argument of the second case above, via \Cref{prop:conductor-open-to-complete} --- and remains so, since $\mathit{opened}_i$ does not change while $\proc_i$ is in window $\omega - 1$ and completed slots only accumulate.
Therefore, both conditions eventually hold at $\proc_i$, so $\proc_i$ activates the rule at line~\ref{line:acs-decide} for $\mathcal{ACS}[\omega]$ and enters window $\omega$ (line~\ref{line:enter_window_omega}). \qed
\end{proof}

Having established that every correct validator enters every window, we now turn to timing: we show that, once the network stabilizes, correct validators enter each window at a precise time.
These timing guarantees rely on the protocol parameters being configured to satisfy the following constraints (lines~\ref{line:assumption-one}--\ref{line:assumption-four}):
\begin{compactitem}
    \item $(p - 1)\tau + \Phi_{oc} + \ell \leq W\tau$ (line~\ref{line:assumption-one});
    \item $(p - 1)\tau + \Phi_{oc} \leq (W - 1)\tau$ (line~\ref{line:assumption-two});
    \item $\Delta < \ell$ (line~\ref{line:assumption-three}); and
    \item $\dtot + \ell \leq (p - 1)\tau$ (line~\ref{line:assumption-four}).
\end{compactitem}
Recall that $p \in \{0, \dots, W-1\}$ is the readiness threshold, $W \in \mathbb{N}_{\geq 1}$ the window size, $\ell$ the latency of the ACS primitive, $\Phi_{oc} = \ell_{\textsc{chorus}} + \dtot$ the open-to-complete delay, and $\Delta$ the known upper bound on message delays after $\mathrm{GST}$.
Intuitively, these constraints ensure that all the work associated with a window, namely proposing to the next ACS instance, deciding, and completing its slots, fits within the window's time span of $W\tau$, so that one window follows the next without gaps once the network is synchronous.

Before proving the timing guarantee, we set up notation for the slots and starting times associated with each window.
Note that, by \Cref{prop:enters-every-window}, every correct validator enters every window $\omega \in \mathbb{N}_{\geq 1}$, so $\mathsf{slots}_i(\omega)$ is defined for every correct validator $\proc_i$ and every window $\omega$.
By \Cref{prop:acs-fate-range}, for every window $\omega$, the common set $\mathsf{slots}(\omega)$ consists of exactly $W$ slots with consecutive slot numbers.
For every window $\omega$ and index $x \in [1, W]$, we write $\mathsf{slot}(\omega, x)$ for the $x$-th slot of $\mathsf{slots}(\omega)$, ordered by slot number.
Then, for every window $\omega$ and index $x \in [1, W]$, we denote by $\mathcal{T}_x(\omega)$ the starting time of $\mathsf{slot}(\omega, x)$, that is,
\[
    \mathcal{T}_x(\omega) \;=\; s.\fdeadline - \Delta, \quad \text{where } s = \mathsf{slot}(\omega, x).
\]
Since consecutive slots are separated by exactly $\tau$ time units, $\mathcal{T}_x(\omega) = \mathcal{T}_1(\omega) + (x - 1)\tau$, for every $x \in [1, W]$.
Finally, we call a window $\omega$ a \emph{post-$\mathrm{GST}$ window} if and only if $\mathcal{T}_1(\omega) \geq \mathrm{GST}$, and a \emph{pre-$\mathrm{GST}$ window} otherwise.

We now show that every correct validator enters each window within $\dtot + \ell$ of its scheduled start (or of $\mathrm{GST}$, if that is later).

\begin{proposition}
\label{prop:window-open-time}
Let $\omega$ be any window.
Then, every correct validator enters window $\omega$ by time
\[
\max(\mathcal{T}_1(\omega), \mathrm{GST}) + \dtot + \ell.
\]
\end{proposition}
\begin{proof}
If $\omega = 1$, every correct validator enters window $1$ at time $0 = \mathcal{T}_1(1)$, so the proposition holds trivially.
Assume henceforth that $\omega > 1$.

Suppose first that some correct validator enters window $\omega$ by time $\max(\mathcal{T}_1(\omega), \mathrm{GST}) + \Delta$.
By \Cref{cor:entry-synchronization}, the entry of one correct validator into window $\omega$ propagates to every correct validator within $\dtot$ time.
Hence, every correct validator enters window $\omega$ by time $\max(\mathcal{T}_1(\omega), \mathrm{GST}) + \Delta + \dtot$, which is earlier than $\max(\mathcal{T}_1(\omega), \mathrm{GST}) + \dtot + \ell$ (as $\Delta < \ell$, line~\ref{line:assumption-three}); the claim follows.
We may therefore assume that no correct validator enters window $\omega$ by time $\max(\mathcal{T}_1(\omega), \mathrm{GST}) + \Delta$; by \Cref{lemma:window-entry}, no correct validator enters any window $\geq \omega$ by that time.

Due to \Cref{prop:enters-every-window}, every correct validator enters window $\omega$, and hence decides from $\mathcal{ACS}[\omega]$ (line~\ref{line:acs-decide}); by \Cref{prop:acs-fate-range}, the decided slot is $\mathsf{slot}(\omega, 1)$, whose starting time is $\mathcal{T}_1(\omega)$.
By the validity property of $\mathcal{ACS}[\omega]$, at least $f + 1$ correct validators proposed to $\mathcal{ACS}[\omega]$, with median proposal $\mathsf{slot}(\omega, 1)$; in particular, some correct validator $\proc_k$ proposed a slot $s'$ with $s'.\fnumber \leq \mathsf{slot}(\omega, 1).\fnumber$ by time $\mathcal{T}_1(\omega)$ (line~\ref{line:acs-propose}).

A correct validator proposes to $\mathcal{ACS}[\omega]$ only when $\mathsf{ready\_for\_next\_window}()$ returns $\mathsf{true}$ (line~\ref{line:ready}).
Since $\mathit{current\_window}_k = \omega - 1$ when $\proc_k$ proposes, $\mathit{opened}_k$ contains exactly $(\omega - 1)W$ slots (\Cref{prop:open-count-window}), so $\proc_k$ has completed the first $(\omega - 2)W + p$ of them by time $\mathcal{T}_1(\omega)$.
Due to \Cref{cor:completion-totality}, every correct validator completes each of these $(\omega - 2)W + p$ slots by time $\max(\mathcal{T}_1(\omega), \mathrm{GST}) + \dtot$.
Moreover, every correct validator has entered window $\omega - 1$ by that time: $\proc_k$ proposes to $\mathcal{ACS}[\omega]$ only from window $\omega - 1$ (line~\ref{line:ready}), so it entered $\omega - 1$ by time $\mathcal{T}_1(\omega)$, and \Cref{cor:entry-synchronization} carries every correct validator into $\omega - 1$ by $\max(\mathcal{T}_1(\omega), \mathrm{GST}) + \dtot$.
Therefore, $\mathsf{ready\_for\_next\_window}()$ returns $\mathsf{true}$ at every correct validator, and each proposes to $\mathcal{ACS}[\omega]$ by time $\max(\mathcal{T}_1(\omega), \mathrm{GST}) + \dtot$.
By the termination property of $\mathcal{ACS}[\omega]$ (its assumptions holding by \Cref{cor:proposal-synchronization} and \Cref{prop:acs-no-premature-abandonment}), every correct validator then decides from $\mathcal{ACS}[\omega]$ by time $\max(\mathcal{T}_1(\omega), \mathrm{GST}) + \dtot + \ell$, which proves the proposition. \qed

\end{proof}

We now establish the key timing invariant underlying recovery: if every correct validator enters a post-$\mathrm{GST}$ window $\omega$ by $\mathcal{T}_p(\omega)$, then every correct validator opens window $\omega + 1$ ``on time'' and there are no ``gaps'' between the two windows.

\begin{proposition}
\label{prop:window-progression}
Fix any post-$\mathrm{GST}$ window $\omega$.
If every correct validator enters window $\omega$ by time $\mathcal{T}_p(\omega)$, then:
\begin{compactitem}
    \item $\mathsf{slot}(\omega + 1, 1).\mathsf{number} = \mathsf{slot}(\omega, W).\mathsf{number} + 1$; and
    \item every correct validator enters window $\omega + 1$ by time $\mathcal{T}_{1}(\omega + 1)$.
\end{compactitem}
\end{proposition}
\begin{proof}
We prove each point in turn.

\smallskip
\noindent \emph{Point 1.}
We establish the first point by showing that every correct validator that proposes to $\mathcal{ACS}[\omega + 1]$ does so with a slot $s^{\star}$ with $s^{\star}.\mathsf{number} = \mathsf{slot}(\omega, W).\mathsf{number} + 1$.
Suppose, for contradiction, that some correct validator $p_i$ proposes a slot $s \neq s^{\star}$.
Note that validator $p_i$ opens each of the first $p$ slots of window $\omega$ by time $\mathcal{T}_p(\omega)$ (at line~\ref{line:trigger-open}, as each slot is opened either when scheduled or at the slot's starting time).
For $p_i$ to propose to $\mathcal{ACS}[\omega + 1]$ at line~\ref{line:acs-propose}, $\mathsf{ready\_for\_next\_window()}$ must return $\mathsf{true}$ (line~\ref{line:ready}).
By \Cref{prop:conductor-open-to-complete}, $p_i$ completes the first $(\omega - 1)W + p$ open slots by time $\mathcal{T}_p(\omega) + \Phi_{oc}$, at which point $\mathsf{ready\_for\_next\_window()}$ returns $\mathsf{true}$ (line~\ref{line:ready-check}), triggering the rule at line~\ref{line:ready}. 
We now verify that the trigger fires before $\mathsf{slot}(\omega, W).\mathsf{deadline} - \Delta$.
By definition, $\mathcal{T}_p(\omega) = \mathcal{T}_1(\omega) + (p - 1)\tau$ and $\mathsf{slot}(\omega, W).\mathsf{deadline} - \Delta = \mathcal{T}_1(\omega) + (W - 1)\tau$.
Hence, $\mathcal{T}_p(\omega) + \Phi_{oc} = \mathcal{T}_1(\omega) + (p - 1)\tau + \Phi_{oc}$.
By assumption at line~\ref{line:assumption-two}, $(p - 1)\tau + \Phi_{oc} \leq (W - 1)\tau$, so:
\[
\mathcal{T}_p(\omega) + \Phi_{oc}
= \mathcal{T}_1(\omega) + (p-1)\tau + \Phi_{oc}
\leq \mathcal{T}_1(\omega) + (W-1)\tau
= \mathcal{T}_W(\omega)
= \mathsf{slot}(\omega, W).\mathsf{deadline} - \Delta.
\]
Consequently, $p_i$ selects $s^{\star}$ with $s^{\star}.\mathsf{number} = \mathsf{slot}(\omega, W).\mathsf{number} + 1$ (lines~\ref{line:sstar-compute}--\ref{line:sstar-update}) and invokes $\mathcal{ACS}[\omega + 1].\mathsf{propose}(s^{\star})$ (line~\ref{line:acs-propose}), contradicting our assumption.

To conclude the first point, by the validity property of $\mathcal{ACS}[\omega + 1]$, at least $f + 1$ correct proposals appear in the decided set.
Hence, the median slot number (line~\ref{line:median-compute}) equals $\mathsf{slot}(\omega, W).\mathsf{number} + 1$, establishing $\mathsf{slot}(\omega + 1, 1).\mathsf{number} = \mathsf{slot}(\omega, W).\mathsf{number} + 1$.

\smallskip
\noindent \emph{Point 2.}
We first show that every correct validator decides from $\mathcal{ACS}[\omega + 1]$ by time $\mathcal{T}_p(\omega) + \Phi_{oc} + \ell$.
Since $\mathcal{T}_p(\omega) = \mathcal{T}_1(\omega) + (p-1)\tau$ and $(p-1)\tau + \Phi_{oc} + \ell \leq W\tau$ (line~\ref{line:assumption-one}), we have $\mathcal{T}_p(\omega) + \Phi_{oc} + \ell \leq \mathcal{T}_1(\omega) + W\tau = \mathcal{T}_1(\omega + 1)$, so this suffices to establish the claim.
(Recall that every correct validator completes the first $(\omega - 1)W + p$ open slots by time $\mathcal{T}_p(\omega) + \Phi_{oc}$; hence $\mathsf{ready\_for\_next\_window}()$ returns $\mathsf{true}$ strictly before time $\mathcal{T}_p(\omega) + \Phi_{oc} + \ell$.)
We distinguish two cases.
\begin{compactitem}
    \item Suppose some correct validator decides from $\mathcal{ACS}[\omega + 1]$ by time $\mathcal{T}_p(\omega) + \Phi_{oc}$.
    By the totality property of $\mathcal{ACS}[\omega + 1]$ (its assumptions holding by \Cref{cor:proposal-synchronization} and \Cref{prop:acs-no-premature-abandonment}), every correct validator then decides by time $\mathcal{T}_p(\omega) + \Phi_{oc} + \Delta < \mathcal{T}_p(\omega) + \Phi_{oc} + \ell$, and the claim holds.

    \item Otherwise, no correct validator decides from $\mathcal{ACS}[\omega + 1]$ by time $\mathcal{T}_p(\omega) + \Phi_{oc}$.
    In particular, no correct validator has activated line~\ref{line:acs-decide} for $\mathcal{ACS}[\omega + 1]$ by that time, so no correct validator has entered any window larger than $\omega$ (line~\ref{line:enter_window_omega}).
    By \Cref{prop:conductor-open-to-complete}, every correct validator completes all $(\omega - 1)W + p$ open slots by time $\mathcal{T}_p(\omega) + \Phi_{oc}$, at which point $\mathsf{ready\_for\_next\_window()}$ returns $\mathsf{true}$ (line~\ref{line:ready-check}), triggering each correct validator to propose to $\mathcal{ACS}[\omega + 1]$ (lines~\ref{line:ready}--\ref{line:acs-propose}).
    Since $\mathcal{ACS}[\omega + 1]$ has latency $\ell$ (its assumptions again holding by \Cref{cor:proposal-synchronization} and \Cref{prop:acs-no-premature-abandonment}), every correct validator decides by time $\mathcal{T}_p(\omega) + \Phi_{oc} + \ell$, and the claim holds. \qed
\end{compactitem}
\end{proof}

Next, we prove that from the second post-$\mathrm{GST}$ window, the protocol ``runs smoothly''.

\begin{proposition}
\label{prop:smooth-windows}
Let $\omega^{\star}$ denote the smallest post-$\mathrm{GST}$ window.
Then, for every window $\omega > \omega^{\star}$, the following holds:
\begin{compactitem}
    \item $\mathsf{slot}(\omega, 1).\mathsf{number} = \mathsf{slot}(\omega - 1, W).\mathsf{number} + 1$; and
    
    \item every correct validator opens every slot $s \in \mathsf{slots}(\omega)$ at time $s.\fdeadline - \Delta$.
\end{compactitem}
\end{proposition}
\begin{proof}
We proceed by induction on $\omega$, with base case $\omega = \omega^{\star} + 1$.
By \Cref{prop:window-open-time}, every correct validator enters window $\omega^{\star}$ by time $\mathcal{T}_1(\omega^{\star}) + \dtot + \ell$.
Since $\dtot + \ell \leq (p - 1)\tau$, every correct validator enters window $\omega^{\star}$ by time $\mathcal{T}_p(\omega^{\star})$.
Therefore, \Cref{prop:window-progression} gives us the following for window $\omega = \omega^{\star} + 1$: (1) $\mathsf{slot}(\omega^{\star} + 1, 1).\fnumber = \mathsf{slot}(\omega^{\star}, W).\fnumber + 1$, and (2) every correct validator enters window $\omega^{\star} + 1$ by time $\mathcal{T}_1(\omega^{\star} + 1)$, which indeed implies that every correct validator opens every slot $s \in \mathsf{slots}(\omega^{\star} + 1)$ at its starting time $s.\fdeadline - \Delta$ (line~\ref{line:conductor-wait-for-open}).
The inductive step is identical: applying \Cref{prop:window-progression} at each subsequent window yields the claim for all $\omega > \omega^{\star} + 1$. \qed
\end{proof}

\Cref{prop:smooth-windows} tells us that the protocol runs smoothly from the second post-$\mathrm{GST}$ window onward: there are no gaps between consecutive windows, and every slot is opened exactly at its starting time.
It remains to understand when this smooth regime begins, i.e., when the second post-$\mathrm{GST}$ window can start.
To this end, we bound when the first post-$\mathrm{GST}$ window can arise.

\begin{proposition}
\label{prop:first-post-gst-window-time}
Let $\omega$ denote the smallest post-$\mathrm{GST}$ window.
Then:
\[
\mathcal{T}_1(\omega) - \mathrm{GST} \leq W\tau.
\]
\end{proposition}
\begin{proof}
If $\omega = 1$, then $\mathcal{T}_1(\omega) = 0 = \mathrm{GST}$, so the proposition holds trivially.

Hence, assume that $\omega > 1$.
Since $\omega$ is the first (i.e., the smallest) post-$\mathrm{GST}$ window and $\omega > 1$ and correct validators enter window $1$ at time $0$, we have that there exists a pre-$\mathrm{GST}$ window.
Let $\omega^{\star}$ denote the greatest pre-$\mathrm{GST}$ window.
Since $\omega^{\star}$ is the greatest pre-$\mathrm{GST}$ window and $\omega$ is the first post-$\mathrm{GST}$ window, we have $\omega = \omega^{\star} + 1$.

We show that every correct validator $p_i$ that proposes to $\mathcal{ACS}[\omega]$ (line~\ref{line:acs-propose}) does so with a slot whose starting time is at most $\mathrm{GST} + W\tau$.
To propose to $\mathcal{ACS}[\omega]$, $p_i$ must have previously entered window $\omega - 1 = \omega^{\star}$ (line~\ref{line:ready}).
By \Cref{prop:window-open-time}, $p_i$ enters window $\omega^{\star}$ by time $\mathrm{GST} + \dtot + \ell$.
Since $\omega^{\star}$ is a pre-$\mathrm{GST}$ window, its first slot starts before $\mathrm{GST}$; as consecutive slots are $\tau$ apart, the starting times of the first $p$ slots of $\omega^{\star}$ all lie before $\mathrm{GST} + (p-1)\tau$, and $p_i$'s entry into $\omega^{\star}$ occurs by $\mathrm{GST} + \dtot + \ell \leq \mathrm{GST} + (p-1)\tau$ as well (line~\ref{line:assumption-four}).
Since a slot is opened at the later of its scheduling and its starting time (line~\ref{line:conductor-wait-for-open}), $p_i$ opens the first $p$ slots of window $\omega^{\star}$ by time $\mathrm{GST} + (p-1)\tau$ (line~\ref{line:trigger-open}).
By \Cref{prop:conductor-open-to-complete}, $p_i$ completes the first $(\omega^{\star} - 1)W + p$ slots by time $\mathrm{GST} + (p-1)\tau + \Phi_{oc}$, which implies that $p_i$ proposes to $\mathcal{ACS}[\omega]$ by time $\mathrm{GST} + (p-1)\tau + \Phi_{oc}$.
The proposed slot $s^{\star}$ is determined at lines~\ref{line:sstar-compute}--\ref{line:sstar-update}; we show its starting time is strictly less than $\mathrm{GST} + W\tau$.
We distinguish two cases based on whether the guard at line~\ref{line:sstar-guard} holds.
If it does not hold, $p_i$ keeps $s^{\star}$ as the first slot whose starting time meets or exceeds the current local time (line~\ref{line:sstar-compute}); since $p_i$ proposes by time $\mathrm{GST} + (p-1)\tau + \Phi_{oc}$, the starting time of $s^{\star}$ is strictly less than $\mathrm{GST} + (p-1)\tau + \Phi_{oc} + \tau = \mathrm{GST} + p\tau + \Phi_{oc}$, which is at most $\mathrm{GST} + W\tau$ by assumption~(2) (line~\ref{line:assumption-two}), since $(p-1)\tau + \Phi_{oc} \leq (W-1)\tau$ implies $p\tau + \Phi_{oc} \leq W\tau$.
If the guard holds, $p_i$ reassigns $s^{\star}$ to the slot immediately after the last slot of window $\omega^{\star}$ (line~\ref{line:sstar-update}); since $\omega^{\star}$ is a pre-$\mathrm{GST}$ window, its last slot has starting time strictly less than $\mathrm{GST} + (W-1)\tau$, and hence $s^{\star}$ has starting time strictly less than $\mathrm{GST} + W\tau$.
Hence, by the median rule at line~\ref{line:median-compute}, $\mathcal{T}_1(\omega) \leq \mathrm{GST} + W\tau$, which proves the proposition. \qed
\end{proof}

We are finally ready to prove \conductor's recovery.

\begin{lemma} [Recovery]
\label{lemma:conductor-recovery}
When run within \cadence (\Cref{algorithm:cadence}), \conductor (\Cref{algorithm:conductor}) satisfies $(2W\tau)$-recovery.
\end{lemma}
\begin{proof}
Fix a correct validator $p_i$ and any slot $s$ with $s.\fdeadline - \Delta \geq \mathrm{GST} + 2W\tau$; we show that $p_i$ opens $s$ and it does so at time $s.\fdeadline - \Delta$.
Let $\omega$ denote the first post-$\mathrm{GST}$ window.
By \Cref{prop:first-post-gst-window-time}, $\mathcal{T}_1(\omega) \leq \mathrm{GST} + W\tau$.
By the first point of \Cref{prop:smooth-windows}, $\mathsf{slot}(\omega + 1, 1).\fnumber = \mathsf{slot}(\omega, W).\fnumber + 1$, so, as consecutive slots start $\tau$ apart,
\[
\mathcal{T}_1(\omega + 1) = \mathcal{T}_W(\omega) + \tau = \mathcal{T}_1(\omega) + W\tau \leq \mathrm{GST} + 2W\tau \leq s.\fdeadline - \Delta.
\]
Moreover, applying the same point to every window beyond $\omega$, each set $\mathsf{slots}(\omega'')$ with $\omega'' \geq \omega + 1$ consists of $W$ consecutively-numbered slots (\Cref{prop:acs-fate-range}) beginning immediately after $\mathsf{slots}(\omega'' - 1)$; hence the sets $\mathsf{slots}(\omega + 1), \mathsf{slots}(\omega + 2), \ldots$ cover every slot with number at least $\mathsf{slot}(\omega + 1, 1).\fnumber$.
Since starting times strictly increase with slot number and $s.\fdeadline - \Delta \geq \mathcal{T}_1(\omega + 1)$, slot $s$ is among them: $s \in \mathsf{slots}(\omega')$ for some $\omega' \geq \omega + 1$.
We first show that $p_i$ opens $s$: by \Cref{prop:enters-every-window}, $p_i$ enters window $\omega'$, thereby scheduling the opening of every slot in $\mathsf{slots}(\omega')$ (line~\ref{line:acs-decide-open}) --- in particular of $s$, which it thus opens (line~\ref{line:trigger-open}).
Finally, by the second point of \Cref{prop:smooth-windows}, this opening occurs exactly at $s$'s starting time: $p_i$ opens $s$ at time $s.\fdeadline - \Delta$. \qed
\end{proof}

\paragraph{Epilogue.}
We conclude by consolidating the results of this subsection into a single statement.

\begin{theorem}[Correctness of \conductor]
\label{thm:conductor-correctness}
When run within \cadence, \conductor is a correct implementation of the orchestrator primitive, with $\mathcal{B} = 2W - p$ and $\mathcal{R} = 2W\tau$.
\end{theorem}
\begin{proof}
Integrity (\Cref{lemma:conductor-integrity}), monotonicity (\Cref{lemma:conductor-monotonicity}), and $(2W - p)$-boundedness (\Cref{lem:boundedness}) are satisfied unconditionally.
Both totality (\Cref{lemma:conductor-totality}) and $(2W\tau)$-recovery (\Cref{lemma:conductor-recovery}) are satisfied when \conductor is run within \cadence. \qed
\end{proof}

\Cref{thm:conductor-correctness} also closes the one loop left open by \chorus's correctness theorem (\Cref{thm:chorus-correctness}), which is conditioned on participation being $\Delta$-synchronized (\Cref{def:delta-synchronized-participation}): within \cadence, that condition is exactly \conductor's totality, so it now comes for free.

\begin{corollary}[Correctness of \chorus within \cadence]
\label{cor:chorus-correctness-within-cadence}
When run within \cadence, \chorus is a correct implementation of the slot consensus primitive.
\end{corollary}
\begin{proof}
By \Cref{thm:chorus-correctness}, it remains to discharge the assumption of $\Delta$-synchronized participation (\Cref{def:delta-synchronized-participation}) for every slot consensus instance $\mathcal{S}[s]$.
Within \cadence, a correct validator starts participating in $\mathcal{S}[s]$ exactly upon opening slot $s$ (line~\ref{line:participate} of \Cref{algorithm:cadence}); hence, the $\dtot$-totality of \conductor (\Cref{lemma:conductor-totality}), with $\dtot = \Delta$, states precisely that participation in $\mathcal{S}[s]$ is $\Delta$-synchronized. \qed
\end{proof}

\noindent
The discharge above completes a seemingly circular story, so let us spell out why it is not.
On the surface, each protocol needs the other's guarantee: \chorus promises synchronized finalizations for a slot \emph{only if} participation in it is already synchronized (\Cref{prop:chorus-totality}), while \conductor synchronizes participation --- the openings of later slots --- \emph{only because} the finalizations of earlier slots were synchronized.
The two needs, however, are never simultaneous: they alternate across windows, and the chain is anchored.
First, \chorus's totality is proven on its own, as a conditional statement, without invoking any guarantee of \conductor.
The induction of \Cref{prop:window-synchronization} then bootstraps from an unconditional anchor: every correct validator enters window $1$ at time $0$, so the openings of window $1$ are synchronized outright; \chorus's conditional totality --- its condition now met --- yields synchronized completions for window $1$; these, in turn, synchronize the entry into (and hence the openings of) window $2$; and so on, window by window.
At every step, window $\omega$ consumes only guarantees already established for windows $1, \ldots, \omega - 1$, together with the conditions of $\omega$ established earlier in the same step --- never the guarantee currently being proven.
(The conditional timing guarantees of the ACS instances are discharged by the same staggering, as seen in the proof of \Cref{prop:window-synchronization}.)
The corollary above is then the final, one-way application: with the participation condition established once and for all (\Cref{lemma:conductor-totality}), \Cref{thm:chorus-correctness} applies --- one finished result applied to another.